\documentclass[preprint,3p]{elsarticle}

\makeatletter
\def\@journal{arXiv.org}
\def\today{}
\makeatother

\newif\ifdetail

\newif\ifnarrow

\def\CHECKED#1{}
\newtheorem{theorem}{Theorem}[section]
\newtheorem{corollary}[theorem]{Corollary}
\newtheorem{lemma}[theorem]{Lemma}
\newtheorem{proposition}[theorem]{Proposition}
\newdefinition{definition}[theorem]{Definition}

\newtheorem{example}[theorem]{Example}
\newproof{proof}{Proof}
\def\[{\begin{eqnarray*}}
\def\]{\end{eqnarray*}}
\newenvironment{Enumerate}{\enumerate[{\kern8pt}1.]\itemsep=3pt}%
    {\endenumerate}
\makeatletter
\def\fnnum{\the\@mpfn}
\def\@beginproof#1#2{\trivlist\let\baselinestretch\@blstr
   \item[\hskip\labelsep{\bf #1.}]\rmfamily}
\def\@opargbeginproof#1#2#3{\trivlist\let\baselinestretch\@blstr
      \item[\hskip\labelsep{\bf #1\ (#3).}]\rmfamily}
\makeatother

\usepackage{makeidx}
\usepackage{amsfonts}
\usepackage{amssymb}
\usepackage{proof}
\usepackage{eqnar}
\usepackage{appendix}
\usepackage[dvipdfm]{hyperref}

\ifdetail
\makeindex
\fi

\makeatletter
\def\pushlabel{\let\olabel=\label
    \let\outer@currentlabel=\@currentlabel
    \def\label##1{\begingroup\let\t@currentlabel=\@currentlabel
	\def\@currentlabel{\outer@currentlabel.\t@currentlabel}\olabel{##1}%
	\endgroup\ignorespaces}}
\makeatother

\makeatletter
\def\appendixeqno{\def\theequation{A.\@arabic\c@equation}%
    \setcounter{equation}0}
\makeatother

\makeatletter
\def\ilabel{\@ifnextchar*\ilabel@\ilabel@@}
\def\ilabel@#1#2{\xdef\ilabel@tmp##1{%
	\noexpand\index{##1|hyperlink{\ilabel@last}}}%
	\ilabel@tmp{#2}\ignorespaces}
\def\ilabel@@#1#2{\hypertarget{#1}{}\label{#1}\gdef\ilabel@last{#1}%
    \index{#2|hyperlink{#1}}\ignorespaces}
\makeatother

\def\itemlabel#1{\olabel{item-#1}\label{#1}\hypertarget{#1}{}\ignorespaces}
\def\itemref#1{\hyperlink{#1}{\ref*{item-#1}}}

\makeatletter
\newtoks\x@footnote@toks
\def\xfootnote#1{\ifx\hline\LT@hline\footnote{#1}\relax\else
    \footnote{}\begingroup\toks0={#1}%
    \edef\x@set@footnote@toks{%
	\global\noexpand\x@footnote@toks={\the\x@footnote@toks
    \noexpand\footnotetext[\the\c@footnote]{\the\toks0}}}%
    \x@set@footnote@toks
    \endgroup\fi}
\def\emitxfootnote{\the\x@footnote@toks\global\x@footnote@toks={}}
\makeatother

\def\cal{\mathcal}

\mathcode`:="303A
\let\othinmskip=\,
\def\,{\penalty5000\othinmskip}
\def\not#1{\mathrel{\setbox0=\hbox{\(\mathord{#1}\)}
                          \setbox1=\hbox{\(\mathord{/}\)}
                \dimen0=\wd0 \advance\dimen0 by -\wd1 \divide\dimen0 by 2
                \hbox{\hspace{\dimen0}\copy1
                                \hspace{-\wd1}\hspace{-\dimen0}\copy0}}}
\def\mathch#1{{\mathchoice
		{\let\msty\displaystyle\let\ssty\scriptstyle #1}
		{\let\msty\textstyle\let\ssty\scriptstyle #1}
		{\let\msty\scriptstyle\let\ssty\scriptscriptstyle #1}
		{\let\msty\scriptscriptstyle\let\ssty\scriptscriptstyle #1}}}
\def\models{\mathrel{\mathch{%
    \setbox0=\hbox{$\msty \mathord{\vert}$}
    \dimen0=\wd0 \divide\dimen0 by 4
    \mathsurround=0pt\hbox to \dimen0{\copy0\hskip-5pt\hfil}\mathord{=}}}}

\let\Section=\section
\def\Subsection#1{\subsection{#1}\vskip5pt}
\def\Case#1{\vskip4pt\par\noindent{\it Case: #1}}
\def\Cases#1{\vskip4pt\par\noindent{\it Cases: #1}}
\def\Subcase#1{\vskip4pt\par\noindent{\it Subcase: #1}}

\def\Hbox#1#2{\setbox0=\hbox{#2}\wd0=#1\copy0}

\def\Qed{\par\vskip-1.75\baselineskip\hfill\qed}

\def\mtight{\medmuskip=0.8\medmuskip \thickmuskip=0.8\thickmuskip}
\def\Mtight{\medmuskip=0.6\medmuskip \thickmuskip=0.6\thickmuskip}

\def\mathbf#1{\mathord{\mathchoice{\mbox{\rm\bf #1}}{\mbox{\rm\bf #1}}
		 {\mbox{\scriptsize\rm\bf #1}}{\mbox{\tiny\rm\bf #1}}}}
\def\Underbrace#1#2{\underbrace{#2}_{{}^{\mbox{\scriptsize \(#1\)}}}}

\ifnarrow\def\NL#1{\\&&#1}\else\def\NL#1{}\fi

\def\na#1{\Eqncomment{#1}\noalign{\vskip3pt}}

\def\bnfdef{\mathrel{{:}{:}{=}}}
\def\bnfor{\mathrel{\mid}}
\def\r#1{\hbox{\rm \((#1)\)}}
\def\Y{\mathbf{Y}}
\def\g{\gamma}
\def\G{\Gamma}
\def\P{\Pi}

\def\IV{\mathbf{Var}}
\def\IE{\mathbf{Exp}}
\def\PT{\mathbf{PTExp}}
\def\CTE{\mathbf{CTExp}^{\scriptstyle\!{\eqtyp}}}
\def\CTEp{\mathbf{CTExp}^{\scriptstyle\!{\peqtyp}}}

\def\TFTE{\mathbf{TF}}
\def\PFTE{\mathbf{PF}}
\def\NFTE{\mathbf{NF}}

\def\TV{\mathbf{T\hskip-.5pt Var}}
\def\TE{\mathbf{TExp}}

\def\ienv{{\mathord{\rho}}}
\def\Ienv{\mathord{\tilde{\ienv}}}
\def\tenv{{\mathord{\eta}}}
\def\ttenv{{\mathord{\xi}}}
\def\Pow#1{{\cal P}(#1)}
\def\to{\penalty2000\mathbin{\rightarrow}\penalty1000}
\def\conj{\mathbin{\wedge}}

\def\impl{\mathbin{\supset}}

\def\Conj{\mathbin{\times}}

\def\Impl{\penalty2000\mathbin{\rightarrow}\penalty1000}
\def\fx{\mu}
\def\fix#1{\fx #1\mskip-2mu{.}\penalty2000\mskip1mu}

\def\II{{\mathord{\,\cal I}}}
\def\IIt{\II^{\scriptscriptstyle{+}}\!}

\def\I#1{\II(#1)}
\def\It#1{\IIt(#1)}
\def\Ip#1{\II'(#1)}

\def\Il#1{\II_{\scriptscriptstyle \!L}(#1)}
\def\Ilp#1{\II'_{\scriptscriptstyle \!L}(#1)}
\def\V{{\cal V}}
\def\NF{\cal N}
\def\Kh{{\cal K}_h}
\def\Kn{{\cal K}_n}
\def\F{{\cal F}}
\def\N{\mathbb{N}}
\def\W{{\cal W}}
\def\VI#1#2{\mathord{[\![{#1}]\!]}^{\V}_{#2}}
\def\zfset#1#2{\{\;#1\;\penalty1000|\penalty1000\;#2\;\}}
\def\Zfset#1#2{\left\{\;#1\;\penalty1000\left|\penalty1000\;#2\;\right.\right\}}
\def\Choice#1{\left\{\setbox0=\hbox{\raise 2pt\hbox{\(\displaystyle
    \renewcommand\arraystretch{1.4}
    \arraycolsep=8pt
    \begin{array}{@{\mskip4mu}ll}#1\end{array}\)}}
    \dimen0=\ht0\advance\dimen0 by -6pt\ht0=\dimen0
    \dimen0=\dp0\advance\dimen0 by -6pt\dp0=\dimen0
    \copy0\right.} 

\def\zero{\mathbf{0}}
\def\suc{\mathbf{s}}
\def\nat{\mathbf{nat}}

\def\dothinrel#1{{\setbox0=\hbox{$\mathsurround=0pt\msty #1$}
	            \setbox1=\hbox{$\mathsurround=0pt\msty \rightarrow$}
		    \ht0=\ht1 \dp0=\dp1 \hbox{\copy0}}}

\def\thinrel#1{\mathrel{\mathch{\dothinrel{#1}}}}

\def\dostrel#1#2{{
	\setbox0=\hbox{$\mathsurround=0pt\ssty O #1$}
	\setbox1=\hbox{$\mathsurround=0pt\scriptstyle O #1$}
	\setbox2=\hbox{$\mathsurround=0pt\msty #2$}
	\setbox3=\hbox{$\mathsurround=0pt\ssty #1$}
		     \dimen0=7 \ht0 \advance\dimen0 by \ht1
		     \multiply\dimen0 by 110 \divide\dimen0 by 800
		     \advance\dimen0 by -\dp0
		     \advance\dimen0 by -\ht2
		     \setbox1=\hbox{\lower\dimen0
			 \hbox{$\mathsurround=0pt\ssty\strut#1$}}
		     \dimen1=\ht3 \advance\dimen1 by -\dimen0
		     \dp1=0pt\ht1=\dimen1
		     \stackrel{\hbox{\copy1}}{#2}}}
\def\strel#1#2{\mathrel{\mathch{\dostrel{#1}{#2}}}}

\def\doannorel#1#2{{
	\setbox0=\hbox{$\mathsurround=0pt\ssty O #2$}
	\setbox1=\hbox{$\mathsurround=0pt\scriptstyle O #2$}
	\setbox2=\hbox{$\mathsurround=0pt\ssty #2$}
	\setbox3=\hbox{$\mathsurround=0pt\msty #1$}
		     \dimen0=3 \ht0 \advance\dimen0 by \ht1
		     \multiply\dimen0 by 200 \divide\dimen0 by 400
		     \advance\dimen0 by -\ht2
		     \advance\dimen0 by -\ht3
		     \setbox0=\hbox{\raise\dimen0
				 \hbox{$\mathsurround=0pt\ssty#2$}}
		     \ht0=\ht2
		     \dp0=0pt
		     \setbox1=\hbox{$\mathsurround=0pt\msty
				     \mathop{#1}\limits_{\hbox{\copy0}}$}
		     \dimen0=0pt
		     \ifdim \ht2 > 2pt
		     \advance\dimen0 by 5pt
		     \fi
		     \dp1=\dimen0
		     \hbox{\copy1}}}

\def\annorel#1#2{\mathrel{\mathch{\doannorel{#1}{#2}}}}
\def\tuple#1{{\langle\,}#1{\,\rangle}}

\def\Ifr#1#2#3{\infer[\mbox{\rm\small\(#1\)}]{#3}{#2}}
\def\ifr#1#2#3{\infer[\mbox{\rm\small(\(#1\))}]{#3}{#2}}
\def\Vdots#1{\lower 4pt
    \hbox{\vbox to #1pt{\leaders\vbox to 4pt{\vfil\hbox{.}}\vfill}}}
\def\Derive#1#2{\setbox0=\hbox{\Vdots{#1}}
              \setbox1=\hbox{\(\displaystyle\strut\mskip8mu minus 4mu#2\)}
              \dimen0=\ht0 \advance\dimen0 by \dp0
              \advance\dimen0 by -\dp1 \advance\dimen0 by -\ht1
              \divide\dimen0 by 2
              \copy0 \raise\dimen0\copy1 \kern-\wd1}
\def\derive#1{\Derive{15}{#1}}

\def\stack#1#2{\setbox0=\hbox{\(\displaystyle\hss #2 \hss\)}
               \setbox1=\hbox{\(\displaystyle \mbox{\(#1\)}
                                                \atop \hbox{\copy0}\)}
               \dimen0=\dp1 \advance\dimen0 by -\dp0
               \raise\dimen0\hbox{\copy1}}

\def\subt#1#2{#1\mathrel{\vdash}\penalty1000 #2}
\def\typ#1#2{#1\mathrel{\vdash}\penalty1000 #2}
\def\Typ#1#2#3{#2\mathrel{\vdash_{#1}\,}\penalty1000 #3}

\def\LAmtyp#1#2{\Typ{\mathbf{LA}\fx\mskip-5mu}{#1}{#2}}
\def\Rlz#1{\mathrel{\mathbf{r}_{#1}}}
\def\O{\mathord{\bullet}}
\def\t{\top}
\def\tvariant{\(\mskip-2mu\top\mskip-2mu\)-variant}

\def\tail#1{{#1}^{t}}
\def\C{{\cal C}}
\def\B{\mathop{\mbox{!}}}
\def\Bar#1{\setbox0=\hbox{\(#1\)} \ifdim\wd0 < 8pt
    \overline{#1}%
    \else
    \overline{#1\hskip-.5pt}\hskip.5pt%
    \fi}
\newbox\subtypbox
\def\subtyp{\strel{\mbox{$\mathsurround=0pt\prec$}}%
    {\setbox0=\hbox{\lower 2pt\hbox{$\mathsurround=0pt=$}}\ht0=2.8pt\copy0}}

\def\psubtyp{\preceq}
\def\rsubtyp{\psubtyp}

\def\eqtyp{\cong}
\def\peqtyp{\simeq}
\def\npeqtyp{\mathrel{\not\simeq}}
\def\geqtyp{\sim}
\def\ngeqtyp{\mathrel{\not\sim}}

\def\Lequiv{\mathrel{\leftrightarrow}}
\def\LAmequiv{\mathrel{\leftrightarrow_{\mathbf{LA}\fx}\!}}

\def\component{\le}

\def\lam#1{\lambda{#1}.\:}
\def\app#1{#1\mskip1mu
	{\setbox0=\hbox{\(\mathsurround=0pt#1\)}\ifdim\wd0>18pt\,\fi}}
\def\pair#1#2{{\langle\,}{#1},\:{#2}{\,\rangle}}
\def\pleft{\mathbf{p}_{\mathbf{1}}\,}
\def\pright{\mathbf{p}_{\mathbf{2}}\,}
\def\proj#1{\mathbf{p}_{\mathbf{#1}}\,}
\def\ileft{\mathbf{i}_{\mathbf{1}}\,}
\def\iright{\mathbf{i}_{\mathbf{2}}\,}

\def\true{\top}

\def\dminus{\mathbin{\strel{.}{\thinrel{-}}}}
\def\ct{\annorel{\rightarrow}{{}_\beta}}
\def\ctc{\strel{\ast}{\ct}}
\def\cteq{\annorel{\leftrightarrow}{{}_\beta}}
\def\beq{\annorel{=}{{}_\beta}}

\def\lm{$\lambda{\fx}$}
\def\lA{$\lambda${\bf A}}

\def\miK4{{\bf miK4}}
\def\GL{{\bf GL}}
\def\iGL{{\bf iGL}}
\def\miGL{{\bf miGL}}

\def\miGLC{{\bf miGLC}}

\def\LA{{\bf LA}}
\def\LAm{{\bf LA}$\fx$}

\newbox\Natbox
\setbox\Natbox=\hbox{$\mathbf{Nat}$}
\def\Nat{\copy\Natbox}
\def\Dom{\mathit{Dom}}
\def\FV#1{\mathit{FV}(#1)}
\def\FTV#1{\mathit{FTV}(#1)}
\def\ETV#1{\mathit{ETV}(#1)}
\def\PETV#1{\mathit{ETV}^{+}(#1)}
\def\NETV#1{\mathit{ETV}^{-}(#1)}
\def\PNETV#1{\mathit{ETV}^{\pm}(#1)}
\def\NPETV#1{\mathit{ETV}^{\mp}(#1)}
\def\Shift#1{\mathord{{}^{\ulcorner}\mskip-2mu#1\mskip-1mu{}^{\urcorner}}}
\newbox\Shiftbox
\setbox\Shiftbox=\hbox{$\Shift{A}$}

\def\Canon#1{#1^{c\scriptscriptstyle\eqtyp}}
\def\Canonp#1{#1^{c\scriptscriptstyle\peqtyp}}
\def\Canong#1{#1^c}

\def\Dp{\mathit{dp}}
\def\Odp{{\mathord{\Dp_{\O}}}}
\def\POdp{{\mathord{\Dp^{+}_{\O}}}}
\def\NOdp{{\mathord{\Dp^{-}_{\O}}}}
\def\PNOdp{{\mathord{\Dp^{\pm}_{\O}}}}
\def\NPOdp{{\mathord{\Dp^{\mp}_{\O}}}}
\def\Idp{{\mathord{\Dp_{\Impl}}}}
\def\PIdp{{\mathord{\Dp^{+}_{\Impl}}}}
\def\NIdp{{\mathord{\Dp^{-}_{\Impl}}}}
\def\PNIdp{{\mathord{\Dp^{\pm}_{\Impl}}}}
\def\NPIdp{{\mathord{\Dp^{\mp}_{\Impl}}}}
\def\PDp{{\mathord{\Dp^{+}}}}
\def\NDp{{\mathord{\Dp^{-}}}}
\def\PNDp{{\mathord{\Dp^{\pm}}}}
\def\NPDp{{\mathord{\Dp^{\mp}}}}
\def\acc{\mathrel{\rhd}}
\def\tacc{\strel{\setbox0=\hbox{\lower 2.5pt\hbox{$\scriptstyle\ast$}}%
    \ht0=0pt\dp0=0pt\copy0}{\rhd}}
\def\pacc{\strel{\setbox0=\hbox{\lower 2pt\hbox{$\scriptscriptstyle+$}}%
    \ht0=0pt\dp0=0pt\copy0}{\rhd}}
\def\opacc{\mathrel{\lhd}}

\def\Comp{\mathit{Comp}}
\def\ptilde#1{\setbox0=\hbox{$#1'$}\ht0=6pt\tilde{\copy0}}

\begin{document}
\begin{frontmatter}
\title{A modal typing system for self-referential
    programs and specifications\footnote{%
	Research supported in part by
	2014 Special Researcher Program of Ryukoku University.
	Some preliminary results to this paper
	have been presented in \cite{nakano-lics00,nakano-tacs01}.}
\ifdetail
\\ {\large(An extended version with detailed proofs)}
\fi
	}
\author{Hiroshi Nakano}
\address{Department of Applied Mathematics and Informatics, \\
    Ryukoku University, Seta, Otsu, Shiga 520-2194, Japan \\
    Email: {\tt nakano@math.ryukoku.ac.jp}}
\begin{abstract}
This paper proposes a modal typing system that enables us
to handle self-referential formulae, including
ones with negative self-references,
which on one hand, would introduce a logical contradiction,
namely Russell's paradox, in the conventional setting,
while on the other hand, are necessary to capture a certain class
of programs such as fixed-point combinators
and objects with so-called binary methods in object-oriented programming.
The proposed system provides
a basis for axiomatic semantics of such a wider range
of programs and a new framework for natural construction
of recursive programs in the proofs-as-programs paradigm.
\end{abstract}
\end{frontmatter}

\ifdetail
\clearpage
\tableofcontents

\clearpage
\fi

\Section{Introduction}\label{intro-sec}

Although recursion, or self-reference, plays an indispensable role in both
programs and their specifications, it also introduces serious difficulties
into their formal treatment.
It is still far from obvious how to capture it
in an axiomatic semantics such as the formulae-as-types notion of
construction \cite{howard}.
Only a rather restricted class of recursive programs (and specifications)
has been captured in this direction
as (co)inductive proofs over the (co)inductive data structures
(see e.g.,
\cite{constable,hayashi:nakano,paulin-mohring,kobayashi:tatsuta,tatsuta}),
and, for example, negative self-references, which would be necessary to handle
a certain range of programs such as fixed-point combinators
and objects with so-called binary methods in object-oriented
programming, still remain out of the scope.
In this paper, a modal logic that provides a basis for
capturing such a wider range of programs
in the proofs-as-programs paradigm is proposed.
The logic is presented as a modal typing system with recursive types
for the purpose of presentation.
Its soundness with respect to a realizability interpretation,
which implies the convergence of well-typed programs, is shown.

\subsection*{Difficulty in binary-methods}
Consider, for example, the specification
\(\Nat(n)\) of objects that represent a natural number \(n\) with
a {\em method} which returns an object of \(\Nat(n{+}m)\)
when one of \(\Nat(m)\) is given.
It could be represented by a self-referential specification such as
\begin{eqnarray*}
\Nat(n) &\;\equiv\;& ((n = 0) + (n > 0 \conj \Nat(n{-}1)) \NL{\mskip80mu}
	{}\times (\forall m.\: \Nat(m) \Impl \Nat(n{+}m))),
\end{eqnarray*}
where we assume that \(n\) and \(m\) range over the set of natural numbers;
\({+}\), \({\times}\) and \(\Impl\) are type constructors for direct sums,
direct products and function spaces, respectively;
\({\conj}\) and \(\forall\) have standard logical (annotative) meanings.
Although it is not obvious whether this self-referential specification
is meaningful in a certain mathematical sense,
it could be a first approximation of the specification we want
since this can be regarded as a refined version of recursive types
which have been widely adopted as a basis for object-oriented type systems
\cite{abadi:cardelli,bruce:cardelli:pierce}.
At any rate, if we define an object \(\zero\) as
\[
    \zero \equiv \pair{\ileft \ast}{\lam{x}x},
\]
then it would satisfy \(\Nat(0)\),
where \(\ileft{\!}\) is the injection into the first
summand of direct sums and \(\ast\) is a constant.
We assume that any program satisfies annotative formulae such as \(n = 0\)
whenever they are true.
We can easily define a function that satisfies
\(\forall n.\: \forall m.\: \Nat(n) \Impl \Nat(m) \Impl \Nat(n{+}m)\)
as
\[
    \mathbf{add}\:x\:y \equiv \pright x\:y,
\]
or
\[
    \mathbf{add}'\:x\:y \equiv \pright y\:x,
\]
where \(\pright{\!}\) extracts the second components, i.e.,
the method of addition in this particular case, from pairs.
We could also define the successor function as a recursive
program as
\[
    \suc\: x \equiv
    \pair{\iright{x}}{\lam{y}{\mathbf{add}\:x\:(\suc\:y)}}
\]
or
\[
    \suc'\: x \equiv
    \pair{\iright{x}}{\lam{y}{\mathbf{add}'\:x\:(\suc'\:y)}}.
\]
In spite of the apparent symmetry between
\(\mathbf{add}\) and \(\mathbf{add}'\), which are both supposed to satisfy
the same specification,
the computational behaviors of \(\suc\) and \(\suc'\) are completely different.
We can observe that \(\suc\) works as expected, but \(\suc'\) does not.

For example, \(\pright (\suc\:\zero)\:\zero\) would be evaluated as
\begin{eqnarray*}
\pright (\suc\:\zero)\:\zero
&\rightarrow& (\lam{y}\mathbf{add}\:\zero\:(\suc\:y))\:\zero \\
&\rightarrow& \mathbf{add}\:\zero\:(\suc\:\zero) \\
&\rightarrow& \pright \zero\:(\suc\:\zero) \\
&\rightarrow& (\lam{x}x)\:(\suc\:\zero) \\
&\rightarrow& \suc\:\zero,
\end{eqnarray*}
whereas
\begin{eqnarray*}
\pright (\suc'\:\zero)\:\zero
&\rightarrow& (\lam{y}\mathbf{add}'\:\zero\:(\suc'\:y))\:\zero \\
&\rightarrow& \mathbf{add}'\:\zero\:(\suc'\:\zero) \\
&\rightarrow& \pright (\suc'\:\zero)\:\zero \\
&\rightarrow& \ldots,
\end{eqnarray*}
and more generally, for any objects \(x\) and \(y\)
of \(\Nat(n)\) (for some \(n\)),
\[
\pright (\suc'\:x)\:y
\rightarrow \ldots \rightarrow \pright (\suc'\:y)\:x
\rightarrow \ldots \rightarrow \pright (\suc'\:x)\:y
\rightarrow \ldots.
\]

This sort of divergence would also be quite common
in (careless) recursive definitions of programs even if we did not
have to handle object-oriented specifications like \(\Nat(n)\).
The peculiarity here is the fact
that the divergence is caused by a program, \(\mathbf{add}'\),
which is supposed to satisfy the same specification as \(\mathbf{add}\).
This example shows such a loss of the compositionality of programs
with respect to the specifications that imply
their termination, or convergence.
It also suggests that, to overcome this difficulty,
\(\mathbf{add}\) and \(\mathbf{add}'\)
should have different specifications, and accordingly the definition
of \(\Nat(n)\) should be revised in some way in order to force it.

\subsection*{{\lm} and its logical inconsistency}
The typing system {\lm} (cf. \cite{barendregt})
is a simply-typed lambda calculus
with recursive types, where any form of self-references, including
negative ones, is permitted.
A non-trivial model for such unrestricted recursive types was
developed by MacQueen, Plotkin and Sethi \cite{mps}, and has been widely
adopted as a theoretical basis for object-oriented type systems
\cite{abadi:cardelli,bruce:cardelli:pierce}.

On the other hand, it is well known that logical formulae with
such unrestricted self-references would introduce a contradiction
(variant of Russell's paradox).
Therefore, logical systems must have
certain restrictions on the forms of self-references (if ever allowed)
in order to keep themselves sound;
for example, \(\mu\)-calculus\cite{pratt,kozen} does not allow
negative self-references (see also \cite{gurevich:shelah}).

Through the formulae-as-types notion, this
paradox corresponds to the fact that every type of {\lm}
is inhabited by a
diverging program which does not produce any information;
for example, the \(\lambda\)-term
\((\lam{x}\app{x}{x})(\lam{x}\app{x}{x})\)
can be typed with every type in {\lm}.
Therefore, even with the model mentioned above, types can be regarded
only as partial specifications of programs, and that is considered the reason
why we lost the compositionality of programs in the \(\Nat(n)\) case,
where we regarded convergence of programs as a part
of their specifications.
This shows a contrast with the success of {\lm} as a basis for
type systems of object-oriented program languages, where
the primary purpose of types, i.e., coarse specifications, is
to prevent run-time type errors, and termination of programs
is out of the scope.

The logical inconsistency of {\lm} also implies that
no matter how much types, or specifications, are refined,
convergence of programs cannot be expressed by them,
and must be handled by endowing the typing system
with some facilities for discussing computational properties of programs.
For example,
Constable et al. adopted this approach in their pioneering works
to incorporate recursive definitions
and partial objects into constructive type theory
\cite{constable:mendler,constable:smith}.
However, in this paper, we will pursue
another approach such that types themselves
can express convergence of programs.

\subsection*{Towards the approximation modality}

Suppose that we have a recursive program \(f\) defined by
\[
     f \equiv F(f),
\]
and want to show that \(f\) satisfies a certain specification \(S\).
Since the denotational meaning of \(f\) is given as the least fixed point
of \(F\), i.e., \(f = \sup_{n<\omega}F^n(\bot)\),
a possible way to do that would be
to apply Scott's fixed-point induction\cite{scott} by showing the following:
\begin{enumerate}
\item[--] \(\bot\) satisfies \(S\),
\item[--] \(F(x)\) satisfies \(S\) provided that \(x\) satisfies \(S\), and
\item[--] \(S\) is chain closed.
\end{enumerate}
However, this does not suffice for our purpose if
\(S\) includes some requirement about the convergence of \(f\),
because obviously
\(\bot\), or even \(F^n(\bot)\), could not satisfy the requirement.
So we need more refined approach.
The failure of the naive fixed-point induction above suggests that
the specification to be satisfied by each \(F^n(\bot)\)
inherently depends on \(n\), and the requirement concerning its
convergence must become stronger when \(n\) increases.
This leads us to a layered version of the fixed-point induction scheme
as follows: in order to show that \(f\) satisfies \(S\),
it suffices to find an infinite sequence \(S_0\), \(S_1\), \(S_2\),
\(\ldots\) of properties, or (virtual) specifications, such that
\begin{enumerate}[{\kern8pt}(\romannumeral 1)]
\item \(S = \bigcap_{n<\omega} S_n\),
\item \(S_{n{+}1} \subseteq S_n\),
\item \(\bot\) satisfies \(S_0\),
\item \(F(x)\) satisfies \(S_{n+1}\)
	provided that \(x\) satisfies \(S_n\), and
\item \(S_n\) is chain closed.
\end{enumerate}
For, since \(F^n(\bot) \in S_n\) for every \(n\) by ({\romannumeral 3})
and ({\romannumeral 4}),
we get \(F^k(\bot) \in S_n\) for every \(k \ge n\) by ({\romannumeral 2}).
This and ({\romannumeral 5}) imply \(f \in S_n\) for every \(n\),
and consequently \(f \in S\) by ({\romannumeral 1}).

In this scheme, the sequence \(S_0\), \(S_1\), \(S_2\), \(\ldots\)
can be regarded as a successive approximation of \(S\),
and \(F\) a (higher-order) program which
constructs a program that satisfies \(S_{n{+}1}\) from one that
satisfies \(S_n\).
It should be also noted that \(F\) works independently of \(n\).
This uniformity of \(F\) over \(n\) leads us to consider a formalization
of this scheme in a modal logic, where
the set of possible worlds (in the sense of Kripke semantics) consists
of all non-negative integers, and \(S_n\) in the induction scheme above
corresponds to the interpretation of \(S\) in the world \(n\).
We now write \(x \Rlz{k} S\) to denote the fact
that \(x\) satisfies the interpretation of \(S\) in the world \(k\),
and define a modality, say \(\O\), as
\begin{center}
\(x \Rlz{k} \O S\quad\) iff \(\quad k = 0\;\) or \(\;x \Rlz{k{-}1} S\).
\end{center}
Condition~({\romannumeral 2}) of the induction scheme says
that \(x \Rlz{k} S\) implies \(x \Rlz{l} S\) for every \(l \le k\);
in other words, the interpretation of specifications
should be {\em hereditary} with respect to the accessibility relation \(>\).
In such a modal framework,
the specification to be satisfied by \(F\) can be represented
by \(\O S \Impl S\) provided that the \(\Impl\)-connective is
interpreted in the standard way in each world, and
our induction scheme can be rewritten as
\begin{eqnarray*}\mbox{%
    if \(\bot \Rlz{0} S\) and \(F \Rlz{k} \O S \Impl S\) for every \(k > 0\),
    then \(f \Rlz{k} S\) for every \(k\).}
\end{eqnarray*}
Furthermore, if we assume that \(S_0\) is a trivial specification
which is satisfiable by any program, then,
shifting the possible worlds downwards by one, we can simplify this to
\begingroup
    \def\theequation{$\ast$}
\begin{eqnarray}
    \label{intro-induction}
    \mbox{if \(F \Rlz{k} \O S \Impl S\) for every \(k\),
	then \(f \Rlz{k} S\) for every \(k\)}.
\end{eqnarray}
    \addtocounter{equation}{-1}
\endgroup
Although this assumption about \(S_0\) somewhat restricts our choice of the
sequence \(S_0\), \(S_1\), \(S_2\), \(\ldots\),
it could be thought rather reasonable
because \(S_0\) must be an almost trivial specification that is
even satisfiable by \(\bot\).
Note that \(S_{n{+}1}\) occurring in the induction now
corresponds to the interpretation of \(S\) in the world \(n\),
and \(S_0\) corresponds to the interpretation of \(\O S\) in the world 0.

We so far considered the set of non-negative integers and the greater-than
relation as the set of possible worlds and the accessibility relation,
respectively.
This setting can be easily generalized to any frame
with a (conversely) well-founded accessibility relation.
The standard interpretation of \(\Impl\) implies
some fundamental properties concerning the \(\O\)-modality
over such frames, which introduce a subsumption, or subtyping, relation
over specifications into our modal framework.
First, the hereditary interpretation of specifications
implies the following property.
\[
    \quad x \Rlz{k} S \quad\mbox{implies}\quad x \Rlz{k} \O S
\]
Second, if we adopt the simple semantics of types, i.e.,
\(x \Rlz{k} S \Impl \true\) for every \(x\), \(S\) and \(k\),
where \(\true\) is the universe of (meanings of) programs,
which is satisfiable by any program, then
\[
x \Rlz{k} \O(S \Impl T) \quad\mbox{implies}\quad x \Rlz{k} \O S \Impl \O T.
\]
Note that this is not always the case because
we could consider non-extensional interpretations,
e.g., F-semantics\cite{hindley}, in which
\(\lam{x}{\bot} \Rlz{k} S \Impl \true\) holds, but
\(\bot \Rlz{k} S \Impl \true\) does not.

Furthermore, if the accessibility relation is linear, i.e., not branching,
such as the case of the greater-than relation between non-negative integers,
there holds the following converse property.
\[
    x \Rlz{k} \O S \Impl \O T \quad\mbox{implies}\quad
    x \Rlz{k} \O(S \Impl T)
\]
In such cases, the meaning of
\(\O S \Impl \O T\) and \(\O(S \Impl T)\) are equivalent.

\subsection*{Specification-level self-references}
This modal framework introduced for program-level self-references
also provides a basis for specification-level self-references.
Suppose that we have a self-referential specification such as
\[
    S = \phi(S).
\]
As we saw in the \(\Nat(n)\) case, negative reference to \(S\)
in \(\phi\) can introduce a contradiction in the conventional setting,
and this is still true in our modal framework.
However, in the world \(n\),
we can now refer to the interpretation of \(S\) in any world \(k < n\)
without worrying about the contradiction.
That is, as long as \(S\) occurs only in scopes of the modal operator
\(\O\) in \(\phi\), the interpretation of \(S\) is well-defined and given
as a fixed point of \(\phi\), which is actually shown to be unique.
For example, if \(S\) is defined as \(S = \O S \Impl T\), then
\(S\) could be interpreted in each world as follows.
\begin{eqnarray*}
S_0 &=& \true \Impl T_0 \\
S_1 &=& S_0 \cap ((\true \Impl T_0) \Impl T_1) \\
S_2 &=& S_1 \cap
	((S_0 \cap ((\true \Impl T_0) \Impl T_1)) \Impl T_2) \\[-4pt]
    &\vdots& \\[-4pt]
S_{n{+}1} &=& S_n \cap (S_n \Impl T_{n{+}1}) \\[-4pt]
    &\vdots&
\end{eqnarray*}
where \(S_k\) and \(T_k\) are the interpretations of \(S\) and \(T\)
in the world \(k\), respectively, and
the notations such as \(\true\) and \(\Impl\) are abused to denote
their expected interpretations as well.
This kind of self-references provides us
a method to define the sequence \(S_0\), \(S_1\), \(S_2\), \(\ldots\,\)
for the refined induction scheme
when we derive properties of recursive programs,
and the induction scheme
would be useless if we did not have such a method.

In the following sections, we will see that this form of specification-level
self-references is quite powerful, and captures a wide range of
specifications including
those which are not representable in the conventional setting
such as ones for \(\mathbf{add}\) and \(\mathbf{add}'\) in the \(\Nat(n)\) case.
Furthermore, the modal version (\ref{intro-induction})
of the induction scheme turns out
to be derivable from other properties of the
\(\O\)-modality and such self-referential specifications, where
the derivation corresponds to fixed-point combinators, such as Curry's \(\Y\).
This also gives us a way to construct recursive programs based on
the proofs-as-programs notion.

\subsection*{The plan of the paper}
The plan of the present paper is as follows.
In the next section,
the syntax of our type expressions, which
can be considered as coarse specifications of programs, is given.
The modal semantics of such type expressions is given
in the manner of realizability in Section~\ref{semantics-sec},
in which we interpret
types over Kripke frames with a (conversely) well-founded
accessibility relation.
We discuss formal derivability of equality and subsumption of types
in Sections~\ref{eqtyp-sec} and \ref{subtyp-sec}, respectively,
and show their soundness with respect to the semantics.
In Section~\ref{lA-sec}, we introduce a typed $\lambda$-calculus {\lA}
equipped with the modality and recursive types,
and show its soundness and subject reduction property
in Section~\ref{basic-prop-sec}.
Then, in Section~\ref{conv-sec},
we show convergence of well-typed \(\lambda\)-terms
according to the forms of their type.
In Section~\ref{program-sec},
we present some examples of program derivation
by extending the pure typing system, where
the \(\Nat(n)\)-example presented in this section will be revisited.
In Section~\ref{logic-sec}, we consider {\lA} as a modal logic
by ignoring left hand sides of ``\(:\)'' from typing, and
show an interesting relationship to an intuitionistic version of the logic of
provability.

\Section{Type expressions}\label{TE-sec}

We start with the syntax of type expressions.
As a preparation for it,
we first give one of {\em pseudo type expressions} \(\,\PT\),
which are obtained by adding a unary type constructor \(\O\)
to those of {\lm}, namely the simply typed \(\lambda\)-calculus
extended with recursive types (cf. \cite{barendregt,cardone:coppo}).
Let \(\TV\) be a countably infinite set of type variable symbols
\(X\), \(Y\), \(Z\), \(\ldots\)\,.
\begin{definition}[Pseudo type expressions]
    \ilabel{pt-def}{pseudo type expressions}
    \ilabel*{type expressions!pseudo}
    \ilabel*{syntax!pseudo type expressions}
    \ilabel*{syntax!PTExp@$\protect\PT$}
The syntax of \(\PT\) is defined by
\begin{Eqnarray*}
    \PT & \bnfdef & \TV & (type variables) \\
    & \bnfor & \PT \Impl \PT & (function types) \\
    & \bnfor & \O \PT & (approximative types) \\
    & \bnfor & \fix{\,\TV}\,\PT & (recursive types).
\end{Eqnarray*}
\end{definition}
We assume that \(\Impl\) associates to the right as usual, and
each (pseudo) type constructor associates according to the following
priority.
\[
    \thickmuskip=8mu
    \mbox{(Low)}\qquad \fix{X}{\;} < {\;}\Impl{\;} < {\;\O} \qquad\mbox{(High)}
\]
For example, \(\O\fix{X}\O X \Impl Y \Impl Z\) is the same
as \(\O(\fix{X}{((\O X) \Impl (Y \Impl Z))})\).
    \ilabel{t-def}{top@$\protect\t$}
    \ilabel*{0 top@$\protect\t$}
We use \(\t\) as an abbreviation for \(\fix{X}\O X\) and
use \(\O^n A\) to denote a (pseudo) type expression
\(\Underbrace{n~\mbox{times}}{\O\ldots\O}A\), where \(n \ge 0\).
In the sequel,
we use \(A\), \(B\), \(C\), \(D\), \(\ldots\)
to denote (pseudo) type expressions,
and denote the set of type variables occurring freely in \(A\)
by \(\FTV{A}\) regarding a type variable \(X\) as bound in \(\fix{X}A\).
We regard \(\alpha\)-convertible (pseudo) type expressions as identical.
    \ilabel{subst-def}{substitution of!type expressions}
    \ilabel*{type expressions!substitution of}
    \ilabel*{0 substitution A@$[A/X]$}
We write \(A[B_1/X_1,\ldots,B_n/X_n]\)
to denote the (pseudo) type expression
obtained from \(A\) by substituting \(B_1,\ldots,B_n\) for each
free occurrence of \(X_1,\ldots,X_n\), respectively,
with necessary \(\alpha\)-conversion to avoid accidental
capture of free type variables.
We assume that \([B_1/X_1,\ldots,B_n/X_n]\) associates with
the preceding type expression
with a higher priority than the modal operator~\(\O\).

As mentioned in Section~\ref{intro-sec},
the modal type operator \({\O}\) causes a shift of possible worlds, and
makes a reference to the one-step coarser world,
in which type expressions are interpreted as coarser specifications, i.e.,
larger sets of values or programs, than the ones in the present world.
The (pseudo) type expression \(\t\) corresponds to
the universe into which \(\lambda\)-terms are interpreted.
Hence, every \(\lambda\)-term should have this type in {\lA}.
We adopt a so called {\em simple semantics} of types
(cf. \cite{hindley,hindley-F}), by which the meaning of \(A \Impl \t\) is
identical to the one of \(\t\).
Thus, some syntactically different type expressions may have the same meaning
as the one of \(\t\).
We call such (pseudo) type expressions {\em {\tvariant}s}, which can be
syntactically distinguished from others as follows.

\begin{definition}[{\tvariant}s]
    \ilabel{tvariant-def}{top@$\top$-variants}
    \ilabel{tail-def}{t@$\protect\tail{A}$}
A (pseudo) type expression \(A\) is {\em a {\tvariant}} if and only if
\(\tail{A}
    = \O^{m_0}\fix{X_1}\O^{m_1}\fix{X_2}\O^{m_2}\ldots\fix{X_n}\O^{m_n}X_i\)
for some \(n\), \(m_0\), \(m_1\), \(m_2\), \(\ldots\), \(m_n\),
\(X_1\), \(X_2\), \(\ldots\), \(X_n\) and \(i\) such that
\(1 \le i \le n\),
\(X_i \notin \{\,\,X_{i{+}1},\,X_{i{+}2},\,\ldots,X_n\,\}\)
and \(m_i+m_{i{+}1}+m_{i{+}2}+\ldots+m_n \ge 1\),
where \(\tail{A}\) is defined as follows.
\[
    \tail{X} = X,\qquad
    \tail{(A \Impl B)} = \tail{B},\qquad
    \tail{(\O A)} = \O\tail{A},\qquad
    \tail{(\fix{X}A)} = \fix{X}\tail{A}.
\]
\end{definition}
Note that \(\tail{A}\) always has the form of
    \(\O^{m_0}\fix{X_1}\O^{m_1}\fix{X_2}\O^{m_2}\ldots\fix{X_n}\O^{m_n}Y\)
for some \(n\), \(m_0\), \(m_1\), \(m_2\), \(\ldots\), \(m_n\),
\(X_1\), \(X_2\), \(\ldots\), \(X_n\) and \(Y\),
and it is decidable whether a pseudo type expression is a {\tvariant} or not.
It will be shown that {\tvariant}s are semantically identical to \(\t\).
We can also easily see that the following propositions hold.

\begin{proposition}\label{tvariant-basic}\pushlabel
\begin{Enumerate}
\item \label{var-tvariant}
    \(X\) is not a {\tvariant}.
\item \label{O-tvariant}
    \(\O A\) is a {\tvariant} if and only if so is \(A\).
\item \label{Impl-tvariant}
    \(A \Impl B\) is a {\tvariant} if and only if so is \(B\).
\item \label{tail-tvariant}
    \(A\) is a {\tvariant} if and only if so is \(\tail{A}\).
\item \label{tvariant-rename}
    \(A\) is a {\tvariant} if and only if so is \(A[Y/X]\).
\end{Enumerate}
\end{proposition}
\begin{proof}
Obvious from Definition~\ref{tvariant-def}.
\qed\CHECKED{2014/04/21}
\end{proof}

Unrestricted use of self-reference in type expressions causes logical
contradictions, i.e., divergence of well-typed programs.
So our typing system only allows self-references in coarser worlds,
i.e., in scopes of the modal operator \({\O}\), or
at positions where the references do not affect the meaning of the whole
type expression.
We say that such references are {\em proper}.
The latter case is included so that the equivalence relations
between type expressions (cf. Definitions~\ref{eqtyp-def} and \ref{peqtyp-def})
preserve properness.
For example, the reference to \(X\) in \(X \Impl \t\) is proper
since \(X \Impl \t\) is semantically identical to \(\t\).

\begin{definition}[Properness]
    \ilabel{proper-def}{proper}
    \ilabel*{type expressions!proper}
A (pseudo) type expression \(A\) is {\em proper}\/ in \(X\)
if and only if\/ \(X\) freely occurs only
(a) in scopes of the \(\O\)-operator in \(A\), or (b) in a {\tvariant}
occurring in \(A\).
In other words,
\begin{Enumerate}
\item A type variable \(Y\) is proper in \(X\)
    if and only if \(Y \not= X\).
\item \(\O A\) is proper in \(X\).
\item \(A \Impl B\) is proper in \(X\) if and only if
    (a) so are both \(A\) and \(B\), or (b) \(B\) is a {\tvariant}.
\item Suppose that \(X \not= Y\).
    Then, \(\fix{Y}A\) is proper in \(X\) if and only if
    (a) so is \(A\) or (b) \(\fix{Y}A\) is a {\tvariant}.
\end{Enumerate}
\end{definition}
For example,
\(\O X\), \(\O(X \Impl Y)\), \(\fix{Y}\O(X \Impl Y)\), \(X \Impl \t\)
and \(\fix{Y}X\Impl\O Y\) are proper in \(X\), and
neither \(X\), \(X \Impl Y\) nor \(\fix{Y}\fix{Z} X \Impl Y\)
is proper in \(X\).
Obviously, the following proposition holds.

\begin{proposition}\label{tvariant-proper}
If\/ \(A\) is a {\tvariant},
then \(A\) is proper in every \(X\).
\end{proposition}
\ifdetail
\begin{proof}
By straightforward induction on the structure of \(A\).
Suppose that \(A\) is a {\tvariant}.

\Case{\(A = X\) for some \(X\).}
This case is impossible by Proposition~\ref{var-tvariant}.

\Case{\(A = \O B\) for some \(B\).}
\(A\) is proper in every \(X\) by Definition~\ref{proper-def}.

\Case{\(A = B \Impl C\) for some \(B\) and \(C\).}
In this case, \(C\) is a {\tvariant} by Proposition~\ref{Impl-tvariant}.
Hence, \(A\) is proper in every \(X\) by Definition~\ref{proper-def}.

\Case{\(A = \fix{X}B\) for some \(X\) and \(B\).}
Trivial by Definition~\ref{proper-def}.
\qed\CHECKED{2014/07/10, 07/21}
\end{proof}
\fi 

We can now define the syntax of type expressions of {\lA}.
A type expression of {\lA} is a pseudo type expression such that
\(A\) is proper in \(X\) for any of its subexpressions
in the form of \(\fix{X}{A}\).
We denote the set of well-formed type expressions by \(\TE\),
which, more precisely, can be defined as follows.

\begin{definition}[Type expressions]
    \ilabel{te-def}{type expressions}
    \ilabel*{TExp@{\bf TExp}}
    \ilabel*{syntax!type expressions}
    \ilabel*{syntax!TExp@$\protect\TE$}
We define the set \(\TE\)
of {\em type expressions} to be the smallest set of pseudo type expressions
that satisfy
\begin{Enumerate}
\item \(X \in \TE\) for every type variable \(X\).
\item If \(A \in \TE\) and \(B \in \TE\), then \(A \Impl B \in \TE\).
\item If \(A \in \TE\), then \(\O A \in \TE\).
\item If \(A \in \TE\) and \(A\) is proper in \(X\), then
\(\fix{X}A \in \TE\).
\end{Enumerate}
\end{definition}
For example, \(X\), \(X \Impl Y\), \((\fix{X}\O X \Impl Y)\Impl Z\),
\(\fix{X}X \Impl \t\)
and \(\fix{X}\O\fix{Y}X \Impl Z\) are type expressions, and
neither \(\fix{X}X \Impl Y\) nor \(\fix{X}\fix{Y}X \Impl Y\)
is a type expression.

When we consider free type variables occurring in a type expression,
we can ignore those occurring in {\tvariant}s, because they do not affect
the meaning of the type expression.
We define two sets of type variables that effectively occur
in a type expression.
One consists of those that occur at positive positions,
and the other at negative positions.

\begin{definition}
    \ilabel{etv-def}{ETV@$\protect\ETV{A}$}
Let \(A\) be a (pseudo) type expression.
We define two sets \(\PETV{A}\) and \(\NETV{A}\) of type variables as follows.
\begin{Eqnarray*}
    \PNETV{A} &=& \{\} & (\mbox{\(A\) is a {\tvariant}}) \\
    \PETV{X} &=& \{\, X \,\} \\
    \NETV{X} &=& \{\} \\
    \PNETV{\O A} &=& \PNETV{A} & (\mbox{\(\,\O A\) is not a {\tvariant}}) \\
    \PNETV{A \Impl B} &=& \NPETV{A} \cup \PNETV{B} \hskip43pt &
			 (\mbox{\(A \Impl B\) is not a {\tvariant}}) \\[3pt]
    \PNETV{\fix{X}A} &=&
	\Choice{\hskip-10pt\begin{array}{ll}
	    (\PNETV{A} \cup \NPETV{A}) - \{\,X\,\} &
		\bigg(\mbox{\tabcolsep=0pt\begin{tabular}{l}
			\(\fix{X}A\) is not a {\tvariant} \\[-4pt]
			and \(X \in \NETV{A}\)
		    \end{tabular}}\bigg) \\
	    \PNETV{A} - \{\,X\,\} &
		\bigg(\mbox{\tabcolsep=0pt\begin{tabular}{l}
			\(\fix{X}A\) is not a {\tvariant} \\[-4pt]
			and \(X \not\in \NETV{A}\)
		    \end{tabular}}\bigg)
	    \end{array}}\hskip-300pt
\end{Eqnarray*}
We also define \(\ETV{A}\) as \(\ETV{A} = \PETV{A} \cup \NETV{A}\).
\end{definition}

For example,
\(\PETV{\fix{X}\,\O(X \Impl Y) \Impl Z} = \{\,Z\,\}\),
\(\NETV{\fix{X}\,\O(X \Impl Y) \Impl Z} = \{\,Y\,\}\),
\(\PNETV{\fix{X}\,(Y \Impl Z) \Impl \O X} = \{\}\), and
\(\PNETV{\fix{X}\,\O(X \Impl Y \Impl Z)} = \{\,Y,\,Z\,\}\).
Obviously, \(\PNETV{A} \subseteq \FTV{A}\).
We can easily check that \(\alpha\)-conversion does not
affect the definition of \(\PNETV{A}\)
by Proposition~\ref{tvariant-rename}.

Some propositions in the present paper will be proved
by induction on the following two kinds of structures of type expressions.
The height \(h(A)\) just reflects the syntactical one and
the rank \(r(A)\) does the semantic one in a fixed world.
\begin{definition}\label{height-rank-def}
    \ilabel{height-def}{h@$h(A)$}
    \ilabel*{type expressions!height of}
    \ilabel{rank-def}{r@$r(A)$}
    \ilabel*{type expressions!rank of}
Let \(A\) be a (pseudo) type expression.
We define
\(h(A)\), the {\em height} of\, \(A\), and
\(r(A)\), the {\em rank} of\, \(A\), as follows.
\begin{Eqnarray*}
h(X) &=& 0 \\
h(\O A) &=& h(A) + 1 \\
h(A \Impl B) &=& \max(h(A),\,h(B)) + 1 \\
h(\fix{X}A) &=& h(A) + 1 \\[6pt]
r(A) &=& 0 & (\mbox{\(A\) is a {\tvariant}}) \\
r(X) &=& 0 \\
r(\O A) &=& 0 & (\mbox{\(\,\O A\) is not a {\tvariant}}) \\
r(A \Impl B) &=& \max(r(A),\,r(B)) + 1
    & (\mbox{\(A \Impl B\) is not a {\tvariant}}) \\
r(\fix{X}A) &=& r(A) + 1 & (\mbox{\(\fix{X}A\) is not a {\tvariant}})
\end{Eqnarray*}
\end{definition}

\begin{proposition}\label{tail-basic}\pushlabel
Let \(A\) and \(B\) be (pseudo) type expressions.
\begin{Enumerate}
\item \itemlabel{tail-subst}
    \(\tail{A[B/X]} = \tail{A}[\tail{B}/X]\).
\item \itemlabel{tail-proper}
    If\/ \(A\) is proper in \(X\), then so is \(\tail{A}\).
\item \itemlabel{tail-te}
    If\/ \(A\) is a type expression, then so is \(\tail{A}\).
\item \itemlabel{tail-etv}
    \(\NETV{\tail{A}} = \{\}\) and \(\PETV{\tail{A}} \subseteq \PETV{A}\).
\item \itemlabel{tvariant-etv}
    \(A\) is a {\tvariant} if and only if \(\ETV{A} = \{\}\),
    provided that \(A \in \TE\).
\end{Enumerate}
\end{proposition}
\begin{proof}
By straightforward induction on \(h(A)\),
and by cases on the form of \(A\) using
Propositions~\ref{tvariant-basic} and \ref{tvariant-proper}.
Use Item~\itemref{tail-proper} for \itemref{tail-te}, and
use Items~\itemref{tail-te} and \itemref{tail-etv}
for \itemref{tvariant-etv}.
\ifdetail

\paragraph{Proof of \protect\itemref{tail-subst}}
\Case{\(A = Y\) for some \(Y\).}
If \(Y = X\), then \(\tail{A[B/X]} = \tail{X[B/X]} = \tail{B}
    = X[\tail{B}/X] = \tail{A}[\tail{B}/X]\).
Otherwise, i.e.,
if \(Y \not= X\), then \(\tail{A[B/X]} = \tail{Y[B/X]} = \tail{Y} = Y
= Y[\tail{B}/X] = \tail{A}[\tail{B}/X]\).

\Case{\(A = \O C\) for some \(C\).}
\begin{Eqnarray*}
    \tail{A[B/X]} &=& \tail{(\O C[B/X])} \\
	&=& \O \tail{(C[B/X])}
	    & (by Definition~\ref{tail-def}) \\
	&=& \O (\tail{C}[\tail{B}/X]) & (by induction hypothesis) \\
	&=& (\O \tail{C})[\tail{B}/X] \\
	&=& \tail{A}[\tail{B}/X]
	    & (by Definition~\ref{tail-def})
\end{Eqnarray*}

\Case{\(A = C \Impl D\) for some \(C\) and \(D\).}
\begin{Eqnarray*}
    \tail{A[B/X]} &=& \tail{(C[B/X] \Impl D[B/X])} \\
	&=& \tail{D[B/X]}
	    & (by Definition~\ref{tail-def}) \\
	&=& \tail{D}[\tail{B}/X] & (by induction hypothesis) \\
	&=& \tail{A}[\tail{B}/X]
	    & (by Definition~\ref{tail-def})
\end{Eqnarray*}

\Case{\(A = \fix{Y}C\) for some \(Y\) and \(C\).}
We can assume that \(Y \not\in \{\,X\,\} \cup \FTV{B}\)
without loss of generality.
\begin{Eqnarray*}
    \tail{A[B/X]} &=& \tail{(\fix{Y}C[B/X])}
	    & (since \(Y \not\in \{\,X\,\} \cup \FTV{B}\)) \\
	&=& \fix{Y}\tail{C[B/X]}
	    & (by Definition~\ref{tail-def}) \\
	&=& \fix{Y}\tail{C}[\tail{B}/X] & (by induction hypothesis) \\
	&=& (\fix{Y}\tail{C})[\tail{B}/X]
	    & (since \(Y \not\in  \{\,X\,\} \cup \FTV{\tail{B}}
		    \subseteq \{\,X\,\} \cup \FTV{B}\)) \\
	&=& \tail{A}[\tail{B}/X]
	    & (by Definition~\ref{tail-def})
\end{Eqnarray*}

\paragraph{Proof of \protect\itemref{tail-proper}}
Suppose that \(A\) is proper in \(X\).

\Case{\(A = Y\) for some \(Y\).}
Obvious since \(\tail{A} = A\).

\Case{\(A = \O B\) for some \(B\).}
Obviously, \(\tail{A}\) is proper in \(X\)
by Definition~\ref{proper-def},
since \(\tail{A} = \O \tail{B}\).

\Case{\(A = B \Impl C\) for some \(B\) and \(C\).}
\(\tail{A} = \tail{C}\) in this case.
If \(A\) is a {\tvariant}, then \(\tail{C}\) is also a {\tvariant}
by Propositions~\ref{Impl-tvariant} and \ref{tail-tvariant}; and hence,
\(\tail{C}\) is proper in \(X\) by Proposition~\ref{tvariant-proper}.
On the other hand,
if \(A\) is not a {\tvariant}, then
\(C\) is proper in \(X\) by Definition~\ref{proper-def};
and hence, so is \(\tail{C}\) by induction hypothesis.

\Case{\(A = \fix{Y}B\) for some \(Y\) and \(B\).}
We can assume that \(Y \not= X\) without loss of generality.
If \(A\) is a {\tvariant}, then so is \(\tail{A}\)
by Proposition~\ref{tail-tvariant}; and hence,
\(\tail{A}\) is proper in \(X\) by Proposition~\ref{tvariant-proper}.
On the other hand,
if \(A\) is not a {\tvariant}, then
\(B\) is proper in \(X\) by Definition~\ref{proper-def};
and hence, so is \(\tail{B}\) by induction hypothesis.
Therefore, \(\tail{A}\) is also proper in \(X\)
by Definition~\ref{proper-def},
since \(\tail{A} = \fix{Y}\tail{B}\),

\paragraph{Proof of \protect\itemref{tail-te}}
Suppose that \(A \in \TE\).
\Case{\(A = X\) for some \(X\).}
Obvious since \(\tail{A} = A\).

\Case{\(A = \O B\) for some \(B\).}
\(\tail{A} = \O \tail{B} \in \TE\),
since \(\tail{B} \in \TE\) by induction hypothesis.

\Case{\(A = B \Impl C\) for some \(B\) and \(C\).}
\(\tail{A} = \tail{C} \in \TE\),
since \(\tail{C} \in \TE\) by induction hypothesis.

\Case{\(A = \fix{X}B\) for some \(X\) and \(B\).}
Note that \(B\) is proper in \(X\) since \(A \in \TE\).
Therefore, \(\tail{A} = \fix{X}\tail{B} \in \TE\),
since \(\tail{B} \in \TE\) by induction hypothesis, and
since \(\tail{B}\) is also proper in \(X\)
by Item~\itemref{tail-proper} of this proposition.

\paragraph{Proof of \protect\itemref{tail-etv}}
If \(A\) is a {\tvariant}, then so is \(\tail{A}\)
by Proposition~\ref{tail-tvariant}; and hence,
\(\PNETV{\tail{A}} = \PNETV{A} = \{\}\) by Definition~\ref{etv-def}.
So we assume that \(A\) is not a {\tvariant} below,
which also implies
\(\tail{A}\) is not either, by Proposition~\ref{tail-tvariant}.

\Case{\(A = X\) for some \(X\).}
Trivial since \(\NETV{A}= \{\}\) and \(\PETV{\tail{A}} = \PETV{A} = \{\,X\,\}\)
in this case.

\Case{\(A = \O B\) for some \(B\).}
Straightforward by induction hypothesis, since
\(\PNETV{A} = \PNETV{B}\) and
\(\PNETV{\tail{A}} = \PNETV{\O \tail{B}} = \PNETV{\tail{B}}\).

\Case{\(A = B \Impl C\) for some \(B\) and \(C\).}
Since neither \(A\) nor \(\tail{A}\) is a {\tvariant},
\(\PNETV{A} = \NPETV{B} \cup \PNETV{C}\)
and
\(\PNETV{\tail{A}} = \PNETV{\tail{C}}\).
Therefore, \(\NETV{\tail{A}} = \{\}\) and
\(\PETV{\tail{A}} \subseteq \PETV{A}\),
since
\(\NETV{\tail{C}} = \{\}\) and \(\PETV{\tail{C}} \subseteq \PETV{C}\)
by induction hypothesis.

\Case{\(A = \fix{X}B\) for some \(X\) and \(B\).}
By Definition~\ref{tail-def},
    \(\PNETV{\tail{A}} = \PNETV{\fix{X}\tail{B}}\).
Note that \(X \not\in \NETV{\tail{B}}\) by induction hypothesis.
Hence, by Definition~\ref{etv-def},
\begin{eqnarray*}
    \PNETV{\tail{A}} &=& \PNETV{\tail{B}} - \{\,X\,\}, ~\mbox{and} \\
    \PNETV{A} &\supseteq& \PNETV{B} - \{\,X\,\}.
\end{eqnarray*}
Therefore,
\(\NETV{\tail{A}} = \{\}\) and \(\PETV{\tail{A}} \subseteq \PETV{A}\)
because
\(\NETV{\tail{B}} = \{\}\) and \(\PETV{\tail{B}} \subseteq \PETV{B}\)
by induction hypothesis.

\paragraph{Proof of \protect\itemref{tvariant-etv}}
The ``only if'' part is obvious from Definition~\ref{etv-def}.
For the ``if'' part, suppose that \(A \in \TE\) and \(\ETV{A} = \{\}\).

\Case{\(A = X\) for some \(X\).}
This case is impossible by Definition~\ref{etv-def}.

\Case{\(A = \O B\) for some \(B\).}
Since \(\ETV{B} = \ETV{A} = \{\}\),
\(B\) is a {\tvariant} by induction hypothesis;
and hence, so is \(A\) by Proposition~\ref{O-tvariant}.

\Case{\(A = B \Impl C\) for some \(B\) and \(C\).}
Since \(\ETV{C} \subseteq \ETV{A} = \{\}\),
\(C\) is a {\tvariant} by induction hypothesis;
and hence, so is \(A\) by Proposition~\ref{Impl-tvariant}.

\Case{\(A = \fix{X}B\) for some \(X\) and \(B\).}
Let \(\tail{A} =
    \fix{X}\O^{m_0}\fix{Z_1}\O^{m_1}\fix{Z_2}\O^{m_2}\ldots\fix{Z_n}\O^{m_n}Y\).
Note that \(\tail{A} \in \TE\) from \(A \in \TE\),
by Item~\itemref{tail-te} of this proposition.
We get \(\ETV{\tail{A}} = \{\}\) from \(\ETV{A} = \{\}\)
by Item~\itemref{tail-etv} of this proposition; that is,
\(Y \in \{\,X,\,Z_1,\,Z_2,\,\ldots,\,Z_n\,\}\).
Hence, \(A\) is a {\tvariant}.
\fi 
\qed\CHECKED{2014/07/11}
\end{proof}

In the rest of the present paper, we use \(A\), \(B\), \(C\), \(\ldots\)
to denote only type expressions.
Type expressions also have the following basic properties
concerning {\tvariant}s and properness.

\begin{proposition}\pushlabel
\begin{Enumerate}
\item \itemlabel{tvariant-subst1}
    If \(A\) is a {\tvariant}, then so is \(A[B/X]\).
\item \itemlabel{tvariant-subst2}
    If\/ \(A[B/X]\) is a {\tvariant}, then
	(a) \(A\) is already a {\tvariant}, or
	(b) \(X \in \PETV{\tail{A}}\) and \(B\) is a {\tvariant}.
\item \itemlabel{tvariant-proper-subst}
    If \(A[B/X]\) is a {\tvariant} and \(A\) is proper in \(X\),
    then \(\fix{X}A\) is also a {\tvariant}.
\item \itemlabel{tvariant-fix}
    \(\fix{X}A\) is a {\tvariant} if and only if so is \(A[\fix{X}A/X]\).
\end{Enumerate}
\end{proposition}
\begin{proof}
Let \(\tail{A}
    = \O^{m_0}\fix{Y_1}\O^{m_1}\fix{Y_2}\O^{m_2}\ldots\fix{Y_n}\O^{m_n}Z\),
where \(Y_i \not\in \FTV{B} \cup \{\,X\,\}\) for every \(i\).
By Proposition~\ref{tail-subst},
\[
    \tail{A[B/X]} = \Choice{\begin{array}{ll}
	\tail{A} & (Z \not= X) \\
	\O^{m_0}\fix{Y_1}\O^{m_1}\fix{Y_2}\O^{m_2}\ldots\fix{Y_n}\O^{m_n}
	    \tail{B} & (Z = X).
	\end{array}}
\]
Therefore,
we get Item~\itemref{tvariant-subst1} by Proposition~\ref{tail-tvariant}
since \(Z \not= X\) if \(A\) is a {\tvariant}, and
similarly Item~\itemref{tvariant-subst2}
since \(Z = X\) implies \(X \in \PETV{\tail{A}}\).
We also get Item~\itemref{tvariant-proper-subst}
since \(\tail{(\fix{X}A)}\) is either \(\fix{X}\tail{(A[B/X])}\) or
\(\fix{X}\O^{m_0}\fix{Y_1}\O^{m_1}\fix{Y_2}\O^{m_2}\ldots\fix{Y_n}\O^{m_n}X\).
The ``if part'' of Item~\itemref{tvariant-fix}
follows \itemref{tvariant-proper-subst}, and the ``only if'' part
is also straightforward since
\(\tail{A[\fix{X}A/X]}\) is either \(\tail{A}\) (in case of \(Z \not= X\)),
or \(\O^{m_0}\fix{Y_1}\O^{m_1}\fix{Y_2}\O^{m_2}\ldots\fix{Y_n}\O^{m_n}
\fix{X}\tail{A}\) (in case of \(Z = X\)).
\qed\CHECKED{2014/07/07}
\end{proof}

\begin{proposition}\label{petv-netv-subst}\pushlabel
\begin{Enumerate}
\item \itemlabel{petv-netv-subst1}
    If \(X \in \PNETV{A}\), \(X \not= Y\) and \(B\) is not a {\tvariant},
    then \(X \in \PNETV{A[B/Y]}\).
\item \itemlabel{petv-netv-subst2}
    If \(X \not\in \PNETV{A}\) and \(X \not\in \ETV{B}\),
    then \(X \not\in \PNETV{A[B/Y]}\).
\end{Enumerate}
\end{proposition}
\begin{proof}
By straightforward induction on \(h(A)\), and by cases on the form of \(A\).
\ifdetail

\paragraph{Proof of \protect\itemref{petv-netv-subst1}}
Suppose that \(B\) is not a {\tvariant},
\(X \in \PNETV{A}\) and \(X \not= Y\).
Note that \(A\) is not a {\tvariant}
by Proposition~\ref{tvariant-etv} since \(\ETV{A} \not= \{\}\).
Furthermore,
\(A[B/Y]\) is not a {\tvariant} either, by Proposition~\ref{tvariant-subst2}.

\Case{\(A = Z\) for some \(Z\).}
Note that \(\NETV{A} = \{\}\) by Definition~\ref{etv-def}.
On the other hand, \(X \in \PETV{A}\) implies \(Z = X\); and hence,
\(\PETV{A[B/Y]} = \PETV{X[B/Y]} = \PETV{X} = \{\,X\,\}\).

\Case{\(A = \O C\) for some \(C\).}
Since \(\PNETV{A} = \PNETV{C}\),
by induction hypothesis, \(X \in \PNETV{C[B/Y]}\); and hence,
\(X \in \PNETV{A[B/Y]}\) because \(\PNETV{A[B/Y]} = \PNETV{C[B/Y]}\).

\Case{\(A = C \Impl D\) for some \(C\) and \(D\).}
Since neither \(A\) nor \(A[B/Y]\) is a {\tvariant},
\(\PNETV{A} = \NPETV{C} \cup \PNETV{D}\)
and
\(\PNETV{A[B/Y]} = \NPETV{C[B/Y]} \cup \PNETV{D[B/Y]}\).
Hence,
we get \(X \in \NPETV{C}\) or \(X \in \PNETV{D}\) from \(X \in \PNETV{A}\).
If \(X \in \NPETV{C}\) is the case, then
by induction hypothesis, \(X \in \NPETV{C[B/Y]}\); and therefore,
\(X \in \PNETV{A[B/Y]}\).
Similarly, we can also get the same from \(X \in \PNETV{D}\).

\Case{\(A = \fix{Z}C\) for some \(Z\) and \(C\).}
We can assume that \(Z \not\in \{\,X,\,Y\,\} \cup \FTV{B}\)
without loss of generality.
Thus, \(A[B/Y] = \fix{Z}{C[B/Y]}\).
If \(Z \in \NETV{C}\), then
we also get \(Z \in \NETV{C[B/Y]}\) by induction hypothesis; and hence,
\(\PNETV{A} = \ETV{C} - \{\,Z\,\}\)
and \(\PNETV{A[B/Y]} = \ETV{C[B/Y]} - \{\,Z\,\}\).
Therefore, we get \(X \in \PNETV{A[B/Y]}\) from \(X \in \PNETV{A}\)
since \(X \in \ETV{C}\) implies \(X \in \ETV{C[B/Y]}\)
by induction hypothesis.
On the other hand, if \(Z \not\in \NETV{C}\), then
\(\PNETV{A} = \PNETV{C} - \{\,Z\,\}\)
and
\(\PNETV{A[B/Y]} \supseteq \PNETV{C[B/Y]} - \{\,Z\,\}\).
Hence, we similarly get \(X \in \PNETV{A[B/Y]}\) from \(X \in \PNETV{A}\)
by induction hypothesis.

\paragraph{Proof of \protect\itemref{petv-netv-subst2}}
Suppose that \(X \not\in \PNETV{A}\) and \(X \not\in \ETV{B}\).
If \(A[B/Y]\) is a {\tvariant}, then
\(X \not\in \PNETV{A[B/Y]}\) by Proposition~\ref{tvariant-etv}.
So we only consider the case when \(A[B/Y]\) is not a {\tvariant} below.
We can also assume that \(A\) is not either,
because \(A[B/Y]\) is a {\tvariant} by Proposition~\ref{tvariant-subst1}
if so is \(A\).

\Case{\(A = Z\) for some \(Z\).}
If \(Z \not= Y\), then trivial since \(A[B/Y] = Z = A\).
On the other hand, if \(Z = Y\), then
\(A[B/Y] = B\); and hence, \(X \not\in \PNETV{A[B/Y]}\)
from \(X \not\in \ETV{B}\).

\Case{\(A = \O C\) for some \(C\).}
Since \(\PNETV{A} = \PNETV{C}\), we get
\(X \not\in \PNETV{C[B/Y]}\)
by induction hypothesis.
Hence, \(X \not\in \PNETV{A[B/Y]}\)
because \(\PNETV{A[B/Y]} = \PNETV{C[B/Y]}\).

\Case{\(A = C \Impl D\) for some \(C\) and \(D\).}
Since neither \(A\) nor \(A[B/Y]\) is a {\tvariant},
\(\PNETV{A} = \NPETV{C} \cup \PNETV{D}\)
and
\(\PNETV{A[B/Y]} = \NPETV{C[B/Y]} \cup \PNETV{D[B/Y]}\).
Hence,
we get \(X \not\in \NPETV{C}\cup\PNETV{D}\) from \(X \in \PNETV{A}\).
Therefore, by induction hypothesis,
\(X \not\in \NPETV{C[B/Y]} \cup \PNETV{D[B/Y]}\), that is,
\(X \not\in \PNETV{A[B/Y]}\).

\Case{\(A = \fix{Z}C\) for some \(Z\) and \(C\).}
We can assume that \(Z \not\in \{\,X,\,Y\,\} \cup \FTV{B}\)
without loss of generality.
Thus, \(A[B/Y] = \fix{Z}{C[B/Y]}\).
If \(Z \in \NETV{C}\), then
\(\PETV{A} = \NETV{A} = \ETV{C} - \{\,Z\,\}\).
Hence, \(X \not\in \ETV{C}\) from \(X \not\in \PNETV{A}\),
and we get \(X \not\in \ETV{C[B/Y]}\)
by induction hypothesis.
Therefore, \(X \not\in \PNETV{A[B/Y]}\),
since \(\PNETV{A[B/Y]} \subseteq \ETV{C[B/Y]} - \{\,Z\,\}\).
On the other hand, if \(Z \not\in \NETV{C}\), then
we get
\(Z \not\in \NETV{C[B/Y]}\) from \(Z \not\in \ETV{B} \subseteq \FTV{B}\)
by induction hypothesis.
Hence,
\(\PNETV{A} = \PNETV{C} - \{\,Z\,\}\)
and
\(\PNETV{A[B/Y]} = \PNETV{C[B/Y]} - \{\,Z\,\}\).
Therefore, we similarly get \(X \not\in \PNETV{A[B/Y]}\)
from \(X \in \PNETV{A}\) by induction hypothesis.
\else 
For Item~\itemref{petv-netv-subst1},
note that \(A\) is not a {\tvariant}
by Proposition~\ref{tvariant-etv} since \(\ETV{A} \not= \{\}\).
Furthermore,
\(A[B/Y]\) is not a {\tvariant} either,
by Proposition~\ref{tvariant-subst2}.
Similarly, for Item~\itemref{petv-netv-subst2},
we only have to handle the case that neither \(A\) nor \(A[B/Y]\)
is a {\tvariant},
because \(\PNETV{A[B/Y]} = \{\}\)
by Proposition~\ref{tvariant-etv} if \(A[B/Y]\) is a {\tvariant},
and because \(A[B/Y]\) is a {\tvariant}
by Proposition~\ref{tvariant-subst1} if so is \(A\).
\fi 
\qed\CHECKED{2014/07/11}
\end{proof}

\begin{proposition}\label{petv-netv-nest}\pushlabel
\begin{Enumerate}
\item \itemlabel{petv-netv-nest1}
    If \(X \in \PETV{A}\) and \(Y \in \PNETV{B}\), then
		\(Y \in \PNETV{A[B/X]}\).
\item \itemlabel{petv-netv-nest2}
    If \(X \in \NETV{A}\) and \(Y \in \PNETV{B}\), then
		\(Y \in \NPETV{A[B/X]}\).
\item \itemlabel{petv-netv-nest3}
    If \(X \not\in \PETV{A}\) and \(Y \not\in \PNETV{A} \cup \NPETV{B}\), then
		\(Y \not\in \PNETV{A\penalty100[B/X]}\).
\item \itemlabel{petv-netv-nest4}
    If \(X \not\in \NETV{A}\) and \(Y \not\in \PNETV{A} \cup \PNETV{B}\), then
		\(Y \not\in \PNETV{A\penalty100[B/X]}\).
\end{Enumerate}
\end{proposition}
\begin{proof}
By simultaneous induction on \(h(A)\), and by cases on the form of \(A\).
\ifdetail

\paragraph{Proof of \protect\itemref{petv-netv-nest1}}

Suppose that \(X \in \PETV{A}\) and \(Y \in \PNETV{B}\).
Hence, neither \(A\) nor \(B\) is a {\tvariant}
by Proposition~\ref{tvariant-etv}; and therefore,
\(A[B/X]\) is not a {\tvariant} either, by Proposition~\ref{tvariant-subst2}.

\Case{\(A = Z\) for some \(Z\).}
In this case,
\(\PNETV{A[B/X]} = \PNETV{B}\), since \(Z = X\) from \(X \in \PETV{A}\).
Hence, \(Y \in \PNETV{A[B/X]}\) from \(Y \in \PNETV{B}\).

\Case{\(A = \O C\) for some \(C\).}
In this case,
since \(\PETV{A} = \PETV{C}\), by the induction hypothesis,
\(Y \in \PNETV{C[B/X]}\), from which we get
\(Y \in \PNETV{A[B/X]}\)
because \(\PNETV{A[B/X]} = \PNETV{C[B/X]}\).

\Case{\(A = C \Impl D\) for some \(C\) and \(D\).}
Since neither \(A\) nor \(A[B/X]\) is a {\tvariant},
\(\PETV{A} = \NETV{C} \cup \PETV{D}\)
and \(\PNETV{A[B/X]} = \NPETV{C[B/X]} \cup \PNETV{D[B/X]}\).
Hence,
we get \(X \in \NETV{C}\cup\PETV{D}\) from \(X \in \PETV{A}\).
Therefore, by the induction hypotheses
for Items~\itemref{petv-netv-nest1} and \itemref{petv-netv-nest2},
\(Y \in \NPETV{C[B/X]} \cup \PNETV{D[B/X]}\), that is,
\(Y \in \PNETV{A[B/X]}\).

\Case{\(A = \fix{Z}C\) for some \(Z\) and \(C\).}
We can assume that \(Z \not\in \{\,X,\,Y\,\} \cup \FTV{B}\)
without loss of generality.
Thus, \(A[B/X] = \fix{Z}{C[B/X]}\).
If \(Z \in \NETV{C}\), then
we also get \(Z \in \NETV{C[B/X]}\) by Proposition~\ref{petv-netv-subst1};
and hence,
\(\PETV{A} = \ETV{C} - \{\,Z\,\}\) and
\(\PNETV{A[B/X]} = \ETV{C[B/X]} - \{\,Z\,\}\).
Therefore,
\(X \in \ETV{C}\) from \(X \in \PETV{A}\),
and we get \(Y \in \ETV{C[B/X]}\)
by the induction hypotheses
for Items~\itemref{petv-netv-nest1} and \itemref{petv-netv-nest2}.
Thus, \(Y \in \PNETV{A[B/X]}\).
On the other hand, if \(Z \not\in \NETV{C}\), then
\(\PETV{A} = \PETV{C} - \{\,Z\,\}\)
and \(\PNETV{A[B/X]} \supseteq \PNETV{C[B/X]} - \{\,Z\,\}\).
Therefore, we similarly get \(Y \in \PNETV{A[B/X]}\)
from \(X \in \PETV{A}\) by the induction hypothesis for the same item.

\paragraph{Proof of \protect\itemref{petv-netv-nest2}}
The proof is almost parallel to the one for Item~\itemref{petv-netv-nest1}.
Suppose that \(X \in \NETV{A}\) and \(Y \in \PNETV{B}\).
Hence, neither \(A\), \(B\) nor \(A[B/X]\) is a {\tvariant}
by Propositions~\ref{tvariant-etv} and \ref{tvariant-subst2}.

\Case{\(A = Z\) for some \(Z\).}
This case contradicts the assumption that \(\NETV{A} \not= \{\}\).

\Case{\(A = \O C\) for some \(C\).}
Since \(\NETV{A} = \NETV{C}\), by induction hypothesis,
\(Y \in \NPETV{C[B/X]}\), from which we get
\(Y \in \NPETV{A[B/X]}\)
because \(\NPETV{A[B/X]} = \NPETV{C[B/X]}\).

\Case{\(A = C \Impl D\) for some \(C\) and \(D\).}
Since neither \(A\) nor \(A[B/X]\) is a {\tvariant},
\(\NETV{A} = \PETV{C} \cup \NETV{D}\)
and
\(\NPETV{A[B/X]} = \PNETV{C[B/X]} \cup \NPETV{D[B/X]}\).
Hence,
we get \(X \in \PETV{C}\cup\NETV{D}\) from \(X \in \NETV{A}\).
Therefore, by the induction hypotheses
for Items~\itemref{petv-netv-nest1} and \itemref{petv-netv-nest2},
\(Y \in \PNETV{C[B/X]} \cup \NPETV{D[B/X]}\), that is,
\(Y \in \NPETV{A[B/X]}\).

\Case{\(A = \fix{Z}C\) for some \(Z\) and \(C\).}
We can assume that \(Z \not\in \{\,X,\,Y\,\} \cup \FTV{B}\)
without loss of generality.
Thus, \(A[B/X] = \fix{Z}{C[B/X]}\).
If \(Z \in \NETV{C}\), then
we also get \(Z \in \NETV{C[B/X]}\) by Proposition~\ref{petv-netv-subst1};
and hence,
\(\NETV{A} = \ETV{C} - \{\,Z\,\}\) and
\(\NPETV{A[B/X]} = \ETV{C[B/X]} - \{\,Z\,\}\).
Therefore,
\(X \in \ETV{C}\) from \(X \in \NETV{A}\),
and we get \(Y \in \ETV{C[B/X]}\) by induction hypothesis.
Thus, \(Y \in \NPETV{A[B/X]}\).
On the other hand, if \(Z \not\in \NETV{C}\), then
\(\NETV{A} = \NETV{C} - \{\,Z\,\}\)
and \(\NPETV{A[B/X]} \supseteq \NPETV{C[B/X]} - \{\,Z\,\}\).
Hence, we similarly get \(Y \in \NPETV{A[B/X]}\)
from \(X \in \NETV{A}\) by induction hypothesis.

\paragraph{Proof of \protect\itemref{petv-netv-nest3}}
Suppose that \(X \not\in \PETV{A}\) and \(Y \not\in \PNETV{A} \cup \NPETV{B}\).
If \(A[B/X]\) is a {\tvariant}, then
\(Y \not\in \PNETV{A[B/X]}\) by Proposition~\ref{tvariant-etv}.
So we only consider the case when \(A[B/X]\) is not a {\tvariant}.
Note that \(A\) is not a {\tvariant} either,
because \(A[B/X]\) is a {\tvariant} by Proposition~\ref{tvariant-subst1}
if so is \(A\).

\Case{\(A = Z\) for some \(Z\).}
In this case, \(Z \not= X\) from \(X \not\in \PETV{A}\); and hence,
\(\PNETV{A[B/X]} = \PNETV{A}\).
Hence, \(Y \not\in \PNETV{A[B/X]}\)
from \(Y \not\in \PNETV{A} \cup \NPETV{B}\).

\Case{\(A = \O C\) for some \(C\).}
Since \(\PNETV{A} = \PNETV{C}\), by induction hypothesis,
\(Y \not\in \PNETV{C[B/X]}\), from which we get
\(Y \not\in \PNETV{A[B/X]}\)
because \(\PNETV{A[B/X]} = \PNETV{C[B/X]}\).

\Case{\(A = C \Impl D\) for some \(C\) and \(D\).}
Since neither \(A\) nor \(A[B/X]\) is a {\tvariant},
\(\PNETV{A} = \NPETV{C} \cup \PNETV{D}\)
and \(\PNETV{A[B/X]} = \NPETV{C[B/X]} \cup \PNETV{D[B/X]}\).
Hence, we get \(X \not\in \NETV{C}\cup\PETV{D}\) from \(X \not\in \PETV{A}\),
and also get \(Y \not\in \NPETV{C} \cup \NPETV{B}\)
and \(Y \not\in \PNETV{D} \cup \NPETV{B}\)
from \(Y \not\in \PNETV{A} \cup \NPETV{B}\).
Therefore, by the induction hypotheses
for Items~\itemref{petv-netv-nest3} and \itemref{petv-netv-nest4},
\(Y \not\in \NPETV{C[B/X]} \cup \PNETV{D[B/X]}\), that is,
\(Y \not\in \PNETV{A[B/X]}\).

\Case{\(A = \fix{Z}C\) for some \(Z\) and \(C\).}
We can assume that \(Z \not\in \{\,X,\,Y\,\} \cup \FTV{B}\)
without loss of generality.
Thus, \(A[B/X] = \fix{Z}{C[B/X]}\).
If \(Z \in \NETV{C}\), then
\(\PETV{A} = \NETV{A} = \ETV{C} - \{\,Z\,\}\); and hence,
we get \(X \not\in \ETV{C}\) from \(X \not\in \PETV{A}\),
and get \(Y \not\in \ETV{C} \cup \NPETV{B}\)
from \(Y \not\in \PNETV{A} \cup \NPETV{B}\).
Therefore, \(Y \not\in \ETV{C[B/X]}\) by induction hypothesis.
We thus get \(Y \not\in \PNETV{A[B/X]}\),
since \(\PNETV{A[B/X]} \subseteq \ETV{C[B/X]} - \{\,Z\,\}\).
On the other hand, if \(Z \not\in \NETV{C}\), then
we get \(Z \not\in \NETV{C[B/X]}\) from \(Z \not\in \ETV{B} \subseteq \FTV{B}\)
by Proposition~\ref{petv-netv-subst2}.
Hence, \(\PNETV{A} = \PNETV{C} - \{\,Z\,\}\)
and \(\PNETV{A[B/X]} = \PNETV{C[B/X]} - \{\,Z\,\}\).
Therefore, we get \(X \not\in \PETV{C}\) from \(X \not\in \PETV{A}\),
and get \(Y \not\in \PNETV{C} \cup \NPETV{B}\)
from \(Y \not\in \PNETV{A} \cup \NPETV{B}\); and hence,
\(Y \not\in \PNETV{C[B/X]}\) by induction hypothesis,
from which we get \(Y \not\in \PNETV{A[B/X]}\).

\paragraph{Proof of \protect\itemref{petv-netv-nest4}}
The proof is almost parallel to the one for Item~\itemref{petv-netv-nest3}.
Suppose that \(X \not\in \NETV{A}\) and \(Y \not\in \PNETV{A} \cup \PNETV{B}\).
It also suffices to consider the case that
neither \(A\) nor \(A[B/X]\) is a {\tvariant}.

\Case{\(A = Z\) for some \(Z\).}
If \(Z = X\), then
\(\PNETV{A[B/X]} = \PNETV{B}\); and hence,
\(Y \not\in \PNETV{A[B/X]}\)
from \(Y \not\in \PNETV{A} \cup \PNETV{B}\).
Otherwise, i.e.,
if \(Z \not= X\), then
\(\PNETV{A[B/X]} = \PNETV{A}\); and hence,
we similarly get \(Y \not\in \PNETV{A[B/X]}\)
from \(Y \not\in \PNETV{A} \cup \PNETV{B}\).

\Case{\(A = \O C\) for some \(C\).}
Since \(\PNETV{A} = \PNETV{C}\), by induction hypothesis,
\(Y \not\in \PNETV{C[B/X]}\), from which we get
\(Y \not\in \PNETV{A[B/X]}\)
because \(\PNETV{A[B/X]} = \PNETV{C[B/X]}\).

\Case{\(A = C \Impl D\) for some \(C\) and \(D\).}
Since neither \(A\) nor \(A[B/X]\) is a {\tvariant},
\(\PNETV{A} = \NPETV{C} \cup \PNETV{D}\)
and \(\PNETV{A[B/X]} = \NPETV{C[B/X]} \cup \PNETV{D[B/X]}\).
Hence, we get \(X \not\in \PETV{C}\cup\NETV{D}\) from \(X \not\in \NETV{A}\),
and also get \(Y \not\in \NPETV{C} \cup \PNETV{B}\)
and \(Y \not\in \PNETV{D} \cup \PNETV{B}\)
from \(Y \not\in \PNETV{A} \cup \PNETV{B}\).
Therefore, by the induction hypotheses
for Items~\itemref{petv-netv-nest3} and \itemref{petv-netv-nest4},
\(Y \not\in \NPETV{C[B/X]} \cup \PNETV{D[B/X]}\), that is,
\(Y \not\in \PNETV{A[B/X]}\).

\Case{\(A = \fix{Z}C\) for some \(Z\) and \(C\).}
We can assume that \(Z \not\in \{\,X,\,Y\,\} \cup \FTV{B}\)
without loss of generality.
Thus, \(A[B/X] = \fix{Z}{C[B/X]}\).
If \(Z \in \NETV{C}\), then
\(\PETV{A} = \NETV{A} = \ETV{C} - \{\,Z\,\}\); and hence,
we get \(X \not\in \ETV{C}\) from \(X \not\in \NETV{A}\),
and get \(Y \not\in \ETV{C} \cup \PNETV{B}\)
from \(Y \not\in \PNETV{A} \cup \PNETV{B}\).
Therefore, \(Y \not\in \ETV{C[B/X]}\) by induction hypothesis.
We thus get \(Y \not\in \PNETV{A[B/X]}\),
since \(\PNETV{A[B/X]} \subseteq \ETV{C[B/X]} - \{\,Z\,\}\).
On the other hand, if \(Z \not\in \NETV{C}\), then
we get \(Z \not\in \NETV{C[B/X]}\) from \(Z \not\in \ETV{B} \subseteq \FTV{B}\)
by Proposition~\ref{petv-netv-subst2}.
Hence, \(\PNETV{A} = \PNETV{C} - \{\,Z\,\}\)
and \(\PNETV{A[B/X]} = \PNETV{C[B/X]} - \{\,Z\,\}\).
Therefore, we get \(X \not\in \NETV{C}\) from \(X \not\in \NETV{A}\),
and get \(Y \not\in \PNETV{C} \cup \PNETV{B}\)
from \(Y \not\in \PNETV{A} \cup \PNETV{B}\); and hence,
\(Y \not\in \PNETV{C[B/X]}\) by induction hypothesis,
from which we get \(Y \not\in \PNETV{A[B/X]}\).
\else 
For Items~\itemref{petv-netv-nest1} and \itemref{petv-netv-nest2},
note that neither \(A\), \(B\) nor \(A[B/X]\) is a {\tvariant}
by Propositions~\ref{tvariant-etv} and \ref{tvariant-subst2}.
In case of \(A = \fix{Z}C\) for some \(Z\) and \(C\),
we use Proposition~\ref{petv-netv-subst1} to show that
\(Z \in \NETV{C}\) implies \(Z \in \NETV{C[B/X]}\).
Similarly,
for Items~\itemref{petv-netv-nest3} and \itemref{petv-netv-nest4},
by Propositions~\ref{tvariant-etv} and \ref{tvariant-subst1},
it suffices to consider the case that
neither \(A\) nor \(A[B/X]\) is a {\tvariant}.
In case of \(A = \fix{Z}C\) for some \(Z\) and \(C\),
we use Proposition~\ref{petv-netv-subst2} to show that
\(Z \not\in \NETV{C}\) implies \(Z \not\in \NETV{C[B/X]}\).
\fi 
\qed\CHECKED{2014/07/11}
\end{proof}

\begin{proposition}\pushlabel
\begin{Enumerate}
\item \itemlabel{rank-proper}
    If\/ \(A\) is proper in \(X\),
    then \(r(A[B/X]) \le r(A)\) for every \(B\).
\item \itemlabel{rank-fix}
    If\/ \(\fix{X}A\) is not a {\tvariant}, then
    \(r(A[\fix{X}A/X]) < r(\fix{X}A)\).
\item \itemlabel{proper-subst1}
    If\/ \(A\) and \(B\) are proper in \(X\), then \(A[B/Y]\)
    is also proper in \(X\) for any \(Y\).
\item \itemlabel{proper-subst2}
    If\/ \(A\) is proper in \(X\), then so is \(A[B/X]\)
    for every \(B\).
\end{Enumerate}
\end{proposition}
\begin{proof}
By straightforward induction on \(h(A)\), and by cases on the form of \(A\)
using Proposition~\ref{tvariant-subst1}, where
Item~\itemref{rank-fix} immediately follows from \itemref{rank-proper}.
Note that every {\tvariant} is proper in any type variable, and that
\(r(A[B/X]) = 0\) for every {\tvariant} \(A\).
Hence, it suffices to only consider the case that \(A\) is not a {\tvariant}.
%
%
\ifdetail

\paragraph{Proof of \protect\itemref{rank-proper}}
Suppose that \(A\) is proper in \(X\).
If \(A[B/X]\) is a {\tvariant}, then it is trivial because
\(r(A[B/X]) = 0\) by Definition~\ref{rank-def}.
Furthermore, if \(A\) is a {\tvariant}, the so is
\(A[B/X]\) by Proposition~\ref{tvariant-subst1}.
Hence, we only consider the case
that neither \(A\) nor \(A[B/X]\) is a {\tvariant}, below.

\Case{\(A = Y\) for some \(Y\).}
Since \(A\) is proper in \(X\),
we get \(Y \not= X\); and hence, \(A[B/X] = A\).

\Case{\(A = \O C\) for some \(C\).}
Trivial since \(r(A[B/X]) = 0\) by Definition~\ref{rank-def} in this case.

\Case{\(A = C \Impl D\) for some \(C\) and \(D\).}
Note that both \(C\) and \(D\) are also proper in \(X\),
since \(A\) is not a {\tvariant}.
Hence, by induction hypothesis and Definition~\ref{rank-def},
\(r(A[B/X]) = r(C[B/X] \Impl D[B/X])
    = \max(r(C[B/X]),\,r(D[B/X])) + 1
    \le \max(r(C),\,r(D)) + 1
    = r(A)\).

\Case{\(A = \fix{Y}C\) for some \(Y\) and \(C\).}
We can assume that \(Y \not\in \{\,X\,\} \cup \FTV{B}\)
without loss of generality.
That is, \(A[B/X] = \fix{Y}{C[B/X]}\).
Note that \(C\) is also proper in \(X\), since \(A\) is not a {\tvariant}.
Hence, by induction hypothesis and Definition~\ref{rank-def},
\(r(A[B/X]) = r(\fix{Y}C[B/X])
    = r(C[B/X]) + 1
    \le r(C) + 1
    = r(\fix{Y}C)\).

\paragraph{Proof of \protect\itemref{rank-fix}}
If \(\fix{X}A\) is not a {\tvariant}, then
\(r(\fix{X}A) = r(A) + 1\) by Definition~\ref{rank-def}; and hence,
\(r(A[\fix{X}A/X]) \le r(A) < r(\fix{X}A)\) by Item~\itemref{rank-proper}
of this proposition.

\paragraph{Proof of \protect\itemref{proper-subst1}}
Suppose that \(A\) and \(B\) are proper in \(X\).
If \(A[B/Y]\) is a {\tvariant}, then \(A[B/Y]\) is proper in \(X\)
by Proposition~\ref{tvariant-proper}.
Furthermore, if \(A\) is a {\tvariant}, the so is
\(A[B/Y]\) by Proposition~\ref{tvariant-subst1}.
Hence, we assume that neither \(A\) nor \(A[B/Y]\) is a {\tvariant}, below.

\Case{\(A = Z\) for some \(Z\).}
If \(Z = Y\), then \(A[B/Y] = B\); and otherwise, \(A[B/Y] = A\).
In either case, \(A[B/Y]\) is proper in \(X\) by assumption.

\Case{\(A = \O C\) for some \(C\).}
Trivial by Definition~\ref{proper-def} since \(A[B/Y] = \O C[B/Y]\).

\Case{\(A = C \Impl D\) for some \(C\) and \(D\).}
Both \(C\) and \(D\) are also proper in \(X\),
since \(A\) is not a {\tvariant}.
Therefore, so are \(C[B/Y]\) and \(D[B/Y]\) by induction hypothesis;
and hence, \(A[B/Y]\) is also proper in \(X\).

\Case{\(A = \fix{Z}C\) for some \(Z\) and \(C\).}
We can assume that \(Z \not\in \{\,X,\,Y\,\} \cup \FTV{B}\)
without loss of generality.
That is, \(A[B/Y] = \fix{Z}{C[B/Y]}\).
Note that \(C\) is also proper in \(X\), since \(A\) is not a {\tvariant}.
Therefore, \(C[B/Y]\) is proper in \(X\) by induction hypothesis;
and hence, so is \(A[B/Y]\) by Definition~\ref{proper-def}.

\paragraph{Proof of \protect\itemref{proper-subst2}}
Suppose that \(A\) is proper in \(X\).
Similarly to the previous item,
we can assume that neither \(A\) nor \(A[B/X]\) is a {\tvariant}.

\Case{\(A = Y\) for some \(Y\).}
Since \(A\) is proper in \(X\),
we get \(Y \not= X\); and hence, \(A[B/X] = A\), which
is proper in \(X\) by assumption.

\Case{\(A = \O C\) for some \(C\).}
Trivial by Definition~\ref{proper-def} since \(A[B/X] = \O C[B/X]\).

\Case{\(A = C \Impl D\) for some \(C\) and \(D\).}
Both \(C\) and \(D\) are also proper in \(X\),
since \(A\) is not a {\tvariant}.
Therefore, so are \(C[B/X]\) and \(D[B/X]\) by induction hypothesis;
and hence, \(A[B/X]\) is also proper in \(X\).

\Case{\(A = \fix{Y}C\) for some \(Y\) and \(C\).}
We can assume that \(Y \not\in \{\,X\,\} \cup \FTV{B}\)
without loss of generality.
That is, \(A[B/X] = \fix{Y}{C[B/X]}\).
Note that \(C\) is also proper in \(X\), since \(A\) is not a {\tvariant}.
Therefore, \(C[B/X]\) is proper in \(X\) by induction hypothesis;
and hence, so is \(A[B/X]\) by Definition~\ref{proper-def}.
\fi 
\qed\CHECKED{2014/07/11}
\end{proof}

\begin{definition}
    \ilabel{depth-def}{depth}
    \ilabel*{dp@$\protect\PNOdp(A), \protect\PNIdp(A)$}
Let \(X\) be a type variable, and \(A\) a type expression.
The {\em positive \(\O\)-depth} \(\POdp(A,\,X)\) and
the {\em negative \(\O\)-depth} \(\NOdp(A,\,X)\)
of \(X\) in \(A\) are defined as follows.
\begin{Eqnarray*}
\def\Infty{\hbox to 70mm{\(\infty\)}}
\PNOdp(A,\,X) &=& \infty & \((\mbox{\(A\) is a {\tvariant}})\) \\
\POdp(X,\,X) &=& 0 & \\
\NOdp(X,\,X) &=& \infty & \\
\PNOdp(Y,\,X) &=& \infty & \((X \not= Y)\) \\
\PNOdp(\O A,\,X) &=& \PNOdp(A,\,X)+1 &
    \((\mbox{\(\,\O A\) is not a {\tvariant}})\) \\
\PNOdp(A \Impl B,\,X) &=& \min(\NPOdp(A,\,X),\,\PNOdp(B,\,X))
    & \((\mbox{\(A \Impl B\) is not a {\tvariant}})\) \\
\PNOdp(\fix{Y}A,\,X) &=&
    \min(\PNOdp(A,\,X),\,\NOdp(A,\,Y) + \NPOdp(A,\,X))\hskip-100pt \\
    && \span\hfill
	$(\mbox{\(X \not= Y\) and \(\fix{Y}A\) is not a {\tvariant}})$
\end{Eqnarray*}
Similarly, the {\em positive} (respectively, {\em negative})
{\em \(\Impl\)-depth} \(\PNIdp(A,\,X)\) of \(X\) in \(A\)
is defined as follows.
\begin{Eqnarray*}
\PNIdp(A,\,X) &=& \infty & \((\mbox{\(A\) is a {\tvariant}})\) \\
\PIdp(X,\,X) &=& 0 \\
\NIdp(X,\,X) &=& \infty \\
\PNIdp(Y,\,X) &=& \infty & \((X \not= Y)\)\\
\PNIdp(\O A,\,X) &=& \PNIdp(A,\,X) &
    \((\mbox{\(\,\O A\) is not a {\tvariant}})\) \\
\PNIdp(A \Impl B,\,X) &=& \min(\NPIdp(A,\,X),\,\PNIdp(B,\,X)) + 1
    & \((\mbox{\(A \Impl B\) is not a {\tvariant}})\) \\
\PNIdp(\fix{Y}A,\,X) &=&
    \min(\PNIdp(A,\,X),\,\NIdp(A,\,Y) + \NPIdp(A,\,X)) \hskip-100pt \\
    && \span\hfill
	$(\mbox{\(X \not= Y\) and \(\fix{Y}A\) is not a {\tvariant}})$
\end{Eqnarray*}
The domains of \({+}\) and \(\min\) are assumed to be naturally extended
to \(\{\, 0,\,1,\,2,\,\ldots,\,\infty \,\}\).
It can be easily verified that \(\alpha\)-conversion does not affect
the definition of \(\PNOdp(A,\,X)\) or \(\PNIdp(A,\,X)\).
We also define the {\em \(\O\)-depth} \(\Odp(A,\,X)\) and
{\em \(\Impl\)-depth} \(\Idp(A,\,X)\) of \(X\) in \(A\) as follows.
\begin{eqnarray*}
\Odp(A,\,X) &=& \min(\POdp(A,\,X),\,\NOdp(A,\,X)) \\
\Idp(A,\,X) &=& \min(\PIdp(A,\,X),\,\NIdp(A,\,X))
\end{eqnarray*}
\end{definition}
For example, let \(A = \fix{X}\,\O(X \Impl \O Y) \Impl Z\) and
\(B = \fix{X}\,\O(X \Impl Y \Impl Z)\).
Then, \(\POdp(A,\,Y) = \PIdp(A,\,Y) = 2\),
\(\NOdp(A,\,Y) = \NIdp(A,\,Y) = \infty\),
\(\POdp(B,\,Y) = 2\), \(\PIdp(B,\,Y) = 3\),
\(\NOdp(B,\,Y) = 1\), and \(\NIdp(B,\,Y) = 2\).

\begin{proposition}\label{depth-basic}\pushlabel
Let \(\Dp\) be either \(\Odp\) or \(\Idp\), and suppose that
\(B\) is not a {\tvariant}.
\begin{Enumerate}
\item \itemlabel{depth-finite-etv}
    \(\PNDp(A,\,X) < \infty\) if and only if \(X \in \PNETV{A}\).
\item \itemlabel{depth-subst0}
    If \(X \not= Y\) and \(Y \not\in \ETV{B}\), then
	\(\PNDp(A[B/X],\,Y) = \PNDp(A,\,Y)\).
\item \itemlabel{depth-subst1}
    If \(X \not= Y\), then \(\PNDp(A[B/X],\,Y) = \min(\PNDp(A,\,Y),\,
	 \PDp(A,\,X) + \PNDp(B,\,Y),\,\NDp(A,\,X) + \NPDp(B,\,Y))\).
\item \itemlabel{depth-subst2}
    \(\PNDp(A[B/X],\,X)
	= \min(\PDp(A,\,X) + \PNDp(B,\,X),\,\NDp(A,\,X) + \NPDp(B,\,X))\).
\item \itemlabel{O-depth-proper}
    \(A\) is proper in \(X\) if and only if \(\Odp(A,\,X) > 0\).
\item \itemlabel{etv-proper}
    If \(X \not\in \ETV{A}\), then \(A\) is proper in \(X\).
\item \itemlabel{NIdp-positive}
    \(\NIdp(A,\,X) > 0\).
\item \itemlabel{I-depth-zero}
    If \(\PIdp(A,\,X) = 0\), then \(A =
	\O^{m_0}\fix{Y_1}\O^{m_1}\fix{Y_2}\O^{m_2}\ldots\fix{Y_n}\O^{m_n}X\)
    for some \(n\), \(m_0\), \(m_1\), \(m_2\), \(\ldots\), \(m_n\),
    \(Y_1\), \(Y_2\), \(\ldots\), \(Y_n\)
    such that \(X \not= Y_i\) for any \(i \in \{\,1,\,2,\,\ldots,\,n\,\}\).
\item \itemlabel{PIdp-positive}
    If \(A\) is proper in \(X\) and \(\fix{X}A\) is not a {\tvariant},
    then \(\PIdp(A,\,X) > 0\).
\end{Enumerate}
\end{proposition}
\begin{proof}
For each item, the proof proceeds
by induction on \(h(A)\), and by cases on the form of \(A\).
Note that all the statements are almost trivial if \(A\) is a {\tvariant}.
Use Item~\itemref{depth-finite-etv}, Propositions~\ref{tvariant-subst1}
and \ref{tvariant-subst2}
for Items~\itemref{depth-subst0},
\itemref{depth-subst1} and \itemref{depth-subst2}.
Use also Item~\itemref{depth-subst1} for \itemref{depth-subst2}.
Use Items~\itemref{depth-finite-etv}, \itemref{NIdp-positive} and
Proposition~\ref{tvariant-etv} for Item~\itemref{I-depth-zero},
from which Item~\itemref{PIdp-positive} immediately follows.
\ifdetail

\paragraph{Proof of \protect\itemref{depth-finite-etv}}
If \(A\) is a {\tvariant}, then
\(\PNETV{A} = \{\}\) and \(\PNDp(A,\,X) = \infty\) by Definitions~\ref{etv-def}
and \ref{depth-def}.
Hence, we assume that \(A\) is not below.

\Case{\(A = Y\) for some \(Y\).}
By Definition~\ref{etv-def},
\(\PETV{A} = \{\,Y\,\}\) and \(\NETV{A} = \{\}\).
On the other hand,
if \(Y = X\), then \(\PDp(A,\,X) = 0\) and \(\NDp(A,\,X) = \infty\).
Otherwise, \(\PNDp(A,\,X) = \infty\).
Hence, \(\PNDp(A,\,X) < \infty\) iff \(X \in \PNETV{A}\).

\Case{\(A = \O B\) for some \(B\).}
Obvious by induction hypothesis, since by definition,
\(\PNETV{A,\,X} = \PNETV{B,\,X}\) and
\(\PNDp(A,\,X) < \infty\) iff \(\PNDp(B,\,X) < \infty\).

\Case{\(A = B \Impl C\) for some \(B\) and \(C\).}
Obvious again by induction hypothesis, since by definition,
\(\PNETV{A,\,X} = \NPETV{B,\,X} \cup \PNETV{C,\,X}\) and
\(\PNDp(A,\,X) < \infty\) iff \(\min(\NPDp(B,\,X),\,\PNDp(C,\,X)) < \infty\).

\Case{\(A = \fix{Y}B\) for some \(Y\) and \(B\).}
We can assume that \(Y \not= X\) without loss of generality.
By Definition~\ref{depth-def},
\begin{eqnarray}
    \label{depth-finite-etv-01}
    \PNDp(A,\,X) &=& \min(\PNDp(B,\,X),\,\NDp(B,\,Y) + \NPDp(B,\,X)).
\end{eqnarray}
If \(Y \in \NETV{B}\), then \(\NDp(B,\,Y) < \infty\) by induction hypothesis.
Hence, by (\ref{depth-finite-etv-01}),
\begin{eqnarray*}
    \PNDp(A,\,X) < \infty &~\mbox{iff}~&
	\min(\PNDp(B,\,X),\,\NPDp(B,\,X)) < \infty, ~\mbox{and} \\
    \PNETV{A} &=& (\PNETV{B} \cup \NPETV{B}) - \{\,Y\,\}
\end{eqnarray*}
by Definition~\ref{etv-def}.
Therefore, since \(X \not= Y\).
\(\PNDp(A,\,X) < \infty\) iff \(X \in \PNETV{A}\)
by induction hypothesis.
On the other hand, if \(Y \not\in \NETV{B}\), then
\(\NDp(B,\,Y) = \infty\) by induction hypothesis.
Hence, by (\ref{depth-finite-etv-01}),
\begin{eqnarray}
    \PNDp(A,\,X) < \infty &~\mbox{iff}~& \PNDp(B,\,X) < \infty,~\mbox{and} \\
    \PNETV{A} &=& \PNETV{B} - \{\,Y\,\}
\end{eqnarray}
by Definition~\ref{etv-def}.
Therefore, \(\PNDp(A,\,X) < \infty\) iff \(X \in \PNETV{A}\)
again by induction hypothesis.

\paragraph{Proof of \protect\itemref{depth-subst0}}
Suppose that \(B\) is not a {\tvariant},
\(X \not= Y\) and \(Y \not\in \ETV{B}\).
Note that \(A[B/X]\) is a {\tvariant} if and only if
so is \(A\) by Propositions~\ref{tvariant-subst1} and \ref{tvariant-subst2}.
Hence, if either is a {\tvariant},
then \(\PNDp(A[B/X],\,Y) = \PNDp(A,\,Y) = \infty\)
by Definition~\ref{depth-def}.
So we only consider the case when neither is a {\tvariant}.

\Case{\(A = Z\) for some \(Z\).}
If \(Z \not= X\), then trivial since \(A[B/X] = Z[B/X] = Z = A\).
If \(Z = X\), then \(A[B/X] = X[B/X] = B\).
Since \(Y \not\in \ETV{A}\) from \(A = X \not= Y\),
and since \(Y \not\in \ETV{B}\),
we get
\(\PNDp(A,\,Y) = \PNDp(B,\,Y) = \infty\)
by Item~\itemref{depth-finite-etv} of this proposition.

\Case{\(A = \O C\) for some \(C\).}
By Definition~\ref{depth-def} and induction hypothesis,
\(\PNOdp(A[B/X],\,Y) = \PNOdp(C[B/X],\,Y) + 1
    = \PNOdp(C,\,Y) + 1 = \PNOdp(A,\,Y)\) and
\(\PNIdp(A[B/X],\,Y) = \PNIdp(C[B/X],\,Y) = \PNIdp(C,\,Y) = \PNIdp(A,\,Y)\).

\Case{\(A = C \Impl D\) for some \(C\) and \(D\).}
Similarly,
\begin{Eqnarray*}
\PNOdp(A[B/X],\,Y) &=& \PNOdp(C[B/X] \Impl D[B/X],\,Y) \\
    &=& \min(\NPOdp(C[B/X],\,Y),\,\PNOdp(D[B/X],\,Y))
    	& (by Definition~\ref{depth-def}) \\
    &=& \min(\NPOdp(C,\,Y),\,\PNOdp(D,\,Y))
    	& (by induction hypothesis) \\
    &=& \PNOdp(A,\,Y)
    	& (by Definition~\ref{depth-def}), and \\
\PNIdp(A[B/X],\,Y) &=& \PNIdp(C[B/X] \Impl D[B/X],\,Y) \\
    &=& \min(\NPIdp(C[B/X],\,Y),\,\PNIdp(D[B/X],\,Y)) + 1
    	& (by Definition~\ref{depth-def}) \\
    &=& \min(\NPIdp(C,\,Y),\,\PNIdp(D,\,Y)) + 1
    	& (by induction hypothesis) \\
    &=& \PNIdp(A,\,Y)
    	& (by Definition~\ref{depth-def}).
\end{Eqnarray*}

\Case{\(A = \fix{Z}C\) for some \(Z\) and \(C\).}
We can assume that \(Z \not\in \{\,X,\,Y\,\}\cup\FTV{B}\)
without loss of generality.
\begin{Eqnarray*}
\PNDp(A[B/X],\,Y) = \PNDp(\fix{Z}C[B/X],\,Y) \mskip-300mu && \\
    &=& \min(\PNDp(C[B/X],\,Y),\,\NDp(C[B/X],\,Z) + \NPDp(C[B/X],\,Y))
    	& (by Definition~\ref{depth-def}) \\
    &=& \min(\PNDp(C,\,Y),\,\NDp(C,\,Z) + \NPDp(C,\,Y))
    	& (by induction hypothesis) \\
    &=& \PNDp(A,\,Y)
    	& (by Definition~\ref{depth-def}).
\end{Eqnarray*}

\paragraph{Proof of \protect\itemref{depth-subst1}}
Suppose that \(B\) is not a {\tvariant} and \(X \not= Y\).
If \(A\) is a {\tvariant}, then so is \(A[B/X]\)
by Proposition~\ref{tvariant-subst1}; and hence,
\(\PNDp(A[B/X],\,Y) = \infty\) and
\(\min(\PNDp(A,\,Y),\,\PDp(A,\,X) + \PNDp(B,\,Y),\,\NDp(A,\,X) + \NPDp(B,\,Y))
    = \min(\infty,\,\infty + \PNDp(B,\,Y),\,\infty + \NPDp(B,\,Y))
    = \infty\)
by Definition~\ref{depth-def}.
Therefore, we assume that \(A\) is not a {\tvariant} below, which implies
that \(A[B/X]\) is not a {\tvariant} either,
by Proposition~\ref{tvariant-subst2}.

\Case{\(A = Z\) for some \(Z\).}
If \(Z \not\in \{\,X,\,Y\,\}\), then
\(\PNDp(A[B/X],\,Y) = \PNDp(Z,\,Y) = \infty\), and
\begin{Eqnarray*}
    \Hbox{20pt}{\(\min(\PNDp(A,\,Y),\,
	\PDp(A,\,X) + \PNDp(B,\,Y),\,\NDp(A,\,X) + \NPDp(B,\,Y))\)} \\
    &=& \min(\PNDp(Z,\,Y),\,
	    \PDp(Z,\,X) + \PNDp(B,\,Y),\,\NDp(Z,\,X) + \NPDp(B,\,Y)) \\
    &=& \min(\infty,\,\infty + \PNDp(B,\,Y),\,\infty + \NPDp(B,\,Y)) \\
    &=& \infty.
\end{Eqnarray*}
If \(Z = X\), then
\(\PNDp(A[B/X],\,Y) = \PNDp(B,\,Y)\), and
\begin{Eqnarray*}
    \Hbox{20pt}%
	{\(\min(\PNDp(A,\,Y),\,
	    \PDp(A,\,X) + \PNDp(B,\,Y),\,\NDp(A,\,X) + \NPDp(B,\,Y))\)} \\
    &=& \min(\PNDp(X,\,Y),\,
	    \PDp(X,\,X) + \PNDp(B,\,Y),\,\NDp(X,\,X) + \NPDp(B,\,Y)) \\
    &=& \min(\infty,\,0 + \PNDp(B,\,Y),\,\infty + \NPDp(B,\,Y)) \\
    &=& \PNDp(B,\,Y).
\end{Eqnarray*}
Finally, if \(Z = Y\), then
\(\PNDp(A[B/X],\,Y) = \PNDp(Y,\,Y)\), and
\begin{Eqnarray*}
    \Hbox{20pt}{\(\min(\PNDp(A,\,Y),\,
	\PDp(A,\,X) + \PNDp(B,\,Y),\,\NDp(A,\,X) + \NPDp(B,\,Y))\)} \\
    &=& \min(\PNDp(Y,\,Y),\,
	    \PDp(Y,\,X) + \PNDp(B,\,Y),\,\NDp(Y,\,X) + \NPDp(B,\,Y)) \\
    &=& \min(\PNDp(Y,\,Y),\,\infty + \PNDp(B,\,Y),\,\infty + \NPDp(B,\,Y)) \\
    &=& \PNDp(Y,\,Y).
\end{Eqnarray*}

\Case{\(A = \O C\) for some \(C\).}
\begin{Eqnarray*}
\Hbox{20pt}{\(\PNOdp(A[B/X],\,Y) = \PNOdp(\O (C[B/X]),\,Y)\)} \\
    &=& \PNOdp(C[B/X],\,Y) + 1
	& (by Definition~\ref{depth-def}) \\
    &=& \min(\PNOdp(C,\,Y),\,
	    \POdp(C,\,X) + \PNOdp(B,\,Y),\,\NOdp(C,\,X) + \NPOdp(B,\,Y)) + 1
	    \mskip-500mu \\
	&& \span\hfill (by induction hypothesis) \\
    &=& \min(\PNOdp(C,\,Y) + 1,\,\POdp(C,\,X) + 1 + \PNOdp(B,\,Y),\,
	    \NOdp(C,\,X) + 1 + \NPOdp(B,\,Y)) \mskip-150mu \\
    &=& \min(\PNOdp(A,\,Y),\,
	    \POdp(A,\,X) + \PNOdp(B,\,Y),\,\NOdp(A,\,X) + \NPOdp(B,\,Y))
	    \mskip-500mu \\
	&& \span\hfill (by Definition~\ref{depth-def}) \\[3pt]
\Hbox{20pt}{\(\PNIdp(A[B/X],\,Y) = \PNIdp(\O (C[B/X]),\,Y)\)} \\
    &=& \PNIdp(C[B/X],\,Y)
	& (by Definition~\ref{depth-def}) \\
    &=& \min(\PNIdp(C,\,Y),\,
	    \PIdp(C,\,X) + \PNIdp(B,\,Y),\,\NIdp(C,\,X) + \NPIdp(B,\,Y))
	    \mskip-500mu \\
	&& \span\hfill (by induction hypothesis) \\
    &=& \min(\PNIdp(A,\,Y),\,
	    \PIdp(A,\,X) + \PNIdp(B,\,Y),\,\NIdp(A,\,X) + \NPIdp(B,\,Y))
	    \mskip-500mu \\
	&& \span\hfill (by Definition~\ref{depth-def})
\end{Eqnarray*}

\Case{\(A = C \Impl D\) for some \(C\) and \(D\).}
\begin{Eqnarray*}
\Hbox{20pt}{\(\PNOdp(A[B/X],\,Y) = \PNOdp(C[B/X] \Impl D[B/X],\,Y)\)} \\
    &=& \min(\NPOdp(C[B/X],\,Y),\,\PNOdp(D[B/X],\,Y)) \\
    &=& \min(\min(\NPOdp(C,\,Y),\, \POdp(C,\,X) + \NPOdp(B,\,Y),\,
		\NOdp(C,\,X) + \PNOdp(B,\,Y)) \mskip-500mu \\
	&& \hphantom{\!\min(}
	    \min(\PNOdp(D,\,Y),\,
		\POdp(D,\,X) + \PNOdp(B,\,Y),\,\NOdp(D,\,X) + \NPOdp(B,\,Y)))
		\mskip-500mu \\
	&& \span\hfill (by induction hypothesis) \\
    &=& \min(\min(\NPOdp(C,\,Y),\,\PNOdp(D,\,Y)),\, \\
	&& \hphantom{\!\min(}
	    \min(\NOdp(C,\,X),\,\POdp(D,\,X)) + \PNOdp(B,\,Y), \\
	&& \hphantom{\!\min(}
	    \min(\POdp(C,\,X),\,\NOdp(D,\,X)) + \NPOdp(B,\,Y)) \\
    &=& \min(\PNOdp(A,\,Y),\,
	    \POdp(A,\,X) + \PNOdp(B,\,Y),\,\NOdp(A,\,X) + \NPOdp(B,\,Y))
	& (by Definition~\ref{depth-def}) \\[3pt]
\Hbox{20pt}{\(\PNIdp(A[B/X],\,Y) = \PNIdp(C[B/X] \Impl D[B/X],\,Y)\)} \\
    &=& \min(\NPIdp(C[B/X],\,Y),\,\PNIdp(D[B/X],\,Y)) + 1 \\
    &=& \min(\min(\NPIdp(C,\,Y),\, \PIdp(C,\,X) + \NPIdp(B,\,Y),\,
		\NIdp(C,\,X) + \PNIdp(B,\,Y)) \mskip-500mu \\
	&& \hphantom{\!\min(}
	    \min(\PNIdp(D,\,Y),\,
		\PIdp(D,\,X) + \PNIdp(B,\,Y),\,\NIdp(D,\,X) + \NPIdp(B,\,Y)))
		+ 1
		\mskip-500mu \\
	&& \span\hfill (by induction hypothesis) \\
    &=& \min(\min(\NPIdp(C,\,Y),\,\PNIdp(D,\,Y)) + 1,\, \\
	&& \hphantom{\!\min(}
	    \min(\NIdp(C,\,X),\,\PIdp(D,\,X)) + 1 + \PNIdp(B,\,Y), \\
	&& \hphantom{\!\min(}
	    \min(\PIdp(C,\,X),\,\NIdp(D,\,X)) + 1 + \NPIdp(B,\,Y)) \\
    &=& \min(\PNIdp(A,\,Y),\,
	    \PIdp(A,\,X) + \PNIdp(B,\,Y),\,\NIdp(A,\,X) + \NPIdp(B,\,Y))
	& (by Definition~\ref{depth-def})
\end{Eqnarray*}

\Case{\(A = \fix{Z}C\) for some \(Z\) and \(C\).}
We can assume that \(Z \not\in \{\,X,\,Y\,\}\cup\FTV{B}\)
without loss of generality.
\begin{Eqnarray*}
\Hbox{20pt}{\(\PNDp(A[B/X],\,Y) = \PNDp(\fix{Z}C[B/X],\,Y)\)} \\
    &=& \min(\PNDp(C[B/X],\,Y),\,\NDp(C[B/X],\,Z) + \NPDp(C[B/X],\,Y))
    	& (by Definition~\ref{depth-def}) \\
    &=& \min(\PNDp(C[B/X],\,Y),\\
	&& \hphantom{\!\min(}
	    \min(\NDp(C,\,Z),\,
		 \PDp(C,\,X) + \NDp(B,\,Z),\,\NDp(C,\,X) + \PDp(B,\,Z))
	     \mskip-500mu \\
	&& \span$\hphantom{\!\min(~}
		+ \NPDp(C[B/X],\,Y))$
	    \hfill (by induction hypothesis) \\
    &=& \min(\PNDp(C[B/X],\,Y), \\
	&& \hphantom{\!\min(}
	    \min(\NDp(C,\,Z),\,\PDp(C,\,X) + \infty,\,\NDp(C,\,X) + \infty)
	     \mskip-500mu \\
	&& \span$\hphantom{\!\min(~}
		+ \NPDp(C[B/X],\,Y))$
	    \hfill (by Item~\itemref{depth-finite-etv}
		    and \(Z \not\in \ETV{B}\)) \\
    &=& \min(\PNDp(C[B/X],\,Y),
	    \NDp(C,\,Z) + \NPDp(C[B/X],\,Y)) \mskip-500mu \\
    &=& \min(\min(\PNDp(C,\,Y),\,
		\PDp(C,\,X) + \PNDp(B,\,Y),\,\NDp(C,\,X) + \NPDp(B,\,Y)),
		\mskip-500mu \\
	&& \hphantom{\!\min(}
	    \NDp(C,\,Z) + \min(\NPDp(C,\,Y),\,
		\PDp(C,\,X) + \NPDp(B,\,Y),\,\NDp(C,\,X) + \PNDp(B,\,Y)))
	    \mskip-500mu \\
    	&&& (by induction hypothesis) \\
    &=& \min(\PNDp(C,\,Y),\,\NDp(C,\,Z) + \NPDp(C,\,Y), \\
	&& \hphantom{\!\min(}
	    \PDp(C,\,X) + \PNDp(B,\,Y),\,
		\NDp(C,\,Z) + \NDp(C,\,X) + \PNDp(B,\,Y), \mskip-500mu \\
	&& \hphantom{\!\min(}
	    \NDp(C,\,X) + \NPDp(B,\,Y),\,
		\NDp(C,\,Z) + \PDp(C,\,X) + \NPDp(B,\,Y)) \mskip-500mu \\
    &=& \min(\min(\PNDp(C,\,Y),\,\NDp(C,\,Z) + \NPDp(C,\,Y)), \\
	&& \hphantom{\!\min(}
	    \min(\PDp(C,\,X),\,\NDp(C,\,Z) + \NDp(C,\,X)) + \PNDp(B,\,Y)),
		\mskip-500mu \\
	&& \hphantom{\!\min(}
	    \min(\NDp(C,\,X),\,\NDp(C,\,Z) + \PDp(C,\,X)) + \NPDp(B,\,Y)))
		\mskip-500mu \\
    &=& \min(\PNDp(A,\,Y),\,
	    \PDp(A,\,X) + \PNDp(B,\,Y),\,
	    \NDp(A,\,X) + \NPDp(B,\,Y)) \mskip-500mu \\
	&&& (by Definition~\ref{depth-def})
\end{Eqnarray*}

\paragraph{Proof of \protect\itemref{depth-subst2}}
Suppose that \(B\) is not a {\tvariant}.
If \(A\) is a {\tvariant}, then so is \(A[B/X]\)
by Proposition~\ref{tvariant-subst1}; and hence,
\(\PNDp(A[B/X],\,X) = \infty\) and
\(\min(\PDp(A,\,X) + \PNDp(B,\,X),\,\NDp(A,\,X) + \NPDp(B,\,X))
    = \min(\infty + \PNDp(B,\,X),\,\infty + \NPDp(B,\,X))
    = \infty\)
by Definition~\ref{depth-def}.
Therefore, we assume that \(A\) is not a {\tvariant} below, which also implies
that \(A[B/X]\) is not either, by Proposition~\ref{tvariant-subst2}.

\Case{\(A = Y\) for some \(Y\).}
If \(Y = X\), then \(\PNDp(A[B/X],\,X) = \PNDp(B,\,X)\), and
\begin{Eqnarray*}
    \Hbox{20pt}%
	{\(\min(\PDp(A,\,X) + \PNDp(B,\,X),\,\NDp(A,\,X) + \NPDp(B,\,X))\)} \\
    &=& \min(\PDp(X,\,X) + \PNDp(B,\,X),\,\NDp(X,\,X) + \NPDp(B,\,X)) \\
    &=& \min(0 + \PNDp(B,\,X),\,\infty + \NPDp(B,\,X)) \\
    &=& \PNDp(B,\,X).
\end{Eqnarray*}
Otherwise, i.e., if \(Y \not= X\), then
\(\PNDp(A[B/X],\,X) = \PNDp(Y,\,X) = \infty\), and
\begin{Eqnarray*}
    \Hbox{20pt}%
	{\(\min(\PDp(A,\,X) + \PNDp(B,\,X),\,\NDp(A,\,X) + \NPDp(B,\,X))\)} \\
    &=& \min(\PDp(Y,\,X) + \PNDp(B,\,X),\,\NDp(Y,\,X) + \NPDp(B,\,X)) \\
    &=& \min(\infty + \PNDp(B,\,X),\,\infty + \NPDp(B,\,X)) \\
    &=& \infty.
\end{Eqnarray*}

\Case{\(A = \O C\) for some \(C\).}
\begin{Eqnarray*}
\Hbox{20pt}{\(\PNOdp(A[B/X],\,X) = \PNOdp(\O (C[B/X]),\,X)\)} \\
    &=& \PNOdp(C[B/X],\,X) + 1
	& (by Definition~\ref{depth-def}) \\
    &=& \min(\POdp(C,\,X) + \PNOdp(B,\,X),\,\NOdp(C,\,X) + \NPOdp(B,\,X)) + 1
	& (by induction hypothesis) \\
    &=& \min(\POdp(C,\,X) + 1 + \PNOdp(B,\,X),\,
	    \NOdp(C,\,X) + 1 + \NPOdp(B,\,X)) \mskip-100mu \\
    &=& \min(\POdp(A,\,X) + \PNOdp(B,\,X),\,\NOdp(A,\,X) + \NPOdp(B,\,X))
	& (by Definition~\ref{depth-def}) \\[3pt]
\Hbox{20pt}{\(\PNIdp(A[B/X],\,X) = \PNIdp(\O (C[B/X]),\,X)\)} \\
    &=& \PNIdp(C[B/X],\,X)
	& (by Definition~\ref{depth-def}) \\
    &=& \min(\PIdp(C,\,X) + \PNIdp(B,\,X),\,\NIdp(C,\,X) + \NPIdp(B,\,X))
	& (by induction hypothesis) \\
    &=& \min(\PIdp(A,\,X) + \PNIdp(B,\,X),\,\NIdp(A,\,X) + \NPIdp(B,\,X))
	& (by Definition~\ref{depth-def})
\end{Eqnarray*}

\Case{\(A = C \Impl D\) for some \(C\) and \(D\).}
\begin{Eqnarray*}
\Hbox{20pt}{\(\PNOdp(A[B/X],\,X) = \PNOdp(C[B/X] \Impl D[B/X],\,X)\)} \\
    &=& \min(\NPOdp(C[B/X],\,X),\,\PNOdp(D[B/X],\,X))
	& (by Definition~\ref{depth-def}) \\
    &=& \min(\min(\POdp(C,\,X) + \NPOdp(B,\,X),\,
		\NOdp(C,\,X) + \PNOdp(B,\,X)) \\
	&& \hphantom{\!\min(}
	    \min(\POdp(D,\,X) + \PNOdp(B,\,X),\,\NOdp(D,\,X) + \NPOdp(B,\,X)))
	& (by induction hypothesis) \\
    &=& \min(\min(\NOdp(C,\,X),\,\POdp(D,\,X)) + \PNOdp(B,\,X), \\
	&& \hphantom{\!\min(}
	    \min(\POdp(C,\,X),\,\NOdp(D,\,X)) + \NPOdp(B,\,X)) \\
    &=& \min(\POdp(A,\,X) + \PNOdp(B,\,X),\,\NOdp(A,\,X) + \NPOdp(B,\,X))
	& (by Definition~\ref{depth-def}) \\[3pt]
\Hbox{20pt}{\(\PNIdp(A[B/X],\,X) = \PNIdp(C[B/X] \Impl D[B/X],\,X)\)} \\
    &=& \min(\NPIdp(C[B/X],\,X),\,\PNIdp(D[B/X],\,X)) + 1
	& (by Definition~\ref{depth-def}) \\
    &=& \min(\min(\PIdp(C,\,X) + \NPIdp(B,\,X),\,
		\NIdp(C,\,X) + \PNIdp(B,\,X)) \\
	&& \hphantom{\!\min(}
	    \min(\PIdp(D,\,X) + \PNIdp(B,\,X),\,\NIdp(D,\,X) + \NPIdp(B,\,X)))
		+ 1 \mskip-200mu \\
	&& \span\hfill (by induction hypothesis) \\
    &=& \min(\min(\NIdp(C,\,X),\,\PIdp(D,\,X)) + 1 + \PNIdp(B,\,X), \\
	&& \hphantom{\!\min(}
	    \min(\PIdp(C,\,X),\,\NIdp(D,\,X)) + 1 + \NPIdp(B,\,X)) \\
    &=& \min(\PIdp(A,\,X) + \PNIdp(B,\,X),\,\NIdp(A,\,X) + \NPIdp(B,\,X))
	& (by Definition~\ref{depth-def})
\end{Eqnarray*}

\Case{\(A = \fix{Y}C\) for some \(Y\) and \(C\).}
We can assume that \(Y \not\in \{\,X\,\}\cup\FTV{B}\)
without loss of generality.
\begin{Eqnarray*}
\Hbox{20pt}{\(\PNDp(A[B/X],\,X) = \PNDp(\fix{Y}C[B/X],\,X)\)} \\
    &=& \min(\PNDp(C[B/X],\,X),\,\NDp(C[B/X],\,Y) + \NPDp(C[B/X],\,X))
    	& (by Definition~\ref{depth-def}) \\
    &=& \min(\PNDp(C[B/X],\,X),\\
	&& \hphantom{\!\min(}
	    \min(\NDp(C,\,Y),\,
		 \PDp(C,\,X) + \NPDp(B,\,Y),\,\NDp(C,\,X) + \PNDp(B,\,Y))
	     \mskip-500mu \\
	&& \hphantom{\!\min(\quad}
		+ \NPDp(C[B/X],\,X))
	    & (by Item~\itemref{depth-subst1}) \\
    &=& \min(\PNDp(C[B/X],\,X), \\
	&& \hphantom{\!\min(}
	    \min(\NDp(C,\,Y),\,\PDp(C,\,X) + \infty,\,\NDp(C,\,X) + \infty)
	     \mskip-500mu \\
	&& \span$\hphantom{\!\min(\mskip3mu}
		+ \NPDp(C[B/X],\,X))$
	    \hfill (by Item~\itemref{depth-finite-etv}
		    and \(Y \not\in \ETV{B}\)) \\
    &=& \min(\PNDp(C[B/X],\,X),
	    \NDp(C,\,Y) + \NPDp(C[B/X],\,X)) \mskip-500mu \\
    &=& \min(\min(\PDp(C,\,X) + \PNDp(B,\,X),\,\NDp(C,\,X) + \NPDp(B,\,X)), \\
	&& \hphantom{\!\min(}
	    \NDp(C,\,Y)
		+ \min(\PDp(C,\,X) + \NPDp(B,\,X),\,\NDp(C,\,X) + \PNDp(B,\,X)))
	    \mskip-500mu \\
    	&&& (by induction hypothesis) \\
    &=& \min(\PDp(C,\,X) + \PNDp(B,\,X),\\
	&& \hphantom{\!\min(}
	    \NDp(C,\,X) + \NPDp(B,\,X),\\
	&& \hphantom{\!\min(}
	    \NDp(C,\,Y) + \PDp(C,\,X) + \NPDp(B,\,X), \\
	&& \hphantom{\!\min(}
	    \NDp(C,\,Y) + \NDp(C,\,X) + \PNDp(B,\,X)) \\
    &=& \min(\min(\PDp(C,\,X),\,\NDp(C,\,Y)+\NDp(C,\,X))+ \PNDp(B,\,X),\\
	&& \hphantom{\!\min(}
	    \min(\NDp(C,\,X),\,\NDp(C,\,Y)+\PDp(C,\,X))+ \NPDp(B,\,X))\\
    &=& \min(\PDp(A,\,X) + \PNDp(B,\,X),\,\NDp(A,\,X) + \NPDp(B,\,X))
    	& (by Definition~\ref{depth-def})
\end{Eqnarray*}

\paragraph{Proof of \protect\itemref{O-depth-proper}}
If \(A\) is a {\tvariant}, then for every \(X\),
\(A\) is proper in \(X\) and \(\PNOdp(A,\,X) = \infty\).
Therefore, we assume that \(A\) is not a {\tvariant} below.

\Case{\(A = Y\) for some \(Y\).}
If \(Y = X\), then
\(A\) is not proper in \(X\) and \(\POdp(A,\,X) = \POdp(X,\,X) = 0\),
that is, \(\Odp(A,\,X) = 0\).
Otherwise, i.e., if \(Y \not= X\), then
\(A\) is proper in \(X\) and \(\PNOdp(A,\,X) = \PNOdp(Y,\,X) = \infty\),
that is, \(\Odp(A,\,X) = \infty > 0\).

\Case{\(A = \O C\) for some \(C\).}
In this case, \(A\) is proper in \(X\) and
\(\PNOdp(A,\,X) = \PNOdp(\O C,\,X) = \PNOdp(C) + 1 > 0\).

\Case{\(A = B \Impl C\) for some \(B\) and \(C\).}
Note that \(C\) is not a {\tvariant} since \(A\) is not.
Hence, \(A\) is proper in \(X\) if and only if so are both \(B\) and \(C\).
On the other hand,
\(\PNOdp(A,\,X) = \PNOdp(B \Impl C,\,X) = \min(\NPOdp(B,\,X),\,\PNOdp(C,\,X))\).
Hence, by induction hypothesis,
\(\Odp(A,\,X) > 0\)
iff \(\min(\Odp(B,\,X),\,\Odp(C,\,X)) > 0\)
iff both \(B\) and \(C\) are proper in \(X\)
iff so is \(A\).

\Case{\(A = \fix{Y}B\) for some \(Y\) and \(B\).}
We can assume that \(Y \not= X\) without loss of generality.
Hence, \(A\) is proper in \(X\) if and only if so is \(B\).
On the other hand,
\(\PNOdp(A,\,X) = \PNOdp(\fix{Y}B,\,X)
    = \min(\PNOdp(B,\,X),\,\NOdp(B,\,Y) + \NPOdp(B,\,X))\).
Since \(B\) is proper in \(Y\),
    \(\NOdp(B,\,Y) > 0\) by induction hypothesis.
Hence, by induction hypothesis again,
\(\Odp(A,\,X) > 0\) iff
\(\Odp(B,\,X) > 0\) iff
\(B\) is proper in \(X\) iff
so is \(A\).

\paragraph{Proof of \protect\itemref{etv-proper}}
Straightforward from Items~\itemref{depth-finite-etv}
and \itemref{O-depth-proper}.

\paragraph{Proof of \protect\itemref{NIdp-positive}}
If \(A\) is a {\tvariant}, then \(\NIdp(A,\,X) = \infty > 0\).
So we assume that \(A\) is not a {\tvariant} below.

\Case{\(A = Y\) for some \(Y\).}
In this case,
\(\NIdp(Y,\,X) = \infty\) by Definition~\ref{depth-def}.

\Case{\(A = \O B\) for some \(B\).}
\(\NIdp(\O B,\,X) = \NIdp(B) > 0\) by induction hypothesis.

\Case{\(A = B \Impl C\) for some \(B\) and \(C\).}
\(\NIdp(B \Impl C,\,X)
    = \min(\PIdp(B,\,X),\,\NIdp(C,\,X)) + 1 > 0\).

\Case{\(A = \fix{Y}B\) for some \(Y\) and \(B\).}
We can assume that \(Y \not= X\) without loss of generality.
\(\NIdp(\fix{Y}B,\,X)
    = \min(\NIdp(B,\,X),\,\NIdp(B,\,Y) + \PIdp(B,\,X)) > 0\),
since \(\NIdp(B,\,X) > 0\) and \(\NIdp(B,\,Y) > 0\) by induction hypothesis.

\paragraph{Proof of \protect\itemref{I-depth-zero}}
Suppose that \(\PIdp(A,\,X) = 0\).
Note that
\(A\) is not a {\tvariant} by Item~\itemref{depth-finite-etv} and
Proposition~\ref{tvariant-etv}.
We show that \(A =
	\O^{m_0}\fix{Y_1}\O^{m_1}\fix{Y_2}\O^{m_2}\ldots\fix{Y_n}\O^{m_n}X\)
for some \(n\), \(m_0\), \(m_1\), \(m_2\), \(\ldots\), \(m_n\),
\(Y_1\), \(Y_2\), \(\ldots\), \(Y_n\)
such that \(X \not= Y_i\) for any \(i \in \{\,1,\,2,\,\ldots,\,n\,\}\).

\Case{\(A = Z\) for some \(Z\).}
We get \(Z = X\) since \(\PIdp(A,\,X) = 0\).
Hence, it suffices to let \(n\) be 0.

\Case{\(A = \O B\) for some \(B\).}
Since \(\PIdp(B,\,X) = \PIdp(A,\,X) = 0\),
by induction hypothesis,
\(B = \O^{m_0'}\fix{Y_1}\O^{m_1}\fix{Y_2}\O^{m_2}\ldots\fix{Y_n}\O^{m_n}X\)
    for some \(n\), \(m_0'\), \(m_1\), \(m_2\), \(\ldots\), \(m_n\),
    \(Y_1\), \(Y_2\), \(\ldots\), \(Y_n\)
    such that \(X \not= Y_i\) for any \(i \in \{\,1,\,2,\,\ldots,\,n\,\}\).
Hence, it suffices to take \(m_0\) as \(m_0 = m_0' + 1\).

\Case{\(A = B \Impl C\) for some \(B\) and \(C\).}
This case is impossible, since
\(\PIdp(A,\,X) = \PIdp(B \Impl C,\,X)
    = \min(\NIdp(B,\,X),\,\PIdp(C,\,X)) + 1 > 0\).

\Case{\(A = \fix{Z}B\) for some \(Z\) and \(B\).}
We can assume that \(Z \not= X\) without loss of generality.
\(\PIdp(A,\,X) = \PIdp(\fix{Z}B,\,X)
    = \min(\PIdp(B,\,X),\,\NIdp(B,\,Z) + \NIdp(B,\,X))\), which implies
\(\PIdp(B,\,X) = 0\),
since \(\PIdp(A,\,X) = 0\) and since \(\NIdp(B,\,X) > 0\)
by Item~\itemref{NIdp-positive} of this proposition.
Therefore, by induction hypothesis,
\(B = \O^{m_1}\fix{Y_2}\O^{m_2}\ldots\fix{Y_n}\O^{m_n}X\)
    for some \(n\), \(m_1\), \(m_2\), \(\ldots\), \(m_n\),
    \(Y_2\), \(Y_3\), \(\ldots\), \(Y_n\)
    such that \(X \not= Y_i\) for any \(i \in \{\,2,\,3,\,\ldots,\,n\,\}\).
It suffices to take
\(m_0\) and \(Y_1\) as \(m_0 = 0\) and \(Y_1 = Z\), respectively.

\paragraph{Proof of \protect\itemref{PIdp-positive}}
Suppose that \(A\) is proper in \(X\).
If \(\PIdp(A,\,X) = 0\), then by Item~\itemref{I-depth-zero}
of this proposition,
\(A = \O^{m_0}\fix{Y_1}\O^{m_1}\fix{Y_2}\O^{m_2}\ldots\fix{Y_n}\O^{m_n}X\)
for some \(n\), \(m_0\), \(m_1\), \(m_2\), \(\ldots\), \(m_n\),
\(Y_1\), \(Y_2\), \(\ldots\), \(Y_n\)
such that \(X \not= Y_i\) for any \(i\).
This implies that \(\fix{X}A\) is a {\tvariant}.
\fi 
\qed\CHECKED{2014/07/11}
\end{proof}

\Section{Semantics of types}\label{semantics-sec}

In this section, we define the semantics of type expressions.
To fix notation, we start with some necessary definitions about
the standard untyped \(\lambda\)-calculus.

\begin{definition}[Untyped \(\lambda\)-terms]
    \ilabel{lambda-term-def}{lambda@$\lambda$-terms}
    \ilabel*{syntax!lambda@$\lambda$-terms}
The syntax of the \(\lambda\)-terms is defined relatively to a set
\(\IV\) of countably infinite {\em individual variables}
(\(f\), \(g\), \(h\), \(x\), \(y\), \(z\), \(\ldots\)).
The set \(\IE\) of {\em \(\lambda\)-terms} is defined
by the following BNF notation.
\begin{Eqnarray*}
    \IE & \bnfdef & \IV & (individual variables) \\
	& \bnfor & {\lam{\IV}{\IE}} & (\(\lambda\)-abstractions) \\
	& \bnfor & {\app{\IE}{\IE}} & (applications)
\end{Eqnarray*}
\end{definition}
We use \(M, N, K, L, \ldots\) to denote \(\lambda\)-terms.
Free and bound occurrences of individual variables and
the notion of \(\alpha\)-convertibility are defined in the standard manner.
Hereafter, we identify \(\lambda\)-terms by this \(\alpha\)-convertibility.
We denote the set of individual variables occurring freely in \(M\)
by \(\FV{M}\), and
use \(M[N_1/x_1,\ldots,N_n/x_n]\) to denote the \(\lambda\)-term
obtained from a \(\lambda\)-term \(M\)
by substituting \(N_1,\ldots,N_n\) for each free occurrence of individual
variables \(x_1,\ldots,x_n\), respectively,
with necessary \(\alpha\)-conversion to avoid accidental
capture of free variables.

\begin{definition}[\(\beta\)-reduction]
    \ilabel{ct-def}{0 beta@$\protect\ct$}
    \ilabel*{beta@$\beta$-reduction}
Let the syntax of {\rm context} \(\C[\,]\) of \(\lambda\)-term be defined
in the standard way as follows.
\[
    \def\mathrel#1{\mskip10mu#1\mskip10mu}
    \C[\,] \bnfdef [\,] \bnfor \lam{\IV}\C[\,]
	\bnfor \app{\C[\,]\,}{\IE} \bnfor \app{\IE}\,{\C[\,]}
\]
The standard notion of \(\beta\)-reduction, a
binary relation \(\ct\) over\/ \(\IE\), is defined by
\[
    \C[\app{(\lam{x}{M}}){N}] \ct \C[M[N/x]],
\]
where \(\C\) is an arbitrary context of \(\lambda\)-term.
\end{definition}
We denote the transitive and reflexive closure of \(\ct\) by \(\ctc\),
and the symmetric closure of \(\ct\) by \(\cteq\).
We define the equivalence relation \(\beq\) to be the transitive
and reflexive closure of \(\cteq\).

\begin{definition}
    \ilabel{map-subst-def}{substitution of!mappings}
    \ilabel*{0 substitution rho@$[t/x]$}
Let \(\rho\) be a mapping from a set \(T\) to a set \(S\),
and let \(x \in T\) and \(v \in S\).
We define a mapping \(\rho[v/x]\) by
\[
\rho[v/x](y) = \Choice{%
	v & (y = x) \\
	\rho(y) & (y \not= x).
    }
\]
\end{definition}

Throughout the present paper, let \(\tuple{\V,\,\cdot,\,\VI{\;}{}}\)
be a syntactical \(\lambda\)-algebra of the untyped \(\lambda\)-calculus.
Each \(\lambda\)-term \(M\) is interpreted
as an element of \(\V\), which is denoted by \(\VI{M}{\rho}\),
where \(\rho\) is {\em an individual environment}
that assigns an element of \(\V\) to each individual variable.
We define syntactical \(\lambda\)-algebra in the standard way
as follows \cite{barendregt}.

\begin{definition}[syntactical \(\lambda\)-algebra]
    \ilabel{lambda-algebra}{lambda algebra@$\lambda$-algebra}
    \ilabel*{syntactical lambda algebra@syntactical $\lambda$-algebra}
A syntactical \(\lambda\)-algebra of \(\IE\) is a tuple
\(\tuple{\V,\,\cdot,\,\VI{\;}{}}\) such that
\begin{Enumerate}
\item \(\V\) : a non-empty set.
\item \({-}\cdot{-} : \V \times \V \to \V\).
\item \(\VI{{-}}{{-}} : \IE \to (\IV \to \V) \to \V\).
\item \(\VI{x}{\ienv} = \ienv(x)\).
\item \(\VI{\app{M}N}{\ienv} = \VI{M}{\ienv} \cdot \VI{N}{\ienv}\).
\item \(\VI{\lam{x}M}{\ienv} \cdot v = \VI{M}{\ienv[v/x]}\).
\item If \(\ienv(x) = \ienv'(x)\) for every \(x \in FV(M)\),
then \(\VI{M}{\ienv} = \VI{M}{\ienv'}\).
\item If \(M \beq N\), then \(\VI{M}{\ienv} = \VI{N}{\ienv}\).
\end{Enumerate}
\end{definition}

\begin{proposition}\label{val-interp-subst}
\(\VI{M[N/x]}{\ienv} = \VI{M}{\ienv[\VI{N}{\ienv}/x]}\).
\end{proposition}
\begin{proof}
Since \(M[N/x] \beq (\lam{y}M)N\),
\(\VI{M[N/x]}{\ienv} =
\VI{[\app{(\lam{y}M)}N}{\ienv} =
\VI{\lam{y}M}{\ienv}\cdot \VI{N}{\ienv} =
\VI{M}{\ienv[\VI{N}{\ienv}/x]}\).
\qed\CHECKED{2014/04/23, 07/21}
\end{proof}

Roughly, every type expression is interpreted as a subset of \(\V\)
of a syntactical \(\lambda\)-algebra in each possible world of a
Kripke frame, where we consider the following two classes of frames.

\begin{definition}
    \ilabel{wf-frame-def}{well-founded frames}
    \ilabel*{frames!well-founded}
    \ilabel{lA-frame-def}{lambda-A frames@\protect\lA-frames}
    \ilabel*{frames!lambda-A@\protect\lA}
    \ilabel{locally-linear-def}{locally linear}
A {\em well-founded}\/ frame is a pair \(\pair{\W}{\acc}\), which consists of
a non-empty set \(\W\) of {\em possible worlds} and
an {\em accessibility relation} \({\acc}\) on \(\W\) such that
\begin{Enumerate}
\item The relation \({\acc}\) is (conversely) well-founded, i.e.,
    there is no infinite sequence such that
    \(p_0 \acc p_1 \acc p_2 \acc p_3 \acc \ldots\,\).
\end{Enumerate}
We say the accessibility relation \(\acc\) is {\em locally linear}
if and only if it satisfies the following additional condition, and
define {\em {\lA}-frames} as well-founded frames whose accessibility relation
is locally linear.
\begin{Enumerate}
\item[2.]
    If\/ \(p \acc q\), then \(p \tacc r \acc q\) for some \(r\) such
    that \(r \acc s\) implies  \(q \tacc s\) for any \(s\),
    where \(\tacc\) denotes
    the reflexive and transitive closure of\/ \(\acc\).
\end{Enumerate}
\end{definition}
The additional condition says that the accessibility relations is not
{\em locally, in a sense,} branching.
Any well-founded and linear binary relation
\(\acc\) satisfies this condition\footnote{%
    It suffices to let \(r\) be a minimal element such that \(r \acc q\).}.
For example, the set of non-negative integers, or ordinals, and
the ``greater than'' relation \(>\) constitutes a {\lA}-frame.
This condition does not always mean that
the accessibility relation is not branching.
For example, the following is a {\lA}-frame, but branching.
\begin{eqnarray*}
\W &=& \zfset{\pair{n}{m}}{n \in \{\,0,\,1,\,2\,\}~\mbox{and}~m \in \N} \\[5pt]
\pair{n}{m} \acc \pair{n'}{m'} &\;\mbox{iff}\;&
 \Choice{%
     n = 0~\mbox{and}~n' \in \{\,1,\,2\,\},~\mbox{or} \\
     n = n'~\mbox{and}~m > m'}
\end{eqnarray*}

Note that we do not impose that \(\acc\) is transitive.
It will be shown that
our interpretation of types does not depend on
whether \(\acc\) is transitive or not (Proposition~\ref{semantics-transitive}).

\begin{definition}
    \ilabel{tenv-def}{type environments}
A mapping \(t\) from \(\W\) to the power set \(\Pow{\V}\)
of \(\V\) is said to be {\em hereditary} if and only if
\[
    p \acc q ~~\mbox{implies}~~ t(p) \subseteq t(q).
\]
A mapping \(\tenv\) from \(\TV \times \W\) to \(\Pow{\V}\)
is called a {\em type environment}, and
also said to be {\em hereditary} if and only if
\[
    p \acc q ~~\mbox{implies}~~ \tenv(X, p) \subseteq \tenv(X, q)\,
    ~\mbox{for any}~X \in \TV.
\]
\end{definition}

In this paper, we only consider hereditary type environments.
Given a well-founded frame and a hereditary type environment \(\tenv\),
each type expression \(A\) is interpreted
as a hereditary mapping from \(\W\) to \(\Pow{\V}\) as follows.

\begin{definition}[Semantics of types]
    \ilabel{rlz-def}{type expressions!semantics of}
    \ilabel*{semantics!\protect\lA}
    \ilabel{I-def}{I@$\protect\I{A}^{\protect\tenv}_p$}
Let \(\pair{\W}{\acc}\) be a well-founded frame,
and \(\tenv\) a hereditary type environment.
We define a hereditary mapping \(\I{A}^{\tenv}\) from \(\W\) to \(\Pow{\V}\)
for each type expression \(A\) by extending \(\tenv\) as follows,
where we prefer to write \(\tenv(X)_p\) and
\(\I{A}^{\tenv}_p\) rather than
\(\tenv(X,\,p)\) and \(\I{A}^{\tenv}(p)\), respectively.
\begin{Eqnarray*}
\I{A}^\tenv_p &=& \V & (\(A\) is a {\tvariant}) \\
\I{X}^\tenv_p &=& \tenv(X)_p \\
\I{\O A}^\tenv_p &=&
	\Zfset{u}{\mbox{if}~p \acc q,~\mbox{then}~u \in \I{A}^\tenv_q}
	& (\(\,\O A\) is not a {\tvariant}) \\
\I{A \Impl B}^\tenv_p &=&
        \Zfset{\!u}{\mbox{if \(p \tacc q\:\), then
		    \(u \cdot v \in \I{B}^\tenv_q\:\)
		    for every \(v \in \I{A}^\tenv_q\)}}\mskip15mu
	& (\(A \Impl B\) is not a {\tvariant}) \\
\I{\fix{X}A}^\tenv_p &=& \I{A[\fix{X}A/X]}^\tenv_p &
	    (\(\fix{X}A\) is not a {\tvariant})
\end{Eqnarray*}
\end{definition}
Note that the \(\I{A}^{\tenv}_p\) has been defined by induction
on the lexicographic ordering of \(\pair{p}{r(A)}\), where
we consider the transitive closure of \(\acc\)
for the ordering of \(p\), which is also well-founded
since so is \(\acc\).
Because \(r(A[\fix{X}A/X]) < r(\fix{X}A)\) by Proposition~\ref{rank-fix}
when \(\fix{X}A\) is not a {\tvariant},
\(\I{\fix{X}A}^\tenv_p\) is well defined.
We can easily verify the following propositions.

\begin{proposition}\label{rlz-another-def}
The equations other than the first one in Definition~\ref{rlz-def}
hold for any type expression, that is,
the following equations hold.
\begin{Eqnarray*}
\I{X}^\tenv_p &=& \tenv(X)_p \\
\I{\O A}^\tenv_p &=&
	\Zfset{u}{\mbox{if}~p \acc q,~\mbox{then}~u \in \I{A}^\tenv_q} \\
\I{A \Impl B}^\tenv_p &=&
        \Zfset{u}{\mbox{if \(p \tacc q\:\), then
	\(u \cdot v \in \I{B}^\tenv_q\:\)
	for every \(v \in \I{A}^\tenv_q\)}} \\
\I{\fix{X}A}^\tenv_p &=& \I{A[\fix{X}A/X]}^\tenv_p
\end{Eqnarray*}
\end{proposition}
\begin{proof}
Straightforward by Propositions~\ref{tvariant-basic} and \ref{tvariant-fix}.
\qed
\end{proof}

\begin{proposition}\pushlabel
Let \(\tenv\) be a hereditary type environment.
\begin{Enumerate}
\item \itemlabel{rlz-subst-env}
    \(\I{A[B/X]}^\tenv_p = \I{A}^{\tenv[\I{B}^\tenv/X]}_p\).
\item \itemlabel{rlz-hereditary}
    \(\I{A}^{\tenv}\) is a hereditary mapping from \(\W\) to \(\Pow{\V}\);
    that is,
    \(p \acc q\) implies \(\I{A}^{\tenv}_p \subseteq \I{A}^{\tenv}_q\).
\end{Enumerate}
\end{proposition}
\begin{proof}
By straightforward induction
on the lexicographic ordering of \(\pair{p}{r(A)}\),
and by cases on the form of \(A\).
If \(A\) is a {\tvariant}, then so is \(A[B/X]\)
by Proposition~\ref{tvariant-subst1}; and hence,
    \(\I{A[B/X]}^\tenv_p = \I{A}^{\tenv[\I{B}^\tenv/X]}_p
    = \I{A}^\tenv_p = \I{A}^\tenv_q = \V\)
    by Definition~\ref{rlz-def}.
\ifdetail
Therefore, we assume that \(A\) is not below.

\paragraph{Proof of \protect\itemref{rlz-subst-env}}
Induction step proceeds by cases on the form of \(A\), where
the proof does not depend on the hereditarity of \(\tenv\).
Let \(\tenv' = \tenv[\I{B}^\tenv/X]\).

\Case{\(A = Y\) for some \(Y\).}
If \(Y = X\), then \(A[B/X] = B\); and hence,
    \(\I{A[B/X]}^\tenv_p = \I{B}^\tenv_p = \I{X}^{\tenv'}_p\).
Otherwise, \(Y \not= X\) and \(A[B/X] = Y\).
Hence,
    \(\I{A[B/X]}^\tenv_p = \I{Y}^\tenv_p = \I{Y}^{\tenv'}_p\).

\Case{\(A = \O C\) for some \(C\).}
In this case,
\begin{Eqnarray*}
    \I{A[B/X]}^\tenv_p
    &=& \I{\O(C[B/X])}^\tenv_p \\
    &=& \Zfset{u}{\mbox{if}~p \acc q,~\mbox{then}~u \in \I{C[B/X]}^\tenv_q}
	& (by Definition~\ref{rlz-def}) \\
    &=& \Zfset{u}{\mbox{if}~p \acc q,~\mbox{then}~u \in
	    \I{C}^{\tenv'}_q}
	& (by induction hypothesis) \\
    &=& \I{\O C}^{\tenv'}_p
	& (by Definition~\ref{rlz-def}).
\end{Eqnarray*}

\Case{\(A = C \Impl D\) for some \(C\) and \(D\).}
In this case,
\begin{Eqnarray*}
    \I{A[B/X]}^\tenv_p
    = \I{C[B/X] \Impl D[B/X]}^\tenv_p \mskip-290mu \\
    &=& \Zfset{\!u}{\mbox{if \(p \tacc q\:\), then
		    \(u \cdot v \in \I{D[B/X]}^\tenv_q\:\)
		    for every \(v \in \I{C[B/X]}^\tenv_q\)}}
	& (by Definition~\ref{rlz-def}) \\
    &=& \Zfset{\!u}{\mbox{if \(p \tacc q\:\), then
		    \(u \cdot v \in \I{D}^{\tenv'}_q\:\)
		    for every \(v \in \I{C}^{\tenv'}_q\)}}
	& \hskip-34pt (by induction hypothesis) \\
    &=& \I{C \Impl D}^{\tenv'}_p
	& (by Definition~\ref{rlz-def}).
\end{Eqnarray*}

\Case{\(A = \fix{Y}C\) for some \(Y\) and \(C\).}
We can assume that \(Y \not\in \FTV{B} \cup \{\,X\,\}\)
without loss of generality.
\begin{Eqnarray*}
    \I{A[B/X]}^\tenv_p
    &=& \I{\fix{Y}C[B/X]}^\tenv_p
	& (since \(Y \not\in \FTV{B} \cup \{\,X\,\}\)) \\
    &=& \I{C[B/X][\fix{Y}C[B/X]/Y]}^\tenv_p & (by Definition~\ref{rlz-def}) \\
    &=& \I{C[\fix{Y}C/Y][B/X]}^\tenv_p
	& (since \(Y \not\in \FTV{B} \cup \{\,X\,\}\)) \\
    &=& \I{C[\fix{Y}C/Y]}^{\tenv'}_p & (by induction hypothesis) \\
    &=& \I{\fix{Y}C}^{\tenv'}_p & (by Definition~\ref{rlz-def})
\end{Eqnarray*}
Note that \(r(C[\fix{Y}C]) < r(\fix{Y}C)\)
by Proposition~\ref{rank-fix}.

\paragraph{Proof of \protect\itemref{rlz-hereditary}}
Induction step proceeds by cases on the form of \(A\).
Suppose that \(p \acc q\).

\Case{\(A = X\) for some \(X\).}
Since \(\tenv\) is hereditary,
    \(\I{A}^\tenv_p = \tenv(X)_p \subseteq \tenv(X)_q = \I{A}^\tenv_q\).

\Case{\(A = \O B\) for some \(B\).}
Suppose that \(u \in \I{\O B}^\tenv_p\).
Then, \(u \in \I{B}^\tenv_q\) by Definition~\ref{rlz-def}.
Therefore,
if \(q \acc r\),
we get \(u \in \I{B}^\tenv_r\) by induction hypothesis, which means
\(u \in \I{\O B}^\tenv_q\).

\Case{\(A = B \Impl C\) for some \(B\) and \(C\).}
Obvious since \(u \in \I{B\Impl C}^\tenv_p\) implies
\(u \in \I{B \Impl C}^\tenv_q\) by Definition~\ref{rlz-def}.

\Case{\(A = \fix{X}B\) for some \(X\) and \(B\).}
Note that \(r(B[\fix{X}B]) < r(\fix{X}B)\) by Proposition~\ref{rank-fix}.
Therefore,
\begin{Eqnarray*}
    \I{\fix{X}B]}^\tenv_p
    &=& \I{B[\fix{X}B/X]}^\tenv_p & (by Definition~\ref{rlz-def}) \\
    &\subseteq& \I{B[\fix{X}B/X]}^\tenv_ q& (by induction hypothesis) \\
    &=& \I{\fix{X}B]}^\tenv_q & (by Definition~\ref{rlz-def}).
\end{Eqnarray*}
\else
Therefore, it suffices to only handle the case that \(A\) is not.
\fi 
\qed\CHECKED{2014/07/11}
\end{proof}

\begin{proposition}\label{semantics-transitive}\pushlabel
Let \(\pair{\W}{\acc}\) be a well-founded frame, and
    \(\pacc\) the transitive closure of\/ \(\acc\).
\begin{Enumerate}
\item \itemlabel{linear-transitive}
    The accessibility relation \(\pacc\) is locally linear
    if and only if so is\/ \(\acc\).
\item \itemlabel{rlz-transitive}
    Let \(\II\) and \(\IIt\) be the interpretations over
    \(\pair{\W}{\acc}\) and \(\pair{\W}{\pacc}\), respectively.
    \(\I{A}^\tenv_p = \It{A}^\tenv_p\,\) for every
    hereditary type environment \(\tenv\) and \(p \in \W\).
\end{Enumerate}
\end{proposition}
\begin{proof}
Straightforward.
Use Propositions~\ref{rlz-another-def} and \ref{rlz-hereditary}
for Item~\itemref{rlz-transitive}.
\ifdetail
\paragraph{Proof of \protect\itemref{linear-transitive}}
For the ``if'' part, suppose that \(\acc\) is locally linear,
and that \(p \pacc q\).
We can find a world \(p'\) such that \(p \tacc p' \acc q\).
Then, since \(\acc\) is locally linear,
there exists some \(r\) such that
\begin{eqnarray}
\label{linear-transitive-01}
    && p' \tacc r \acc q,~\mbox{and} \\
\label{linear-transitive-02}
    && r \acc s ~\mbox{implies}~ q \tacc s ~\mbox{for any}~s.
\end{eqnarray}
We get \(p \tacc r \pacc q\) from (\ref{linear-transitive-01}).
Furthermore, if \(r \pacc s\), then
\(r \acc s' \tacc s\) for some \(s'\); and hence,
\(q \tacc s'\) by (\ref{linear-transitive-02}), which implies
\(q \tacc s\).
We thus get that \(\pacc\) is also locally linear.

As for the ``only if'' part, suppose that \(\pacc\) is locally linear,
and that \(p \acc q\).
Since \(\pacc\) is locally linear,
there exists some \(r'\) such that
\begin{eqnarray}
\label{linear-transitive-03}
    && p \tacc r' \pacc q,~\mbox{and} \\
\label{linear-transitive-04}
    && r' \pacc s ~\mbox{implies}~ q \tacc s ~\mbox{for any}~s.
\end{eqnarray}
We can find a world \(r\) such that
\(p \tacc r' \tacc r \acc q\), that is,
\(p \tacc r \acc q\), by (\ref{linear-transitive-03}).
Furthermore, if \(r \acc s\), then
\(r' \tacc r \acc s\), that is, \(r' \pacc s\); and hence,
\(q \tacc s\) by (\ref{linear-transitive-04}).
Therefore, \(\acc\) is also locally linear.

\paragraph{Proof of \protect\itemref{rlz-transitive}}
Straightforward induction on \(h(A)\), and by cases on the form of \(A\).
Use Propositions~\ref{rlz-another-def} and \ref{rlz-hereditary}.
The only non-trivial case is as follows.
\begin{Eqnarray*}
u \in \I{\O A}^\tenv_p
    &~\mbox{iff}~& p \acc q ~\mbox{implies}~u \in \I{A}^\tenv_q
	    & (by Proposition~\ref{rlz-another-def}) \\
    &~\mbox{iff}~& p \acc q \tacc r ~\mbox{implies}~u \in \I{A}^\tenv_r
	    & (by Proposition~\ref{rlz-hereditary}) \\
    &~\mbox{iff}~& p \pacc r ~\mbox{implies}~u \in \I{A}^\tenv_r \\
    &~\mbox{iff}~& u \in \It{\O A}^\tenv_p
	    & (by Proposition~\ref{rlz-another-def})
\end{Eqnarray*}
\Qed
\else
\qed
\fi 
\CHECKED{2014/07/14}
\end{proof}

Adopting such an interpretation of types over a well-founded frame, or
a {\lA}-frame, we can verify the subtyping relations between
\(\O (A \Impl B)\) and \(\O A \Impl \O B\)
discussed in Section~\ref{intro-sec} as follows.

\begin{proposition}\label{subtyp-K-soundness}
Let \(\pair{\W}{\acc}\) be a well-founded frame.
Then,
\(\I{\O(A \Impl B)}^\tenv_p \subseteq \I{\O A \Impl \O B}^\tenv_p\,\)
for any \(A\), \(B\), \(\tenv\) and \(p \in \W\).
\end{proposition}
\begin{proof}
Suppose that \(u \in \I{\O(A \Impl B)}^{\tenv}_p\).
To show that \(u \in \I{\O A \Impl \O B}^{\tenv}_p\),
suppose also that \(p \tacc r\) and \(v \in \I{\O A}^{\tenv}_r\).
By Proposition~\ref{rlz-another-def},
it suffices to show that
\(u \cdot v \in \I{B}^{\tenv}_q\strut\) for every \(q \opacc r\).
Since
\(u \in \I{\O(A \Impl B)}^{\tenv}_p\) and \(p \tacc r\),
we get
\(u \in \I{\O(A \Impl B)}^{\tenv}_r\strut\) by Proposition~\ref{rlz-hereditary}.
Hence, if \(r \acc q\), then
    \(u \in \I{A \Impl B}^{\tenv}_q\); and therefore,
\(u \cdot v \in \I{B}^{\tenv}_q\)
by Proposition~\ref{rlz-another-def}
because \(v \in \I{A}^{\tenv}_q\) from \(v \in \I{\O A}^{\tenv}_r\strut\).
\qed\CHECKED{2014/07/11}
\end{proof}

\begin{theorem}\label{peqtyp-KL-soundness}
Let \(\pair{\W}{\acc}\) be a well-founded frame.
Then,
\(\I{\O(A \Impl B)}^\tenv_p = \I{\O A \Impl \O B}^\tenv_p\,\)
for any \(A\), \(B\), hereditary \(\tenv\) and \(p \in \W\)
if and only if\/ \(\acc\) is locally linear.
\end{theorem}
\begin{proof}
For the ``if'' part, suppose that \(\acc\) satisfies Condition~2
of Definition~\ref{wf-frame-def}.
We get
\(\I{\O(A \Impl B)}^\tenv_p \subseteq \I{\O A \Impl \O B}^\tenv_p\)
by Proposition~\ref{subtyp-K-soundness}.
To show the opposite,
suppose also that \(u \in \I{\O A \Impl \O B}^{\tenv}_p\),
\(p \acc q' \tacc q\) and \(v \in \I{A}^{\tenv}_q\).
It suffices to show that \(u \cdot v \in \I{B}^{\tenv}_q\strut\).
Let \(p'\) be the element of \(\W\) such that \(p \tacc p' \acc q\).
By Condition~2 of Definition~\ref{wf-frame-def},
there exists some \(r \in \W\) such that
\begin{Eqnarray}
    && p \tacc r, \label{peqtyp-KL-soundness-01} \\
    && r \acc q, ~\mbox{and} \label{peqtyp-KL-soundness-02} \\
    && r \acc s~\mbox{implies}~q \tacc s~\mbox{for any}~s.
	\label{peqtyp-KL-soundness-03}
\end{Eqnarray}
By Proposition~\ref{rlz-hereditary}, we get
\(u \in \I{\O A \Impl \O B}^{\tenv}_r\) from
\(u \in \I{\O A \Impl \O B}^{\tenv}_p\) and (\ref{peqtyp-KL-soundness-01}),
and get \(v \in \I{\O A}^{\tenv}_r\)
from \(v \in \I{A}^{\tenv}_q\) and (\ref{peqtyp-KL-soundness-03}).
Therefore,
\(u \cdot v \in \I{\O B}^{\tenv}_r\) by Proposition~\ref{rlz-another-def}.
We now get \(u \cdot v \in \I{B}^{\tenv}_q\) from this
and (\ref{peqtyp-KL-soundness-02}) by Proposition~\ref{rlz-another-def}.

For the ``only if'' part, suppose that \(\acc\) does not satisfy
Condition~2, i.e.,
there exist some \(p\), \(q \in \W\) such that \(p \acc q\), and
\begin{Eqnarray}
    \label{peqtyp-KL-soundness-04}
    ~\mbox{for every}~r \in \W,~\mbox{if}~p \tacc r \acc q,
	~\mbox{then}~ r \acc s~\mbox{and}~q \not\tacc s
	~\mbox{for some}~s \in \W.
\end{Eqnarray}
Since \(p \tacc p \acc q\), there exists some \(s\) such that
	\(p \acc s\) and \(q \not\tacc s\).
Then, consider the hereditary type environment \(\tenv\) defined as follows.
\begin{Eqnarray}
    \label{peqtyp-KL-soundness-05}
    \tenv(X)_t &=& \Choice{\V & (q \tacc t) \\ \{\} & (q \not\tacc t)} \\[5pt]
    \label{peqtyp-KL-soundness-06}
    \tenv(Y)_t &=& \Choice{\V & (q \acc t) \\ \{\} & (q \not\acc t)}
\end{Eqnarray}
We can see that \(\I{\O(X \Impl Y)}^{\tenv}_p = \{\}\), while
\(\I{\O X \Impl \O Y }^{\tenv}_p = \V\).
In fact, \(\I{X \Impl Y}^{\tenv}_q = \{\}\strut\),
since \(\I{X}^{\tenv}_q = \tenv(X)_q = \V\) and
\(\I{Y}^{\tenv}_q = \tenv(Y)_q = \{\}\) by
(\ref{peqtyp-KL-soundness-05}) and (\ref{peqtyp-KL-soundness-06}), respectively;
and hence, \(\I{\O(X \Impl Y)}^{\tenv}_p = \{\}\) from \(p \acc q\).
Furthermore,
\(\I{\O X \Impl \O Y }^{\tenv}_p = \V\) can be also shown as follows.
Suppose that \(p \tacc t\) and \(\I{\O X }^{\tenv}_t \not= \{\}\strut\),
which means that
\(t \acc s\) implies \(\tenv(X)_s = \V\) for any \(s \in \W\)
by (\ref{peqtyp-KL-soundness-05}); that is,
\[
    t \acc s~\mbox{implies}~q \tacc s~\mbox{for any}~s.
\]
Hence, \(t \not\acc q\) by (\ref{peqtyp-KL-soundness-04}); and therefore,
\(t \acc s\) implies \(q \acc s\) for any \(s\).
Then, by (\ref{peqtyp-KL-soundness-06}),
\[
    t \acc s~\mbox{implies}~\tenv(Y)_s = \V\,~\mbox{for any}~s.
\]
We thus get \(\I{\O Y }^{\tenv}_t = \V\) from
\(\I{\O X }^{\tenv}_t \not= \{\}\) and \(p \tacc t\).
Hence, \(\I{\O X \Impl \O Y }^{\tenv}_p = \V\).
\qed\CHECKED{2014/07/12, 07/24}
\end{proof}

\Section{Equality of types}\label{eqtyp-sec}

In Section~\ref{lA-sec}, we will introduce
a modal typing system {\lA}, which
respect the semantics of type expressions based on {\lA}-frames.
In this section, we discuss formal derivability of equality
between type expressions,
which will be incorporated into the modal typing system.

\begin{definition}[\(\eqtyp\)]
    \ilabel{eqtyp-def}{type expressions!equality of}
    \ilabel*{0 equality@$\protect\eqtyp$}
    \ilabel{Il-def}{I l@$\protect\Il{A}^{\protect\tenv}_p$}
The equivalence relation \(\eqtyp\) on \(\TE\) is defined
to be the smallest binary relation that satisfies the following:
\begingroup
\leftmargini=60pt
\labelsep=10pt
\begin{itemize}\itemsep=0pt
\item[\r{{\eqtyp}\mbox{-reflex}}] \(A \eqtyp A\).
\item[\r{{\eqtyp}\mbox{-symm}}]
    If\/ \(A \eqtyp B\), then \(B \eqtyp A\).
\item[\r{{\eqtyp}\mbox{-trans}}]
    If\/ \(A \eqtyp B\) and \(B \eqtyp C\), then \(A \eqtyp C\).
\item[\r{{\eqtyp}\mbox{-}{\O}}]
    If\/ \(A \eqtyp B\), then \(\O A \eqtyp \O B\).
\item[\r{{\eqtyp}\mbox{-}{\Impl}}]
    If\/ \(A \eqtyp C\) and \(B \eqtyp D\),
    then \(A \Impl B \eqtyp C \Impl D\).
\item[\r{{\eqtyp}\mbox{-}{\Impl}{\t}}]
    \(A \Impl \t \eqtyp \t\).
\item[\r{{\eqtyp}\mbox{-fix}}]
    \(\fix{X}A \eqtyp A[\fix{X}A/X]\).
\item[\r{{\eqtyp}\mbox{-uniq}}]
    If \(A \eqtyp C[A/X]\) and \(C\) is
    proper in \(X\), then \(A \eqtyp \fix{X}C\).
\end{itemize}
\endgroup
\end{definition}
For example, \(\fix{X}{Y \Impl \O X} \eqtyp \t\)
by \r{{\eqtyp}\mbox{-uniq}} since
\((Y \Impl \O X)[\t/X] \eqtyp \t\) as follows.
\begin{Eqnarray*}
(Y \Impl \O X)[\t/X] \ifnarrow\mskip-80mu\\\fi
    &=& Y \Impl \O \fix{X}\O X \\
    &\eqtyp& Y \Impl \fix{X}\O X
	& (by \r{{\eqtyp}\mbox{-reflex}}, \r{{\eqtyp}\mbox{-fix}},
	    \r{{\eqtyp}\mbox{-symm}} and \r{{\eqtyp}\mbox{-}{\Impl}}) \\
    &\eqtyp& \fix{X}\O X  & (by \r{{\eqtyp}\mbox{-}{\Impl}{\t}}) \\
    &=& \t
\end{Eqnarray*}
The intended meaning of \(A \eqtyp B\) is that
the interpretations of \(A\) and \(B\)
are identical in any world of any well-founded frames.
We will see in Proposition~\ref{geqtyp-fix-congr} that
the following rule is derivable from the above.
\begin{quote}
\(A \eqtyp B\) implies \(\fix{X}A \eqtyp \fix{X}B\),
provided that \(A\) or \(B\) is proper in \(X\).
\end{quote}

\begin{definition}[\(\peqtyp\)]
    \ilabel{peqtyp-def}{0 equality L@$\protect\peqtyp$}
We similarly define another equivalence relation \(\peqtyp\) by adding
the following condition to the definition of \(\eqtyp\).
\begingroup
\leftmargini=60pt\labelsep=10pt
\begin{itemize}
\item[\r{{\peqtyp}\mbox{-{\bf K}/{\bf L}}}]
    \(\O(A \Impl B) \peqtyp \O A \Impl \O B\)
\end{itemize}
\endgroup
\end{definition}

This rule reflects the equivalence shown in Theorem~\ref{peqtyp-KL-soundness},
and is only valid for {\lA}-frames.
Noted that \(A \eqtyp B\) implies \(A \peqtyp B\).
Roughly, two type expressions are equivalent modulo \(\eqtyp\),
if their (possibly infinite) type expressions
obtained by indefinite unfolding recursive types occurring in them
are identical modulo the rule \r{{\eqtyp}\mbox{-}{\Impl}\t}.
Also, they are equivalent modulo \(\peqtyp\),
if their indefinite unfoldings are identical modulo the rules
\r{{\eqtyp}\mbox{-}{\Impl}\t} and \r{{\peqtyp}\mbox{-{\bf K}/{\bf L}}}.

\Subsection{Basic properties of the equality}\label{geqtyp-basic-sec}
\ilabel{geqtyp-def}{0 equality Z@$\protect\geqtyp$}
In this subsection,
we discuss some basic properties of \(\eqtyp\) and \(\peqtyp\).
Most of the results are common to both \(\eqtyp\) and \(\peqtyp\).
In the sequel, let \(\geqtyp\) denotes either \(\eqtyp\) or \(\peqtyp\).

\begin{proposition}\label{geqtyp-t}\pushlabel
\begin{Enumerate}
\item \itemlabel{geqtyp-t1}
    \(\t \geqtyp \O \t\).
\item \itemlabel{geqtyp-t2}
    \(\t \geqtyp \fix{X}\O^n X\) for every \(n \ge 1\).
\item \itemlabel{geqtyp-O-fix}	
\(A \geqtyp \O^{n{+}1} A\) if and only if \(A \geqtyp \t\).
\end{Enumerate}
\end{proposition}
\begin{proof}
Straightforward by \r{{\eqtyp}\mbox{-fix}} and \r{{\eqtyp}\mbox{-uniq}}.
\ifdetail

\paragraph{Proof of \protect\itemref{geqtyp-t1}}
\(\t = \fix{X}\O X \geqtyp \O \fix{X}\O X = \O \t\) by \r{{\eqtyp}\mbox{-fix}}.

\paragraph{Proof of \protect\itemref{geqtyp-t2}}
Since \(\t \geqtyp \O^n \t\) by Item~\itemref{geqtyp-t1},
we get \(\t \geqtyp \fix{X}\O^n X\) by \r{{\eqtyp}\mbox{-uniq}}.

\paragraph{Proof of \protect\itemref{geqtyp-O-fix}}
The ``if'' part is obvious from Item~\itemref{geqtyp-t1}.
For the ``only if'' part, if \(A \geqtyp \O^{n+1} A\),
then, \(A \geqtyp \fix{X}\O^{n+1}X \geqtyp \t\)
by \r{{\eqtyp}\mbox{-uniq}} and Item~\itemref{geqtyp-t2}.
\fi
\qed\CHECKED{2014/07/08}
\end{proof}

In the succeeding sections, we occasionally use the following proposition
in proofs without mention.

\begin{proposition}\label{geqtyp-subst}
If\/ \(A \geqtyp B\) and \(C \geqtyp D\),
then \(A[C/X] \geqtyp B[D/X]\).
\end{proposition}
\begin{proof}
By induction on the derivation of \(A \geqtyp B\),
and by cases on the last rule applied in it.
The readers should refer to Appendix~\ref{geqtyp-subst-sec}
for a detailed proof.
\qed
\end{proof}

Before showing the soundness of \(\eqtyp\) and \(\peqtyp\)
with respect to the intended semantics of type expressions,
we confirm that
the equality preserves whether a type expression is a {\tvariant} or not,
since some notions introduced so far, such as properness,
\(\PNOdp\), \(\PNIdp\) and \(ETV^{\pm}\), are
defined in a way where {\tvariant}s are treated as an exceptional case.

\begin{proposition}\label{geqtyp-tvariant}
Suppose that \(A \geqtyp B\). Then,
\(A\) is a {\tvariant} if and only if so is \(B\).
\end{proposition}
\begin{proof}
By simultaneous induction on the derivation of \(A \geqtyp B\)
with the claim
\[
    \POdp(\tail{A}, Z) = \POdp(\tail{B}, Z)~\,\mbox{for any}~Z,
\]
and by cases on the last rule used in the derivation.
Note that, by Propositions~\ref{tail-etv} and \ref{depth-finite-etv},
\[
    \NOdp(\tail{A}, Z) = \NOdp(\tail{B}, Z) = \infty~\,\mbox{for any}~Z.
\]
Use Propositions~\ref{tvariant-basic} and \ref{tvariant-fix}.
The cases other than \r{{\eqtyp}\mbox{-fix}} or \r{{\eqtyp}\mbox{-uniq}}
are straightforward.

In case of \r{{\eqtyp}\mbox{-fix}},
there exist some \(X\) and \(C\) such that \(A = \fix{X}C\) and
\(B = C[A/X]\). First, \(A\) is a {\tvariant} if and only if so is \(B\)
by Proposition~\ref{tvariant-fix}.
Therefore, if either \(A\) or \(B\) is a {\tvariant}, so are \(\tail{A}\)
and \(\tail{B}\) by Proposition~\ref{tail-tvariant}; and hence,
\(\POdp(\tail{A}, Z) = \POdp(\tail{B}, Z) = \infty\) by
Definition~\ref{depth-def}.
On the other hand, if neither \(A\) nor \(B\) is a {\tvariant},
then assuming \(X \not= Z\) without loss of generality,
\begin{Eqnarray*}
\POdp(\tail{B},Z) &=&
    \POdp(\tail{C}[\tail{A}/X],Z)
	    & (by Proposition~\ref{tail-subst}) \\
    &=& \min(\POdp(\tail{C},Z),
	\POdp(\tail{C},X) + \POdp(\tail{A},Z),
	\NOdp(\tail{C},X) + \NOdp(\tail{A},Z)) \mskip-500mu \\
	    &&& (by Proposition~\ref{depth-subst1}) \\
    &=& \min(\POdp(\tail{C},Z),
	\POdp(\tail{C},X) + \POdp(\tail{A},Z))
	    & (since \(\NOdp(\tail{C},X) = \NOdp(\tail{A},Z) = \infty\)) \\
    &=& \min(\POdp(\tail{C},Z),\,
	 \POdp(\tail{C},X) +
	    \min(\POdp(\tail{C}, Z),\, \NOdp(\tail{C}, X)+\NOdp(\tail{C}, Z)))
	     \mskip-500mu \\
	     &&& (by Definition~\ref{depth-def}) \\
    &=& \min(\POdp(\tail{C},Z),\,\POdp(\tail{C},X) + \POdp(\tail{C}, Z))
	    & (since \(\NOdp(\tail{C},X) = \NOdp(\tail{C},Z) = \infty\)) \\
    &=& \POdp(\tail{C},Z) \\
    &=& \min(\POdp(\tail{C},Z),\,\NOdp(\tail{C},X) + \NOdp(\tail{C}, Z))
	    & (since \(\NOdp(\tail{C},X) = \NOdp(\tail{C},Z) = \infty\)) \\
    &=& \POdp(\tail{A}, Z) & (by Definition~\ref{depth-def}).
\end{Eqnarray*}

In case of \r{{\eqtyp}\mbox{-uniq}},
there exist some \(X\) and \(C\) such that \(B = \fix{X}C\),
\(A \geqtyp C[A/X]\), and \(C\) is proper in \(X\).
By induction hypothesis, we have
\begin{eqnarray}
\label{geqtyp-tvariant-01}
    && \mbox{\(A\) is a {\tvariant} iff so is \(C[A/X]\), and} \\
\label{geqtyp-tvariant-02}
    && \POdp(\tail{A},\,Z) = \POdp(\tail{C[A/X]},\,Z)~\mbox{for any \(Z\).}
\end{eqnarray}
First, if \(A\) is a {\tvariant}, then
so is \(C[A/X]\) by (\ref{geqtyp-tvariant-01}); and therefore,
so is \(\fix{X}{C}\) by Proposition~\ref{tvariant-proper-subst}.
To show the converse,
assume that
\(\fix{X}C\) is a {\tvariant} whereas \(A\) is not.
We show that this assumption leads to a contradiction.
Let \(\tail{C} =
    \O^{m_0}\fix{X_1}\O^{m_1}\fix{X_2}\O^{m_2}\ldots\fix{X_n}\O^{m_n}Y\),
and
\(\tail{A} =
    \O^{m_0'}\fix{X_1'}\O^{m_1'}\fix{X_2'}\O^{m_2'}\ldots\fix{X_{n'}'}
    \O^{m_{n'}'}Y'\).
Note that \(Y \in \{\,X,X_1,X_2,\ldots,X_n\,\}\) and
\(Y' \not\in \{\,X_1',X_2',\ldots,X_{n'}'\,\}\)
since \(\fix{X}{C}\) is a {\tvariant} and \(A\) is not.
Furthermore, since \(A\) is not, \(\tail{C[A/X]}\)
is not a {\tvariant} either, by (\ref{geqtyp-tvariant-01}) and
Proposition~\ref{tail-tvariant}.
Therefore,
\(\tail{C}\) is not a {\tvariant} either,
by Propositions~\ref{tail-subst} and \ref{tvariant-subst1}; and hence,
\(Y = X\) and \(Y \not= X_i\) for every \(i\).
Note also that \(m_0+m_1+\ldots+m_n > 0\) by Propositions~\ref{tail-proper}
and \ref{O-depth-proper} since \(C\) is proper in \(X\).
Hence,
\begin{Eqnarray*}
\POdp(\tail{A},\, Y') &=& \POdp(\tail{C[A/X]},\, Y')
	& (by (\ref{geqtyp-tvariant-02})) \\
    &=& \POdp(\tail{C}[\tail{A}/X],\, Y')
	& (by Proposition~\ref{tail-subst}) \\
    &=& \POdp(\O^{m_0}\fix{X_1}\O^{m_1}\fix{X_2}\O^{m_2}\ldots\fix{X_n}
	    \O^{m_n}\tail{A},\, Y') \ifnarrow\hskip-70pt\fi \\
    &=& \POdp(\tail{A},\, Y') + m_0+m_1+\ldots+m_n \ifnarrow\hskip-70pt\fi
	& (by Definition~\ref{depth-def}).
\end{Eqnarray*}
However, this is impossible
since \(\POdp(\tail{A},Y') = m_0'+m_1'+\ldots+m_{n'}'< \infty\)
and \(m_0+m_1+\ldots+m_n > 0\).
This completes the proof of the first part.
For the second part, we assume
that neither \(\tail{A}\) nor \(\tail{(\fix{X}C)}\) is a {\tvariant},
since \(\POdp(\tail{A}, Z) = \POdp(\tail{(\fix{X}C)}, Z) = \infty\)
by Definition~\ref{depth-def} if they are.
Note that \(\tail{C}\) is also proper in \(X\) by Proposition~\ref{tail-proper};
and hence, \(\Odp(\tail{C},\,X) > 0\) by Proposition~\ref{O-depth-proper}.
Then, assuming that \(X \not= Z\) without loss of generality,
\begin{Eqnarray*}
    \POdp(\tail{A}, Z) &=& \POdp(\tail{C[A/X]}, Z)
	    & (by (\ref{geqtyp-tvariant-02})) \\
    &=& \POdp(\tail{C}[\tail{A}/X], Z) & (by Proposition~\ref{tail-subst}) \\
    &=& \min(\POdp(\tail{C},Z),\, \POdp(\tail{C},X) + \POdp(\tail{A},Z),
	\NOdp(\tail{C},X) + \NOdp(\tail{A},Z)) \mskip-500mu \\
	    &&& (by Proposition~\ref{depth-subst1}) \\
    &=& \min(\POdp(\tail{C},Z),\, \POdp(\tail{C},X) + \POdp(\tail{A},Z))
	    & (since \(\NOdp(\tail{C},X) = \NOdp(\tail{A},Z) = \infty\)) \\
    &=& \POdp(\tail{C},Z)
	    & \hskip-300pt\hfill (since \(\POdp(\tail{C}, X) > 0\);
		valid even if \(\POdp(\tail{A},Z) = \infty\)) \\
    &=& \min(\POdp(\tail{C},Z),\,\NOdp(\tail{C},X) + \NOdp(\tail{C}, Z))
	    & (since \(\NOdp(\tail{C},X) = \NOdp(\tail{C},Z) = \infty\)) \\
    &=& \POdp(\fix{X}\tail{C},Z) &
	    (by Definition~\ref{depth-def}) \\
    &=& \POdp(\tail{B},Z) &
	    (by Definition~\ref{tail-def}).
\end{Eqnarray*}
\Qed\CHECKED{2014/07/12}
\end{proof}

\Subsection{Soundness of the derivation rules for equality}

We now proceed to show that the equivalence relations \(\eqtyp\) and
\(\peqtyp\) on type expressions well respect the semantics of types
according to their intended frame classes.
The last one step before the proof is to prove the following lemma,
which says that if a type expression is proper in a type variable,
then the interpretation of the type expression in a possible world
does not depend on the interpretation of the type variable in that world.

\begin{lemma}\label{rlz-proper-subst-lemma}
Let \(\pair{\W}{\acc}\) be a well-founded frame,
\(A\) a type expression, and \(p \in \W\).
Let \(\tenv\) and \(\tenv'\) be hereditary type environments
such that
\(\tenv(X)_q = \tenv'(X)_q\)\/ for every \(X \in \TV\) and \(q \in \W\)
such that \(p \acc q\).
If\/ for every \(X\), either (a) \(A\) is proper in \(X\),
or (b) \(\tenv(X)_p = \tenv'(X)_p\),
then \(\I{A}^{\tenv}_p = \I{A}^{\tenv'}_p\).
\end{lemma}
\begin{proof}
By induction on the lexicographic ordering of \(\pair{p}{r(A)}\).
Suppose that (a) or (b) holds for every \(X \in \TV\).
If \(A\) is a {\tvariant}, then
\(\I{A}^{\tenv}_p = \I{A}^{\tenv'}_p = \V\) by Definition~\ref{rlz-def}.
Therefore, we show that
\(\I{A}^{\tenv}_p = \I{A}^{\tenv'}_p\) by cases on the form of \(A\)
assuming that \(A\) is not a {\tvariant}.

\Case{\(A = Y\) for some \(Y\).}
In this case, \(A\) is not proper in \(Y\).
Hence, \(\I{A}^\tenv_p = \tenv(Y)_p = \tenv'(Y)_p = \I{A}^{\tenv'}_p\)
from (b).

\Case{\(A = \O B\) for some \(B\).}
Since \(\tenv(X)_q = \tenv'(X)_q~\) for every \(X\)
and \(q \opacc p\),
we get \(\I{B}^{\tenv}_q = \I{B}^{\tenv'}_q\) for every \(q \opacc p\)
by induction hypothesis.
Therefore,
\(\I{\O B}^{\tenv}_p
    = \zfset{u}{u \in \I{B}^{\tenv}_q~~\mbox{for every}~q \opacc p}
    = \zfset{u}{u \in \I{B}^{\tenv'}_q~~\mbox{for every}~q \opacc p}
    = \I{\O B}^{\tenv'}_p\)
by Definition~\ref{rlz-def}.

\Case{\(A = B \Impl C\) for some \(B\) and \(C\).}
Note that \(C\) is not a {\tvariant} since \(A\) is not.
Therefore, \(r(B)\), \(r(C) < r(A)\),
and (a) implies both \(B\) and \(C\) are
also proper in \(X\).
Hence, by induction hypothesis, we get
\(\I{B}^{\tenv}_q = \I{B}^{\tenv'}_q\) and
\(\I{C}^{\tenv}_q = \I{C}^{\tenv'}_q\) for every \(q\)
such that \(p \tacc q\).
Therefore, \(\I{B \Impl C}^{\tenv}_p = \I{B \Impl C}^{\tenv'}_p\)
by Definition~\ref{rlz-def}.

\Case{\(A = \fix{Y}C\) for some \(Y\) and \(C\).}
We can assume that \(X \not= Y\) without loss of generality.
Note that \(r(C[\fix{Y}C/Y]) < r(\fix{Y}C)\)
by Proposition~\ref{rank-fix}, and that
\(C[\fix{Y}C/Y]\) is proper in \(X\) if so is \(A\),
by Definition~\ref{proper-def} and Proposition~\ref{proper-subst1}.
Hence,
\begin{Eqnarray*}
\I{\fix{Y}C}^{\tenv}_p
    &=& \I{C[\fix{Y}C/Y]}^{\tenv}_p & (by Definition~\ref{rlz-def}) \\
    &=& \I{C[\fix{Y}C/Y]}^{\tenv'}_p & (by induction hypothesis) \\
    &=& \I{\fix{Y}C}^{\tenv'}_p & (by Definition~\ref{rlz-def}).
\end{Eqnarray*}
\Qed\CHECKED{2014/05/01}
\end{proof}

\begin{theorem}[Soundness of \(\eqtyp\)]
    \ilabel{eqtyp-soundness}{soundness!equal@$\protect\eqtyp$}
Let \(\pair{\W}{\acc}\) be a well-founded frame,
and consider the interpretation in the frame.
If\/ \(A \eqtyp B \), then \(\I{A}^\tenv = \I{B}^\tenv\)
for any type environment \(\tenv\).
\end{theorem}
\begin{proof}
By induction on the derivation of \(A \eqtyp B\),
and by cases on the last rule in the derivation.
Suppose that \(A \eqtyp B \).
If either \(A\) or \(B\) is a {\tvariant}, then so are both by
Proposition~\ref{geqtyp-tvariant}; and hence,
\(\I{A}^{\tenv}_p = \I{B}^{\tenv}_p = \V\) for every \(p\)
by Definition~\ref{rlz-def}.
Therefore, we show that
\(\I{A}^{\tenv} = \I{B}^{\tenv}\) assuming that neither is a {\tvariant}.

\Cases{\r{{\eqtyp}\mbox{-reflex}},
    \r{{\eqtyp}\mbox{-symm}} and \r{{\eqtyp}\mbox{-trans}}.}
Trivial.

\Case{\r{{\eqtyp}\mbox{-}{\O}}.}
In this case, there exist some \(A'\) and \(B'\) such that
\(A = \O A'\), \(B = \O B'\) and \(A' \eqtyp B'\).
We have \(\I{A'}^\tenv_q = \I{B'}^\tenv_q\) for every \(q \in \W\)
by induction hypothesis.
Therefore,
\(\I{\O A'}^\tenv_p
    = \zfset{u}{u \in \I{A'}^\tenv_q~~\mbox{for every}~q \opacc p}
    = \zfset{u}{u \in \I{B'}^\tenv_q~~\mbox{for every}~q \opacc p}
    = \I{\O B'}^\tenv_p\)
by Definition~\ref{rlz-def}.

\Case{\r{{\eqtyp}\mbox{-}{\Impl}}.}
Similar to the previous case.

\Case{\r{{\eqtyp}\mbox{-}{\Impl}{\t}}.}
Impossible because we assumed that neither \(A\) nor \(B\) is a {\tvariant}.

\Case{\r{{\eqtyp}\mbox{-fix}}.}
Obvious from Definition~\ref{rlz-def}.

\Case{\r{{\eqtyp}\mbox{-uniq}}.}
There exist some \(X\) and \(C\) such that
\(B = \fix{X}C\), \(A \eqtyp C[A/X]\) and \(C\) is proper in \(X\).
By induction hypothesis,
\begin{eqnarray}\label{eqtyp-soundness-01}
    \I{A}^{\tenv'}_p &=& \I{C[A/X]}^{\tenv'}_p~
    \mbox{for every \(\tenv'\)}.
\end{eqnarray}
We show that \(\I{A}^\tenv_p = \I{\fix{X}C}^\tenv_p~\)
for every \(p \in \W\) by induction on \(p\).
The induction hypothesis in this induction is
\begin{eqnarray}
\label{eqtyp-soundness-02}
    \I{A}^\tenv_q = \I{\fix{X}C}^\tenv_q
    ~~\mbox{for every \(q \opacc p\)}.
\end{eqnarray}
Therefore,
\begin{Eqnarray*}
\I{A}^\tenv_p
    &=& \I{C[A/X]}^\tenv_p & (by (\ref{eqtyp-soundness-01})) \\
    &=& \I{C}^{\tenv[\I{A}^\tenv/X]}_p
	& (by Proposition~\ref{rlz-subst-env}) \\
    &=& \I{C}^{\tenv[\I{\fix{X}C}^\tenv/X]}_p
	& (by (\ref{eqtyp-soundness-02})
	    and Lemma~\ref{rlz-proper-subst-lemma}) \\
    &=& \I{C[\fix{X}C/X]}^\tenv_p
	& (by Proposition~\ref{rlz-subst-env}) \\
    &=& \I{\fix{X}C}^\tenv_p
	& (by Definition~\ref{rlz-def}).
\end{Eqnarray*}
\Qed\CHECKED{2014/05/01}
\end{proof}

\begin{theorem}[Soundness of \(\peqtyp\) with respect to {\lA}-frames]
    \ilabel{peqtyp-soundness}{soundness!equal@$\protect\peqtyp$}
Let \(\pair{\W}{\acc}\) be a {\lA}-frame,
\(\tenv\) a hereditary type environment.
If\/ \(A \peqtyp B \), then \(\I{A}^\tenv = \I{B}^\tenv\).
\end{theorem}
\begin{proof}
The proof is quite parallel to the one of Theorem~\ref{eqtyp-soundness}.
The only additional task is
to check the rule \r{{\peqtyp}\mbox{-{\bf K}/{\bf L}}},
which is straightforward from Theorem~\ref{peqtyp-KL-soundness}.
\qed\CHECKED{2014/07/12}
\end{proof}

\Subsection{Properties of type expressions preserved by the equality}

At the end of Subsection~\ref{geqtyp-basic-sec}, we showed that
the equality preserves whether a type expression is a {\tvariant} or not.
In this subsection,
we shall confirm the fact that the equivalence relations also preserve
other basic properties of type expressions such as
properness, \(\PNOdp\), \(\PNIdp\) and \(ETV^{\pm}\).
Most of them are proved by induction on the derivation of the equality.

\begin{proposition}\label{geqtyp-tail}
\(A \geqtyp B\) implies \(\tail{A} \geqtyp \tail{B}\).
\end{proposition}
\begin{proof}
By induction on the derivation of \(A \geqtyp B\),
and by cases on the last rule used in the derivation.
We employ Proposition~\ref{tail-basic} in the case that last rule is
\r{{\eqtyp}\mbox{-fix}} or \r{{\eqtyp}\mbox{-uniq}}.
\ifdetail
Suppose that \(A \geqtyp B\).

\Cases{\r{{\eqtyp}\mbox{-reflex}},
    \r{{\eqtyp}\mbox{-symm}} and \r{{\eqtyp}\mbox{-trans}}.}
Trivial.

\Case{\r{{\eqtyp}\mbox{-}{\O}}.}
In this case, \(A = \O A'\), \(B = \O B'\) and \(A' \eqtyp B'\)
for some \(A'\) and \(B'\).
By induction hypothesis, \(\tail{A'} \geqtyp \tail{B'}\).
Hence, \(\tail{A} \geqtyp \tail{B}\)
since \(\tail{A} = \O \tail{A'}\) and \(\tail{B} = \O \tail{B'}\)
by Definition~\ref{tail-def}.

\Case{\r{{\eqtyp}\mbox{-}{\Impl}}.}
Similar to the previous case.

\Case{\r{{\eqtyp}\mbox{-}{\Impl}{\t}}.}
In this case, \(A = A' \Impl \t\) and \(B = \t\) for some \(A'\).
Hence, \(\tail{A} = \tail{t}\) by Definition~\ref{tail-def}.

\Case{\r{{\eqtyp}\mbox{-fix}}.}
\(A = \fix{X}C\) and \(B = C[\fix{X}C/X]\) for some \(X\) and \(C\).
Hence,
\begin{Eqnarray*}
    \tail{A} &=& \fix{X}\tail{C} & (by Definition~\ref{tail-def}) \\
	&\geqtyp& \tail{C}[\fix{X}\tail{C}/X] & (by \r{{\eqtyp}\mbox{-fix}}) \\
	&\geqtyp& \tail{C}[\tail{(\fix{X}C)}/X]
	    & (by Definition~\ref{tail-def}) \\
	&=& \tail{C[\fix{X}C/X]} & (by Proposition~\ref{tail-subst}) \\
	&=& \tail{B}.
\end{Eqnarray*}

\Case{\r{{\eqtyp}\mbox{-uniq}}.}
\(B = \fix{X}C\) for some \(X\) and \(C\) such that \(A \geqtyp C[A/X]\) and
\(C\) is proper in \(X\).
Note that \(\tail{C}\) is also proper in \(X\)
by Proposition~\ref{tail-proper}.
Hence,
\begin{Eqnarray*}
    \tail{A} &\geqtyp& \tail{C[A/X]} & (by induction hypothesis) \\
	&=& \tail{C}[\tail{A}/X] & (by Proposition~\ref{tail-subst}) \\
	&=& \fix{X}\tail{C} & (by \r{{\eqtyp}\mbox{-uniq}}) \\
	&=& \tail{(\fix{X}C)} & (by Definition~\ref{tail-def}) \\
	&=& \tail{B}.
\end{Eqnarray*}

\Case{\r{{\peqtyp}\mbox{-{\bf K}/{\bf L}}}.}
This case only applies to \(\peqtyp\).
In this case, \(A = \O(C \Impl D)\) and \(B = \O C \Impl \O D\) for some
\(C\) and \(D\).
Hence, \(\tail{A} = \O \tail{D} = \tail{B}\).
\fi 
\qed\CHECKED{2014/07/08}
\end{proof}

\begin{proposition}\label{geqtyp-depth}
    If\/ \(A \geqtyp B\), then \(\PNOdp(A,\,X) = \PNOdp(B,\,X)\) and
    \(\PNIdp(A,\,X) = \PNIdp(B,\,X)\).
\end{proposition}
\begin{proof}
By induction on the derivation of \(A \geqtyp B\),
and by cases on the rule applied last.
Suppose that \(A \geqtyp B\).
The cases other than
\r{{\eqtyp}\mbox{-fix}} or \r{{\eqtyp}\mbox{-uniq}} are again straightforward.
Let \(\Dp\) be either \(\Odp\) or \(\Idp\).
If either \(A\) or \(B\) is a {\tvariant}, then
so are they by Proposition~\ref{geqtyp-tvariant}; and therefore,
\(\PNDp(A,\,X) = \PNDp(B,\,X) = \infty\) by Definition~\ref{depth-def}.
Hence, we assume that neither is.
\ifdetail

\Cases{\r{{\eqtyp}\mbox{-reflex}},
    \r{{\eqtyp}\mbox{-symm}} and \r{{\eqtyp}\mbox{-trans}}.}
Trivial.

\Case{\r{{\eqtyp}\mbox{-}{\O}}.}
In this case, \(A = \O A'\), \(B = \O B'\) and \(A' \eqtyp B'\)
for some \(A'\) and \(B'\).
By induction hypothesis, \(\PNDp(A',\, X) = \PNDp(B',\,X)\).
Hence, \(\PNDp(A,\,X) = \PNDp(B,\,X)\)
by Definition~\ref{depth-def}.

\Case{\r{{\eqtyp}\mbox{-}{\Impl}}.}
Similar to the previous case.

\Case{\r{{\eqtyp}\mbox{-}{\Impl}{\t}}.}
We do not have to consider this case
since neither \(A\) nor \(B\) is {\tvariant} by assumption.

\Case{\r{{\eqtyp}\mbox{-fix}}.} There
\else 
In case of \r{{\eqtyp}\mbox{-fix}}, there
\fi 
exist some \(Y\) and \(C\) such that \(A = \fix{Y}C\) and \(B = C[A/Y]\),
where we can assume that \(Y \not= X\) without loss of generality.
Note that \(\Dp(C, Y) > 0\) by Propositions~\ref{O-depth-proper},
\ref{NIdp-positive} and \ref{PIdp-positive},
since \(C\) is proper in \(Y\).
Therefore,
\begin{Eqnarray*}
\Hbox{20pt}{\(\PNDp(B,X) = \PNDp(C[A/Y],X)\)} \\
    &=& \min(\PNDp(C,X),
	\PDp(C,Y) + \PNDp(A,X),
	\NDp(C,Y) + \NPDp(A,X))
	    & (by Proposition~\ref{depth-subst1}) \\
    &=& \min(\PNDp(C,X),\, \\
	&& \hphantom{\min(}
	 \PDp(C,Y) +
	    \min(\PNDp(C, X),\, \NDp(C, Y)+\NPDp(C, X)),
	    \hskip-200pt \\
	&& \hphantom{\min(}
	 \NDp(C,Y) +
	    \min(\NPDp(C, X),\, \NDp(C, Y)+\PNDp(C, X)))
	    & (by Definition~\ref{depth-def}) \\
    &=& \min(\PNDp(C,X),\, \\
	&& \hphantom{\min(}
	 \PDp(C,Y) +
	    \min(\PNDp(C, X),\, \NDp(C, Y)+\NPDp(C, X)),
	    \hskip-200pt \\
	&& \hphantom{\min(}
	 \NDp(C,Y) + \NPDp(C, X))&
	 (since \(\NDp(C, Y) > 0\)) \\
    &=& \min(\PNDp(C,X),\, \NDp(C,Y) + \NPDp(C, X)) &
	 (since \(\PDp(C, Y) > 0\)) \\
    &=& \PNDp(\fix{Y}C, X) & (by Definition~\ref{depth-def}).
\end{Eqnarray*}

\ifdetail
\Case{\r{{\eqtyp}\mbox{-uniq}}.} There
\else 
In case of \r{{\eqtyp}\mbox{-uniq}}, there 
\fi 
exist some \(Y\) and \(C\) such that \(B = \fix{Y}C\),
\(A \geqtyp C[A/Y]\), and \(C\) is proper in \(Y\),
where we again assume that \(Y \not= X\) without loss of generality.
In this case, \(C[A/Y]\) is not a {\tvariant} either,
by Proposition~\ref{geqtyp-tvariant}.
Note that \(\Dp(C, Y) > 0\) by Propositions~\ref{O-depth-proper},
\ref{NIdp-positive} and \ref{PIdp-positive},
since \(C\) is proper in \(Y\).
Therefore,
\begin{Eqnarray*}
\Hbox{10pt}{\(\PNDp(A,\,X) = \PNDp(C[A/Y], X)\)}
	    &&& (by induction hypothesis) \\
    &=& \min(\PNDp(C,X),\, \PDp(C,Y) + \PNDp(A,X), \NDp(C,Y) + \NPDp(A,X))
	    & (by Proposition~\ref{depth-subst1}) \\
    &=& \min(\PNDp(C,X),\, \NDp(C,Y) + \NPDp(A,X))
	    & \hskip-500pt\hfill
		(since \(\PDp(C, Y) > 0\);
		    valid even if \(\PNDp(A,X) = \infty\)) \\
    &=& \min(\PNDp(C,X),\, \NDp(C,Y) + \min(\NPDp(C,X),\,
		\NDp(C,Y) + \PNDp(A,X))) \mskip-20mu
	    & \hskip-100pt\hfill (by the equation so far) \\
    &=& \min(\PNDp(C,X),\,\NDp(C,Y) + \NPDp(C,X))
	    & \hskip-500pt\hfill
		(since \(\NDp(C, Y) > 0\);
		    valid even if \(\PNDp(A,X) = \infty\)) \\
    &=& \PNDp(\fix{Y}C,X)
	    & (by Definition~\ref{depth-def}).
\end{Eqnarray*}
\ifdetail
\Case{\r{{\peqtyp}\mbox{-{\bf K}/{\bf L}}}.}
This case only applies to \(\peqtyp\).
In this case, \(A = \O(C \Impl D)\) and \(B = \O C \Impl \O D\) for some
\(C\) and \(D\).
Hence,
\begin{eqnarray*}
    \PNOdp(\O (C \Impl D),\,X) &=& \PNOdp(C \Impl D,\,X) + 1 \\
	&=& \min(\NPOdp(C,\,X),\,\PNOdp(D,\,X)) + 1 \\
	&=& \min(\NPOdp(C,\,X) + 1,\,\PNOdp(D,\,X) + 1) \\
	&=& \min(\NPOdp(\O C,\,X),\,\PNOdp(\O D,\,X)) \\
	&=& \NPOdp(\O C \Impl \O D,\,X),~\mbox{and} \\[3pt]
    \PNIdp(\O (C \Impl D),\,X) &=& \PNIdp(C \Impl D,\,X) \\
	&=& \min(\NPIdp(C,\,X),\,\PNIdp(D,\,X)) + 1\\
	&=& \min(\NPIdp(C,\,X) + 1,\,\PNIdp(D,\,X) + 1) \\
	&=& \min(\NPIdp(\O C,\,X),\,\PNIdp(\O D,\,X)) \\
	&=& \NPIdp(\O C \Impl \O D,\,X).
\end{eqnarray*}
\fi 
\Qed\CHECKED{2014/07/08}
\end{proof}

\begin{proposition}\pushlabel
Suppose that \(A \geqtyp B\).
\begin{Enumerate}
\item \label{geqtyp-etv}
    \(\PNETV{A} = \PNETV{B}\).
\item \label{geqtyp-proper}
    \(A\) is proper in \(X\) if and only if so is \(B\).
\end{Enumerate}
\end{proposition}
\begin{proof}
Straightforward by Propositions~\ref{depth-finite-etv},
\ref{O-depth-proper} and \ref{geqtyp-depth}.
\qed\CHECKED{2014/04/28, 07/08}
\end{proof}

By the facts obtained so far,
we can show the following derived rules.

\begin{proposition}\pushlabel
\begin{Enumerate}
\item \itemlabel{geqtyp-fix-congr}
    If\/ \(A \geqtyp B\), then \(\fix{X}A \geqtyp \fix{X}B\).
\item \itemlabel{geqtyp-fix-nest}
    \(\fix{X}A[X/Y] \geqtyp \fix{X}A[A[X/Y]/Y]\).
\end{Enumerate}
\end{proposition}
\begin{proof}
For Item~\itemref{geqtyp-fix-congr}, suppose that \(A \geqtyp B\).
Then, \(\fix{X}A \geqtyp A[\fix{X}A/X] \geqtyp B[\fix{X}A/X]\)
by \r{{\eqtyp}\mbox{-fix}} and Proposition~\ref{geqtyp-subst}.
Therefore, \(\fix{X}A \geqtyp \fix{X}B\) by \r{{\eqtyp}\mbox{-uniq}}.
For Item~\itemref{geqtyp-fix-nest},
whether \(X = Y\) or not,
\begin{Eqnarray*}
    \fix{X}A[X/Y] &\geqtyp& A[X/Y][\fix{X}A[X/Y]/X]
	& (by \r{{\eqtyp}\mbox{-fix}}) \\
    &\geqtyp& A[\fix{X}A[X/Y]/Y][\fix{X}A[X/Y]/X]
	& (since \(X \not\in \FTV{\fix{X}A[X/Y]}\)) \\
    &\geqtyp& A[A[X/Y][\fix{X}A[X/Y]/X]/Y][\fix{X}A[X/Y]/X]
	& (by \r{{\eqtyp}\mbox{-fix}} and Proposition~\ref{geqtyp-subst}) \\
    &\geqtyp& A[A[X/Y]/Y][\fix{X}A[X/Y]/X]
	& \hskip-300pt\hfill
	(since \(X \not\in \FTV{A[X/Y][\fix{X}A[X/Y]/X]}\)).
\end{Eqnarray*}
Hence, \(\fix{X}A[X/Y] \geqtyp \fix{X}A[A[X/Y]/Y]\) by \r{{\eqtyp}\mbox{-uniq}},
again.
\qed\CHECKED{2014/07/13}
\end{proof}

Proposition~\ref{geqtyp-tvariant} says
that the relation \(\geqtyp\) preserves
whether a type expression is a {\tvariant} or not.
Furthermore, we can show that
a type expression \(A\) is a {\tvariant}
if and only if \(A \geqtyp \t\).
It does not depend on which of \(\eqtyp\) and \(\peqtyp\)
we consider as the equality between types.
Because it is syntactically decidable whether a type expression
is a {\tvariant} or not,
it is decidable whether \(A \geqtyp \t\) or not,

\begin{lemma}\label{geqtyp-tail-t-lemma}\pushlabel
Let \(\tail{A}
    = \O^{m_0}\fix{X_1}\O^{m_1}\fix{X_2}\O^{m_2}\ldots\fix{X_n}\O^{m_n}Y\)
for some \(n\), \(m_0\), \(m_1\), \(m_2\), \(\ldots\), \(m_n\),
\(X_1\), \(X_2\), \(\ldots\), \(X_n\) and \(Y\).
\begin{Enumerate}
\item \label{geqtyp-tail-t1-lemma}
    \(A \geqtyp \t\) if the following is the case.
    \begin{enumerate}
    \item[(a)] \(Y = X_i\) for some \(i\) such that
	\(Y \notin \{\,\,X_{i{+}1},\,X_{i{+}2},\,\ldots,X_n\,\}\),
	and \(m_i+m_{i{+}1}+m_{i{+}2}+\ldots+m_n \ge 1\).
    \end{enumerate}
\item \label{geqtyp-tail-t2-lemma} 
    \(A[\t/Y] \geqtyp \t\) if the following is the case.
    \begin{enumerate}
    \item[(b)] \(Y \notin \{\,\,X_1,\,X_2,\,\ldots,X_n\,\}\).
    \end{enumerate}
\end{Enumerate}
\end{lemma}
\begin{proof}
By simultaneous induction on \(h(A)\), and by cases on the form of \(A\).

\Case{\(A = Z\) for some \(Z\).}
In this case, \(Y = Z\), \(n = 0\) and \(m_0 = 0\) since \(\tail{A} = Z\).
Hence, (a) is not the case.
The second item is trivial since \(A[\t/Y] = Y[\t/Y] = \t\).

\Case{\(A = B \Impl C\) for some \(B\) and \(C\).}
In this case, \(\tail{A} = \tail{C}\).
If (a) holds, then
\(C \geqtyp \t\) by induction hypothesis.
Hence, \(A \geqtyp B \Impl \t \geqtyp \t\)
by \r{{\eqtyp}\mbox{-}{\Impl}{\t}}.
On the other hand,
if (b) is the case, then
\(C[\t/Y] \geqtyp \t\) by induction hypothesis.
Hence, \(A[\t/Y] = B[\t/Y] \Impl C[\t/Y] \geqtyp B[\t/Y] \Impl \t
\geqtyp \t\) by \r{{\eqtyp}\mbox{-}{\Impl}{\t}}.

\Case{\(A = \O B\) for some \(B\).}
In this case, \(\tail{A} = \O \tail{B}\).
If (a) holds, then
\(B \geqtyp \t\) by induction hypothesis.
Therefore, \(A \geqtyp \O \t \geqtyp \t\) by Proposition~\ref{geqtyp-t1}.
In case (b),
\(B[\t/Y] \geqtyp \t\) by induction hypothesis.
Hence, \(A[\t/Y] \geqtyp \O B[\t/Y] \geqtyp \O \t \geqtyp \t\)
by Proposition~\ref{geqtyp-t1}, again.

\Case{\(A = \fix{Z}B\) for some \(Z\) and \(B\).}
In this case, \(\tail{(\fix{Z}B)} = \fix{Z}\tail{B}\), \(m_0 = 0\),
\(Z = X_1\), \(n \ge 1\) and
\(\tail{B} = \O^{m_1}\fix{X_2}\O^{m_2}\ldots\fix{X_n}\O^{m_n}Y\).
First, suppose that (a) holds.
If \(i \ge 2\), then
\(B \geqtyp \t\) by induction hypothesis.
Hence, \(B[\t/Z] \geqtyp \t[\t/Z] = \t\)
by Proposition~\ref{geqtyp-subst}; and therefore,
\(\fix{Z}B \geqtyp \t\) by \r{{\eqtyp}\mbox{-uniq}}.
On the other hand, if \(i = 1\), i.e., \(Y = Z\), then
\(B[\t/Z] \geqtyp \t\) by induction hypothesis,
since \(Z = Y \not\in \{\,\,X_2,\,X_3,\,\ldots,X_n\,\}\).
Therefore, again \(A = \fix{Z}B \geqtyp \t\) by \r{{\eqtyp}\mbox{-uniq}}.
Second, suppose that (b) holds.
We get \(B[\t/Y] \geqtyp \t\) by induction hypothesis.
Hence, \(A[\t/Y] = (\fix{Z}B)[\t/Y] = \fix{Z}B[\t/Y] \geqtyp
    \fix{Z}\t \geqtyp \t[\fix{Z}\t/Z] = \t\)
by Proposition~\ref{geqtyp-fix-congr} and \r{{\eqtyp}\mbox{-fix}}.
\qed\CHECKED{2014/04/26, 04/29, 07/08}
\end{proof}

\begin{theorem}\label{geqtyp-t-tvariant}
A type expression \(A\) is a {\tvariant} if and only if \(A \geqtyp \t\).
\end{theorem}
\begin{proof}
The ``only if'' part is straightforward from
Lemma~\ref{geqtyp-tail-t1-lemma}.
For the ``if'' part, suppose that \(A \geqtyp \t\), and
let \(\tail{A}
    = \O^{m_0}\fix{X_1}\O^{m_1}\fix{X_2}\O^{m_2}\ldots\fix{X_n}\O^{m_n}Y\).
We get \(\tail{A} \geqtyp \tail{\t} = \t\) by Proposition~\ref{geqtyp-tail}
and Definition~\ref{tail-def}.
Hence, \(Y = X_i\) for some \(i\) (\(1 \le i \le n\))
by Definition~\ref{etv-def} and
Proposition~\ref{geqtyp-etv}, because \(\PETV{\t} = \{\}\).
Therefore, taking the largest such \(i\), we also establish
\(Y \not\in \{\,\,X_{i{+}1},\,X_{i{+}2},\,\ldots,X_n\,\}\) and
\(m_i+m_{i{+}1}+m_{i{+}2}+\ldots+m_n \ge 1\) by Proposition~\ref{tail-te}.
\qed\CHECKED{2014/04/29, 07/08}
\end{proof}

By Theorem~\ref{geqtyp-t-tvariant},
we can show other basic properties about the equality of
type expressions as follows.
In the succeeding sections, we might use Theorem~\ref{geqtyp-t-tvariant}
in proofs without mention.

\begin{proposition}\label{geqtyp-basic}\pushlabel
\begin{Enumerate}
\item \itemlabel{geqtyp-var-t}
    \(X \ngeqtyp \t\). 
\item \itemlabel{geqtyp-O-t}
    \(\O A \geqtyp \t\) if and only if \(A \geqtyp \t\).
\item \itemlabel{geqtyp-Impl-t}
    \(A \Impl B \geqtyp \t\) if and only if \(B \geqtyp \t\).
\item \itemlabel{geqtyp-tail-t}
    \(A \geqtyp \t\) if and only if \(\tail{A} \geqtyp \t\).
\end{Enumerate}
\end{proposition}
\begin{proof}
Straightforward from
Proposition~\ref{tvariant-basic} and Theorem~\ref{geqtyp-t-tvariant}.
\qed\CHECKED{2014/04/30, 07/08}
\end{proof}

\begin{proposition}\label{geqtyp-no-etv-subst}
    If \(X \not\in \ETV{A}\), then \(A \geqtyp A[B/X]\).
\end{proposition}
\begin{proof}
If \(A\) is a {\tvariant}, then \(A \geqtyp A[B/X]\)
by Proposition~\ref{tvariant-subst1} and Theorem~\ref{geqtyp-t-tvariant}.
\ifdetail
So we only consider the case when \(A\) is not.
The rest of the proof proceeds by induction on \(h(A)\),
and by cases on the form of \(A\).

\Case{\(A = Y\) for some \(Y\).}
In this case, \(Y \not= X\) since \(X \not\in \ETV{A}\).
Hence, \(A[B/X] = A \geqtyp A\).

\Case{\(A = \O C\) for some \(C\).}
Since \(\ETV{A} = \ETV{C}\), we get
\(A[B/X] = \O C[B/X] \geqtyp \O C = A\)
by induction hypothesis.

\Case{\(A = C \Impl D\) for some \(C\) and \(D\).}
In this case, \(\ETV{A} = \ETV{C} \cup \ETV{D}\)
since \(A\) is not a {\tvariant}.
Hence, by induction hypothesis,
\(C[B/X] \geqtyp C\) and \(D[B/X] \geqtyp D\).
Therefore,
\(A[B/X] = C[B/X] \Impl D[B/X] \geqtyp C \Impl D = A\).

\Case{\(A = \fix{Y}C\) for some \(Y\) and \(C\).}
We can assume that \(Y \not\in \{\,X\,\} \cup \FTV{B}\)
without loss of generality.
Hence, \(X \not\in \ETV{C}\) because
\(\ETV{A} = \ETV{C} - \{\,Y\,\}\) by Definition~\ref{etv-def}.
Therefore,
\(A[B/X] = \fix{Y}C[B/X] \geqtyp \fix{Y}C = A\)
by induction hypothesis and Proposition~\ref{geqtyp-fix-congr}.
\else 
For the case when \(A\) is not,
by straightforward induction on \(h(A)\), and by cases on the form of \(A\).
\fi 
\qed\CHECKED{2014/04/29, 07/08}
\end{proof}

\Subsection{Canonical forms of type expressions}

So far we have confirmed that the equality defined by the derivation rules
preserves the basic properties of type expressions, such as
{\tvariant}s, properness,
\(\PNOdp\), \(\PNIdp\) and \(ETV^{\pm}\).
However, we know almost nothing about when two type expressions are
equal to each other.
For example, one might conjecture the following properties.
\begin{Enumerate}
\item \(X \geqtyp Y\) if and only if\/ \(X = Y\).
\item \(\O A \geqtyp \O B\) if and only if\/ \(A \geqtyp B\).
\item \(A \Impl B \geqtyp C \Impl D\) if and only if\/
    (a) \(A \geqtyp C\) and \(B \geqtyp D\), or
    (b) \(B \geqtyp D \geqtyp \t\).
\end{Enumerate}
Furthermore,
\begin{Enumerate}
\item[4.] \(X \ngeqtyp \O A\).
\item[5.] \(X \ngeqtyp A \Impl B\).
\item[6.] \(\O A \eqtyp B \Impl C\) implies \(A \eqtyp C \eqtyp \t\),
    which should not be the case for \(\peqtyp\)
    because of \r{{\peqtyp}\mbox{-{\bf K}/{\bf L}}}.
\end{Enumerate}
In fact, Items~1, 4 and 5 can be shown by Theorems~\ref{eqtyp-soundness}
and \ref{peqtyp-soundness}
considering an appropriate type environment under
a certain non-trivial interpretation\footnote{%
    Proposition~\ref{geqtyp-var-t} can be also shown in the same way.}.
However, we need some more preparation to realize that the other expectations
are also fulfilled.
Because of the existence of \r{{\eqtyp}\mbox{-}{\Impl}{\t}}
(and also \r{{\peqtyp}\mbox{-{\bf K}/{\bf L}}} in case of \(\peqtyp\)),
the equivalence relation \(\eqtyp\) (or \(\peqtyp\)) is not identical to
the equality as (possibly infinite) labeled trees obtained by unfolding
recursive types.
So it is convenient to define canonical representation of type expressions.

\begin{definition}[Canonical type expressions]
    \ilabel{canon-type-def}{canonical type expressions}
    \ilabel*{type expressions!canonical}
    \ilabel{cte-def}{CTExp@$\protect\CTE, \protect\CTEp$}
    \ilabel*{syntax!CTExp@$\protect\CTE,\protect\CTEp$}
We define a set \(\CTE\) (respectively, \(\CTEp\)) of
{\em canonical} type expressions as follows.
\begin{eqnarray*}
    \CTE & \bnfdef & \t \bnfor \O^n\: \TV
			\bnfor \O^n\, (\TE_1 \Impl \TE_2) \\
    \CTEp & \bnfdef & \t \bnfor \O^n\: \TV
			\bnfor \TE_1 \Impl \TE_2
\end{eqnarray*}
where \(n\) is an arbitrary non-negative integer and
\(\TE_2\) is not a {\tvariant}.
\end{definition}

The following definition can be regarded as an effective procedure
for finding canonical forms of given type expressions, which will be
immediately verified by Proposition~\ref{geqtyp-canon}.
Later, also by Proposition~\ref{canon-congr},
it will be shown that
\(A \ngeqtyp B\,\) if canonical forms of two type expressions \(A\) and \(B\)
fit in different categories of the three forms above.

\begin{definition}
    \ilabel{canon-def}{c@$\protect\Canon{A}, \Canonp{A}$}
For each type expression \(A\), we define \(\Canon{A}\) as follows.
\begin{Eqnarray*}
\Canon{A} &=& \t & (\(A\) is a {\tvariant}) \\
\Canon{X} &=& X \\
\Canon{(\O A)} &=& \O \Canon{A} & (\(\,\O A\) is not a \tvariant) \\
\Canon{(A \Impl B)} &=& A \Impl B & (\(A \Impl B\) is not a \tvariant) \\
\Canon{(\fix{X}A)} &=&
    \Canon{A}[\fix{X}A/X] & (\(\fix{X}A\) is not a \tvariant) \\[6pt]
\na{We similarly define \(\Canonp{A}\)
by adjusting the definition of \(\Canon{(\O A)}\) as follows.}
\noalign{\vskip6pt}
\Canonp{(\O A)} &=&
    \Choice{%
	\O \O^n X & \mbox{(\(\Canonp{A} = \O^n X\))} \\
	\O B \Impl \O C \mskip60mu& \mbox{(\(\Canonp{A} = B \Impl C\))}}
	\mskip-500mu
\end{Eqnarray*}
\end{definition}
For example, let \(A = \O\,\fix{X}\,\O(X \Impl \O Y)\) and
\(B = X \Impl \fix{Y}\,X \Impl \O (Z \Impl Y)\).
Then, \(\Canon{A} = \O\O((\fix{X}\,\O(X \Impl \O Y)) \Impl \O Y)\),
\(\Canonp{A} = \O\O(\fix{X}\,\O(X \Impl \O Y)) \Impl \O\O\O Y\), and
\(\Canon{B} = \Canonp{B} = \t\).

\ilabel{canong-def}{c Z@$\protect\Canong{A}$}
In order to make the description concise,
we will use \(\Canong{A}\),
to denote either \(\Canon{A}\) or \(\Canonp{A}\), according to
the context in which we consider \(\eqtyp\) or \(\peqtyp\), respectively.

\begin{proposition}\label{geqtyp-canon}
\(\Canong{A}\) is a canonical type expression
such that \(\Canong{A} \geqtyp A\), that is,
\(\Canon{A} \eqtyp A\) and \(\Canonp{A} \peqtyp A\).
\end{proposition}
\begin{proof}
If \(A\) is a {\tvariant}, then \(\Canong{A} = \t \geqtyp A\) by
Definition~\ref{canon-def} and Theorem~\ref{geqtyp-t-tvariant}.
Hence, we assume that \(A\) is not a {\tvariant};
and therefore, \(A \ngeqtyp \t\).
The proof proceeds by induction on \(h(A)\),
and by cases on the form of \(A\).
Most cases are straightforward.
In the case when \(A = \fix{X}B\) for some \(X\) and \(B\),
we get that \(\Canong{B}\) is canonical and \(\Canong{B} \geqtyp B\)
by induction hypothesis.
Hence, we get \(\Canong{A} = \Canong{B}[A/X] \geqtyp B[A/X] \geqtyp A\)
by Proposition~\ref{geqtyp-subst} and \r{{\eqtyp}\mbox{-fix}},
and get \(\Canong{B} \not= \t\) from \(A \ngeqtyp \t\),
since \(\Canong{B} = \t\) implies
\(A = \fix{X}B \geqtyp \fix{X}\Canong{B} = \fix{X}\t \geqtyp \t\)
by Proposition~\ref{geqtyp-fix-congr} and \r{{\eqtyp}\mbox{-fix}}.
Therefore,
\begin{Enumerate}
\item[(a)] \(\Canong{B} = \O^n Y\) for some \(n\) and \(Y\), or
\item[(b)] \(\Canong{B} = \O^n(C \Impl D)\)
for some \(n\), \(C\) and \(D\) such that \(D\) is not a {\tvariant},
where \(n = 0\) in case of \(\peqtyp\).
\end{Enumerate}
In case (a), \(Y \not= X\) since \(Y = X\) implies \(A = \fix{X}B
\geqtyp \fix{X}\O^n X \geqtyp \t\)
by Propositions~\ref{geqtyp-fix-congr} and \ref{geqtyp-t2}.
Hence, \(\Canong{A}\) is canonical
since \(\Canong{A} = \Canong{B}[A/X] = (\O^n Y)[A/X] = \O^n Y\).
As for case (b),
since neither \(A\) nor \(D\) is a {\tvariant},
\(D[A/X]\) is not a {\tvariant} either, by Proposition~\ref{tvariant-subst2},
which implies that \(\Canong{B}[A/X]\) is canonical.
Thus, \(\Canong{A}\) is a canonical type expression.
\qed\CHECKED{2014/04/29, 07/08}
\end{proof}

The following lemma will be used in Proposition~\ref{canon-congr}.

\begin{lemma}\label{canon-subst-lemma}
If neither \(A[\t/X]\) nor \(B\) is a {\tvariant}, then
    \(\Canong{A[B/X]} = \Canong{A}[B/X]\).
\end{lemma}
\begin{proof}
By induction on \(h(A)\), and by cases on the form of \(A\).
Suppose that neither \(A[\t/X]\) nor \(B\) is a {\tvariant}, which
also implies that neither \(A\) nor \(A[B/X]\) is a {\tvariant}
by Propositions~\ref{tvariant-subst1} and \ref{tvariant-subst2}.

\Case{\(A = Y\) for some \(Y\).}
In this case,
we get \(Y \not= X\) since \(A[\t/X]\) is not a {\tvariant}.
Therefore, \(\Canong{A[B/X]} = \Canong{Y[B/X]} = \Canong{Y} = Y\) and
\(\Canong{A}[B/X] = \Canong{Y}[B/X] = Y[B/X] = Y\).

\Case{\(A = \O C\) for some \(C\).}
Since \(A[\t/X]\) is not a {\tvariant},
\(C[\t/X]\) is not either, by Proposition~\ref{O-tvariant}.
Hence, by induction hypothesis,
\begin{eqnarray}
    \label{canon-subst-lemma-01}
    \Canong{C[B/X]} &=& \Canong{C}[B/X].
\end{eqnarray}
Furthermore, \(C\) is not a {\tvariant} since \(A\) is not.
Hence, there are two subcases on the form of \(\Canong{C}\) as follows.
\begin{enumerate}[(\romannumeral 1)]\itemsep=6pt
\item \(\Canong{C} = \O^n Y\) for some \(n\) and \(Y\).
    By Definition~\ref{canon-def},
    \(\Canong{A} = \O^{n{+}1} Y\) in this case.
    By Propositions~\ref{geqtyp-canon} and \ref{geqtyp-subst}, we get
    \(A[\t/X] \geqtyp \Canong{A}[\t/X] = (\O^{n{+}1} Y)[\t/X]\).
    Thus, \(Y \not= X\) since \(A[\t/X]\) is not a {\tvariant}.
    Hence, \(\Canong{C[B/X]} = \Canong{C}[B/X]
    = (\O^n Y)[B/X] = \O^n Y\) by (\ref{canon-subst-lemma-01}).
    Therefore, \(\Canong{A[B/X]} = \Canong{(\O C)[B/X]} =
    \Canong{(\O C[B/X])} = \O^{n{+}1} Y\),
    and \(\Canong{A}[B/X] = (\O^{n{+}1} Y)[B/X]
    = \O^{n{+}1} Y\).
\item \(\Canong{C} = \O^n(D\Impl E)\) for some \(n\), \(D\) and \(E\),
    where \(n = 0\) in case of \(\peqtyp\).
    We get \(\Canon{A} = \O^{n{+}1}(D \Impl E)\)
    and \(\Canonp{A} = \O D \Impl \O E\) by Definition~\ref{canon-def}.
    On the other hand, we also get
    \(\Canong{C[B/X]} = \Canong{C}[B/X] = \O^n(D[B/X] \Impl E[B/X])\)
    by (\ref{canon-subst-lemma-01}).
    Therefore, in case of \(\eqtyp\),
    \(\Canon{A[B/X]}
    = \Canon{(\O C[B/X])}
    = \O^{n{+}1}(D[B/X] \Impl E[B/X])\)
    and \(\Canon{A}[B/X]
    = (\O^{n{+}1}(D \Impl E))[B/X]
    = \O^{n{+}1}(D[B/X] \Impl E[B/X])\).
    Similarly, for \(\peqtyp\),
    \(\Canonp{A[B/X]}\penalty0
    = \Canonp{(\O C[B/X])} = \O D[B/X] \Impl \O E[B/X]\)
    and \(\Canonp{A}[B/X]
    = (\O D \Impl \O E)[B/X] = \O D[B/X] \Impl \O E[B/X]\).
\end{enumerate}

\Case{\(A = C \Impl D\) for some \(C\) and \(D\).}
Since \(A\) is not a {\tvariant}, \(D\) is not, either.
Hence, \(\Canong{A} = C \Impl D\) by Definition~\ref{canon-def}.
Therefore, \(\Canong{A[B/X]} = \Canong{(C \Impl D)[B/X]}
= \Canong{(C[B/X] \Impl D[B/X])} = C[B/X] \Impl D[B/X]\),
and \(\Canong{A}[B/X] = (C \Impl D)[B/X] = C[B/X] \Impl D[B/X]\).

\Case{\(A = \fix{Y}C\) for some \(Y\) and \(C\).}
We assume that \(Y \not\in \{\,X\,\} \cup \FTV{B}\)
without loss of generality.
Since \(A[\t/X]\) is not a {\tvariant}, \(C[\t/X]\) is not, either.
Therefore, \(\Canong{A[B/X]} = \Canong{(\fix{Y}C[B/X])}
= \Canong{C[B/X]}[\fix{Y}C[B/X]/Y] = \Canong{C[B/X]}[A[B/X]/Y]
= \Canong{C}[B/X]\penalty100[A[B/X]/Y] = \Canong{C}\penalty0[B/X,A[B/X]/Y]\)
by Definition~\ref{canon-def} and the induction hypothesis on \(C\).
On the other hand,
\(\Canong{A}[B/X] = \Canong{C}[A/Y][B/X] = \Canong{C}[B/X,A[B/X]/Y]\).
\qed\CHECKED{2014/04/29, 07/08}
\end{proof}

\begin{proposition}\label{canon-congr}
Suppose that \(A \geqtyp B\).
\begin{enumerate}[(\romannumeral 1)]\itemsep=3pt
    \item If \(\Canong{A} = \t\), then \(\Canong{B} = \t\).
    \item If \(\Canong{A} = \O^n X\) for some \(n\) and \(X\),
	then \(\Canong{B} = \O^n X\).
    \item If \(\Canong{A} = \O^n(C \Impl D)\) for some \(n\), \(C\) and \(D\),
	then \(\Canong{B} = \O^n(C' \Impl D')\) for some \(C'\) and \(D'\)
such that \(C \geqtyp C'\) and \(D \geqtyp D' \ngeqtyp \t\),
where \(n = 0\) in case of\/ \(\peqtyp\).
\end{enumerate}
\end{proposition}
\begin{proof}
If \(A \geqtyp B\), and if
\(A\) or \(B\) is a {\tvariant}, then \(\Canong{A} = \Canong{B} = \t\)
by Definition~\ref{canon-def} and Theorem~\ref{geqtyp-t-tvariant}.
Hence, we assume that neither \(A\) nor \(B\) is a {\tvariant}, i.e.,
\(A \ngeqtyp \t\) and \(B \ngeqtyp \t\).
We prove the following more general claim
by induction on the derivation of \(A \geqtyp B\) or \(B \geqtyp A\).
\begin{quote}
    If \(A \geqtyp B \ngeqtyp \t\) or \(B \geqtyp A \ngeqtyp \t\),
    then ({\romannumeral 2}) and ({\romannumeral 3}) hold.
\end{quote}
Suppose that \(A \geqtyp B \ngeqtyp \t\) or \(B \geqtyp A \ngeqtyp \t\).
By cases on the last rule applied in the derivation.
Most cases are trivial.
\ifdetail

\Cases{\r{{\eqtyp}\mbox{-reflex}},
    \r{{\eqtyp}\mbox{-symm}} and \r{{\eqtyp}\mbox{-trans}}.}
Trivial.

\Case{\r{{\eqtyp}\mbox{-}{\O}}.}
If \(A \geqtyp B \ngeqtyp \t\), then
\(A = \O A'\), \(B = \O B'\) and \(A' \eqtyp B'\) for some \(A'\) and \(B'\).
Note that \(\Canong{A'} \not= \t\) by Propositions~\ref{geqtyp-canon}
and \ref{geqtyp-t1} since \(A \ngeqtyp \t\).
If \(\Canong{A'} = \O^n X\) for some \(n\) and \(X\), then
    \(\Canong{A} = \O^{n+1} X\) by Definition~\ref{canon-def}, and
\(\Canong{B'} = \O^n X\) by induction hypothesis.
Hence, \(\Canong{B} = \O^{n+1} X\) by Definition~\ref{canon-def}.
On the other hand,
If \(\Canong{A'} = \O^n (C \Impl D)\) for some \(n\), \(C\) and \(D\),
where \(n = 0\) in case of \(\peqtyp\), then
\(\Canon{A} = \O^{n+1} (C \Impl D)\), \(\Canonp{A} = \O C \Impl \O D\),
and \(\Canong{B'} = \O^n (C' \Impl D')\) for some \(C'\) and \(D'\)
such that \(C \geqtyp C'\) and \(D \geqtyp D' \ngeqtyp \t\)
by induction hypothesis.
Therefore,
\(\Canon{B} = \O^{n+1} (C' \Impl D')\) and \(\Canonp{B} = \O C' \Impl \O D'\).
Thus, both ({\romannumeral 2}) and ({\romannumeral 3}) hold.
The proof for the case when \(B \geqtyp A \ngeqtyp \t\) is symmetrical.

\Case{\r{{\eqtyp}\mbox{-}{\Impl}}.}
In this case, \(A = C \Impl D\) and \(B = E \Impl F\) for some \(C\), \(D\),
\(E\) and \(F\) such that \(C \geqtyp E\) and \(D \geqtyp F\).
Note that \(D \geqtyp F \ngeqtyp \t\)
by Propositions~\ref{Impl-tvariant} and \ref{geqtyp-t-tvariant}
since \(A \geqtyp B \ngeqtyp \t\).
By Definition~\ref{canon-def},
\(\Canong{A} = C \Impl D\) and \(\Canong{B} = E \Impl F\).
Hence, both ({\romannumeral 2}) and ({\romannumeral 3}) hold, again.

\Case{\r{{\eqtyp}\mbox{-}{\Impl}{\t}}.}
We do not have to consider this case
since neither \(A\) nor \(B\) is {\tvariant} by assumption.
\else
The only interesting cases are the following three.
\fi 

\Case{\r{{\eqtyp}\mbox{-fix}}.}
In this case, for some \(Y\) and \(C\),
\(A = \fix{Y}C\) and \(B = C[\fix{Y}C/Y]\), or vice versa.
Note that \(C[\t/Y] \ngeqtyp \t\),
since \(C[\t/Y] \geqtyp \t\) implies \(\fix{Y}C \geqtyp \t\)
by \r{{\eqtyp}\mbox{-uniq}}.
Therefore,
\(\Canong{(\fix{Y}C)} = \Canong{C}[\fix{Y}C/Y] = \Canong{C[\fix{Y}C/Y]}\)
by Lemma~\ref{canon-subst-lemma}.

\Case{\r{{\eqtyp}\mbox{-uniq}}.}
In this case, there exist some \(Y\), \(C\) and \(D\)
such that \(D \geqtyp C[D/Y]\), and either
\(A = D\) and \(B = \fix{Y}C\), or vice versa.
Similarly to the previous case, \(C[\t/Y]\) is not a {\tvariant}; and hence,
\(\Canong{C[D/Y]} = \Canong{C}[D/Y]\) by Lemma~\ref{canon-subst-lemma}.
We also get \(\Canong{C} \not= \t\)
and \(\Canong{C}\not= \O^n Y\) for any \(n\),
since \(\fix{Y}C\) is not a {\tvariant}.
That is, there are two subcases on the form of \(\Canong{C}\) to consider.
\begin{Enumerate}
\item[(a)] \(\Canong{C} = \O^n X\) for some \(n\) and \(X\) such that
\(X \not= Y\).
\item[(b)] \(\Canong{C} = \O^n(E \Impl F)\) for some \(n\), \(E\) and \(F\)
such that \(F \ngeqtyp \t\), where \(n = 0\) in case of \(\peqtyp\).
\end{Enumerate}
In case (a),
\(\Canong{C[D/Y]} = \Canong{C}[D/Y] = (\O^n X)[D/Y] = \O^n X\); and hence,
\(\Canong{D} = \O^n X\) by induction hypothesis.
On the other hand,
\(\Canong{(\fix{Y}C)} = \Canong{C}[\fix{Y}C/Y]
    = (\O^n X)[\fix{Y}C/Y] = \O^n X\).
As for case (b),
\(\Canong{C[D/Y]} = \Canong{C}[D/Y] = (\O^n(E \Impl F))[D/Y]
= \O^n(E[D/Y] \Impl F[D/Y])\);
and hence, by induction hypothesis,
\(\Canong{D} = \O^n(E' \Impl F')\) for some \(E'\) and \(F'\)
such that \(E' \geqtyp E[D/Y]\) and \(F' \geqtyp F[D/Y] \ngeqtyp \t\).
On the other hand,
\(\Canong{(\fix{Y}C)} = \Canong{C}[\fix{Y}C/Y]
    = (\O^n(E \Impl F))[\fix{Y}C/Y]
    = \O^n(E[\fix{Y}C/Y] \Impl F[\fix{Y}C/Y])\).
Note that
\(E' \geqtyp E[D/Y] \geqtyp E[\fix{Y}C/Y]\) and
\(F' \geqtyp F[D/Y] \geqtyp F[\fix{Y}C/Y]\) by Proposition~\ref{geqtyp-subst}
since \(D \geqtyp \fix{Y}C\).

\Case{\r{{\peqtyp}\mbox{-{\bf K}/{\bf L}}}.}
This case only applies to \(\peqtyp\).
In this case,
\(A = \O(C \Impl D)\) and \(B = \O C \Impl \O D\), or vice versa.
Therefore,
\(\Canonp{A} = \Canonp{B} = \O C \Impl \O D\) by Definition~\ref{canon-def}.
\qed\CHECKED{2014/07/08}
\end{proof}

\begin{proposition}\label{canon-geqtyp-equiv}
\(A \geqtyp B\) if and only if either
\begin{enumerate}[(\romannumeral 1)]\itemsep=3pt
\item \(\Canong{A} = \Canong{B} = \t\),
\item \(\Canong{A} = \Canong{B} = \O^n X\)
for some \(n\) and \(X\), or
\item \(\Canong{A} = \O^n(C \Impl D)\)
and \(\Canong{B} = \O^n(E \Impl F)\)
for some \(n\), \(C\), \(D\), \(E\) and \(F\)
such that \(C \geqtyp E\) and \(D \geqtyp F\).
Furthermore, \(n = 0\) in case of\/ \(\peqtyp\).
\end{enumerate}
\end{proposition}
\begin{proof}
The ``if'' part is straightforward by
Proposition~\ref{geqtyp-canon}, \r{{\eqtyp}\mbox{-}{\O}} and
\r{{\eqtyp}\mbox{-}{\Impl}}.
We get the ``only if'' part by
Propositions~\ref{geqtyp-canon} and \ref{canon-congr}.
\qed\CHECKED{2014/04/30, 07/08}
\end{proof}

\Subsection{Other properties of the equality of types}

In the rest of this section, we shall continue to examine other properties
of \(\eqtyp\) and \(\peqtyp\) that are necessary for the proofs
in the succeeding sections.
The conjectures raised at the beginning of the previous subsection will
be proved.

\begin{proposition}\label{geqtyp-O-Impl}\pushlabel
\begin{Enumerate}
\item \label{eqtyp-O-Impl}
If\/ \(\O A \eqtyp B \Impl C\) then \(A \eqtyp C \eqtyp \t\).
\item \label{peqtyp-O-Impl}
If\/ \(\O A \peqtyp B \Impl C\), then (a) \(A \peqtyp C \peqtyp \t\), or
    (b) \(A \peqtyp D \Impl E\) for some \(D\) and \(E\) such that
    \(\O D \peqtyp B\) and \(\O E \peqtyp C\).
\end{Enumerate}
\end{proposition}
\begin{proof}
For Item~1, suppose that \(\O A \eqtyp B \Impl C\).
Since either \(A \not\eqtyp \t\) or \(C \not\eqtyp \t\) implies
that neither \(\O A\) nor \(B \Impl C\) is a {\tvariant}
by Proposition~\ref{tvariant-basic} and Theorem~\ref{geqtyp-t-tvariant},
we get \(\Canon{(\O A)} = \O \Canon{A}\) and
\(\Canon{(B \Impl C)} = B \Impl C\) by Definition~\ref{canon-def}.
However, it contradicts Proposition~\ref{canon-geqtyp-equiv}.
As for Item~2, suppose that \(\O A \peqtyp B \Impl C\), and
that either \(A \npeqtyp \t\) or \(C \npeqtyp \t\).
Since either one implies that \(B \Impl C\) is not a {\tvariant},
we get \(\Canonp{(B \Impl C)} = B \Impl C\).
Therefore, by Propositions~\ref{geqtyp-canon}, \ref{canon-geqtyp-equiv}
and Definition~\ref{canon-def}, we get
\(A \peqtyp \Canonp{A} = D \Impl E\) for some \(D\) and \(E\) such that
\(\Canonp{(\O A)} = \O D \Impl \O E\), \(\O D \peqtyp B\),
and \(\O E \peqtyp C\).
\qed\CHECKED{(2014/07/09, 07/21)}
\end{proof}

\begin{proposition}\label{not-geqtyp}\pushlabel
Let \(m\) and \(n\) be non-negative integers.
\begin{Enumerate}
\item \label{not-geqtyp-var-Impl} \(\O^m X \ngeqtyp \O^n (A \Impl B)\).
\item \label{not-geqtyp-var-t} \(\O^m X \ngeqtyp \t\).
\item \label{geqtyp-var-var}
	\(\O^m X \geqtyp \O^n Y\) if and only if\/ \(X = Y\) and\/ \(m = n\).
\item \label{not-geqtyp-var-O} \(X \ngeqtyp \O A\).
\end{Enumerate}
\end{proposition}
\begin{proof}
Straightforward from Proposition~\ref{canon-congr}
and Definition~\ref{canon-def}\footnote{%
Proposition~\ref{not-geqtyp} can be also shown by Theorems~\ref{eqtyp-soundness}
and \ref{peqtyp-soundness}
considering an appropriate type environment under
a certain non-trivial interpretation.}.
\qed\CHECKED{2014/04/30}
\end{proof}

\begin{definition}
    \ilabel{Shift-def}{0 shift@$\protect\strut\protect\Shift{A}^n$}
Let \(n\) be a non-negative integer.
We define \(\Shift{\!A}^n\) as follows.
\begin{Eqnarray*}
\Shift{A}^n &=& \t & (\(A\) is a {\tvariant}) \\
\Shift{X}^n &=& X \\
\Shift{\O A}^0 &=& \O A & (\(\,\O A\) is not a {\tvariant}) \\
\Shift{\O A}^{n{+}1} &=& \Shift{A}^n & (\(\,\O A\) is not a {\tvariant}) \\
\Shift{A \Impl B}^n &=& \Shift{A}^n \Impl \Shift{B}^n
    & (\(A \Impl B\) is not a {\tvariant}) \\
\Shift{\fix{X}A}^n &=& \Shift{A[\fix{X}A/X]}^n
    & (\(\fix{X}A\) is not a {\tvariant})
\end{Eqnarray*}
\end{definition}
Note that \(\Shift{\!A}^n\) is defined by induction on
the lexicographic ordering of \(\pair{n}{r(A)}\).
For example,
\begin{eqnarray*}
\Shift{\fix{X}{\O X \Impl Y}}^0
    &=& \Shift{\O (\fix{X}{\O X \Impl Y}) \Impl Y}^0 \\
    &=& \Shift{\O (\fix{X}{\O X \Impl Y})}^0 \Impl \Shift{Y}^0 \\
    &=& \O (\fix{X}{\O X \Impl Y}) \Impl Y, \\
\Shift{\fix{X}{\O(\O X \Impl \O\O Y)}}^3
    &=& \Shift{\O(\O (\fix{X}{\O(\O X \Impl \O\O Y)}) \Impl \O\O Y)}^3 \\
    &=& \Shift{\O (\fix{X}{\O(\O X \Impl \O\O Y)}) \Impl \O\O Y}^2 \\
    &=& \Shift{\O (\fix{X}{\O(\O X \Impl \O\O Y)})}^2 \Impl \Shift{\O\O Y}^2 \\
    &=& \Shift{\fix{X}{\O(\O X \Impl \O\O Y)}}^1 \Impl \Shift{\O Y}^1 \\
    &=& \Shift{\O(\O (\fix{X}{\O(\O X \Impl \O\O Y)}) \Impl \O\O Y)}^1
	\Impl \Shift{Y}^0 \\
    &=& \Shift{\O (\fix{X}{\O(\O X \Impl \O\O Y)}) \Impl \O\O Y}^0 \Impl Y \\
    &=& (\Shift{\O (\fix{X}{\O(\O X \Impl \O\O Y)})}^0 \Impl \Shift{\O\O Y}^0)
	\Impl Y \\
    &=& (\O (\fix{X}{\O(\O X \Impl \O\O Y)}) \Impl \O\O Y) \Impl Y.
\end{eqnarray*}

\begin{proposition}\label{shift-basic}\pushlabel
\begin{Enumerate}
\item\itemlabel{shift-0}
    \(\Shift{A}^0 \geqtyp A\).
\item\itemlabel{shift-O}
    \(\Shift{\O A}^{n{+}1} \geqtyp \Shift{A}^n\).
\item\itemlabel{shift-Impl}
    \(\Shift{A \Impl B}^n \geqtyp \Shift{A}^n \Impl \Shift{B}^n\).
\item\itemlabel{shift-fix}
    \(\Shift{\fix{X}A}^n \geqtyp \Shift{A[\fix{X}A/X]}^n\).
\end{Enumerate}
\end{proposition}
\begin{proof}
Item~\itemref{shift-0} is by straightforward induction on \(r(A)\).
Note that other items are trivial
for type expressions other than {\tvariant}s.
Use Propositions~\ref{tvariant-basic}, \ref{tvariant-fix} and
Theorem~\ref{geqtyp-t-tvariant}.
\ifdetail

\paragraph{Proof of \protect\itemref{shift-0}}
By straightforward induction on \(r(A)\), and by cases on the form of \(A\).
If \(A\) is a {\tvariant}, then \(\Shift{A}^0 = \t \geqtyp A\)
by Definition~\ref{Shift-def} and Theorem~\ref{geqtyp-t-tvariant}.
Hence, we only consider the case when \(A\) is not.

\Case{\(A = X\) for some \(X\).}
By Definition~\ref{Shift-def}, \(\Shift{A}^0 = \Shift{X}^0 = X = A\).

\Case{\(A = \O B\) for some \(B\).}
Similarly, \(\Shift{A}^0 = \Shift{\O B}^0 = \O B = A\).

\Case{\(A = B \Impl C\) for some \(B\) and \(C\).}
By Definition~\ref{Shift-def} and induction hypothesis,
\(\Shift{A}^0 = \Shift{B \Impl C}^0
    = \Shift{B}^0 \Impl \Shift{C}^0 \geqtyp B \Impl C = A\).

\Case{\(A = \fix{X}B\) for some \(X\) and \(B\).}
Similarly,
\(\Shift{A}^0 = \Shift{\fix{X}B}^0
    = \Shift{B[A/X]}^0 \geqtyp B[A/X] \geqtyp \fix{X}B = A\)
by induction hypothesis and \r{{\eqtyp}\mbox{-fix}}.

\paragraph{Proof of \protect\itemref{shift-O}}
If \(A\) is {\tvariant}, then so is \(\O A\) by Proposition~\ref{O-tvariant};
and hence, \(\Shift{\O A}^{n+1} = \t\) and \(\Shift{A}^n = \t\)
by Definition~\ref{Shift-def}.
Otherwise, trivial by Definition~\ref{Shift-def} again.

\paragraph{Proof of \protect\itemref{shift-Impl}}
If \(A \Impl B\) is {\tvariant},
then so is \(B\) by Proposition~\ref{Impl-tvariant};
and hence,
\(\Shift{A \Impl B}^n = \t\) and
\(\Shift{A}^n \Impl \Shift{B}^n = \Shift{A}^n \Impl \t \geqtyp \t\)
by Definition~\ref{Shift-def} and \r{{\eqtyp}\mbox{-}{\Impl}{\t}}.
Otherwise, trivial by definition.

\paragraph{Proof of \protect\itemref{shift-fix}}
If \(\fix{X}A\) is {\tvariant},
then so is \(A[\fix{X}A/X]\) by Proposition~\ref{tvariant-fix};
and hence, \(\Shift{\fix{X}A}^n = \t\) and \(\Shift{A[\fix{X}A/X]}^n = \t\)
by Definition~\ref{Shift-def}.
Otherwise, trivial again.
\fi 
\qed\CHECKED{2014/04/30, 07/09}
\end{proof}

\begin{lemma}\label{shift-proper-subst}
Let \(m\) and \(n\) be non-negative integers,
\(A_1,\ldots,A_m,B_1,\ldots,B_m\) and \(C\) be type expressions such that
\(A_i \geqtyp B_i\) and \(\Shift{A_i}^k \geqtyp \Shift{B_i}^k\)
for every \(k < n\) and \(i\) (\(1 \le i \le m\)).
Let \(\vec{X} = X_1,\ldots,X_m\), \(\vec{A} = A_1,\ldots,A_m\) and
\(\vec{B} = B_1,\ldots,B_m\).
If\/ for every \(i\), either (a) \(C\) is proper in \(X_i\),
or (b) \(\Shift{A_i}^n \geqtyp \Shift{B_i}^n\),
then \(\Shift{C[\vec{A}/\vec{X}]}^n \geqtyp \Shift{C[\vec{B}/\vec{X}]}^n\)
\end{lemma}
\begin{proof}
By induction on the lexicographic ordering of \(\pair{n}{h(C)}\),
and by cases on the form of \(C\).
Suppose (a) or (b) holds for every \(i\) (\(1 \le i \le m\)).
If \(C\) is a {\tvariant} or \(n = 0\), then
we get \(\Shift{C[\vec{A}/\vec{X}]}^n \geqtyp \Shift{C[\vec{B}/\vec{X}]}^n\)
by Proposition~\ref{tvariant-subst1}, or
by Propositions~\ref{shift-0} and \ref{geqtyp-subst}, respectively.
Hence, we assume that \(C\) is not a {\tvariant} and \(n > 0\).

\Case{\(C = Y\) for some \(Y\).}
If \(Y \not\in \{\,\vec{X}\,\}\), then obvious since
\(C[\vec{A}/\vec{X}] = C = C[\vec{B}/\vec{X}]\).
If \(Y = X_i\) for some \(i\), then
\(C[\vec{A}/\vec{X}] = A_i\) and \(C[\vec{B}/\vec{X}] = B_i\).
Hence, \(\Shift{C[\vec{A}/\vec{X}]}^n \geqtyp \Shift{C[\vec{B}/\vec{X}]}^n\)
from (b) since \(C\) is not proper in \(X_i\).

\Case{\(C = \O D\) for some \(D\).}
In this case, \(\Shift{D[\vec{A}/\vec{X}]}^{n{-}1}
    \geqtyp \Shift{D[\vec{B}/\vec{X}]}^{n{-}1}\) by induction hypothesis.
Hence,
\(\Shift{C[\vec{A}/\vec{X}]}^n \geqtyp \Shift{D[\vec{A}/\vec{X}]}^{n{-}1}
\geqtyp \Shift{D[\vec{B}/\vec{X}]}^{n{-}1}\penalty1000
\geqtyp \Shift{C[\vec{B}/\vec{X}]}^n\)
by Proposition~\ref{shift-O}.

\Case{\(C = D \Impl E\) for some \(D\) and \(E\).}
Note that, for every \(i\), both \(D\) and \(E\) are proper in \(X_i\)
if and only if so is \(C\), since \(C\) is not a {\tvariant}.
Therefore,
\(\Shift{D[\vec{A}/\vec{X}]}^n \geqtyp \Shift{D[\vec{B}/\vec{X}]}^n\) and
\(\Shift{E[\vec{A}/\vec{X}]}^n \geqtyp \Shift{E[\vec{B}/\vec{X}]}^n\)
by induction hypothesis; and hence,
\(\Shift{C[\vec{A}/\vec{X}]}^n
\geqtyp \Shift{D[\vec{A}/\vec{X}]}^n \Impl \Shift{E[\vec{A}/\vec{X}]}^n
\geqtyp \Shift{D[\vec{B}/\vec{X}]}^n \Impl \Shift{E[\vec{B}/\vec{X}]}^n
\geqtyp \Shift{C[\vec{B}/\vec{X}]}^n\)
by Proposition~\ref{shift-Impl} and \r{{\eqtyp}\mbox{-}{\Impl}}.

\Case{\(C = \fix{Y}D\) for some \(Y\) and \(D\).}
We can assume that for any \(i\),
\(Y \not\in \{\,X_i\,\} \cup \FTV{A_i}\cup \FTV{B_i}\)
without loss of generality.
Note that, for every \(i\),
\(D\) is proper in \(X_i\) if and only if so is \(C\),
since \(C\) is not a {\tvariant} and \(Y \not= X_i\).
By induction hypothesis,
\(\Shift{C[\vec{A}/\vec{X}]}^k \geqtyp \Shift{C[\vec{B}/\vec{X}]}^k\)
for every \(k < n\).
Since \(D\) is proper in \(Y\),
\begin{Eqnarray*}
\Shift{C[\vec{A}/\vec{X}]}^n
    &\geqtyp& \Shift{D[\vec{A}/\vec{X}][C[\vec{A}/\vec{X}]/Y]}^n
				 & (by Proposition~\ref{shift-fix}) \\
    &=& \Shift{D[\vec{A}/\vec{X},C[\vec{A}/\vec{X}]/Y]}^n
				 & (since \(Y \not\in \FTV{\vec{A}}\)) \\
    &\geqtyp& \Shift{D[\vec{B}/\vec{X},C[\vec{B}/\vec{X}]/Y]}^n
					 & (by induction hypothesis) \\
    &=& \Shift{D[\vec{B}/\vec{X}][C[\vec{B}/\vec{X}]/Y]}^n
				 & (since \(Y \not\in \FTV{\vec{B}}\)) \\
    &\geqtyp& \Shift{C[\vec{B}/\vec{X}]}^n
				 & (by Proposition~\ref{shift-fix}).
\end{Eqnarray*}
\Qed\CHECKED{2014/07/09}
\end{proof}

\begin{proposition}\label{shift-geqtyp}
If \(A \geqtyp B\), then \(\Shift{A}^n \geqtyp \Shift{B}^n\) for every \(n\).
\end{proposition}
\begin{proof}
Suppose that \(A \geqtyp B\).
Since obviously \(\Shift{A}^0 \geqtyp \Shift{B}^0\)
by Proposition~\ref{shift-0},
we assume that \(n > 0\), below.
By induction on the derivation of \(A \geqtyp B\),
and by cases on the last rule applied in the derivation.
\ifdetail

\Cases{\r{{\eqtyp}\mbox{-reflex}},
    \r{{\eqtyp}\mbox{-symm}} and \r{{\eqtyp}\mbox{-trans}}.}
Trivial.

\Case{\r{{\eqtyp}\mbox{-}{\O}}.}
In this case, there exist some \(A'\) and \(B'\) such that
\(A = \O A'\), \(B = \O B'\) and \(A' \geqtyp B'\).
Hence,
\(\Shift{A}^n \geqtyp \Shift{A'}^{n{-}1}
\geqtyp \Shift{B'}^{n{-}1} \geqtyp \Shift{B}^n\)
by Proposition~\ref{shift-O} and induction hypothesis.

\Case{\r{{\eqtyp}\mbox{-}{\Impl}}.}
Similarly to the previous case,
there exist some \(C\), \(D\), \(E\) and \(F\) such that
\(A = C \Impl D\), \(B = E \Impl F\), \(C \geqtyp E\) and \(D \geqtyp F\).
By induction hypothesis,
\(\Shift{C}^n \geqtyp \Shift{E}^n\) and \(\Shift{D}^n \geqtyp \Shift{F}^n\).
Hence,
\(\Shift{A}^n \geqtyp \Shift{C}^n \Impl \Shift{D}^n
    \geqtyp \Shift{E}^n \Impl \Shift{F}^n \geqtyp \Shift{B}^n\)
by Proposition~\ref{shift-Impl} and \r{{\eqtyp}\mbox{-}{\Impl}}.

\Case{\r{{\eqtyp}\mbox{-}{\Impl}{\t}}.}
In this case, \(A = C \Impl \t\) and \(B = \t\) for some \(C\), that is,
both are {\tvariant}s.
Hence, \(\Shift{A}^n = \Shift{B}^n =\t\)
by Definition~\ref{Shift-def}.

\Case{\r{{\eqtyp}\mbox{-fix}}.}
\(A = \fix{X}C\) and \(B = C[A/X]\) for some \(X\) and \(C\).
Therefore, \(\Shift{A}^n \geqtyp \Shift{C[A/X]}^n = \Shift{B}^n\)
by Proposition~\ref{shift-fix}.

\Case{\r{{\eqtyp}\mbox{-uniq}}.}
There 
\else
Use Proposition~\ref{shift-basic}.
The only non-trivial case is \r{{\eqtyp}\mbox{-uniq}}.
In this case, there 
\fi 
exist some \(X\) and \(C\) such that
\(B = \fix{X}C\), \(A \geqtyp C[A/X]\) and \(C\) is proper in \(X\).
By induction hypothesis,
\begin{eqnarray}\label{shift-geqtyp-01}
    \Shift{A}^n &\geqtyp& \Shift{C[A/X]}^n\!.
\end{eqnarray}
We show that \(\Shift{A}^n \geqtyp \Shift{B}^n\) by induction on \(n\).
By the induction hypothesis on \(n\), we have
\begin{eqnarray}\label{shift-geqtyp-02}
\Shift{A}^k \geqtyp \Shift{B}^k ~\mbox{for every \(k < n\)}.
\end{eqnarray}
Therefore,
\begin{Eqnarray*}
\Shift{A}^n
    &\geqtyp& \Shift{C[A/X]}^n & (by (\ref{shift-geqtyp-01})) \\
    &\geqtyp& \Shift{C[B/X]}^n & (by Lemma~\ref{shift-proper-subst}
				    and (\ref{shift-geqtyp-02})) \\
    &\geqtyp& \Shift{B}^n & (by Proposition~\ref{shift-fix}).
\end{Eqnarray*}
\ifdetail

\Case{\r{{\peqtyp}\mbox{-{\bf K}/{\bf L}}}.}
\(A = \O(C \Impl D)\) and \(B = \O C \Impl \O D\) for some \(C\) and \(D\).
This case only applies to \(\peqtyp\).
We get
\(\Shift{A}^n
    \peqtyp \Shift{C \Impl D}^{n{-}1}
    \peqtyp \Shift{C}^{n{-}1} \Impl \Shift{D}^{n{-}1}
    \peqtyp \Shift{\O C}^n \Impl \Shift{\O D}^n
    \peqtyp \Shift{B}^n\)
by Propositions~\ref{shift-O} and \ref{shift-Impl}.
\qed\CHECKED{2014/07/09}
\else
\Qed
\fi 
\end{proof}

\begin{proposition}\label{geqtyp-congr}\pushlabel
\begin{Enumerate}
\item \label{geqtyp-O-congr}
    \(\O A \geqtyp \O B\) if and only if\/ \(A \geqtyp B\).
\item \label{geqtyp-Impl-congr}
    \(A \Impl B \geqtyp C \Impl D\) if and only if\/
    (a) \(A \geqtyp C\) and \(B \geqtyp D\), or
    (b) \(B \geqtyp D \geqtyp \t\).
\end{Enumerate}
\end{proposition}
\begin{proof}
The ``if'' part of Item~1 is obvious from \r{{\eqtyp}\mbox{-}{\O}},
and so is the ``only if'' part since
\(A \geqtyp \Shift{A}^0 \geqtyp \Shift{\O A}^1 \geqtyp \Shift{\O B}^1
\geqtyp \Shift{B}^0 \geqtyp B\)
by Propositions~\ref{shift-0}, \ref{shift-O} and \ref{shift-geqtyp}.
The ``if'' part of Item~2 is also obvious
from \r{{\eqtyp}\mbox{-}{\Impl}} and \r{{\eqtyp}\mbox{-}{\Impl}{\t}}.
As for the ``only if'' part of Item~2,
suppose that \(A \Impl B \geqtyp C \Impl D\).
If \(A \Impl B\) or \(C \Impl D\) is a {\tvariant}, then
so are both \(B\) and \(D\)
by Proposition~\ref{Impl-tvariant} and Theorem~\ref{geqtyp-t-tvariant};
and hence, \(B \geqtyp D \geqtyp \t\).
On the other hand, if neither is a {\tvariant},
we get \(A \geqtyp C\) and \(B \geqtyp D\)
by Proposition~\ref{canon-congr} since
\(\Canong{(A \Impl B)} = A \Impl B\) and \(\Canong{(C \Impl D)} = C \Impl D\).
\qed\CHECKED{2014/04/30, 07/09}
\end{proof}

We shall conclude this section with the following lemma,
which will be used in the proof of
Propositions~\ref{subtyp-O-var} and \ref{subtyp-O-Impl}
in the succeeding section.

\begin{lemma}\pushlabel
Let \(A\) be proper in \(X\), and suppose that \(A[B/X] \geqtyp B\).
\begin{Enumerate}
\item \itemlabel{geqtyp-subst-var-invariant}
    If\/ \(\Canong{B} = \O^n Y\), then \(\Canong{A} = \O^n Y\).
\item \itemlabel{geqtyp-subst-Impl-invariant}
    If\/ \(\Canonp{B} = \O^n (C \Impl D)\), then
	\(\Canonp{A} = \O^n (C' \Impl D')\)
	for some \(n\), \(C'\) and \(D'\) such that
	\(C'[B/X] \geqtyp C\) and \(D'[B/X] \geqtyp D\), where
	\(n = 0\) is case of \(\peqtyp\).
\end{Enumerate}
\end{lemma}
\begin{proof}
By cases on the form of \(\Canong{A}\).
Suppose that \(A\) is proper in \(X\) and
\(A[B/X] \geqtyp B\).

\Case{\(\Canong{A} = \t\).}
Either \(\Canong{B} = \O^n Y\) or
\(\Canong{B} = \O^n (C \Impl D)\) implies \(B \ngeqtyp \t\)
by Proposition~\ref{canon-congr}; and hence,
\(\Canong{A} \not= \t\) from \(A[B/X] \geqtyp B\)
by Propositions~\ref{tvariant-subst1}, \ref{geqtyp-canon}
and Theorem~\ref{geqtyp-t-tvariant}.

\Case{\(\Canong{A} = \O^m Z\) for some \(m\) and \(Z\).}
If \(Z = X\), then \(m > 0\) since \(A\) is proper in \(X\), and
\(B \geqtyp A[B/X] \geqtyp (\O^m X)[B/X] = \O^m B\), which
implies \(B \geqtyp \t\) by Proposition~\ref{geqtyp-O-fix}.
Hence, \(Z \not= X\); and therefore,
\(B \geqtyp A[B/X] \geqtyp (\O^m Z)[B/X] = \O^m Z\), which implies
that \(\Canong{B} = \O^m Z\) by Proposition~\ref{canon-congr}.
Thus, both Items~1 and 2 hold.

\Case{\(\Canong{A} = \O^m(E \Impl F)\) for some \(m\), \(E\) and \(F\).}
In this case,
\(B \geqtyp A[B/X] \geqtyp \O^m(E \Impl F)[B/X] = \O^m(E[B/X] \Impl F[B/X])\).
Hence, by Proposition~\ref{canon-congr},
\(\Canong{B} = \O^m (C \Impl D)\) for some \(C\) and \(D\) such that
\(C \eqtyp E[B/X]\) and \(D \eqtyp F[B/X]\).
Therefore, both Items~1 and 2 hold.
Note that \(m = 0\) in case of \(\peqtyp\).
\qed\CHECKED{2014/05/02, 07/09}
\end{proof}

\Section{Subtyping}\label{subtyp-sec}

As mentioned in Section~\ref{intro-sec},
our intended interpretation of the \(\O\)-modality introduces
a subtyping relation into types.
We define the subtyping relation by a set of derivation rules
in a similar way to \cite{amadio:cardelli}.

\begin{definition}
    \ilabel{subtyp-assump-def}{subtyping assumption}
A {\em subtyping assumption} is a finite set of pairs
of type variables such that
any type variable appears at most once in the set.
To denote the subtyping assumption
\(\zfset{\tuple{X_i,\,Y_i}}{i = 1,\,2,\,\ldots,\,n}\),
we write
\(\{\,X_1 \psubtyp Y_1,\, X_2 \psubtyp Y_2,\,\ldots,\,X_n \psubtyp Y_n\,\}\).
We use \(\g\), \(\g'\), \(\g_1\), \(\g_2\), \(\ldots\)
to denote subtyping assumptions, and \(\FTV{\g}\) to denote
the set of type variables occurring in \(\g\).
\end{definition}

\begin{definition}[\(\psubtyp\)]
    \ilabel{subtyp-def}{subtyping judgment}
    \ilabel*{0 subtyp@$\protect\psubtyp$}
    \ilabel*{0 yields subtyping@$\protect\subt{\g}{A \psubtyp B}$}
We define the derivability of {\em subtyping judgment}
\(\subt{\g}{A \psubtyp B}\) by the following derivation rules.
\[
\ifr{{\rsubtyp}\mbox{-assump}}
    {}
    {\subt{\g \cup \{X \psubtyp Y\}}{X \psubtyp Y}}
\mskip50mu
\ifr{{\rsubtyp}\mbox{-}\t}
    {}
    {\subt{\g}{A \psubtyp \t}}
\]
\[
\ifr{{\rsubtyp}\mbox{-reflex}}
    {A \peqtyp B}
    {\subt{\g}{A \psubtyp B}}
\mskip50mu
\ifr{{\rsubtyp}\mbox{-trans}}
    {\subt{\g_1}{A \psubtyp B}
     & \subt{\g_2}{B \psubtyp C}}
    {\subt{\g_1 \cup \g_2}{A \psubtyp C}}
\]
\[
\ifr{{\rsubtyp}\mbox{-}{\O}}
    {\subt{\g}{A \psubtyp B}}
    {\subt{\g}{\O A \psubtyp \O B}}
\mskip70mu
\ifr{{\rsubtyp}\mbox{-}{\Impl}}
    {\subt{\g_1}{A' \psubtyp A}
     & \subt{\g_2}{B \psubtyp B'}}
    {\subt{\g_1 \cup \g_2}{A \Impl B \psubtyp A' \Impl B'}}
\]
\[
\Ifr{\r{{\rsubtyp}\mbox{-}{\fx}}~~~\left(
	\mbox{\tabcolsep=0pt
	\ifnarrow
	    \begin{tabular}{p{0.36\hsize}}
	\else
	    \begin{tabular}{p{0.44\hsize}}
	\fi
	    \(X \not\in \FTV{\g} \cup \FTV{B}\),
		\(Y \not\in \FTV{\g} \cup \FTV{A}\),
		and \(A\) and \(B\) are proper
	    in \(X\) and \(Y\), respectively
	    \end{tabular}}
	\right)}
    {\subt{\g \cup \{X \psubtyp Y\}}{A \psubtyp B}}
    {\subt{\g}{\fix{X} A \psubtyp \fix{Y} B}}
\]
\[
\ifr{{\rsubtyp}\mbox{-approx}}
    {}
    {\subt{\g}{A \psubtyp \O A}}
\]
Note that \(\g \cup \{X \psubtyp Y\}\) and \(\g_1 \cup \g_2\) in
the rules must be well-formed subtyping assumptions, i.e., any type variable
must not have more than one occurrence in them.
We also define a binary relation \(\psubtyp\) over \(\TE\) as
\(A \psubtyp B\) if and only if \(\subt{\{\}}{\!A \psubtyp B}\) is derivable.
\end{definition}

The subtyping relation \(\psubtyp\)
reflects the interpretation of type expressions in {\lA}-frames.
Most of the subtyping rules are standard.
The rule \r{{\rsubtyp}\mbox{-}{\fx}} corresponds to the ``Amber rule''
\cite{cardelli}.
The rules \r{{\rsubtyp}\mbox{-}\t},
\r{{\rsubtyp}\mbox{-}{\O}} and
\r{{\rsubtyp}\mbox{-approx}} reflect our intended
meaning of the \(\O\)-modality discussed in Section~\ref{intro-sec}.
In the rest of this section,
we discuss some basic properties of \(\psubtyp\).

We can also consider another subtyping judgment, say
\(\subt{\g}{A \subtyp B}\), by substituting \(A \eqtyp B\) for
the antecedent of the \r{{\rsubtyp}\mbox{-reflex}}-rule and by
instead adding the following subtyping rule, which
corresponds to the axiom schema {\bf K} of normal modal logic.
\[
    \ifr{{\subtyp}\mbox{-\bf K}}
	{}
	{\subt{\g}{\O(A \Impl B) \subtyp \O A \Impl \O B}}
\]
However, in order to concentrate on the main purpose of the present paper,
we will not discuss such a variant in the sequel\footnote{%
In \cite{nakano-tacs01}, 
this variant is called S-\(\lambda{\bullet}\mu^{\!+}\).
}.

\begin{definition}
    \ilabel{subtyp-models-def}{0 models subtyp@$\protect\strut%
	\protect\tenv \protect\models \protect\g$}
Let \(\pair{\W}{\acc}\) be a {\lA}-frame,
\(\tenv\) a type environment, and \(\g\) a subtyping assumption.
We write \(\tenv \models \g\) if and only if
\(\tenv(X)_p \subseteq \tenv(Y)_p\) for every \(p \in \W\), \(X\) and \(Y\)
such that \(\{\,X \psubtyp Y\,\} \subseteq \g\).
\end{definition}

\begin{theorem}[Soundness of \(\psubtyp\)]
    \ilabel{psubtyp-soundness}{soundness!subtyping@$\protect\psubtyp$}
Consider the interpretation \(\II\) over a {\lA}-frame \(\pair{\W}{\acc}\).
Let \(\tenv\) be a hereditary type environment.
If\/ \(\subt{\g}{A \psubtyp B}\) and \(\tenv \models \g\),
then \(\I{A}^\tenv_p \subseteq \I{B}^\tenv_p\;\) for every \(p \in \W\).
\end{theorem}
\begin{proof}
By induction on the derivation of \(\subt{\g}{A \psubtyp B}\),
and by cases on the last subtyping rule applied in the derivation.

\Case{\r{{\rsubtyp}\mbox{-assump}}.}
In this case, \(A = X\) and \(B = Y\) for some \(X\) and \(Y\) such that
\(\{\,X \psubtyp Y\,\} \subseteq \g\).
We get \(\tenv(X)_p \subseteq \tenv(Y)_p\) for every \(p\)
from \(\tenv \models \g\).

\Case{\r{{\rsubtyp}\mbox{-}\t}.}
Obvious from Definition~\ref{rlz-def}.

\Case{\r{{\rsubtyp}\mbox{-reflex}}.}
Obvious from Theorem~\ref{peqtyp-soundness}.

\Case{\r{{\rsubtyp}\mbox{-trans}}.}
For some \(C\), \(\g_1\) and \(\g_2\) such that \(\g = \g_1 \cup \g_2\),
the derivation ends with
\[
\Ifr{\r{{\rsubtyp}\mbox{-trans}}.}
    {\subt{\g_1}{A \psubtyp C}
     & \subt{\g_2}{C \psubtyp B}}
    {\subt{\g_1 \cup \g_2}{A \psubtyp B}}
\]
Since \(\g = \g_1 \cup \g_2\),
we get \(\tenv \models \g_1\) and \(\tenv \models \g_2\)
from \(\tenv \models \g\).
Therefore,
\(\I{A}^\tenv_p \subseteq \I{C}^\tenv_p \subseteq \I{B}^\tenv_p\)
by induction hypothesis.

\Case{\r{{\rsubtyp}\mbox{-}{\O}}.}
For some \(A'\) and \(B'\) such that \(A = \O A'\) and \(B = \O B'\),
the derivation ends with
\[
\Ifr{\r{{\rsubtyp}\mbox{-}{\O}}.}
    {\subt{\g}{A' \psubtyp B'}}
    {\subt{\g}{\O A' \psubtyp \O B'}}
\]
We get \(\I{A'}^\tenv_q \subseteq \I{B'}^\tenv_q\) for every \(q\)
by induction hypothesis.
Therefore, by Proposition~\ref{rlz-another-def},
\[
    \I{\O A'}^\tenv_p
    = \zfset{u}{u \in \I{A'}^\tenv_q~~\mbox{for every}~q \opacc p}
    \subseteq \zfset{u}{u \in \I{B'}^\tenv_q~~\mbox{for every}~q \opacc p}
    = \I{\O B'}^\tenv_p.
\]

\Case{\r{{\rsubtyp}\mbox{-}{\Impl}}.}
For some \(A_1\), \(A_2\), \(B_1\), \(B_2\), \(\g_1\) and \(\g_2\)
such that \(A = A_1 \Impl A_2\), \(B = B_1 \Impl B_2\) and
\(\g = \g_1 \cup \g_2\),
the derivation ends with
\[
\Ifr{\r{{\rsubtyp}\mbox{-}{\Impl}}.}
    {\subt{\g_1}{B_1 \psubtyp A_1}
     & \subt{\g_2}{A_2 \psubtyp B_2}}
    {\subt{\g_1 \cup \g_2}{A_1 \Impl A_2 \psubtyp B_1 \Impl B_2}}
\]
Similarly to the previous case, we get
\(\I{B_1}^\tenv_q \subseteq \I{A_1}^\tenv_q\) and
\(\I{A_2}^\tenv_q \subseteq \I{B_2}^\tenv_q\) for every \(q\),
by induction hypothesis.
Therefore, by Proposition~\ref{rlz-another-def},
\(\I{A_1 \Impl A_2}^\tenv_p
    = \zfset{u}{u \cdot v \in \I{A_2}^\tenv_q~~\mbox{for every}~
	    v \in \I{A_1}^\tenv_q~~\mbox{whenever}~p \tacc q}
    \subseteq \zfset{u}{u \cdot v \in \I{B_2}^\tenv_q~~\mbox{for every}~
	    v \in \I{B_1}^\tenv_q~~\mbox{whenever}~p \tacc q}
    = \I{B_1 \Impl B_2}^\tenv_p\).

\Case{\r{{\rsubtyp}\mbox{-}{\fx}}.}
For some \(X'\), \(Y'\), \(A'\) and \(B'\) such that
\(A = \fix{X'}A'\) and \(B = \fix{Y'}B'\), the derivation ends with
\[
\ifr{{\rsubtyp}\mbox{-}{\fx}}
    {\subt{\g \cup \{\,X' \psubtyp Y'\,\}}{A' \psubtyp B'}}
    {\subt{\g}{\fix{X'} A' \psubtyp \fix{Y'} B'}}
\]
where \(X' \not\in \FTV{\g} \cup \FTV{B'}\),
\(Y' \not\in \FTV{\g} \cup \FTV{A'}\), and furthermore,
\(A'\) and \(B'\) are proper in \(X'\) and \(Y'\), respectively.
We show that
\(\I{A}^{\tenv}_p \subseteq \I{B}^{\tenv}_p\) for every \(p \in \W\)
by induction on \(p\).
Suppose that \(p \in \W\).
By the induction hypothesis on \(p\),
\begin{eqnarray}
\label{psubtyp-soundness-01}
    \I{A}^{\tenv}_q \subseteq \I{B}^{\tenv}_q
    ~\mbox{for every}~ q \opacc p.
\end{eqnarray}
Let \(\tenv'\) be the type environment defined as follows.
\begin{eqnarray*}
    \tenv'(X')_q &=& \Choice{%
		\I{A}^\tenv_q \qquad& (p \acc q) \\
		\{\} & (p \not\acc q)
		} \\[4pt]
    \tenv'(Y')_q &=& \Choice{%
		\I{B}^\tenv_q \qquad& (p \acc q) \\
		\{\} & (p \not\acc q)
		} \\[4pt]
    \tenv'(Z)_q &=& \tenv(Z)_q\mskip89mu(Z \not\in \{\,X',\,Y'\,\})
\end{eqnarray*}
Note that
\begin{eqnarray}\label{psubtyp-soundness-02}
    \tenv'(Z)_q &=& \tenv[\I{A}^\tenv/X',\I{B}^\tenv/Y'](Z)_q
    ~~\mbox{for every \(Z\) and \(q \opacc p\)},
\end{eqnarray}
and that \(\tenv' \models \g \cup \{\,X' \psubtyp Y'\,\}\)
by (\ref{psubtyp-soundness-01}).
Hence, by the induction hypothesis on the derivation,
\begin{eqnarray}
\label{psubtyp-soundness-03}
    \I{A'}^{\tenv'}_p \subseteq \I{B'}^{\tenv'}_p.
\end{eqnarray}
On the other hand,
\begin{Eqnarray*}
\I{A}^\tenv_p
    &=& \I{A'[A/X']}^\tenv_p
	    & (by Proposition~\ref{rlz-another-def}) \\
    &=& \I{A'[A/X',\,B/Y']}^\tenv_p
	    & (by Proposition~\ref{geqtyp-no-etv-subst} and
		Theorem~\ref{peqtyp-soundness}) \\
    &=& \I{A'}^{\tenv[\I{A}^\tenv/X',\,\I{B}^\tenv/Y']}_p
	    & (by Proposition~\ref{rlz-subst-env}) \\
    &=& \I{A'}^{\tenv'}_p
	    & (by (\ref{psubtyp-soundness-02})
		and Lemma~\ref{rlz-proper-subst-lemma}).
\end{Eqnarray*}
Similarly, \(\I{B}^\tenv_p = \I{B'}^{\tenv'}_p\).
Therefore, \(\I{A}^\tenv_p \subseteq \I{B}^\tenv_p\)
by (\ref{psubtyp-soundness-03}).

\Case{\r{{\rsubtyp}\mbox{-approx}}.}
\(B = \O A\) in this case.
Hence, \(\I{A}^\tenv_p
    \subseteq \zfset{u}{u \in \I{A}^\tenv_q~~\mbox{for every}~q \opacc p}
    = \I{B}^\tenv_p\)
by Propositions~\ref{rlz-hereditary} and \ref{rlz-another-def}.
\qed\CHECKED{2014/05/01, 07/09}
\end{proof}

Although not being necessary for showing the soundness of subtyping
judgments, the following basic propositions should also be established in order
to discuss formal derivability of subtyping, later.

\begin{proposition}\label{subtyp-basic}\pushlabel
\begin{Enumerate}
\item \itemlabel{subtyp-rename}
    If\/ \(\subt{\g\cup\{\,X \psubtyp Y\,\}}{A \psubtyp B}\) is derivable, and
    \(\g\cup\{\,X' \psubtyp Y'\,\}\) is a well-formed subtyping assumption,
    then \(\subt{\g\cup\{\,X' \psubtyp Y'\,\}}%
	{A[X'/X,\,Y'/Y] \psubtyp B[X'/X,\,Y'/Y]}\)
    is also derivable without changing the height of derivation.
    {\rm\bf(renaming)}
\item \itemlabel{subtyp-weakening}
    If\/ \(\subt{\g}{A \psubtyp B}\) is derivable, and
    \(\g'\) is a subtyping assumption such that
    \(\g \subseteq \g'\), then \(\subt{\g'}{A \psubtyp B}\)
    is also derivable without changing the height of derivation.
    {\rm\bf(weakening)}
\item \itemlabel{subtyp-inst}
    If\/ \(\subt{\g\cup\{\,X \psubtyp Y\,\}}{A \psubtyp B}\)
    is derivable, then
    so is \(\subt{\g}{A[C/X,\,C/Y] \psubtyp B[C/X,\,C/Y]}\)
    without changing the height of derivation.
\item \itemlabel{subtyp-subst}
    If\/ \(\subt{\g\cup\{\,X \psubtyp Y\,\}}{A \psubtyp B}\) and
    \(\subt{\g}{C \psubtyp D}\) are derivable, then
    so is \(\subt{\g}{A[C/X,\,D/Y] \psubtyp B[C/X,\,D/Y]}\).
    {\rm\bf(substitution)}
\end{Enumerate}
\end{proposition}
\begin{proof}
By induction on the height of the derivation of
\(\subt{\g\cup
\{\,X \psubtyp Y\,\}}{A \psubtyp B}\) or \(\subt{\g}{A \psubtyp B}\), and by
cases on the last rule applied.
The proofs for Items~\itemref{subtyp-rename} and \itemref{subtyp-weakening}
proceed by simultaneous induction.
We use Item~\itemref{subtyp-rename} to rename the bound type variables
in case of \r{{\rsubtyp}\mbox{-}{\fx}}, and
Item~\itemref{subtyp-weakening} to unify
the typing assumptions of the two antecedents when the last
rule is \r{{\rsubtyp}\mbox{-trans}} or \r{{\rsubtyp}\mbox{-}{\Impl}}.
Use Proposition~\ref{geqtyp-subst} for the case \r{{\rsubtyp}\mbox{-reflex}}
in the proofs of Items~\itemref{subtyp-rename}, \itemref{subtyp-inst}
and \itemref{subtyp-subst}.
\ifdetail

\paragraph{Proof of \protect\itemref{subtyp-rename}}
By simultaneous induction with Item~\itemref{subtyp-weakening}.
Suppose that \(\subt{\g\cup\{\,X \psubtyp Y\,\}}{A \psubtyp B}\) is derivable,
and \(\g\cup\{\,X' \psubtyp Y'\,\}\) is a well-formed subtyping assumption.
Let \(A' = A[X'/X,\,Y'/Y]\) and \(B' = B[X'/X,\,Y'/Y]\).

\Case{\r{{\rsubtyp}\mbox{-assump}}.}
In this case, the derivation has the form of
\[
    \ifr{{\rsubtyp}\mbox{-assump}}
	{}
	{\subt{\g' \cup \{Z_1 \psubtyp Z_2\}}{Z_1 \psubtyp Z_2}}
\]
for some \(Z_1\), \(Z_2\) and \(\g'\) such that
\(A = Z_1\), \(B = Z_2\) and
\(\g' \cup \{Z_1 \psubtyp Z_2\} = \g \cup \{\,X \psubtyp Y\,\}\).
If \(X \in \{\,Z_1,\,Z_2\,\}\), then \(X = Z_1\) and \(Y = Z_2\);
and hence, \(A' = X'\) and \(B' = Y'\).
Therefore, \(\subt{\g \cup \{X' \psubtyp Y'\}}{A' \psubtyp B'}\)
is derivable by the same rule.
On the other hand, if \(X \not\in \{\,Z_1,\,Z_2\,\}\), then
\(\{\,X,\,Y\,\} \cap \{\,Z_1,\,Z_2\,\} = \{\}\);
and hence, \(A' = Z_1\), \(B' = Z_2\) and
\(\{\,Z_1 \psubtyp Z_2\,\} \subseteq \g\).
Therefore, \(\subt{\g \cup \{X' \psubtyp Y'\}}{A' \psubtyp B'}\)
is derivable by \r{{\rsubtyp}\mbox{-assump}} again.

\Case{\r{{\rsubtyp}\mbox{-}\t}.}
In this case, \(B = \t\).
Hence, trivial because \(B' = \t[X'/X,\,Y'/Y] = \t\).

\Case{\r{{\rsubtyp}\mbox{-reflex}}.}
In this case, \(A \peqtyp B\).
Hence, \(A' \peqtyp B'\) by Proposition~\ref{geqtyp-subst}; and therefore,
\(\subt{\g \cup \{X' \psubtyp Y'\}}{A' \psubtyp B'}\) is derivable
by \r{{\rsubtyp}\mbox{-reflex}}.

\Case{\r{{\rsubtyp}\mbox{-trans}}.}
The derivation ends with
\[
\ifr{{\rsubtyp}\mbox{-trans}}
    {\subt{\g_1}{A \psubtyp C}
     & \subt{\g_2}{C \psubtyp B}}
    {\subt{\g_1 \cup \g_2}{A \psubtyp B}}
\]
for some \(C\), \(\g_1\) and \(\g_2\) such that
\(\g \cup \{\,X \psubtyp Y\,\} = \g_1 \cup \g_2\).
By the induction hypothesis for Item~\itemref{subtyp-weakening},
we can assume that
\(\g_1 = \g_2 = \g \cup \{\,X \psubtyp Y\,\}\).
Therefore, by the induction hypothesis for this item,
\(\subt{\g \cup \{X' \psubtyp Y'\}}{A' \psubtyp C[X'/X,\,Y'/Y]'}\) and
\(\subt{\g \cup \{X' \psubtyp Y'\}}{C[X'/X,\,Y'/Y] \psubtyp B'}\)
are derivable
without changing the heights of the original derivations of the antecedents.
Hence, so is
\(\subt{\g \cup \{X' \psubtyp Y'\}}{A' \psubtyp B'}\) by the same rule
without changing the height of the original derivation of the consequent.

\Case{\r{{\rsubtyp}\mbox{-}{\O}}.}
The derivation ends with
\[
\ifr{{\rsubtyp}\mbox{-}{\O}}
    {\subt{\g \cup \{\,X \psubtyp Y\,\}}{C \psubtyp D}}
    {\subt{\g \cup \{\,X \psubtyp Y\,\}}{\O C \psubtyp \O D}}
\]
for some \(C\) and \(D\) such that \(A = \O C\) and \(B = \O D\).
By induction hypothesis,
\(\subt{\g \cup \{\,X' \psubtyp Y'\,\}}%
    {C[X'/X,\,Y'/Y] \psubtyp D[X'/X,\,Y'/Y]}\)
is derivable with the original height.
Hence, so is
\(\subt{\g \cup \{\,X' \psubtyp Y'\,\}}{A' \psubtyp B'}\) by the same rule
without changing the height of derivation.

\Case{\r{{\rsubtyp}\mbox{-}{\Impl}}.}
The derivation ends with
\[
\ifr{{\rsubtyp}\mbox{-}{\Impl}}
    {\subt{\g_1}{B_1 \psubtyp A_1}
     & \subt{\g_2}{A_2 \psubtyp B_2}}
    {\subt{\g_1 \cup \g_2}{A_1 \Impl A_2 \psubtyp B_1 \Impl B_2}}
\]
for some \(A_1\), \(A_2\), \(B_1\), \(B_2\), \(\g_1\) and \(\g_2\)
such that \(A = A_1 \Impl A_2\), \(B = B_1 \Impl B_2\) and
\(\g \cup \{\,X \psubtyp Y\,\} = \g_1 \cup \g_2\).
We can assume that
\(\g_1 = \g_2 = \g \cup \{\,X \psubtyp Y\,\}\)
by the induction hypothesis for Item~\itemref{subtyp-weakening}.
Hence, by the induction hypothesis for this item,
\(\subt{\g \cup \{X' \psubtyp Y'\}}%
    {B_1[X'/X,\,Y'/Y] \psubtyp A_1[X'/X,\,Y'/Y]}\) and
\(\subt{\g \cup \{X' \psubtyp Y'\}}%
    {A_2[X'/X,\,Y'/Y] \psubtyp B_2[X'/X,\,Y'/Y]}\)
are derivable
without changing the heights.
Therefore, so is
\(\subt{\g \cup \{X' \psubtyp Y'\}}{A' \psubtyp B'}\) by the same rule
without changing the height of the original derivation.

\Case{\r{{\rsubtyp}\mbox{-}{\fx}}.}
The derivation ends with
\[
\ifr{{\rsubtyp}\mbox{-}{\fx}}
    {\subt{\g \cup \{\,X \psubtyp Y\,\} \cup \{\,Z_1 \psubtyp Z_2\,\}}
	{C \psubtyp D}}
    {\subt{\g \cup \{\,X \psubtyp Y\,\}}{\fix{Z_1}C \psubtyp \fix{Z_2}D}}
\]
for some \(Z_1\), \(Z_2\), \(C\) and \(D\) such that
\(A = \fix{Z_1}C\), \(B = \fix{Z_2}D\),
where \(Z_1 \not\in \FTV{\g \cup \{\,X \psubtyp Y\,\}} \cup \FTV{D}\),
\(Z_2 \not\in \FTV{\g \cup \{\,X \psubtyp Y\,\}} \cup \FTV{C}\),
and furthermore,
\(C\) and \(D\) are proper in \(Z_1\) and \(Z_2\), respectively.
By induction hypothesis,
we can assume that \(\{\,Z_1,\,Z_2\,\} \cap \{\,X',\,Y'\,\} = \{\}\)
without loss of generality.
Hence, \(A' = \fix{Z_1}C[X'/X,\,Y'/Y]\) and \(B' = \fix{Z_2}D[X'/X,\,Y'/Y]\).
On the other hand, by induction hypothesis again,
\(\subt{\g \cup \{\,X' \psubtyp Y'\,\} \cup \{\,Z_1 \psubtyp Z_2\,\}}%
	{C[X'/X,\,Y'/Y] \psubtyp D[X'/X,\,Y'/Y]}\)
is derivable without changing the height; and hence,
so is \(\subt{\g \cup \{\,X' \psubtyp Y'\,\}}{A' \psubtyp B'}\)
by the same rule.

\Case{\r{{\rsubtyp}\mbox{-approx}}.}
In this case, \(B = \O A\); and hence, \(B' = \O A'\).
Therefore, trivial because
\(\subt{\g\cup \{\,X' \psubtyp Y'\,\}}{A' \psubtyp \O A'}\)
is derivable by \r{{\rsubtyp}\mbox{-approx}}.

\paragraph{Proof of \protect\itemref{subtyp-weakening}}
By simultaneous induction with Item~\itemref{subtyp-rename}.
Suppose that \(\subt{\g}{A \psubtyp B}\) is derivable,
and \(\g \subseteq \g'\).

\Case{\r{{\rsubtyp}\mbox{-assump}}.}
In this case, the derivation has the form of
\[
    \ifr{{\rsubtyp}\mbox{-assump}}
	{}
	{\subt{\g}{X \psubtyp Y}}
\]
for some \(X\) and \(Y\) such that
\(A = X\), \(B = Y\), and \(\{X \psubtyp Y\} \subseteq \g\).
Hence, \(\subt{\g'}{A \psubtyp B}\) is also derivable by the same rule,
since \(\g \subseteq \g'\).

\Cases{\r{{\rsubtyp}\mbox{-}\t}, \r{{\rsubtyp}\mbox{-reflex}}
    and \r{{\rsubtyp}\mbox{-approx}}.}
Trivial.

\Cases{\r{{\rsubtyp}\mbox{-trans}}, \r{{\rsubtyp}\mbox{-}{\O}} and
    \r{{\rsubtyp}\mbox{-}{\Impl}}.}
Straightforward by induction hypothesis.

\Case{\r{{\rsubtyp}\mbox{-}{\fx}}.}
The derivation ends with
\[
\ifr{{\rsubtyp}\mbox{-}{\fx}}
    {\subt{\g \cup \{\,X \psubtyp Y\,\}}{C \psubtyp D}}
    {\subt{\g}{\fix{X}C \psubtyp \fix{Y}D}}
\]
for some \(X\), \(Y\), \(C\) and \(D\) such that
\(A = \fix{X}C\), \(B = \fix{Y}D\),
where \(X \not\in \FTV{\g} \cup \FTV{D}\), \(Y \not\in \FTV{\g} \cup \FTV{C}\),
and furthermore,
\(C\) and \(D\) are proper in \(X\) and \(Y\), respectively.
By the induction hypothesis for Item~\itemref{subtyp-rename},
we can assume that \(\{\,X,\,Y\,\} \cap (\g' - \g) = \{\}\)
without loss of generality.
Therefore, by the induction hypothesis for this item,
\(\subt{\g' \cup \{\,X \psubtyp Y\,\}}{C \psubtyp D}\)
is derivable without changing the height; and hence,
so is \(\subt{\g'}{A \psubtyp B}\) by the same rule.

\paragraph{Proof of \protect\itemref{subtyp-inst}}
Suppose that \(\subt{\g\cup\{\,X \psubtyp Y\,\}}{A \psubtyp B}\) is derivable.
Let \(A' = A[C/X,\,C/Y]\) and \(B' = B[C/X,\,C/Y]\).

\Case{\r{{\rsubtyp}\mbox{-assump}}.}
In this case, the derivation has the form of
\[
    \ifr{{\rsubtyp}\mbox{-assump}}
	{}
	{\subt{\g' \cup \{Z_1 \psubtyp Z_2\}}{Z_1 \psubtyp Z_2}}
\]
for some \(Z_1\), \(Z_2\) and \(\g'\) such that
\(A = Z_1\), \(B = Z_2\) and
\(\g' \cup \{Z_1 \psubtyp Z_2\} = \g \cup \{\,X \psubtyp Y\,\}\).
If \(X \in \{\,Z_1,\,Z_2\,\}\), then \(X = Z_1\) and \(Y = Z_2\);
and hence, \(A' = B' = C\).
Therefore, \(\subt{\g}{A' \psubtyp B'}\)
is derivable by \r{{\rsubtyp}\mbox{-reflex}}.
On the other hand, if \(X \not\in \{\,Z_1,\,Z_2\,\}\), then
\(\{\,X,\,Y\,\} \cap \{\,Z_1,\,Z_2\,\} = \{\}\);
and hence, \(A' = Z_1\), \(B' = Z_2\) and
\(\{\,Z_1 \psubtyp Z_2\,\} \subseteq \g\).
Therefore, \(\subt{\g}{A' \psubtyp B'}\)
is derivable by \r{{\rsubtyp}\mbox{-assump}}.

\Case{\r{{\rsubtyp}\mbox{-}\t}.}
In this case, \(B = \t\).
Hence, trivial because \(B' = \t[C/X,\,C/Y] = \t\).

\Case{\r{{\rsubtyp}\mbox{-reflex}}.}
In this case, \(A \peqtyp B\).
Hence, \(A' \peqtyp B'\) by Proposition~\ref{geqtyp-subst}; and therefore,
\(\subt{\g}{A' \psubtyp B'}\) is derivable by the same rule.

\Case{\r{{\rsubtyp}\mbox{-trans}}.}
The derivation ends with
\[
\ifr{{\rsubtyp}\mbox{-trans}}
    {\subt{\g_1}{A \psubtyp D}
     & \subt{\g_2}{D \psubtyp B}}
    {\subt{\g_1 \cup \g_2}{A \psubtyp B}}
\]
for some \(D\), \(\g_1\) and \(\g_2\) such that
\(\g \cup \{\,X \psubtyp Y\,\} = \g_1 \cup \g_2\).
By Item~\itemref{subtyp-weakening} of this proposition,
we can assume that
\(\g_1 = \g_2 = \g \cup \{\,X \psubtyp Y\,\}\)
without loss of generality.
Therefore, by induction hypothesis,
\(\subt{\g}{A' \psubtyp D[C/X,\,C/Y]}\) and
\(\subt{\g}{D[C/X,\,C/Y] \psubtyp B'}\)
are derivable
without changing the heights of derivations; and hence,
so is \(\subt{\g}{A' \psubtyp B'}\) by the same rule
without changing the height.

\Case{\r{{\rsubtyp}\mbox{-}{\O}}.}
The derivation ends with
\[
\ifr{{\rsubtyp}\mbox{-}{\O}}
    {\subt{\g \cup \{\,X \psubtyp Y\,\}}{D \psubtyp E}}
    {\subt{\g \cup \{\,X \psubtyp Y\,\}}{\O D \psubtyp \O E}}
\]
for some \(D\) and \(E\) such that \(A = \O D\) and \(B = \O E\).
By induction hypothesis,
\(\subt{\g}{D[C/X,\,C/Y] \psubtyp E[C/X,\,C/Y]}\)
is derivable without changing the height of derivation.
Hence, so is \(\subt{\g}{A' \psubtyp B'}\) by the same rule.

\Case{\r{{\rsubtyp}\mbox{-}{\Impl}}.}
The derivation ends with
\[
\ifr{{\rsubtyp}\mbox{-}{\Impl}}
    {\subt{\g_1}{B_1 \psubtyp A_1}
     & \subt{\g_2}{A_2 \psubtyp B_2}}
    {\subt{\g_1 \cup \g_2}{A_1 \Impl A_2 \psubtyp B_1 \Impl B_2}}
\]
for some \(A_1\), \(A_2\), \(B_1\), \(B_2\), \(\g_1\) and \(\g_2\)
such that \(A = A_1 \Impl A_2\), \(B = B_1 \Impl B_2\) and
\(\g \cup \{\,X \psubtyp Y\,\} = \g_1 \cup \g_2\).
We can assume that
\(\g_1 = \g_2 = \g \cup \{\,X \psubtyp Y\,\}\)
by Item~\itemref{subtyp-weakening} of this proposition.
Hence, by induction hypothesis,
\(\subt{\g}{B_1[C/X,\,C/Y] \psubtyp A_1[C/X,\,C/Y]}\) and
\(\subt{\g}{A_2[C/X,\,C/Y] \psubtyp B_2[C/X,\,C/Y]}\)
are derivable without changing the heights.
Therefore, so is \(\subt{\g}{A' \psubtyp B'}\) by the same rule.

\Case{\r{{\rsubtyp}\mbox{-}{\fx}}.}
The derivation ends with
\[
\ifr{{\rsubtyp}\mbox{-}{\fx}}
    {\subt{\g \cup \{\,X \psubtyp Y\,\} \cup \{\,Z_1 \psubtyp Z_2\,\}}
	{D \psubtyp E}}
    {\subt{\g \cup \{\,X \psubtyp Y\,\}}{\fix{Z_1}D \psubtyp \fix{Z_2}E}}
\]
for some \(Z_1\), \(Z_2\), \(D\) and \(E\) such that
\(A = \fix{Z_1}D\), \(B = \fix{Z_2}E\),
where \(Z_1 \not\in \FTV{\g \cup \{\,X \psubtyp Y\,\}} \cup \FTV{E}\),
\(Z_2 \not\in \FTV{\g \cup \{\,X \psubtyp Y\,\}} \cup \FTV{D}\),
and furthermore,
\(D\) and \(E\) are proper in \(Z_1\) and \(Z_2\), respectively.
By Item~\itemref{subtyp-rename} of this proposition,
we can assume that \(\{\,Z_1,\,Z_2\,\} \cap \FTV{C} = \{\}\)
without loss of generality.
Hence, \(A' = \fix{Z_1}D[C/X,\,C/Y]\) and \(B' = \fix{Z_2}E[C/X,\,C/Y]\).
On the other hand,
by induction hypothesis,
\(\subt{\g \cup \{\,Z_1 \psubtyp Z_2\,\}}{D[C/X,\,C/Y] \psubtyp E[C/X,\,C/Y]}\)
is derivable without changing the height; and hence,
so is \(\subt{\g}{A' \psubtyp B'}\) by the same rule.

\Case{\r{{\rsubtyp}\mbox{-approx}}.}
Trivial since \(B' = \O A'\) in this case.

\paragraph{Proof of \protect\itemref{subtyp-subst}}
Suppose that \(\subt{\g\cup\{\,X \psubtyp Y\,\}}{A \psubtyp B}\)
and \(\subt{\g}{C \psubtyp D}\) are derivable.
Let \(A' = A[C/X,\,D/Y]\) and \(B' = B[C/X,\,D/Y]\).

\Case{\r{{\rsubtyp}\mbox{-assump}}.}
In this case, the derivation has the form of
\[
    \ifr{{\rsubtyp}\mbox{-assump}}
	{}
	{\subt{\g' \cup \{Z_1 \psubtyp Z_2\}}{Z_1 \psubtyp Z_2}}
\]
for some \(Z_1\), \(Z_2\) and \(\g'\) such that
\(A = Z_1\), \(B = Z_2\) and
\(\g' \cup \{Z_1 \psubtyp Z_2\} = \g \cup \{\,X \psubtyp Y\,\}\).
If \(X \in \{\,Z_1,\,Z_2\,\}\), then \(X = Z_1\) and \(Y = Z_2\);
and hence, \(A' = C\) and \(B' = D\).
Therefore, \(\subt{\g}{A' \psubtyp B'}\) is derivable by assumption.
On the other hand, if \(X \not\in \{\,Z_1,\,Z_2\,\}\), then
\(\{\,X,\,Y\,\} \cap \{\,Z_1,\,Z_2\,\} = \{\}\);
and hence, \(A' = Z_1\), \(B' = Z_2\) and
\(\{\,Z_1 \psubtyp Z_2\,\} \subseteq \g\).
Therefore, \(\subt{\g}{A' \psubtyp B'}\)
is derivable by \r{{\rsubtyp}\mbox{-assump}}.

\Case{\r{{\rsubtyp}\mbox{-}\t}.}
In this case, \(B = \t\).
Hence, trivial because \(B' = \t[C/X,\,D/Y] = \t\).

\Case{\r{{\rsubtyp}\mbox{-reflex}}.}
In this case, \(A \peqtyp B\).
Hence, \(A' \peqtyp B'\) by Proposition~\ref{geqtyp-subst}; and therefore,
\(\subt{\g}{A' \psubtyp B'}\) is derivable by \r{{\rsubtyp}\mbox{-reflex}}.

\Case{\r{{\rsubtyp}\mbox{-trans}}.}
The derivation ends with
\[
\ifr{{\rsubtyp}\mbox{-trans}}
    {\subt{\g_1}{A \psubtyp E}
     & \subt{\g_2}{E \psubtyp B}}
    {\subt{\g_1 \cup \g_2}{A \psubtyp B}}
\]
for some \(E\), \(\g_1\) and \(\g_2\) such that
\(\g \cup \{\,X \psubtyp Y\,\} = \g_1 \cup \g_2\).
By Item~\itemref{subtyp-weakening} of this proposition,
we can assume that
\(\g_1 = \g_2 = \g \cup \{\,X \psubtyp Y\,\}\)
without loss of generality.
Therefore,
\(\subt{\g}{A' \psubtyp E[C/X,\,D/Y]}\) and
\(\subt{\g}{D[C/X,\,D/Y] \psubtyp B'}\) are derivable
by induction hypothesis; and hence,
so is \(\subt{\g}{A' \psubtyp B'}\) by the same rule.

\Case{\r{{\rsubtyp}\mbox{-}{\O}}.}
The derivation ends with
\[
\ifr{{\rsubtyp}\mbox{-}{\O}}
    {\subt{\g \cup \{\,X \psubtyp Y\,\}}{E \psubtyp F}}
    {\subt{\g \cup \{\,X \psubtyp Y\,\}}{\O E \psubtyp \O F}}
\]
for some \(E\) and \(F\) such that \(A = \O E\) and \(B = \O F\).
By induction hypothesis,
\(\subt{\g}{E[C/X,\,D/Y] \psubtyp F[C/X,\,D/Y]}\) is derivable.
Hence, so is \(\subt{\g}{A' \psubtyp B'}\) by the same rule.

\Case{\r{{\rsubtyp}\mbox{-}{\Impl}}.}
The derivation ends with
\[
\ifr{{\rsubtyp}\mbox{-}{\Impl}}
    {\subt{\g_1}{B_1 \psubtyp A_1}
     & \subt{\g_2}{A_2 \psubtyp B_2}}
    {\subt{\g_1 \cup \g_2}{A_1 \Impl A_2 \psubtyp B_1 \Impl B_2}}
\]
for some \(A_1\), \(A_2\), \(B_1\), \(B_2\), \(\g_1\) and \(\g_2\)
such that \(A = A_1 \Impl A_2\), \(B = B_1 \Impl B_2\) and
\(\g \cup \{\,X \psubtyp Y\,\} = \g_1 \cup \g_2\).
We can assume that
\(\g_1 = \g_2 = \g \cup \{\,X \psubtyp Y\,\}\)
by Item~\itemref{subtyp-weakening} of this proposition.
Hence, by induction hypothesis,
\(\subt{\g}{B_1[C/X,\,D/Y] \psubtyp A_1[C/X,\,D/Y]}\) and
\(\subt{\g}{A_2[C/X,\,D/Y] \psubtyp B_2[C/X,\,D/Y]}\)
are derivable.
Therefore, so is \(\subt{\g}{A' \psubtyp B'}\) by the same rule.

\Case{\r{{\rsubtyp}\mbox{-}{\fx}}.}
The derivation ends with
\[
\ifr{{\rsubtyp}\mbox{-}{\fx}}
    {\subt{\g \cup \{\,X \psubtyp Y\,\} \cup \{\,Z_1 \psubtyp Z_2\,\}}
	{E \psubtyp F}}
    {\subt{\g \cup \{\,X \psubtyp Y\,\}}{\fix{Z_1}E \psubtyp \fix{Z_2}F}}
\]
for some \(Z_1\), \(Z_2\), \(E\) and \(F\) such that
\(A = \fix{Z_1}E\), \(B = \fix{Z_2}F\),
where \(Z_1 \not\in \FTV{\g \cup \{\,X \psubtyp Y\,\}} \cup \FTV{F}\),
\(Z_2 \not\in \FTV{\g \cup \{\,X \psubtyp Y\,\}} \cup \FTV{E}\),
and furthermore,
\(E\) and \(F\) are proper in \(Z_1\) and \(Z_2\), respectively.
By Item~\itemref{subtyp-rename} of this proposition,
we can assume that \(\{\,Z_1,\,Z_2\,\} \cap (\FTV{C} \cup \FTV{D}) = \{\}\)
without loss of generality.
Hence, \(A' = \fix{Z_1}E[C/X,\,D/Y]\) and \(B' = \fix{Z_2}F[C/X,\,D/Y]\).
On the other hand,
by induction hypothesis,
\(\subt{\g \cup \{\,Z_1 \psubtyp Z_2\,\}}{E[C/X,\,D/Y] \psubtyp F[C/X,\,D/Y]}\)
is derivable; and hence,
so is \(\subt{\g}{A' \psubtyp B'}\) by the same rule.

\Case{\r{{\rsubtyp}\mbox{-approx}}.}
Trivial since \(B' = \O A'\) in this case.
\fi 
\qed\CHECKED{2014/05/01, 07/09}
\end{proof}

In the rest of the present section,
we shall show three more basic propositions about
subtyping, which give us some intuitions about the structure of
subtyping introduced by the derivation rules, and are
crucial to prove the subject reduction property
of the typing system in the next section.

\begin{proposition}\label{subtyp-gt-t}
If\/ \(\subt{\g}{\t \psubtyp A}\) is derivable, then \(A \geqtyp \t\).
\end{proposition}
\begin{proof}
Suppose that \(\subt{\g}{A \psubtyp B}\) is derivable.
By Theorem~\ref{geqtyp-t-tvariant},
it suffices to show that \(A \peqtyp \t\) implies \(B \peqtyp \t\).
The proof proceeds by induction on the height of the derivation,
and by cases on the last rule applied in the derivation.
The case \r{{\rsubtyp}\mbox{-assump}}
is impossible by Proposition~\ref{not-geqtyp-var-t}.
The cases \r{{\rsubtyp}\mbox{-}{\t}} and \r{{\rsubtyp}\mbox{-reflex}}
are trivial.
The case \r{{\rsubtyp}\mbox{-trans}} is straightforward
by induction hypothesis.
Use Propositions~\ref{geqtyp-O-t} and \ref{geqtyp-Impl-t}
for the cases
\r{{\rsubtyp}\mbox{-}{\O}}, \r{{\rsubtyp}\mbox{-}{\Impl}},
and \r{{\rsubtyp}\mbox{-approx}}.
If the last rule is \r{{\rsubtyp}\mbox{-}{\fx}}, then
the derivation ends with
\[
\ifr{{\rsubtyp}\mbox{-}{\fx}}
    {\subt{\g \cup \{\,X \psubtyp Y\,\}}{A' \psubtyp B'}}
    {\subt{\g}{\fix{X}A' \psubtyp \fix{Y}B'}}
\]
for some \(X\), \(Y\), \(A'\) and \(B'\) such that
\(A = \fix{X}A'\), \(B = \fix{Y}B'\),
\(X \not\in \FTV{B'}\), \(Y \not\in \FTV{A'}\), and furthermore,
\(A'\) and \(B'\) are proper in \(X\) and \(Y\), respectively.
By Proposition~\ref{subtyp-inst}, we can get a derivation of
\(\subt{\g}{A'[\t/X] \psubtyp B'[\t/Y]}\) from the one of
\(\subt{\g \cup \{\,X \psubtyp Y\,\}}{A' \psubtyp B'}\)
without changing the height of derivation.
On the other hand, since \(A \peqtyp A'[A/X]\) and \(A \peqtyp \t\),
we also get \(A'[\t/X] \peqtyp \t\) by Proposition~\ref{geqtyp-subst}.
Therefore, by induction hypothesis, \(B'[\t/X] \peqtyp \t\); and hence,
\(B = \fix{Y}B' \peqtyp \t\) by \r{{\eqtyp}\mbox{-uniq}}.
\qed\CHECKED{2014/06/20, 07/09}
\end{proof}

\begin{proposition}\label{subtyp-O-var}
Suppose that \(\subt{\g}{A \psubtyp B}\) be a derivable subtyping judgment,
and \(B \npeqtyp \t\).
\begin{Enumerate}
\item If\/ \(\Canonp{A} = \O^m X\), then
    \(\Canonp{B} = \O^n Y\) for some \(n\) and \(Y\) such that
    \(m \le n\), and either \(X = Y\) or \(\{\,X \psubtyp Y\,\} \subseteq \g\).
\item If\/ \(\Canonp{B} = \O^n Y\), then
    \(\Canonp{A} = \O^m X\) for some \(m\) and \(X\) such that
    \(m \le n\), and either \(X = Y\) or \(\{\,X \psubtyp Y\,\} \subseteq \g\).
\end{Enumerate}
\end{proposition}
\begin{proof}
By induction on the derivation of \(\subt{\g}{A \psubtyp B}\),
and by cases on the last rule applied in the derivation.
\ifdetail

\Case{\r{{\rsubtyp}\mbox{-assump}}.}
In this case, both \(A\) and \(B\) are type variables.
Hence,
if \(\Canonp{A} = \O^m X\), then \(m = 0\) and \(A = X\)
by Propositions~\ref{geqtyp-canon} and \ref{geqtyp-var-var}; and therefore,
\(\{\,X \psubtyp Y\,\} \subseteq \g\) and \(B = Y\) for some \(Y\).
Similarly,
if \(\Canonp{B} = \O^n Y\), then \(n = 0\) and \(B = Y\), which
implies
\(\{\,X \psubtyp Y\,\} \subseteq \g\) and \(A = X\) for some \(X\).

\Case{\r{{\rsubtyp}\mbox{-}\t}.}
Impossible since \(B \ngeqtyp \t\).

\Case{\r{{\rsubtyp}\mbox{-reflex}}.}
In this case, \(A \peqtyp B\).
Hence, both Items~1 and 2 hold by Proposition~\ref{canon-congr}.

\Case{\r{{\rsubtyp}\mbox{-trans}}.}
In this case,
\(\subt{\g_1}{A \psubtyp C}\) and \(\subt{\g_2}{C \psubtyp B}\) are
derivable for some \(\g_1\), \(\g_2\) and \(C\)
such that \(\g = \g_1 \cup \g_2\).
We get \(C \ngeqtyp \t\) from \(B \ngeqtyp \t\)
by Proposition~\ref{subtyp-gt-t}.
If \(\Canonp{A} = \O^m X\), then
by the induction hypothesis on \(\subt{\g_1}{A \psubtyp C}\),
for some \(k\) and \(Z\),
\[
    \Canonp{C} = \O^k Z,~m \le k,~\mbox{and either}~
    X = Z~\mbox{or}~\{\,X \psubtyp Z\,\} \subseteq \g_1.
\]
Hence, by the induction hypothesis on \(\subt{\g_2}{C \psubtyp B}\),
for some \(n\) and \(Y\),
\[
    \Canonp{B} = \O^n Y,~k \le n,~\mbox{and either}~
    Z = Y~\mbox{or}~\{\,Z \psubtyp Y\,\} \subseteq \g_2.
\]
Note that either \(X = Z\) or \(Z = Y\) holds, because
\(\{\,X \psubtyp Z\,\} \cup \{\,Z \psubtyp Y\,\}\) cannot be
a well-formed subtyping assumption.
We thus get Item~1.
The proof for Item~2 is just symmetrical.

\Case{\r{{\rsubtyp}\mbox{-}{\O}}.}
In this case, \(\subt{\g}{A' \psubtyp B'}\) are
derivable for some \(A'\) and \(B'\) such that
\(A = \O A'\) and \(B = \O B'\).
We get \(B' \ngeqtyp \t\) from \(B \ngeqtyp \t\)
by Proposition~\ref{geqtyp-t1}.
If \(\Canonp{A} = \O^m X\), then \(m > 0\) and
\(\Canonp{A'} = \O^{m-1} X\) by Definition~\ref{canon-def}.
Therefore, by induction hypothesis,
for some \(n'\) and \(Y\),
\[
    \Canonp{B'} = \O^{n'} Y,~m-1 \le n',~\mbox{and either}~
    X = Y~\mbox{or}~\{\,X \psubtyp Y\,\} \subseteq \g.
\]
We thus get Item~1 by taking \(n\) as \(n = n' +1\),
since \(\Canonp{B'} = \O^{n'}Y\) implies
\(\Canonp{B} = \O^{n'+1} Y\) by Definition~\ref{canon-def}.
The proof for Item~2 is just symmetrical.

\Case{\r{{\rsubtyp}\mbox{-}{\Impl}}.}
This case is impossible
by Propositions~\ref{not-geqtyp-var-Impl} and \ref{geqtyp-canon}.

\else 
Use Propositions~\ref{geqtyp-canon} and \ref{geqtyp-var-var}
for the case \r{{\rsubtyp}\mbox{-assump}}, and
use Propositions~\ref{canon-congr}, \ref{subtyp-gt-t} and
\ref{geqtyp-t1} for the cases \r{{\rsubtyp}\mbox{-reflex}},
\r{{\rsubtyp}\mbox{-trans}} and
\r{{\rsubtyp}\mbox{-}{\O}}, respectively.
The case \r{{\rsubtyp}\mbox{-}{\Impl}}
is impossible by Propositions~\ref{not-geqtyp-var-Impl} and \ref{geqtyp-canon}.
\fi 
\ifdetail

\Case{\r{{\rsubtyp}\mbox{-}{\fx}}.}
The 
\else
If the last rule is \r{{\rsubtyp}\mbox{-}{\fx}}, then the
\fi
derivation ends with
\[
\ifr{{\rsubtyp}\mbox{-}{\fx}}
    {\subt{\g \cup \{\,X' \psubtyp Y'\,\}}{A' \psubtyp B'}}
    {\subt{\g}{\fix{X'}A' \psubtyp \fix{Y'}B'}}
\]
for some \(X'\), \(Y'\), \(A'\) and \(B'\) such that
\(A = \fix{X'}A'\), \(B = \fix{Y'}B'\),
\(X' \not\in \FTV{B'}\), \(Y' \not\in \FTV{A'}\), and furthermore,
\(A'\) and \(B'\) are proper in \(X'\) and \(Y'\), respectively.
For Item~1, suppose that \(\Canonp{A} = \O^m X\) and \(B \npeqtyp \t\).
Note that \(X \not= X'\) by Poposition~\ref{geqtyp-canon} and \ref{geqtyp-etv},
and that \(B' \npeqtyp \t\) since \(B' \geqtyp \t\) implies \(B \geqtyp \t\).
Since \(A'[A/X'] \geqtyp A \geqtyp \Canonp{A} = \O^m X\)
and \(A'\) is proper in \(X'\), we get \(\Canonp{A'} = \O^m X\)
by Lemma~\ref{geqtyp-subst-var-invariant};
and therefore, by induction hypothesis,
\(\Canonp{B'} = \O^n Y\) for some \(n \ge m\), and
\(X = Y\) or \(\{\,X \psubtyp Y\,\} \subseteq \g \cup \{\,X' \psubtyp Y'\,\}\).
We get \(Y' \not= Y\) from \(B \not\eqtyp \t\) because otherwise
\(B \eqtyp \fix{Y}\Canonp{B'} = \fix{Y}\O^n Y \eqtyp \t\)
by Propositions~\ref{geqtyp-canon}, \ref{geqtyp-fix-congr} and
\ref{geqtyp-t2}; and hence, \(\Canonp{B} = \Canonp{B'}[B/Y'] = \O^n Y\).
Note that
\(\{\,X \psubtyp Y\,\} \subseteq \g \cup \{\,X' \psubtyp Y'\,\}\) implies
\(\{\,X \psubtyp Y\,\} \subseteq \g\) because \(X' \not= X\).
Symmetrically, we can show Item~2 in this case.
\ifdetail

\Case{\r{{\rsubtyp}\mbox{-approx}}.}
In this case, \(B = \O A\).
Hence, by Definition~\ref{canon-def},
\(\Canonp{A} = \O^m X\) implies \(\Canonp{B} = \O^{m+1} X\), and
on the other hand,
\(\Canonp{B} = \O^n Y\) implies \(n > 0\) and
\(\Canonp{A} = \O^{n-1} Y\).
\fi 
\qed\CHECKED{2014/06/20, 07/09}
\end{proof}

\begin{proposition}\label{subtyp-O-Impl}
Suppose that
\(\subt{\g}{A \psubtyp B}\) is derivable, and \(B \npeqtyp \t\).
\begin{Enumerate}\itemsep=2pt plus 2pt
\item If\/ \(\Canonp{A} = C \Impl D\), then
there exist some \(k\), \(E\) and \(F\) such that
\begin{Enumerate}
\item[(1a)] \(\Canonp{B} = E \Impl F\),
\item[(1b)] \(\subt{\g}{E \psubtyp \O^k C}\) and
    \(\subt{\g}{\O^k D \psubtyp F}\) are derivable.
\end{Enumerate}
\item If\/ \(\Canonp{B} = E \Impl F\), then
there exist some \(k\), \(C\) and \(D\) such that
\begin{Enumerate}
\item[(2a)] \(\Canonp{A} = C \Impl D\),
\item[(2b)] \(\subt{\g}{E \psubtyp \O^k C}\) and
    \(\subt{\g}{\O^k D \psubtyp F}\) are derivable.
\end{Enumerate}
\end{Enumerate}
Note that \(k\) ranges over non-negative integers.
\end{proposition}
\begin{proof}
By induction on the derivation of \(\subt{\g}{A \psubtyp B}\),
and by cases on the last rule used in the derivation.
Lemma~\ref{geqtyp-subst-Impl-invariant} is crucial to the case
\r{{\rsubtyp}\mbox{-}{\fx}}.
We will employ Theorem~\ref{geqtyp-t-tvariant}
and Proposition~\ref{geqtyp-canon} in this proof without specific mention.
First, note that \(A \npeqtyp \t\) from \(B \npeqtyp \t\)
by Proposition~\ref{subtyp-gt-t}.
Hence,
\(\Canonp{A} = C \Impl D\) and \(\Canonp{B} = E \Impl F\) imply
\(D \npeqtyp \t\) and \(F \npeqtyp \t\), respectively,
by Proposition~\ref{Impl-tvariant}.

\Case{\r{{\rsubtyp}\mbox{-assump}}.}
In this case, both \(A\) and \(B\) are type variables.
Hence, trivial by Proposition~\ref{not-geqtyp-var-Impl}.

\Case{\r{{\rsubtyp}\mbox{-}\t}.}
This case is impossible since \(B \ngeqtyp \t\).

\Case{\r{{\rsubtyp}\mbox{-reflex}}.}
In this case, \(A \peqtyp B\).
Hence, we can get both Items~1 and 2
by Proposition~\ref{canon-congr}.

\Case{\r{{\rsubtyp}\mbox{-trans}}.}
In this case,
\(\subt{\g_1}{A \psubtyp G}\) and \(\subt{\g_2}{G \psubtyp B}\) are
derivable for some \(\g_1\), \(\g_2\) and \(G\)
such that \(\g = \g_1 \cup \g_2\).
We get \(G \ngeqtyp \t\) from \(B \ngeqtyp \t\)
by Proposition~\ref{subtyp-gt-t}.
If \(\Canonp{A} = C \Impl D\), then
by the induction hypothesis on \(\subt{\g_1}{A \psubtyp G}\),
there exist some \(k'\), \(E'\) and \(F'\) such that
\begin{eqnarray*}
    && \Canonp{G} = E' \Impl F', \\
    && \subt{\g_1}{E' \psubtyp \O^{k'} C}~\mbox{and}~
	     \subt{\g_1}{\O^{k'} D \psubtyp F'}~\mbox{are derivable.}
\end{eqnarray*}
Therefore, by the induction hypothesis on \(\subt{\g_2}{G \psubtyp B}\),
there exist some \(k''\), \(E\) and \(F\) such that
\begin{eqnarray*}
    && \Canonp{B} = E \Impl F, \\
    && \subt{\g_2}{E \psubtyp \O^{k''} E'}~\mbox{and}~
	     \subt{\g_2}{\O^{k''} F' \psubtyp F}~\mbox{are derivable.}
\end{eqnarray*}
Taking \(k\) as \(k = k' + k''\), we get (1a) and (1b).
The proof for Item~2 is just symmetrical.

\Case{\r{{\rsubtyp}\mbox{-}{\O}}.}
In this case, \(\subt{\g}{A' \psubtyp B'}\) is
derivable for some \(A'\) and \(B'\) such that
\(A = \O A'\) and \(B = \O B'\).
We get \(B' \ngeqtyp \t\) from \(B \ngeqtyp \t\)
by Proposition~\ref{geqtyp-t1}.
If \(\Canonp{A} = C \Impl D\), then
\(\Canonp{A'} = C' \Impl D'\) for some \(C'\) and \(D'\) such that
\(C = \O C'\) and \(D = \O D'\) by Definition~\ref{canon-def}.
Therefore, by induction hypothesis,
there exist some \(k\), \(E'\) and \(F'\) such that
\begin{eqnarray*}
    && \Canonp{B'} = E' \Impl F', \\
    && \subt{\g}{E' \psubtyp \O^{k} C'}~\mbox{and}~
	 \subt{\g}{\O^{k} D' \psubtyp F'}~\mbox{are derivable.}
\end{eqnarray*}
Hence, it suffices for Item~1 to take \(E\) and \(F\) as
\(E = \O E'\) and \(F = \O F'\), respectively,
since \(\Canonp{B} = \Canonp{(\O B')} = \O E \Impl \O F\).
The proof for Item~2 is just symmetrical.

\Case{\r{{\rsubtyp}\mbox{-}{\Impl}}.}
In this case,
\(\subt{\g_1}{B_1 \psubtyp A_1}\) and \(\subt{\g_2}{A_2 \psubtyp B_2}\)
are derivable for some \(\g_1\), \(\g_2\),
\(A_1\), \(A_2\), \(B_1\) and \(B_2\) such that
\(\g = \g_1 \cup \g_2\), \(A = A_1 \Impl A_2\) and \(B = B_1 \Impl B_2\).
If \(\Canonp{A} = C \Impl D\), then
we get \(C = A_1\) and \(D = A_2\) by Definition~\ref{canon-def}.
Hence, it suffices to take \(k\), \(E\) and \(F\) as
\(k = 0\), \(E = B_1\) and \(F = B_2\).
Symmetrically,
if \(\Canonp{B} = E \Impl F\), we get \(E = B_1\) and \(F = B_2\);
and hence, it suffices to take \(k\), \(C\) and \(D\) as
\(k = 0\), \(C = A_1\) and \(D = A_2\).

\Case{\r{{\rsubtyp}\mbox{-}{\fx}}.}
In this case, \(\subt{\g\cup\{\,X \psubtyp Y\,\}}{A' \psubtyp B'}\)
is derivable for some \(X\), \(Y\), \(A'\) and \(B'\) such that
\begin{eqnarray}
\label{subtyp-O-Impl-01}
    A &=& \fix{X}A', \\
\label{subtyp-O-Impl-02}
    B &=& \fix{Y}B', \\
\label{subtyp-O-Impl-03}
    X &\not\in& \FTV{\g} \cup \FTV{B'},~\mbox{and} \\
\label{subtyp-O-Impl-04}
    Y &\not\in& \FTV{\g} \cup \FTV{A'}.
\end{eqnarray}
Note that \(A'\) and \(B'\) are proper in \(X\) and \(Y\), respectively,
and that \(A' \npeqtyp \t\) and \(B' \npeqtyp \t\) from
\(A \npeqtyp \t\) and \(B \npeqtyp \t\), respectively,
by Definition~\ref{tvariant-def}.
For Item~1, suppose that \(\Canonp{A} = C \Impl D\).
Then, \(\Canonp{A'}[A/X] = \Canonp{A} = C \Impl D\)
by Definition~\ref{canon-def} and (\ref{subtyp-O-Impl-01}).
Therefore, by Lemma~\ref{geqtyp-subst-Impl-invariant},
there exist some \(C'\) and \(D'\) such that
\begin{eqnarray}
    \Canonp{A'} &=\;& C' \Impl D' \nonumber \\
\label{subtyp-O-Impl-05}
    C &\peqtyp& C'[A/X] ~\mbox{and}~ D \peqtyp D'[A/X].
\end{eqnarray}
Hence, by induction hypothesis,
there exist some \(k\), \(E'\) and \(F'\) such that
\begin{enumerate}[{\kern8pt(1}a{${}'$})]\itemsep=3pt
\item \(\Canonp{B'} = E' \Impl F'\),
\item \(\subt{\g\cup\{\,X \psubtyp Y\,\}}{E' \psubtyp \O^{k} C'}\) and
    \(\subt{\g\cup\{\,X \psubtyp Y\,\}}{\O^{k} D' \psubtyp F'}\)
    are derivable.
\end{enumerate}
Let \(E\) and \(F\) be as
\(E = E'[B/Y]\) and \(F = F'[B/Y]\), respectively.
Then, (1a) can be shown as follows.
\begin{Eqnarray*}
\Canonp{B}
    &=& \Canonp{B'}[B/Y]
	    & (by (\ref{subtyp-O-Impl-02}) and Definition~\ref{canon-def}) \\
    &=& (E' \Impl F')[B/Y] & (by (\(\mbox{1a}'\))) \\
    &=& E \Impl F
\end{Eqnarray*}
Note that
\(F \npeqtyp \t\) from \(B \npeqtyp \t\) by Proposition~\ref{Impl-tvariant};
and hence, \(X \not\in \ETV{E} \cup \ETV{F} \cup \FTV{B}\)
from (\ref{subtyp-O-Impl-03}),
since \(\ETV{E} \cup \ETV{F} = \ETV{E \Impl F}
    = \ETV{\Canonp{B}} = \ETV{B} \subseteq \FTV{B} \subseteq \FTV{B'}\)
by Definition~\ref{etv-def} and Proposition~\ref{geqtyp-etv}.
Therefore,
\begin{eqnarray*}
    E &\peqtyp& E[A/X] = E'[B/Y][A/X] = E'[A/X,\,B/Y],~\mbox{and} \\
    F &\peqtyp& F[A/X] = F'[B/Y][A/X] = F'[A/X,\,B/Y]
\end{eqnarray*}
by Proposition~\ref{geqtyp-no-etv-subst}.
On the other hand, similarly, we get
\(Y \not\in \ETV{C} \cup \ETV{D} \cup \FTV{A}\)
from (\ref{subtyp-O-Impl-04}); and therefore, by (\ref{subtyp-O-Impl-05}),
\begin{eqnarray*}
    C &\peqtyp& C[B/Y] \peqtyp C'[A/X][B/Y] = C'[A/X,\,B/Y],~\mbox{and} \\
    D &\peqtyp& D[B/Y] \peqtyp D'[A/X][B/Y] = D'[A/X,\,B/Y].
\end{eqnarray*}
Hence, (1b) can be also established, for
\(\subt{\g}{E'[A/X,\,B/Y] \psubtyp \O^k C'[A/X,\,B/Y]}\) and
\(\subt{\g}{\O^k D'[A/X,\,B/Y] \psubtyp F'[A/X,\,B/Y]}\) are derivable
from (\(\mbox{1b}'\)) by Proposition~\ref{subtyp-subst}.
We thus get (1a) and (1b) in this case.
The proof for Item~2 is just symmetrical.

\Case{\r{{\rsubtyp}\mbox{-approx}}.}
In this case, \(B = \O A\).
If \(\Canonp{A} = C \Impl D\), then
\(\Canonp{B} = \O C \Impl \O D\) by Definition~\ref{canon-def};
and hence, it suffices to take \(k\), \(E\) and \(F\) as
\(k = 1\), \(E = \O C\) and \(F = \O D\).
Symmetrically,
If \(\Canonp{B} = E \Impl F\), then
by Definition~\ref{canon-def},
\(\Canonp{A} = C \Impl D\) for some \(C\) and \(D\) such that
\(E = \O C\) and \(F = \O D\).
Hence, it suffices to take \(k\) as \(k = 1\).
\qed\CHECKED{2014/06/20, 07/10}
\end{proof}

It might be supposed that the following hold
concerning the subtyping of type expressions.
\begin{enumerate}
\item[] {\bf (Wrong)}\quad
    If \(\O A \psubtyp B\), then
    \(B \peqtyp \O C\)  for some \(C\) such that \(A \psubtyp C\).
\item[] {\bf (Wrong)}\quad
    If\/ \(A \psubtyp \O B\), then either
    (a) \(A \psubtyp B\), or
    (b) \(A \peqtyp \O C\) for some \(C\) such that \(C \psubtyp B\).
\end{enumerate}
However, neither is the case.
Consider the following counter-examples, respectively.
\begin{eqnarray*}
&& \O (X \Impl Y) \peqtyp \O X \Impl \O Y \psubtyp X \Impl \O Y \\
&& (X \Impl \O Y) \Impl \O Z \psubtyp (\O X \Impl \O Y) \Impl \O Z
			  \peqtyp \O((X \Impl Y) \Impl Z)
\end{eqnarray*}

\Section{The typing system {\lA}}\label{lA-sec}

In this section,
we finally introduce a typed \(\lambda\)-calculus
equipped with the modality and recursive types.

\begin{definition}[Typing contexts]
    \ilabel{typ-context}{typing contexts}
    \ilabel*{Dom@$\protect\Dom(\G)$}
A {\em typing context} is a finite mapping that assigns a type
expression to each individual variable of its domain.
We use \(\G\), \(\G'\), \(\ldots\) to denote typing contexts, and
\(\{\,x_1:A_1,\ldots,x_m:A_m\,\}\) to denote a typing context
whose domain is \(\{\,x_1,\ldots,x_m\,\}\) and
that assigns \(A_i\) to \(x_i\) for every \(i\),
where \(A_1,\ldots,A_m\) are type expressions,
and \(x_1,\ldots,x_m\) are distinct individual variables.
We use \(\Dom(\G)\) to denote the domain of \(\G\).
\end{definition}

\begin{definition}[Typing judgment]
    \ilabel{typ-judgment}{typing judgment}
    \ilabel*{0 yields typing@$\protect\typ{\G}{M:A}$}
A {\em typing judgment} of {\lA} has the following form,
\[
    \typ{\{\,x_1:A_1,\,x_2:A_2,\,\ldots,\,x_n:A_n\,\}}{M:B},
\]
where
\(\{\,x_1:A_1,\,x_2:A_2,\,\ldots,\,x_n:A_n\,\}\) is a typing context,
\(M\) a \(\lambda\)-term, and \(B\) a type expression.
We occasionally write the same judgment omitting \(\{\,\}\) as follows.
\[
    \typ{x_1:A_1,\,x_2:A_2,\,\ldots,\,x_n:A_n}{M:B}
\]
Note that \(n\) can be \(0\).
We use \(\O\G\) to denote the typing context
\(\{\,x_1:\O A_1,\,x_2:\O A_2,\,\ldots,\,x_n:\O A_n\,\}\) when
\(\G = \{\,x_1:A_1,\,x_2:A_2,\,\ldots,\,x_n:A_n\,\}\),
and write \(\G'\! \psubtyp \G\) if and only if
\(\Dom(\G') = \Dom(\G)\) and \(\G'(x) \psubtyp \G(x)\)
for every \(x \in \Dom(\G)\).
\end{definition}

We now define the typing rules of {\lA}.
According to the intended meaning of \(\O\),
three typing rules, namely \r{\mbox{shift}}, \r{\t} and \r{\rsubtyp},
are added to those of the simply typed \(\lambda\)-calculus.

\begin{definition}[Typing rules]
    \ilabel{typ-rules}{typing rules}
    \ilabel{La-def}{lambda-A@\protect\lA}
The typing system {\lA} is defined by
the following derivation rules.
\[
\ifr{\mbox{var}}
    {\strut}
    {\typ{\G \cup \{\,x : A\,\}}{x : A}}
\mskip80mu
\ifr{\mbox{shift}}
    {\typ{\O \G}{M : \O A}}
    {\typ{\G}{M : A}}
\]
\[
\ifr{\t}
    {\strut}
    {\typ{\G}{M : \t}}
\mskip80mu
\ifr{\psubtyp}
    {\typ{\G}{M : A} & A \psubtyp B}
    {\typ{\G}{M : B}}
\]
\[
\ifr{\Impl\,\mbox{I}}
    {\typ{\G \cup \{\,x : A\,\}}{M : B}}
    {\typ{\G}{\lam{x}{M} : A \Impl B}}
\mskip30mu
\ifr{\Impl\mbox{E}}
    {\typ{\G_1}{M : A \Impl B}
     & \typ{\G_2}{N : A}}
    {\typ{\G_1 \cup \G_2}{\app{M}{N} : B}}
\]
\end{definition}

The \r{\t}-rule can be considered almost redundant,
since in fact \(\typ{\G}{M : \t}\) is derivable without using \r{\t}
for any closed \(\lambda\)-term \(M\).
In derivation trees of typing judgments,
the rule \r{\psubtyp} is sometimes written as \r{\eqtyp} or \r{\peqtyp}
where it is applied for two type expressions \(A\) and \(B\) such that
\(A \eqtyp B\) or \(A \peqtyp B\), respectively.
Note that two derivation rules, namely
\r{\mbox{shift}} and \r{\rsubtyp}, are not related to
term (program) construction.
These rules only reflect the relationship among types (specifications),
just like the case of parametric polymorphism, and
constitute a non-constructive (non-computational, or non-informative)
part of logic of programming.

The readers may wonder the following rule, which corresponds to
the necessitation rule of normal modal logic, is missing.
\[
\ifr{\mbox{nec}}
    {\typ{\G_1}{M : A}}
    {\typ{\O \G_1\cup\G_2}{M : \O A}}
\]
However, it will be shown that this rule is redundant to {\lA} as
Proposition~\ref{lA-nec-subst-redundant}.
On the other hand,
the axiom schema {\bf K} of normal modal logic is incorporated
into {\lA} in a stronger form
as the \r{{\peqtyp}\mbox{-{\bf K}/{\bf L}}}-rule for type equality.

The \r{\mbox{shift}}-rule represents the fact that
for every hereditary interpretation over a given well-founded frame,
which is not necessarily a {\lA}-frame,
we can extend it by adding new worlds so that
every possible world \(p\) in the original frame
has another world from which \(p\) is accessible.
Since the interpretation of \(\typ{\G}{M:A}\) in the world \(p\) is
identical to the one of \(\typ{\O \G}{M:\O A}\) in such a world,
\(\typ{\G}{M:A}\) is valid whenever so is \(\typ{\O \G}{M:\O A}\).
In \cite{nakano-tacs01}, the \(\r{\mbox{shift}}\)-rule was treated as
an optional rule to the core typing system,
which has however turned to be sound with respect to any hereditary
interpretation over well-founded frames\footnote{%
The \r{\mbox{shift}}-rule becomes redundant
if we add intersection types to {\lA}.}.
The reader should refer to the proof of Theorem~\ref{soundness-theorem}
for the details.

\begin{example}\label{Y-derivable}
\def\FF{\lam{x}{\app{f}{(\app{x}{x})}}}
We can derive Curry's fixed-point combinator \(\Y\) in {\lA};
more precisely, the following is derivable.
\[
    \typ{}{\lam{f}{\app{(\FF)}{(\FF)}}
	\mathrel{\,:\,}(\O X \Impl X) \Impl X}
\]
Let a type \(A = \fix{Y}\O Y \Impl X\) and
a derivation \(\P\) as follows.
\[
\ifnarrow\lower4pt\else\lower10pt\fi\hbox{\(\Pi = {}\)}\;
\begin{array}[c]{c}
	\ifr{\Impl\mbox{I}}
	    {\ifr{\Impl\mbox{E}}
		{\ifr{\mbox{var}}
			{}
			{\typ{f : \O X \Impl X}{f : \O X \Impl X}}
		 & \ifnarrow
		     \mskip-230mu
		     \def\Dr#1{\stack{\mskip70mu#1}{\Derive{30}}}%
		 \else
		     \def\Dr#1{#1}%
		 \fi
		 \Dr{\ifr{\Impl\mbox{E}}
			{\ifr{\rsubtyp}
			    {\ifr{\mbox{var}}{}{\typ{x:\O A}{x : \O A}}}
			    {\typ{x:\O A}{x : \O\O A\Impl \O X}}
			&
			   \ifr{\rsubtyp}
				{\ifr{\mbox{var}}{}{\typ{x:\O A}{x : \O A}}}
				{\typ{x:\O A}{x : \O\O A}}}
			{\typ{x:\O A}{\app{x}{x} : \O X}}}}
		{\typ{f : \O X \Impl X,\,x:\O A}
		     {\app{f}{(\app{x}{x})} : X}}}
	    {\typ{f : \O X \Impl X}{\FF : \O A \Impl X}}
\end{array}
\]
Then, we can derive \(\Y\) as follows.
\def\GG{f : \O X \Impl X}
\def\L{\lam{x}\app{f}{(\app{x}{x})}}
\def\LL{(\L)}
\vspace{-10pt}
\[
\quad
\ifr{\Impl\mbox{I}}
    {\ifr{\Impl\mbox{E}}
	{\stack{\derive{\Pi}\mskip65mu}
		{\typ{\GG}{\L : \O A \Impl X}}
	 & \ifnarrow
	     \mskip-150mu
	     \def\Dr#1{\stack{\mskip60mu#1}{\Derive{25}}}%
	 \else
	     \def\Dr#1{#1}%
	 \fi
	 \Dr{\ifr{\rsubtyp}
	       {\stack{\derive{\Pi}\mskip65mu}
		      {\typ{\GG}{\L : \O A \Impl X}}}
	       {\typ{\GG}{\L : \O A}}}}
	{\typ{\GG}{\app{\LL}{\LL} : X}}}
    {\typ{}{\lam{f}{\app{\LL}{\LL}} : (\O X \Impl X) \Impl X}}
\]
\end{example}
\def\QQ{\lam{x}\lam{f}\app{f}{(\app{\app{x}{x}}{f})}}
We can also observe that
Turing's fixed-point combinator \(\app{(\QQ)\penalty0}{(\QQ)}\)
has the same type.
More generally,
we can derive \(\typ{}{\Y_{\!n}:(\O X \Impl X) \Impl X}\) for every \(n\),
where \(\Y_0 = \Y\) and
\(\Y_{\!n{+}1} = \Y_n(\lam{y}{\lam{f}\app{f}{(\app{y}{f})}})\).

The type \((\O X \Impl X) \Impl X\)
gives a concise axiomatic meaning to the fixed-point combinators;
it says that they can produce an element of \(X\)
with a given function that works as an information
pump from \(\O X\) to \(X\); in other words,
they provide the induction scheme discussed in Section~\ref{intro-sec}.
The type thus enables us to construct recursive programs
using the fixed-point combinators
without analyzing their computational behavior.
We shall see some examples of
such recursive programs in Section~\ref{program-sec}.

\ifdetail
\begingroup
\def\M{\mathbf{M}}
\def\B{\mathbf{B}}
\def\N{\mathbf{N}}
\begin{example}\label{non-standard-Y}
Let\/ \(\N = \app{\app{\B}{\M}}{(\app{\app{\B}{(\app{\B}{\M})}}{\B})}\),
where\/ \(\M = \lam{x}{xx}\) and\/
\(\B = \lam{f}\lam{g}\lam{x}\app{f}{(\app{g}{x})}\).
The \(\lambda\)-term \(\N\) is known to have the same B\"{o}hm tree as \(\Y\)
while it does not satisfy\/
\(\app{\N}{f} \ctc \app{f}{(\app{\N}{f})}\) {\rm \cite{statman-1993}}.
We can observe that \(N\) also has the type \((\O X \Impl X) \Impl X\).
Let 
\[
A = \fix{Y}\O Y \Impl X,~~
B = (\O A \Impl \O X) \Impl X,~\mbox{and}~~
C = (\O A \Impl \O X) \Impl A.
\]
First, \(\M\) can be typed with both \(A \Impl X\) and \(B \Impl X\) 
as follows.
\[
    \ifr{\Impl\mbox{I}}
		 {\ifr{\Impl\mbox{E}}
		    {\ifr{\eqtyp}
			{\ifr{\mbox{var}}{}{\typ{x:A}{x :A}}}
			{\typ{x:A}{x : \O A\Impl X}}
		    &
		       \ifr{\rsubtyp}
			    {\ifr{\mbox{var}}{}{\typ{x:A}{x : A}}}
			    {\typ{x:A}{x : \O A}}}
		    {\typ{x:A}{\app{x}{x} : X}}}
	    {\typ{}{\M : A \Impl X}}
\]
Since \(B = (\O A \Impl \O X) \Impl X
\psubtyp \O (A \Impl X) \Impl X
\psubtyp \O (\O A \Impl X) \Impl X
\eqtyp \O A \Impl X
\eqtyp A\),
\[
   \Ifr{\r{\rsubtyp}.}
	{\stack{\derive{}}{\typ{}{\M : A \Impl X}}
	    & A \Impl X \psubtyp B \Impl X}
	{\typ{}{\M : B \Impl X}}
\]
Then, because \(\B\) can be typed with
\((D \Impl E) \Impl (F \Impl D) \Impl F \Impl E\)
for any \(D\), \(E\) and \(F\),
\[
 \ifr{\Impl\mbox{E}}
     {\ifr{\Impl\mbox{E}}
	    {\stack{\derive{}}
		 {\typ{}{\B : (B \Impl X) \Impl ((\O X \Impl X) \Impl B)
		    \Impl (\O X \Impl X) \Impl X}}
	     & \stack{\derive{}}{\typ{}{\M : B \Impl X}}}
	     {\typ{}{\app{\B}{\M} : ((\O X \Impl X) \Impl B)
		\Impl (\O X \Impl X) \Impl X}}
     & \mskip-600mu
     \ifr{\Impl\mbox{E}}
	 {\ifr{\Impl\mbox{E}}
	     {\stack{\mskip-170mu\derive{}}
		 {\typ{}{\B : (C \Impl B) \Impl ((\O X \Impl X) \Impl C)
		    \Impl (\O X \Impl X) \Impl B}}
	     & \mskip-230mu
	     \ifr{\Impl\mbox{E}}
		  {\stack{\derive{}}{\typ{}{\B : (A \Impl X) \Impl C \Impl B}}
		   & \stack{\derive{}}{\typ{}{\M : A \Impl X}}}
		  {\mskip180mu
		      \stack{\typ{}{\app{\B}{\M} : C \Impl B}}
			    {\mskip30mu\Derive{30}{}}}}
	     {\typ{}{\app{\B}{(\app{\B}{\M})}
		 : ((\O X \Impl X) \Impl C) \Impl (\O X \Impl X) \Impl B}}
	    & \mskip-290mu
	     \ifr{\eqtyp}
		  {\stack{\derive{}}
			 {\typ{}{\B:(\O X \Impl X) \Impl
			      (\O A \Impl \O X) \Impl \O A \Impl X}}}
		  {\mskip150mu
		      \stack{\typ{}{\B:(\O X \Impl X) \Impl C}}
			    {\mskip110mu\Derive{100}{}}}}
	 {\mskip330mu\stack{\typ{}{\app{\app{\B}{(\app{\B}{\M})}}{\B}:
			 (\O X \Impl X) \Impl B}}
		{\mskip185mu\Derive{60}{}}}}
     {\typ{}{\app{\app{\B}{\M}}{(\app{\app{\B}{(\app{\B}{\M})}}{\B})}
	     : (\O X \Impl X) \Impl X}}
\]
\end{example}
\endgroup
\fi

\Section{Basic properties of {\lA}}\label{basic-prop-sec}

In this section, we show
the soundness of the typing system {\lA}
with respect to the semantics of types
presented in Section~\ref{semantics-sec}, and
its subject reduction property.
We begin with some basic property of the typing system.

\begin{definition}[Skeletons and sizes of derivations]
    \ilabel{typ-skeleton-def}{skeleton}
    \ilabel{typ-size-def}{size}
The {\em skeleton} of a derivation is defined to be the tree
obtained from the derivation by erasing all the judgments occurring in it.
That is, a skeleton is a finite tree with nodes labeled by names
of the inference rules listed in Definition~\ref{typ-rules}.
The {\em size} of a derivation is defined to be the number of
nodes involved in the skeleton of the derivation.
\end{definition}

\begin{proposition}\label{typ-basic}\pushlabel
Let\/ \(\G \cup \{\,x:A,\,y:A\,\}\) be a well-formed typing context, that is,
\(x \in \Dom(\G)\) implies \(\G(x) = A\), and the same for \(y\).
\begin{Enumerate}
\item \itemlabel{typ-var-rename}
    If\/ \(\typ{\G\cup\{\,x:A\,\}}{M:B}\) is derivable, then
    so is \(\typ{\G\cup\{\,y:A\,\}}{M[y/x]:B}\)
    without changing the skeleton of the derivation.
    {\rm\bf(renaming)}
\item \itemlabel{typ-weakening}
    If\/ \(\typ{\G}{M:B}\) is derivable,
    then so is \(\typ{\G\cup\{\,x:A\,\}}{M:B}\)
    without changing the skeleton of the derivation.
    {\rm\bf(weakening)}
\item \itemlabel{typ-var-separate}
    If\/ \(\typ{\G\cup\{\,x:A\,\}}{M[x/y]:B}\) is derivable, then so is
    \(\typ{\G\cup\{\,x:A,\,y:A\,\}}{M:B}\)
    without changing the skeleton of the derivation.
    {\rm\bf(separation)}
\item \itemlabel{typ-var-erase}
    Suppose that\/ \(x \not\in \FV{M}\).
    If\/
    \(\typ{\G\cup\{\,x:A\,\}}{M:B}\) is derivable, then so is \(\typ{\G}{M:B}\)
    without changing the skeleton of the derivation.
    {\rm\bf(erasing)}
\end{Enumerate}
\end{proposition}
\begin{proof}
By induction on the skeletons of the derivations, and by cases on the rule
applied last in them.
The proofs of Items~\itemref{typ-var-rename}
and \itemref{typ-weakening} proceed by simultaneous induction.
We use Item~\itemref{typ-var-rename} for \itemref{typ-var-separate},
and use Item~\itemref{typ-weakening} for \itemref{typ-var-separate}
and \itemref{typ-var-erase}.
Note that since we have the \r{\t}-rule,
\(\FV{M} \subseteq \Dom(\G)\) is not always the case when
\(\typ{\G}{M:B}\) is derivable for some \(B\).
\ifdetail
\paragraph{Proof of \protect\itemref{typ-var-rename}}
The proof proceeds by simultaneous induction
with Item~\itemref{typ-weakening}.
Suppose that \(\typ{\G \cup \{\,x:A\,\}}{M : B}\) is derivable,
and that \(\G \cup \{\,y:A\,\}\) is a well-formed typing context.

\Case{\r{\mbox{var}}.}
The derivation has the form of
\[
\ifr{\mbox{var}}
    {}
    {\typ{\G' \cup \{\,x:A,\,z:B\,\}}{z : B}}
\]
for some \(z\) and \(\G'\)
such that \(\G\cup\{\,x:A\,\} = \G'\cup\{\,x:A,\,z:B\,\}\) and \(z = M\).
If \(z = x\), then \(A = B\) and \(M[y/x] = y\); and hence,
    \(\typ{\G \cup \{\,y:A\,\}}{M[y/x] : B}\) is derivable by \r{\mbox{var}}.
Otherwise, \(M[y/x] = z\); and hence,
\(\typ{\G \cup \{\,y:A\,\}}{M[y/x] : B}\) is also derivable
by \r{\mbox{var}} since \(\{\,z:B\,\} \subseteq \G\).

\Case{\r{\mbox{shift}}.}
The derivation ends with
\[
\Ifr{\r{\mbox{shift}}.}
    {\typ{\O\G\cup\{\,x:\O A\,\}}{M : \O B}}
    {\typ{\G\cup\{\,x:A\,\}}{M : B}}
\]
By induction hypothesis, we have a derivation of
\(\typ{\O\G\cup\{\,y:\O A\,\}}{M[y/x] : \O B}\),
    the skeleton of which is identical to
    the one of the original antecedent; and hence,
by applying \r{\mbox{shift}},
\(\typ{\G\cup\{\,y:A\,\}}{M[y/x] : B}\) is derivable without
changing the skeleton of the original derivation
of \(\typ{\G\cup\{\,x:A\,\}}{M : B}\).

\Case{\r{\t}.}
\(\typ{\G\cup\{\,y:A\,\}}{M[y/x]:B}\) is also derivable by \r{\t}
since \(B = \t\) in this case.

\Case{\r{\rsubtyp}.}
The derivation ends with
\[
\ifr{\rsubtyp}
    {\typ{\G\cup\{\,x:A\,\}}{M : B'}
	& B' \psubtyp B}
    {\typ{\G\cup\{\,x:A\,\}}{M : B}}
\]
for some \(B'\).
\(\typ{\G\cup\{\,y:A\,\}}{M[y/x] : B'}\) is derivable by induction hypothesis
without changing the skeleton of the derivation of the antecedent.
Hence, by \r{\rsubtyp}, we can also derive
\(\typ{\G\cup\{\,y:A\,\}}{M[y/x] : B}\)
without changing the skeleton of the original derivation.

\Case{\r{\Impl\mbox{I}}.}
The derivation ends with
\[
\ifr{\Impl\mbox{I}}
    {\typ{\G\cup\{\,x:A,\,z:C\,\}}{K : D}}
    {\typ{\G\cup\{\,x:A\,\}}{\lam{z}K : C \Impl D}}
\]
for some \(C\), \(D\), \(z\) and \(K\) such that \(B = C \Impl D\)
and \(M = \lam{z}K\).
Let \(z'\) be a fresh individual variable.
By induction hypothesis,
\(\typ{\G\cup\{\,x:A,\,z':C\,\}}{K[z'/z] : D}\) is derivable
without changing the skeleton.
Hence, by induction hypothesis again,
so is \(\typ{\G\cup\{\,y:A,\,z':C\,\}}{K[z'/z][y/x] : D}\), from which
we can derive \(\typ{\G\cup\{\,y:A\,\}}{\lam{z'}K[z'/z][y/x] : B}\)
by the same \r{{\Impl}\mbox{I}}.
Note that \(M[y/x] = (\lam{z}K)[y/x] = \lam{z'}K[z'/z][y/x]\).

\Case{\r{\Impl\mbox{E}}.}
The derivation ends with
\[
\ifr{\Impl\mbox{E}}
    {\typ{\G_1}{M_1 : C \Impl B}
     & \typ{\G_2}{M_2 : C}}
    {\typ{\G \cup \{\,x:A\,\}}{\app{M_1}{M_2} : B}}
\]
for some \(\G_1\), \(\G_2\), \(C\), \(M_1\) and \(M_2\) such that
\(\G \cup \{\,x:A\,\} = \G_1 \cup \G_2\) and \(M = \app{M_1}{M_2}\).
By the induction hypothesis for Item~\itemref{typ-weakening},
we can assume that \(\G_1 = \G_2 = \G \cup \{\,x:A\,\}\).
Hence, by the induction hypothesis for Item~\itemref{typ-var-rename},
we can derive \(\typ{\G\cup\{\,y:A\,\}}{M_1[y/x] : C \Impl B}\) and
\(\typ{\G\cup\{\,y:A\,\}}{M_2[y/x] : C}\) without changing the skeletons.
Therefore, by applying \r{{\Impl}\mbox{E}},
\(\typ{\G\cup\{\,y:A\,\}}{M[y/x] : B}\) is also derivable
without changing the original skeleton.

\paragraph{Proof of \protect\itemref{typ-weakening}}
The proof proceeds by simultaneous induction
with Item~\itemref{typ-var-rename}.
Suppose that \(\typ{\G}{M : B}\) is derivable,
and that \(\G \cup \{\,x:A\,\}\) is a well-formed typing context.

\Case{\r{\mbox{var}}.}
\(\typ{\G \cup \{\,x:A\,\}}{M:B}\) is also derivable by \r{\mbox{var}}
since \(\G \cup \{\,x:A\,\}\) is a well-formed typing context.

\Case{\r{\mbox{shift}}.}
The derivation ends with
\[
\Ifr{\r{\mbox{shift}}.}
    {\typ{\O\G}{M : \O B}}
    {\typ{\G}{M : B}}
\]
Note that
\(\O\G\cup\{\,x:\O A\,\}\) is also a well-formed typing context.
Therefore, by induction hypothesis,
\(\typ{\O\G\cup\{\,x:\O A\,\}}{M : \O B}\) is derivable
without changing the skeleton; and hence,
so is \(\typ{\G\cup\{\,x:A\,\}}{M : B}\) by \r{\mbox{shift}}.

\Case{\r{\t}.}
\(\typ{\G \cup \{\,x:A\,\}}{M:B}\) is also derivable by \r{\t}
since \(B = \t\).

\Case{\r{\rsubtyp}.}
In this case, \(\typ{\G}{M : B'}\) is derivable for some \(B'\)
such that \(B' \psubtyp B\).
Hence, by induction hypothesis, so is
\(\typ{\G\cup\{\,x:A\,\}}{M : B'}\)
without changing the skeleton, from which
we can derive
\(\typ{\G\cup\{\,x:A\,\}}{M : B}\) by the same rule.

\Case{\r{\Impl\mbox{I}}.}
The derivation ends with
\[
\ifr{\Impl\mbox{I}}
    {\typ{\G\cup\{\,y:C\,\}}{K : D}}
    {\typ{\G}{\lam{y}K : C \Impl D}}
\]
for some \(C\), \(D\), \(y\) and \(K\) such that \(B = C \Impl D\)
and \(M = \lam{y}K\).
Let \(y'\) be a fresh individual variable.
By the induction hypothesis for Item~\itemref{typ-var-rename},
\(\typ{\G\cup\{\,y':C\,\}}{K[y'/y] : D}\) is derivable
without changing the skeleton.
Note also that \(\G\cup\{\,x:A,\,y':C\,\}\) is also a well-formed typing context
since \(y'\) is fresh.
Hence, by the induction hypothesis for Item~\itemref{typ-weakening},
we have a derivation of
\(\typ{\G\cup\{\,x:A,\,y':C\,\}}{K[y'/y] : D}\)
whose skeleton is identical to the one for the original antecedent.
Therefore, by applying the same \r{{\Impl}\mbox{I}},
we can derive \(\typ{\G\cup\{\,x:A\,\}}{\lam{y'}K[y'/y] : B}\)
without changing the skeleton for the original consequent.
Note that \(M = \lam{y}K = \lam{y'}K[y'/y]\).

\Case{\r{\Impl\mbox{E}}.}
The derivation ends with
\[
\ifr{\Impl\mbox{E}}
    {\typ{\G_1}{M_1 : C \Impl B}
     & \typ{\G_2}{M_2 : C}}
    {\typ{\G}{\app{M_1}{M_2} : B}}
\]
for some \(\G_1\), \(\G_2\), \(C\), \(M_1\) and \(M_2\) such that
\(\G = \G_1 \cup \G_2\) and \(M = \app{M_1}{M_2}\).
By induction hypothesis,
\(\typ{\G_1\cup\{\,x:A\,\}}{M_1 : C \Impl B}\) and
\(\typ{\G_2\cup\{\,x:A\,\}}{M_2 : C}\) are derivable
without changing the skeletons.
Therefore, so is
\(\typ{\G\cup\{\,x:A\,\}}{M : B}\) by applying the same \r{{\Impl}\mbox{E}}.

\paragraph{Proof of \protect\itemref{typ-var-separate}}
Suppose that \(\typ{\G\cup\{\,x:A\,\}}{M[x/y] : B}\) is derivable, and that
\(\G \cup \{\,x:A,\,y:A\,\}\) is a well-formed typing context.

\Case{\r{\mbox{var}}.}
The derivation has the form of
\[
\ifr{\mbox{var}}
    {}
    {\typ{\G' \cup \{\,x:A,\,z:B\,\}}{z : B}}
\]
for some \(z\) and \(\G'\)
such that \(\G\cup\{\,x:A\,\} = \G'\cup\{\,x:A,\,z:B\,\}\) and \(z = M[x/y]\).
If \(M = y\), then we get \(x = z\) from \(z = M[x/y]\); and hence, \(A = B\).
Therefore,
\(\typ{\G \cup \{\,x:A,\,y:A\,\}}{M : B}\) is derivable by \r{\mbox{var}}.
On the other hand,
if \(M \not= y\), then \(z = M[x/y]\) implies \(M = z\); and hence,
\(\typ{\G \cup \{\,x:A,\,y:A\,\}}{M : B}\) is also derivable by \r{\mbox{var}}
since \(\{\,z:B\,\} \subseteq \G\cup\{\,x:A\,\}\).

\Case{\r{\mbox{shift}}.}
The derivation ends with
\[
\Ifr{\r{\mbox{shift}}.}
    {\typ{\O\G\cup\{\,x:\O A\,\}}{M[x/y] : \O B}}
    {\typ{\G\cup\{\,x:A\,\}}{M[x/y] : B}}
\]
Note that
\(\O\G\cup\{\,x:\O A,\,y:\O A\,\}\) is also a well-formed typing context.
Therefore, by induction hypothesis,
\(\typ{\O\G\cup\{\,x:\O A,\,y:\O A\,\}}{M : \O B}\)
is derivable without changing the skeleton;
and hence, so is \(\typ{\G\cup\{\,x:A,\,y:A\,\}}{M : B}\) by the same rule.

\Case{\r{\t}.}
\(\typ{\G\cup\{\,x:A,\,y:A\,\}}{M:B}\) is also derivable
by \r{\t} since \(B = \t\).

\Case{\r{\rsubtyp}.}
The derivation ends with
\[
\ifr{\rsubtyp}
    {\typ{\G\cup\{\,x:A\,\}}{M[x/y] : B'}
	& B' \psubtyp B}
    {\typ{\G\cup\{\,x:A\,\}}{M[x/y] : B}}
\]
for some \(B'\).
By induction hypothesis,
\(\typ{\G\cup\{\,x:A,\,y:A\,\}}{M : B'}\)
is derivable without changing the skeleton.
Hence, so is \(\typ{\G\cup\{\,x:A,\,y:A\,\}}{M : B}\)
by the same \r{\rsubtyp}.

\Case{\r{\Impl\mbox{I}}.}
The derivation ends with
\[
\ifr{\Impl\mbox{I}}
    {\typ{\G\cup\{\,x:A,\,z:C\,\}}{K : D}}
    {\typ{\G\cup\{\,x:A\,\}}{\lam{z}K : C \Impl D}}
\]
for some \(C\), \(D\), \(z\) and \(K\) such that \(B = C \Impl D\)
and \(M[x/y] = \lam{z}K\).
By Item~\itemref{typ-var-rename} of this proposition, we can assume
that \(z \not\in \{\,x,\,y\,\}\) without loss of generality.
Since \(M[x/y] = \lam{z}K\),
we get \(M = \lam{z}K'\) for some \(K'\) such that \(K = K'[x/y]\).
Therefore, by induction hypothesis,
\(\typ{\G\cup\{\,x:A,\,y:A,\,z:C\,\}}{K' : D}\) is derivable
without changing the skeleton; and hence,
so is \(\typ{\G\cup\{\,x:A,\,y:A\,\}}{\lam{z}K' : B}\) by the same rule.

\Case{\r{\Impl\mbox{E}}.}
The derivation ends with
\[
\ifr{\Impl\mbox{E}}
    {\typ{\G_1}{M_1 : C \Impl B}
     & \typ{\G_2}{M_2 : C}}
    {\typ{\G \cup \{\,x:A\,\}}{\app{M_1}{M_2} : B}}
\]
for some \(\G_1\), \(\G_2\), \(C\), \(M_1\) and \(M_2\) such that
\(\G \cup \{\,x:A\,\} = \G_1 \cup \G_2\) and \(M[x/y] = \app{M_1}{M_2}\).
By Item~\itemref{typ-weakening} of this proposition,
we can assume that \(\G_1 = \G_2 = \G \cup \{\,x:A\,\}\).
Note that \(M[x/y] = \app{M_1}{M_2}\) implies
\(M = \app{M_1'}{M_2'}\)
for some \(M_1'\) and \(M_2'\) such that
\(M_1 = M_1'[x/y]\) and \(M_2 = M_2'[x/y]\).
Therefore, by induction hypothesis,
\(\typ{\G \cup \{\,x:A,\,y:A\,\}}{M_1' : C \Impl B}\) and
\(\typ{\G \cup \{\,x:A,\,y:A\,\}}{M_2' : C}\) are derivable
without changing the skeletons; and hence, so is
\(\typ{\G\cup\{\,x:A,\,y:A\,\}}{M : B}\)
by applying the same \r{{\Impl}\mbox{E}}.

\paragraph{Proof of \protect\itemref{typ-var-erase}}
Suppose that \(\typ{\G\cup\{\,x:A\,\}}{M : B}\) is derivable,
and \(x \not\in \FV{M}\).

\Case{\r{\mbox{var}}.}
The derivation has the form of
\[
\ifr{\mbox{var}}
    {}
    {\typ{\G' \cup \{\,x:A,\,y:B\,\}}{y : B}}
\]
for some \(y\) and \(\G'\)
such that \(\G\cup\{\,x:A\,\} = \G'\cup\{\,x:A,\,y:B\,\}\) and \(y = M\).
Note that \(x \not= y\); and hence, \(\{\,y:B\,\} \subseteq \G\)
from \(x \not\in \FV{M}\).
Therefore, we can derive \(\typ{\G}{M : B}\) by \r{\mbox{var}}.

\Case{\r{\mbox{shift}}.}
The derivation ends with
\[
\Ifr{\r{\mbox{shift}}.}
    {\typ{\O\G\cup\{\,x:\O A\,\}}{M : \O B}}
    {\typ{\G\cup\{\,x:A\,\}}{M : B}}
\]
By induction hypothesis,
\(\typ{\O\G}{M : \O B}\) is derivable without changing the skeleton;
and hence, so is \(\typ{\G}{M : B}\) by the same rule.

\Case{\r{\t}.}
\(\typ{\G}{M:B}\) is also derivable by \r{\t} since \(B = \t\).

\Case{\r{\rsubtyp}.}
The derivation ends with
\[
\ifr{\rsubtyp}
    {\typ{\G\cup\{\,x:A\,\}}{M : B'}
	& B' \psubtyp B}
    {\typ{\G\cup\{\,x:A\,\}}{M : B}}
\]
for some \(B'\).
By induction hypothesis,
\(\typ{\G}{M : B'}\) is derivable without changing the skeleton.
Hence, so is \(\typ{\G}{M : B}\) by the same \r{\rsubtyp}.

\Case{\r{\Impl\mbox{I}}.}
The derivation ends with
\[
\ifr{\Impl\mbox{I}}
    {\typ{\G\cup\{\,x:A,\,y:C\,\}}{K : D}}
    {\typ{\G\cup\{\,x:A\,\}}{\lam{y}K : C \Impl D}}
\]
for some \(C\), \(D\), \(y\) and \(K\) such that \(B = C \Impl D\)
and \(M = \lam{y}K\).
If \(x \not= y\), we get \(x \not\in \FV{K}\) from \(x \not\in \FV{\lam{y}K}\).
In this case,
by induction hypothesis,
\(\typ{\G\cup\{\,y:C\,\}}{K : D}\) is derivable
without changing the skeleton; and hence,
so is \(\typ{\G}{\lam{y}K : B}\) by the same rule.
On the other hand, if \(x = y\), then \(A = C\) and the antecedent is
identical to \(\typ{\G\cup\{\,y:C\,\}}{K : D}\).
Therefore, we can also derive \(\typ{\G}{\lam{y}K : B}\) by \r{{\Impl}\mbox{I}}.

\Case{\r{\Impl\mbox{E}}.}
The derivation ends with
\[
\ifr{\Impl\mbox{E}}
    {\typ{\G_1}{M_1 : C \Impl B}
     & \typ{\G_2}{M_2 : C}}
    {\typ{\G \cup \{\,x:A\,\}}{\app{M_1}{M_2} : B}}
\]
for some \(\G_1\), \(\G_2\), \(C\), \(M_1\) and \(M_2\) such that
\(\G \cup \{\,x:A\,\} = \G_1 \cup \G_2\) and \(M = \app{M_1}{M_2}\).
By Item~\itemref{typ-weakening} of this proposition,
we can assume that \(\G_1 = \G_2 = \G \cup \{\,x:A\,\}\).
Since \(x \not\in \FV{M_1}\) and \(x \not\in \FV{M_2}\),
by induction hypothesis,
\(\typ{\G}{M_1 : C \Impl B}\) and
\(\typ{\G}{M_2 : C}\) are derivable without changing the skeletons.
Therefore, so is
\(\typ{\G}{M : B}\) by applying the same \r{{\Impl}\mbox{E}}.
\fi 
\qed\CHECKED{2014/07/12}
\end{proof}

Now we can show the soundness of {\lA} with respect to the
semantics \(\II\) of type expressions over {\lA}-frames.

\begin{theorem}[Soundness of {\lA}]
    \ilabel{soundness-theorem}{soundness!lambda-A@\protect\lA}
If\/ \(\typ{\{\,x_1 : A_1,\,\ldots,\,x_n : A_n\,\}}{M:B}\)
is derivable in {\lA}, then
for any {\lA}-frame \(\pair{\W}{\acc}\),
\(\VI{M}{\ienv} \in \I{B}^\tenv_p\,\)
    for every \(p \in \W\), hereditary \(\tenv\) and \(\ienv\) whenever
    \(\ienv(x_i) \in \I{A_i}^\tenv_{p}\;\)
    for every \(i\) (\(i = 1,\,2,\,\ldots,\,n\)).
\end{theorem}
\begin{proof}
By induction on the derivation, and by cases on the last rule used in it.
Most cases are straightforward.
Use Theorem~\ref{psubtyp-soundness}
for the case of \(\r{{\rsubtyp}}\).
Let \(\G = \{\,x_1 : A_1,\,\ldots,\,x_n : A_n\,\}\), and
\(\typ{\G}{M:B}\) be a derivable judgment.
Suppose that \(\pair{\W}{\acc}\) is a {\lA}-frame,
and that \(\ienv(x_i) \in \I{A_i}^\tenv_{p}\) (\(i = 1,\,2,\,\ldots,n\)).

\Case{\r{\mbox{var}}.}
In this case,
\(M = x_j\) and \(B = A_j\) for some \(j\) (\(1 \le j \le n\)).
Therefore, by assumption, \(\VI{M}{\ienv} = \VI{x_j}{\ienv} = \ienv(x_j)
    \in \I{A_j}^\tenv_p = \I{B}^\tenv_p\).

\Case{\r{\mbox{shift}}.} 
In this case, the derivation ends with
\[
    \Ifr{\r{\mbox{shift}}.}
        {\typ{\{\,x_1:\O A_1,\,x_2:\O A_2,\,\ldots,\,x_n:\O A_n,\,\}}
	     {M : \O B}}
        {\typ{\{\,x_1:A_1,\,x_2:A_2,\,\ldots,\,x_n:A_n\,\}}
	     {M : B}}
\]
For the given \(\pair{\W}{\acc}\), \(\tenv\) and \(p\),
we construct another {\lA}-frame \(\pair{\W'}{\acc'}\)
and a hereditary type environment \(\tenv'\) for \(\pair{\W'}{\acc'}\)
by extending \(\pair{\W}{\acc}\) and \(\tenv\), respectively,
as follows.
\begin{eqnarray*}
\W' &=& \{\,\star\,\} \cup \W \qquad(\star \not\in \W) \\[4pt]
q \acc' r ~&\mbox{iff}& ~\Choice{%
		\mbox{\(q = \star\) and \(r = p\), or} \\
		\mbox{\(q\), \(r \in \W\) and \(q \acc r\)}
	    } \\[3pt]
    \tenv'(X)_q &=& ~\Choice{%
	\tenv(X)_p \qquad(q = \star) \\
	\tenv(X)_q \qquad(q \in \W)
    }
\end{eqnarray*}
where \(\star\) is a fresh world added to the original ones.
Obviously, \(\acc'\) is (conversely) well-founded, and
\(\tenv'\) hereditary.
Observe also that \(\acc'\) is locally linear; and hence,
\(\pair{\W'}{\acc'}\) constitutes a {\lA}-frame.
Then, let \(\Ip{A}^{\tenv'}\) be the interpretation of \(A\)
in the {\lA}-frame \(\pair{\W'}{\acc'}\)
under the type environment \(\tenv'\).
Note that \(\Ip{A}^{\tenv'}_q = \I{A}^\tenv_q\)
for every \(q \in \W\) and \(A\).
Since \(\ienv(x_i) \in \I{A_i}^\tenv_p\),
we get \(\ienv(x_i) \in \Ip{A_i}^{\tenv'}_q\) for every \(q \opacc' \star\)
by the definition of \(\acc'\).
Hence, \(\ienv(x_i) \in \Ip{\O A_i}^{\tenv'}_\star\).
By induction hypothesis,
we get \(\VI{M}{\ienv} \in \Ip{\O B}^{\tenv'}_\star\); and hence,
\(\VI{M}{\ienv} \in \Ip{B}^{\tenv'}_p = \I{B}^\tenv_p\)
because \(\star \acc' p\strut\).

\Case{\r{\t}.}
Obvious since \(\I{\t}^\tenv_p = \V\) for every \(p \in \W\)
by Definition~\ref{rlz-def}.

\Case{\r{\rsubtyp}.} For some \(B'\), the derivation ends with
\[
    \Ifr{\r{{\rsubtyp}}.}
        {\typ{\G}{M : B'}
	 & \subt{}{B' \psubtyp B}}
        {\typ{\G}{M : B}}
\]
We get \(\VI{M}{\ienv} \in\I{B}^\tenv_p\)
by Theorem~\ref{psubtyp-soundness} because
\(\VI{M}{\ienv} \in\I{B'}^\tenv_p\) by induction hypothesis.

\Case{\r{\Impl\mbox{I}}.}
The derivation ends with
\[
\ifr{\Impl\mbox{I}}
    {\typ{\G \cup \{\,y:B_1\,\}}{L:B_2}}
    {\typ{\G}{\lam{y}{L} : B_1 \Impl B_2}}
\]
for some \(y\), \(L\), \(B_1\) and \(B_2\) such that
\(M = \lam{y}{L}\) and \(B = B_1 \Impl B_2\).
Suppose that \(p \tacc r\) and \(v \in \I{B_1}^\tenv_r\).
By Proposition~\ref{rlz-another-def}, it suffices to show that
\(\VI{\lam{y}{L}}{\ienv} \cdot v \in \I{B_2}^\tenv_r\).
Since \(p \tacc r\), by Proposition~\ref{rlz-hereditary},
\(\ienv[v/y](x_i) = \ienv(x_i) \in \I{A_i}^\tenv_p \subseteq
     \I{A_i}^\tenv_r\)
for every \(i\) such that \(x_i \not= y\), and
\(\ienv[v/y](y) = v \in \I{B_1}^\tenv_p \subseteq \I{B_1}^\tenv_r\).
Hence, by induction hypothesis, \(\VI{L}{\ienv[v/y]} \in \I{B_2}^\tenv_r\).
Therefore,
\(\VI{\lam{y}{L}}{\ienv} \cdot v = \VI{L}{\ienv[v/y]} \in \I{B_2}^\tenv_r\)
by Definition~\ref{lambda-algebra}.

\Case{\r{\Impl\mbox{E}}.}
The derivation ends with
\[
    \ifr{\Impl\mbox{E}}
	{\typ{\G_1}{M_1:C \Impl B}
	 & \typ{\G_2}{M_2:C}}
        {\typ{\G_1\cup\G_2}{\app{M_1}{M_2} : B}}
\]
for some \(\G_1\), \(\G_2\), \(M_1\), \(M_2\) and \(C\)
such that \(\G = \G_1 \cup \G_2\) and \(M = \app{M_1}{M_2}\).
Since \(\G_1, \G_2 \subseteq \G\), by induction hypothesis,
\(\VI{M_1}{\ienv} \in \I{C \Impl B}^\tenv_p\) and
\(\VI{M_2}{\ienv} \in \I{C}^\tenv_p\).
Hence, \(\VI{\app{M_1}{M_2}}{\ienv}
    = \VI{M_1}{\ienv}\cdot\VI{M_2}{\ienv} \in \I{B}^\tenv_p\)
by Definition~\ref{lambda-algebra} and Proposition~\ref{rlz-another-def}.
\qed\CHECKED{2014/06/24, 07/21}
\end{proof}

Theorem~\ref{soundness-theorem} assures us that the modularity of programs
is preserved even if we regard type expressions, or specifications,
as asserting the convergence of programs.
For example, if a type \(B\) comprises values that are normalizable to
canonical ones,
and we have a program \(M\) of a type \(A \Impl B\), then
we can expect that \(M\) terminates and returns such a canonical value
when we provide a value of \(A\).
In Section~\ref{conv-sec},
we will show such convergence properties of well-typed \(\lambda\)-terms
by a discussion on the soundness with respect to
a term model of the untyped \(\lambda\)-calculus.

In the rest of the present section,
we will show other important properties of the typing system,
in particular, the subject reduction property of the system, which is also
necessary to derive a certain convergence property of
well-typed \(\lambda\)-terms (cf. Theorem~\ref{typ-maximal}).
First, we show that the following two rules are derivable in {\lA}.
This makes the proof of the subject reduction much easier.
\[
\ifr{\mbox{nec}}
    {\typ{\G_1}{M : A}}
    {\typ{\O \G_1\cup\G_2}{M : \O A}}
\mskip50mu
\ifr{\mbox{subst}}
    {\typ{\G_1\cup\{\,x:A\,\}}{M : B} & \typ{\G_2}{N : A}}
    {\typ{\G_1\cup\G_2}{M[N/x] : B}}
\]

\begin{proposition}\label{lA-nec-subst-redundant}
The two rules \r{\mbox{nec}} and \r{\mbox{subst}} are derivable in {\lA}.
That is, the following two hold.
\begin{Enumerate}\pushlabel
\item \itemlabel{lA-nec-redundant}
    Let \(\O \G_1 \cup \G_2\) be a well-formed typing context.
    If\/ \(\typ{\G_1}{M:A}\) is derivable,
    then so is\/ \(\typ{\O\G_1 \cup \G_2}{M:\O A}\).
\item \itemlabel{lA-subst-redundant}
    Let \(\G_1 \cup \G_2\) be a well-formed typing context.
    If\/ \(\typ{\G_1 \cup \{\,x:A\,\}}{M:B}\) and \(\typ{\G_2}{N:A}\)
    are derivable, then so is\/ \(\typ{\G_1 \cup \G_2}{M[N/x]:B}\).
\end{Enumerate}
\end{proposition}
\begin{proof}
Each proof proceeds by induction on the size of the derivation for \(M\), and
each induction step is shown by cases on the last rule applied in it.
For Item~\itemref{lA-nec-redundant},
suppose that \(\typ{\G_1}{M:A}\) is derivable.

\Case{\r{\mbox{var}}.}
In this case, \(M\) is an individual variable
and \(\{\,M:A\,\} \subseteq \G_1\).
Therefore, \(\typ{\O\G_1 \cup \G_2}{M:\O A}\) is derivable
by \(\r{\mbox{var}}\).

\Case{\r{\mbox{shift}}.}
The derivation ends with
\[
\Ifr{\r{\mbox{shift}}.}
    {\typ{\O \G_1}{M : \O A}}
    {\typ{\G_1}{M : A}}
\]
Since \(\typ{\O\O \G_1 \cup \O \G_2}{M : \O\O A}\)
is derivable by induction hypothesis,
so is \(\typ{\O \G_1 \cup \G_2}{M : \O A}\) by \r{\mbox{shift}}.

\Case{\r{\t}.}
In this case, \(A = \t\).
Hence,
\(\typ{\O\G_1 \cup \G_2}{M:\O A}\) is derivable by \r{\t} and \r{\rsubtyp}
since \(\O A = \O \t \peqtyp \t\) by Proposition~\ref{geqtyp-t1}.

\Case{\r{\rsubtyp}.}
Straightforward by induction hypothesis since \(A' \psubtyp A\) implies
\(\O A'\psubtyp \O A\) for every \(A'\) by \r{{\rsubtyp}\mbox{-}{\O}}.

\Case{\r{\Impl\mbox{I}}.}
The derivation ends with
\[
\ifr{\Impl\mbox{I}}
    {\typ{\G_1 \cup \{\,x:B\,\}}{K : C}}
    {\typ{\G_1}{\lam{x}K : B \Impl C}}
\]
for some \(B\), \(C\), \(x\) and \(K\) such that \(A = B \Impl C\)
and \(M = \lam{x}K\).
By induction hypothesis,
\(\typ{\O \G_1 \cup \{\,x:\O B\,\} \cup \G_2}{K : \O C}\) is derivable.
Hence, so is \(\typ{\O \G_1 \cup \G_2}{M:\O A}\)
by \r{\Impl\mbox{I}} and \r{\rsubtyp}
since \(\O B \Impl \O C \peqtyp \O(B \Impl C)\).
Note that the proof would fail at this point
if we did not have the \(\r{{\peqtyp}\mbox{-{\bf K}/{\bf L}}}\)-rule,
and even if we added an extra subtyping rule
\(\O (B \Impl C) \subtyp \O B \Impl \O C\) corresponding
to the axiom schema {\bf K} of normal modal logic.

\Case{\r{\Impl\mbox{E}}.}
In this case,
the derivation ends with
\[
\ifr{\Impl\mbox{E}}
    {\typ{\G_{11}}{M_1 : B \Impl A}
     & \typ{\G_{12}}{M_2 : B}}
    {\typ{\G_{11} \cup \G_{12}}{\app{M_1}{M_2} : A}}
\]
for some \(\G_{11}\), \(\G_{12}\), \(B\), \(M_1\) and \(M_2\) such that
\(\G_1 = \G_{11} \cup \G_{12}\) and \(M = \app{M_1}{M_2}\).
By induction hypothesis,
\(\typ{\O \G_{11} \cup \G_2}{M_1 : \O(B \Impl A)}\) and
\(\typ{\O \G_{12} \cup \G_2}{M_2 : \O B}\) are derivable.
Hence, so is \(\typ{\O\G_1 \cup \G_2}{M:\O A}\)
by \r{\rsubtyp} and \r{\Impl\mbox{E}}
since \(\O(B \Impl A) \psubtyp \O B \Impl \O A\).
This completes the proof of Item~\itemref{lA-nec-redundant}.
\ifdetail

For Item~\itemref{lA-subst-redundant},
suppose that \(\typ{\G_1 \cup \{\,x:A\,\}}{M:B}\) and \(\typ{\G_2}{N:A}\)
are derivable.
By induction on the size of the derivation of
\(\typ{\G_1 \cup \{\,x:A\,\}}{M:B}\),
and by cases on the last rule applied in it.
Let \(\G = \G_1 \cup \G_2\).

\Case{\r{\mbox{var}}.}
In this case, \(M\) is an individual variable.
If \(M = x\), then straightforward by Proposition~\ref{typ-weakening}
since \(A = B\) and \(M[N/x] = N\).
Otherwise, \(M[N/x] = M\) and \(\{\,M:B\,\} \subseteq \G_1\).
Hence, \(\typ{\G}{M[N/x]:B}\) is derivable by \(\r{\mbox{var}}\).

\Case{\r{\mbox{shift}}.}	
The 
\else 

The proof for Item~\itemref{lA-subst-redundant} is
almost straightforward.
The only non-trivial case is when the last rule is \r{\mbox{shift}}.
In this case, the
\fi 
derivation ends with
\[
\Ifr{\r{\mbox{shift}}.}
    {\typ{\O \G_1 \cup \{\,x:\O A\,\}}{M : \O B}}
    {\typ{\G_1 \cup \{\,x:A\,\}}{M : B}}
\]
Since \(\typ{\G_2}{N:A}\) is derivable, so is \(\typ{\O \G_2}{N:\O A}\)
by Item~\itemref{lA-nec-redundant} of this proposition.
Therefore,
\(\typ{\O \G}{M[N/x] : \O B}\) is also derivable by induction hypothesis;
and hence, so is
\(\typ{\G}{M[N/x] : B}\) by \(\r{\mbox{shift}}\).
Note that the \r{\mbox{nec}}-rule,
or Item~\itemref{lA-nec-redundant} of this proposition, is crucial
to this case.
\ifdetail

\Case{\r{\t}.}
Trivial since \(B = \t\).

\Case{\r{\rsubtyp}.}
%
%
The derivation ends with
\[
\ifr{\rsubtyp}
    {\typ{\G_1 \cup \{\,x:A\,\}}{M : B'} & B' \psubtyp B}
    {\typ{\G_1 \cup \{\,x:A\,\}}{M : B}}
\]
for some \(B'\).
Since \(\typ{\G_2}{N:A}\) is derivable,
so is \(\typ{\G}{M[N/x] : B'}\) by induction hypothesis;
and hence, so is \(\typ{\G}{M[N/x] : B}\) by \(\r{\rsubtyp}\).

\Case{\r{\Impl\mbox{I}}.}
The derivation ends with
\[
\ifr{\Impl\mbox{I}}
    {\typ{\G_1 \cup \{\,x:A\,\}\cup\{\,y:C\,\}}{K : D}}
    {\typ{\G_1 \cup \{\,x:A\,\}}{\lam{y}K : C \Impl D}}
\]
for some \(C\), \(D\), \(y\) and \(K\) such that \(B = C \Impl D\)
and \(M = \lam{y}K\).
By Proposition~\ref{typ-var-rename},
we can assume that \(y \not\in \FV{N}\) without loss of generality.
By induction hypothesis,
\(\typ{\G \cup \{\,y:C\,\}}{K[N/x] : D}\) is derivable; and hence,
so is \(\typ{\G}{\lam{y}{(K[N/x])} : C \Impl D}\)
by applying \(\r{\Impl\mbox{I}}\).
Note that \(M[N/x] = (\lam{y}K)[N/x] = \lam{y}{(K[N/x])}\) since
\(y \not\in \FV{N}\) by assumption.

\Case{\r{\Impl\mbox{E}}.}
The derivation ends with
\[
\ifr{\Impl\mbox{E}}
    {\typ{\G_{11}}{M_1 : C \Impl B}
     & \typ{\G_{12}}{M_2 : C}}
    {\typ{\G_{11} \cup \G_{12}}{\app{M_1}{M_2} : B}}
\]
for some \(\G_{11}\), \(\G_{12}\), \(C\),
\(M_1\) and \(M_2\) such that
\(\G_1 \cup \{\,x:A\,\} = \G_{11} \cup \G_{12}\) and \(M = \app{M_1}{M_2}\).
By Proposition~\ref{typ-weakening}, we can assume that
\(\G_{11} = \G_{12} = \G_1 \cup \{\,x:A\,\}\) without loss of generality.
Therefore, by induction hypothesis,
\(\typ{\G}{M_1[N/x] : C \Impl B}\)
and \(\typ{\G}{M_2[N/x] : C}\) are derivable; and hence, so is
\(\typ{\G}{\app{(M_1[N/x])}{(M_2[N/x])} : B}\)
by applying \(\r{\Impl\mbox{E}}\).
Note that \(M[N/x] = (\app{M_1}{M_2})[N/x] = \app{(M_1[N/x])}{(M_2[N/x])}\).
\fi 
\qed\CHECKED{2014/06/24, 07/19}
\end{proof}

Note that Proposition~\ref{lA-nec-subst-redundant}
depends on the \r{{\peqtyp}\mbox{-{\bf K}/{\bf L}}}-rule.
Without this equality, even with the normal subtyping rule
\(\O (A \Impl B) \subtyp \O A \Impl \O B\),
and even if we dropped the \r{\mbox{shift}}-rule from {\lA},
the \r{\mbox{nec}}-rule would not be derivable in {\lA}.
Furthermore, without the \r{{\peqtyp}\mbox{-{\bf K}/{\bf L}}}-rule,
even if we added the \r{\mbox{nec}}-rule to {\lA},
the \r{\mbox{subst}}-rule would remain underivable, because
we could not handle the case of the \r{\mbox{nec}}-rule in
the proof of Proposition~\ref{lA-subst-redundant}.
For example, we have the following two derivations in such a variant of
{\lA}.
\[
    \ifr{\mbox{nec}}
	{\ifr{\Impl\,\mbox{I}}
	     {\ifr{\Impl\,\mbox{E}}
		  {\typ{g:X\Impl Z}{g:X\Impl Z}
		      & \typ{x:X}{x:X}}
		  {\typ{g:X \Impl Z,\:x:X}{\app{g}{x}:Z}}}
	     {\typ{x:X}{\lam{g}\app{g}{x}:(X \Impl Z) \Impl Z}}}
	{\typ{x:\O X}{\lam{g}\app{g}{x}:\O((X \Impl Z) \Impl Z)}}
\mskip15mu
    \ifr{\Impl\,\mbox{E}}
        {\typ{f:Y\Impl \O X}{f:Y\Impl \O X}
	  & \typ{y:Y}{y:Y}}
	{\typ{f:Y \Impl \O X,\:y:Y}{\app{f}{y}:\O X}}
\]
However, we cannot derive the following judgment
without the \r{\mbox{subst}}-rule\footnote{%
    In \cite{nakano-tacs01}, although the author conjectured that
    two typing systems S-\(\lambda{\bullet}\mu\) and F-\(\lambda{\bullet}\mu\)
    described in the paper enjoy some basic properties (Proposition~2),
    such as the substitution lemma and the subject reduction property,
    it turned to be wrong.
    However, fortunately, this does not affect other results, including
    the main ones, presented in the paper.}.
\[
    \typ{f:Y \Impl \O X,\:y:Y}
	{\lam{g}\app{g}{(\app{f}{y})} : \O((X \Impl Z) \Impl Z)}
\]

Our final task in this section is to prove the subject reduction property of
{\lA}.
To this end, we prepare the following three lemmas.

\begin{lemma}\label{typ-subtyp}
Suppose that\/ \(\G' \psubtyp \G\).
If\/ \(\typ{\G}{M:B}\) is derivable, then so is \(\typ{\G'}{M:B}\).
\end{lemma}
\begin{proof}
Suppose that \(\typ{\G\cup\{\,x:A\,\}}{M:B}\) is derivable,
\(A' \psubtyp A\), and \(x \not\in \Dom(\G)\).
It suffices to show that \(\typ{\G\cup\{\,x:A'\,\}}{M:B}\) is also derivable.
Let \(y\) be a fresh individual variable.
Since \(\typ{\{\,y:A'\,\}}{y:A}\) is derivable by \r{\mbox{var}}
and \r{{\rsubtyp}}, so is \(\typ{\G\cup\{\,y:A'\,\}}{M[y/x]:B}\)
by Proposition~\ref{lA-subst-redundant}.
Therefore, \(\typ{\G\cup\{\,x:A'\,\}}{M:B}\) is also derivable
by Proposition~\ref{typ-var-rename}.
\qed\CHECKED{2014/06/24, 07/19}
\end{proof}
By this proposition, we realize that the following extension
of the \(\r{\rsubtyp}\) rule does not affect the derivability of
typing judgments.\hskip-4pt\footnote{%
The typing system presented in \cite{nakano-tacs01} was
formulated with this variant.}
\[
\ifr{\rsubtyp}
    {\typ{\G}{M : A} & \G' \psubtyp \G & A \psubtyp A'}
    {\typ{\G'}{M : A'}}
\]

\begin{lemma}\label{subtyp-O-Impl-lemma}
Suppose that \(A \Impl B \psubtyp \O^n(C \Impl D) \npeqtyp \t\).
Then\/ \(\O^n C \psubtyp \O^k A\) and \(\O^k B \psubtyp \O^n D\)
for some \(k\).
\end{lemma}
\begin{proof}
Note that neither \(A \Impl B\) nor \(C \Impl D\) is a {\tvariant}
by Propositions~\ref{O-tvariant}, \ref{subtyp-gt-t}
and Theorem~\ref{geqtyp-t-tvariant}.
Hence,
\(\Canonp{(A \Impl B)} = \Canonp{A} \Impl \Canonp{B}\)
and \(\Canonp{(\O^n (C \Impl D))} = \O^n \Canonp{C} \Impl \O^n \Canonp{D}\).
Therefore, \(\O^n \Canonp{C} \psubtyp \O^{k} \Canonp{A}\)
and \(\O^{k} \Canonp{B} \psubtyp \O^n \Canonp{D}\)
for some \(k\) by Proposition~\ref{subtyp-O-Impl}, that is,
\(\O^n C \psubtyp \O^{k} A\) and \(\O^{k} B \psubtyp \O^n D\)
by Proposition~\ref{geqtyp-canon}.
\qed\CHECKED{2014/06/24, 07/19}
\end{proof}

%
%
\begin{lemma}[Anti-abstraction lemma]\label{lA-abst-lemma}
Let\/ \(\typ{\G}{\lam{x}{M}: A}\) be a derivable judgment, and
suppose that \(x \not\in \Dom(\G)\).
Then, \(\typ{\O^n\G \cup \{\,x:B\,\}}{M:C}\) is also derivable
for some \(n\), \(B\) and \(C\) such that
\(B \Impl C \psubtyp \O^n A\).
\end{lemma}
\begin{proof}
If \(A \peqtyp \t\), then \(\typ{\G \cup \{\,x:\t\,\}}{M:\t}\) is
derivable by \r{\t}.
Therefore, we assume that \(A \npeqtyp \t\).
The proof proceeds by induction on the derivation,
and by cases on the last rule applied in it.
Note that neither \r{\mbox{var}} nor \r{\Impl\mbox{E}} is possible
because of the form of the \(\lambda\)-term \(\lam{x}{M}\).

\Case{\r{\mbox{shift}}.}
This is the only case that we have to consider \(n\) not being 0.
In this case, the derivation ends with
\[
\Ifr{\r{\mbox{shift}}.}
    {\typ{\O \G}{\lam{x}{M} : \O A}}
    {\typ{\G}{\lam{x}{M} : A}}
\]
Note that \(\O A \npeqtyp \t\) from \(A \npeqtyp \t\).
By induction hypothesis,
\(\typ{\O^{n'}\O \G \cup \{\,x:B\,\}}{M:C}\) is derivable
for some \(n'\), \(B\) and \(C\) such that
\(B \Impl C \psubtyp \O^{n'}\O A\).
Therefore, it suffices to take \(n\) as \(n = n' +1\).

\Case{\r{\t}.}
This case is impossible by the assumption that \(A \npeqtyp \t\).

\Case{\r{\rsubtyp}.}
In this case, the derivation ends with
\[
\ifr{\rsubtyp}
    {\typ{\G}{\lam{x}{M} : A'} & A' \psubtyp A}
    {\typ{\G}{\lam{x}{M} : A}}
\]
for some \(A'\).
By induction hypothesis,
\(\typ{\O^n \G \cup \{\,x:B\,\}}{M:C}\) is derivable
for some \(n\), \(B\) and \(C\) such that \(B \Impl C \psubtyp \O^n A'\), which
implies \(B \Impl C \psubtyp \O^n A\) since \(A' \psubtyp A\).

\Case{\r{\Impl\mbox{I}}.}
Trivial since the derivation ends with
\[
\ifr{\Impl\mbox{I}}
    {\typ{\G\cup\{\,x:B\,\}}{M : C}}
    {\typ{\G}{\lam{x}{M} : B \Impl C}}
\]
for some \(B\) and \(C\) such that \(A = B \Impl C\).
\qed\CHECKED{2014/06/24, 07/20}
\end{proof}

Finally, we can proceed to the proof of the subject reduction property of
{\lA} with three lemmas above.

\begin{theorem}[Subject reduction]
    \ilabel{subj-red-theorem}{subject reduction}
Suppose that \(\typ{\G}{M:A}\) is derivable, and that
\(M \ct L\).
Then, \(\typ{\G}{L:A}\) is also derivable.
\end{theorem}
\begin{proof}
By induction on the derivation of \(\typ{\G}{M:A}\), and
by cases on the last rule applied in it.

\Case{\r{\mbox{var}}.}
This case is impossible since \(M \ct L\).

\Case{\r{\t}.}
In this case, \(A = \t\).
Hence, \(\typ{\G}{L:A}\) is also derivable by \r{\t}.

\Case{\r{\mbox{shift}} or \r{\rsubtyp}.}
The derivation ends with either of the following (for some \(A'\)
in case of \r{\rsubtyp}).
\[
    \ifr{\mbox{shift}}
        {\typ{\O \G}{M : \O A}}
        {\typ{\G}{M : A}}
    \mskip50mu
    \ifr{\rsubtyp}
	{\typ{\G}{M : A'} & A' \psubtyp A}
	{\typ{\G}{M: A}}
    \]
Hence, straightforward
from the induction hypothesis by applying the same rule.

\Case{\r{\Impl\mbox{I}}.}
The derivation ends with
\[
    \ifr{\Impl\mbox{I}}
	{\typ{\G\cup\{\,x:B\,\}}{M':C}}
	{\typ{\G}{\lam{x}M':B\Impl C}}
\]
for some \(x\), \(M'\), \(B\) and \(C\) such that
\(M = \lam{x}{M'}\) and \(A = B \Impl C\).
Since \(\lam{x}M' \ct L\), there exists some \(L'\) such that
\(M' \ct L'\) and \(L = \lam{x}{L'}\).
Hence, by induction hypothesis,
\(\typ{\G\cup\{\,x:B\,\}}{L':C}\) is derivable; and therefore,
so is \(\typ{\G}{\lam{x}L':B \Impl C}\) by applying \r{\Impl\mbox{I}}.

\Case{\r{\Impl\mbox{E}}.}
The derivation ends with
\[
    \ifr{\Impl\mbox{E}}
	{\typ{\G_1}{M_1:B \Impl A}
	& \typ{\G_2}{M_2:B}}
	{\typ{\G_1 \cup \G_2}{\app{M_1}{M_2}:A}}
\]
for some \(M_1\), \(M_2\), \(\G_1\), \(\G_2\) and \(B\) such that
\(M = \app{M_1}{M_2}\) and \(\G = \G_1 \cup \G_2\).
If \(A \peqtyp \t\), then \(\typ{\G}{L:A}\) is derivable by \r{\t}.
Therefore, we assume that \(A \npeqtyp \t\) in the sequel.
Note that it also implies that \(B \Impl A \npeqtyp \t\).
Since \(\app{M_1}{M_2} \ct L\),
there are three possible subcases as follows:
\begin{Enumerate}
\item[(a)] \(M_1\ct L_1\) and
    \(L = \app{L_1}{M_2}\) for some \(L_1\).
\item[(b)] \(M_2\ct L_2\) and
    \(L = \app{M_1}{L_2}\) for some \(L_2\).
\item[(c)] \(M_1 = \lam{x}{L'}\) and
    \(L = L'[M_2/x]\) for some \(x\) and \(L'\).
\end{Enumerate}
By induction hypothesis,
\(\typ{\G_1}{L_1:B \Impl A}\) and \(\typ{\G_2}{L_2:B}\)
are derivable for (a) and (b), respectively.
Therefore, we can derive
\(\typ{\G}{L:A}\) by applying \r{\Impl\mbox{E}} in these cases.
As for (c), we can assume that \(x \not\in \Dom(\G_1)\)
without loss of generality.
Hence, by Lemma~\ref{lA-abst-lemma},
there exist some \(n\), \(C\) and \(D\) such that
\begin{eqnarray}
\label{subj-red-12}
    && C \Impl D \psubtyp \O^n (B \Impl A),~\mbox{and} \\
\label{subj-red-13}
    && \typ{\O^n \G_1\cup\{\,x:C\,\}}{L':D}~\mbox{is derivable.}
\end{eqnarray}
Therefore, by (\ref{subj-red-12}) and Lemma~\ref{subtyp-O-Impl-lemma},
\begin{eqnarray}
\label{subj-red-15}
    && \O^n B \psubtyp \O^k C~\mbox{and}~\O^k D \psubtyp \O^n A
\end{eqnarray}
for some \(k\).
We can derive \(\typ{\O^{n+k} \G_1\cup\{\,x:\O^k C\,\}}{L':\O^k D}\)
from (\ref{subj-red-13})
by applying Proposition~\ref{lA-nec-redundant}, namely \r{\mbox{nec}},
\(k\) times; and hence,
\(\typ{\O^n\G_1\cup\{\,x:\O^k C\,\}}{L':\O^k D}\) is also derivable
by Lemma~\ref{typ-subtyp} since \(\O^n \G_1 \psubtyp \O^{n+k}\G_1\).
On the other hand,
\(\typ{\O^n \G_2}{M_2:\O^n B}\) is derivable
from \(\typ{\G_2}{M_2:B}\) by \r{\mbox{nec}}; and hence, so is
\(\typ{\O^n \G_2}{M_2:\O^k C}\) by \r{\rsubtyp} and (\ref{subj-red-15}).
Therefore, \(\typ{\O^n \G}{L'[M_2/x] : \O^k D}\) is derivable
by Proposition~\ref{lA-subst-redundant}, namely \r{\mbox{subst}};
and hence, so is
\(\typ{\O^n \G}{L'[M_2/x] : \O^n A}\) by \r{\rsubtyp} and (\ref{subj-red-15}),
from which \(\typ{\G}{L'[M_2/x] : A}\) is derivable
by applying \r{\mbox{shift}} \(n\) times.
\qed\CHECKED{2014/06/24, 07/20}
\end{proof}

The subject reduction property is independent of the existence
of the \r{\mbox{shift}}-rule.
Even if \r{\mbox{shift}} is dropped from {\lA},
Proposition~\ref{lA-nec-subst-redundant} still holds, and
the statements and the proofs of
Lemmas~\ref{subtyp-O-Impl-lemma}, \ref{lA-abst-lemma}
and Theorem~\ref{subj-red-theorem} remain valid
by regarding \(n\) as 0.

\Section{Convergence of well-typed $\lambda$-terms}\label{conv-sec}

When we gave the interpretation of type expressions
in Definition~\ref{rlz-def},
the only restriction forced on what type variables mean was that
the interpretation must be hereditary with respect to
the accessibility relation.
Therefore, just as in the case of the simply typed lambda calculus,
we are free to incorporate convergence properties of inhabitants into
the meaning of type expressions even though recursive types are involved.
Thus, type expressions can still say something about convergence, and
the soundness theorem assures the convergence of well-typed
\(\lambda\)-terms according to their types.
In this section, we give such results as follows.

\begin{Enumerate}
\item Every \(\lambda\)-term of types other than \(\t\) is head normalizable,
	i.e., solvable.
\item Every \(\lambda\)-term of types without positive occurrences of \(\t\)
    has a B\"{o}hm tree without unsolvable terms.
\item Every \(\lambda\)-term of types without occurrences of
    the modal operator \(\O\) is normalizable.
\end{Enumerate}

\Subsection{A term model}

The results about the convergence of well-typed \(\lambda\)-terms
are proved by applying the soundness of {\lA} (Theorem~\ref{soundness-theorem})
to an appropriate interpretation of type expressions,
considering a certain term model of the untyped \(\lambda\)-calculus.
So we first introduce the term model used for that purpose.

\begin{proposition}\label{beq-subst}
If\/ \(M \beq M'\) and \(N \beq N'\), then \(M[N/x] \beq M'[N'/x]\).
\end{proposition}
\begin{proof}
\ifdetail
It suffices to show that
(a) \(N \beq N'\) implies \(M[N/x] \beq M[N'/x]\), and
(b) \(M \ct M'\) implies \(M[N/x] \beq M'[N/x]\).
For Item~(a), suppose that \(N \beq N'\).
We show that \(M[N/x] \beq M[N'/x]\)
by induction on the structure of \(M\).

\Case{\(M = y\) for some individual variable \(y\).}
If \(y = x\), then
\(M[N/x] = N \beq N' = M[N'/x]\) by assumption.
Otherwise, \(M[N/x] = M[N'/x] = y\).

\Case{\(M = \app{M_1}{M_2}\) for some \(M_1\) and \(M_2\).}
By induction hypothesis,
\(M_1[N/x] \beq M_1[N'/x]\) and \(M_2[N/x] \beq M_2[N'/x]\).
Hence, \(M[N/x] \beq M[N'/x]\).

\Case{\(M = \lam{y}{K}\) for some \(y\) and \(K\).}
We can assume that \(y \not\in \FV{N} \cup \FV{N'}\) without loss of generality.
Hence, \(M[N/X] = \lam{y}(K[N/x]) \beq \lam{y}(K[N'/x]) = M[N'/x]\),
since \(K[N/x] \beq K[N'/x]\) by induction hypothesis.
This completes the proof of Item~(a).

As for Item~(b), suppose that \(M \ct M'\).
We show that \(M[N/x] \beq M'[N/x]\)
by induction on the structure of \(M\) again.

\Case{\(M = y\) for some individual variable \(y\).}
This case contradicts the assumption that \(M \ct M'\).

\Case{\(M = \app{M_1}{M_2}\) for some \(M_1\) and \(M_2\).}
Since \(M \ct M'\), there are the following three subcases.
\begin{enumerate}[{\kern4pt}(\romannumeral 1)]
\item
    \(M_1\ct M_1'\) and \(M' = \app{M_1'}{M_2}\) for some \(M_1'\).
\item
    \(M_2\ct M_2'\) and \(M' = \app{M_1}{M_2'}\) for some \(M_2'\).
\item
    \(M_1 = \lam{y}K\) and \(M' = K[M_2/y]\) for some \(y\) and \(K\).
\end{enumerate}
If ({\romannumeral 1}) is the case, then
\(M_1[N/x] \beq M_1'[N/x]\) by induction hypothesis; and hence,
\(M[N/x] \beq M'[N/x]\).
The case ({\romannumeral 2}) is similar.
As for case ({\romannumeral 3}),
we can assume that \(y \not\in \FV{N} \cup \{\,x\,\}\)
without loss of generality.
Hence,
\(M[N/x] = \app{(\lam{y}K[N/x])}{(M_2[N/x])}
    \ct K[N/x][M_2[N/x]/y] = K[M_2/y][N/x] = M'[N/x]\).

\Case{\(M = \lam{y}{K}\) for some \(y\) and \(K\).}
In this case, \(K \ct K'\) and \(M' = \lam{y}K'\) for some \(K'\), and
\(K[N/x] \beq K'[N/x]\) by induction hypothesis.
We can also assume that \(y \not\in \FV{N}\).
Hence, \(M[N/x] = \lam{y}K[N/x] \beq \lam{y}K'[N/x] = M'[N/x]\).
\else
Show that
\(N \beq N'\) implies \(M[N/x] \beq M[N'/x]\), and
\(M \ct M'\) implies \(M[N/x] \beq M'[N/x]\)
by induction on the structure of \(M\).
\fi 
\qed\CHECKED{2014/07/10}
\end{proof}

\begin{definition}
    \ilabel{term-model-def}{term model}
The term model
\(\tuple{\V,\,\cdot,\,\VI{\;}{}}\) of \(\IE\) is defined as follows.
\begin{Enumerate}
\item \(\V = \IE/{\beq}\)
\item \([M] \cdot [N] = [\app{M}{N}]\)
\item \(\VI{M}{\ienv} =
	[M[\Ienv(z_1)/z_1,\Ienv(z_2)/z_2,\ldots,\Ienv(z_k)/z_k]]\)
\end{Enumerate}
\end{definition}
where \([M]\) and \(\Ienv(z)\) denote
the equivalence class including \(M\) and
the representative member of \(\ienv(z)\), respectively,
and \(z_1\), \(z_2\), \(\ldots\), \(z_k\) are distinct individual variables
such that \(\FV{M} = \{\,z_1,\,z_2,\,\ldots,z_k\,\}\).
Note that Proposition~\ref{beq-subst} assures the well-definedness of Item~3.

\begin{proposition}\label{term-model-lambda-algebra}
The term model \(\tuple{\V,\,\cdot,\,\VI{\;}{}}\)
is a syntactical \(\lambda\)-algebra.
\end{proposition}
\begin{proof}
It suffices to show that
the eight conditions of Definition~\ref{lambda-algebra} are
satisfied by the term model.
Conditions~1 through 5 are straightforward.
As for Condition~6, let \(x \not\in \FV{\Ienv(z_i)} \cup \{\,z_i\,\}\)
for every \(i\)
(\(i = 1,\,2,\,\ldots,\,k\)).
\begin{Eqnarray*}
\VI{\lam{x}M}{\ienv}\cdot [N]
&=& [(\lam{x}{M})[\Ienv(z_1)/z_1,\Ienv(z_2)/z_2,\ldots,\Ienv(z_k)/z_k]]
    \cdot [N]
    & (by Definition~\ref{term-model-def}) \\
&=& [\lam{x}{(M[\Ienv(z_1)/z_1,\Ienv(z_2)/z_2,\ldots,\Ienv(z_k)/z_k])}]
    \cdot [N] \mskip-300mu \\
    &&& \hskip-92pt
	(since \(x \not\in \FV{\Ienv(z_i)} \cup \{\,z_i\,\}\)
	for every \(i\)) \\
&=& [\app{(\lam{x}{(M[\Ienv(z_1)/z_1,\Ienv(z_2)/z_2,\ldots,\Ienv(z_k)/z_k])})}
	 {N}]
     & (by Definition~\ref{term-model-def}) \\
&=& [M[\Ienv(z_1)/z_1,\Ienv(z_2)/z_2,\ldots,\Ienv(z_k)/z_k][N/x]]
    \mskip-300mu \\
&=& [M[\Ienv(z_1)/z_1,\Ienv(z_2)/z_2,\ldots,\Ienv(z_k)/z_k,N/x]]
    \mskip-300mu \\
    &&& \hskip-92pt
	(since \(x \not\in \FV{\Ienv(z_i)} \cup \{\,z_i\,\}\)
	for every \(i\)) \\
&=& \VI{M}{\ienv[[N]/x]} & (by Definition~\ref{term-model-def})
\end{Eqnarray*}
Conditions~7 and 8 are also straightforward
from Definition~\ref{term-model-def} and Proposition~\ref{beq-subst}.
\qed\CHECKED{2014/07/10}
\end{proof}

\Subsection{Tail finite types}

In this subsection,
we show that every \(\lambda\)-term of
a certain class of type expressions, which will later be defined
as {\em tail finite} types,
is {\em head normalizable}.
The set of head normalizable terms is defined in the standard manner as follows.

\begin{definition}[Head normal forms]
    \ilabel{hnf-def}{head normal forms}
A \(\lambda\)-term \(M\) is a {\em head normal form} if and only if\/
\(M\) has the form of
\(\lam{x_1}\lam{x_2}\ldots\lam{x_m}\app{y\,}{N_1\,N_2\,\ldots\,N_n}\),
where \(m,\,n \ge 0\).
We say that {\em \(M\) has a head normal form}, or
{\em is head normalizable},
if \(M \ctc M'\) for some head normal form \(M'\).
\end{definition}

We also define B\"{o}hm trees of \(\lambda\)-terms
in the standard manner according to this definition of
head normal forms, in which
\(\lambda\)-terms without head normal forms are denoted by \(\bot\).

\begin{definition}[Tail finite types]
    \ilabel{tf-def}{tail finite}
A type expression \(A\) is {\em tail finite} if and only if
\(A \eqtyp \O^{m_0}(B_1 \Impl \O^{m_1}(B_2 \Impl \O^{m_2}(B_3 \Impl
\ldots \Impl \O^{m_{n{-}1}}(B_n \Impl \O^{m_n} X)\ldots)))\)
for some \(n\), \(m_0\), \(m_1\), \(m_2\), \(\ldots\), \(m_n\),
\(B_1\), \(B_2\), \(\ldots\), \(B_n\) and \(X\).
\end{definition}

While we used \(\eqtyp\) to define the tail finiteness,
it is also possible to use \(\peqtyp\) instead of \(\eqtyp\).
Later, by Proposition~\ref{tf-tfte-not-top-equiv}, we shall see that
either of the two definitions leads us to the identical notion
of tail finiteness.
Actually, it will be shown that
a type expression \(A\) is tail finite if and only if
\(A\) is not a {\tvariant}.
To this end, we need an alternative definition of tail finiteness as follows.

\begin{definition}
    \ilabel{tfte-def}{TF@$\protect\TFTE$}
    \ilabel{tfte-syntax-def}{syntax!TF@$\protect\TFTE$}
Let \(V\) be a set of type variables.
We define a subset \(\TFTE^V\) of \(\TE\) as follows.
\begin{Eqnarray*}
    \TFTE^V  &\bnfdef& X & (\(X \in \TV - V\)) \\
	&\bnfor& \O  \TFTE^V \\
	&\bnfor& \TE \Impl \TFTE^V \\
	&\bnfor& \fix{X}A
		& \big(\(\fix{X}A \in \TE\) and
			\(A \in \TFTE^{V\cup\{\,X\,\}}\)\big).
\end{Eqnarray*}
\end{definition}
We can easily check that \(\TFTE^V\) is closed under
\(\alpha\)-conversion of type expressions.
We denote \(\TFTE^{\{\}}\) by \(\TFTE\), which will be shown to be
identical to the set of tail finite type expressions.
Roughly, \(\TFTE^V\) is the set of type expressions that are tail finite
even when some of type variables listed in \(V\) are instantiated
by \(\t\), which will also be shown in Proposition~\ref{tf-tfte-not-top-equiv}.

\begin{proposition}\label{tfte-tf}
Let \(V\) be a set of type variables, and suppose that \(A \in \TFTE^V\).
Then
\(A \geqtyp \O^{m_0}(B_1 \Impl \O^{m_1}(B_2 \Impl \O^{m_2}(B_3 \Impl
    \ldots \Impl \O^{m_{n{-}1}}(B_n \Impl \O^{m_n} X)\ldots)))\)
for some \(n\), \(m_0\), \(m_1\), \(m_2\), \(\ldots\), \(m_n\),
\(B_1\), \(B_2\), \(B_3\), \(\ldots\), \(B_n\) and \(X \not\in V\).
\end{proposition}
\begin{proof}
By induction on \(h(A)\), and by cases on the form of \(A\).

\Case{\(A = Y\) for some \(Y\).}
Obvious since we get \(Y \not\in V\) by Definition~\ref{tfte-def}.

\Case{\(A = \O C\) for some \(C\).}
We get \(C \in \TFTE^V\) by Definition~\ref{tfte-def}.
Hence, by induction hypothesis,
\(C \geqtyp \O^{m_0'}(B_1 \Impl \O^{m_1}(B_2 \Impl \O^{m_2}(B_3 \Impl
\ldots \Impl \O^{m_{n{-}1}}(B_n \Impl \O^{m_n} X)\ldots)))\)
for some \(n\), \(m_0'\), \(m_1\), \(m_2\), \(\ldots\), \(m_n\),
\(B_1\), \(B_2\), \(\ldots\), \(B_n\) and \(X \not\in V\).
Therefore, it suffices to take \(m_0 = m_0' + 1\).

\Case{\(A = C \Impl D\) for some \(C\) and \(D\).}
We get \(D \in \TFTE^V\) by Definition~\ref{tfte-def}.
By induction hypothesis,
\(D \geqtyp \O^{m_1}(B_2 \Impl \O^{m_2}(B_3 \Impl \O^{m_3}(B_4 \Impl
\ldots \Impl \O^{m_{n{-}1}}(B_n \Impl \O^{m_n} X)\ldots)))\)
for some \(n\), \(m_1\), \(m_2\), \(\ldots\), \(m_n\),
\(B_2\), \(B_3\), \(\ldots\), \(B_n\) and \(X \not\in V\).
Therefore, it suffices to take \(m_0 = 0\) and \(B_1 = C\).

\Case{\(A = \fix{Z}{C}\) for some \(Z\) and \(C\).}
We get \(C \in \TFTE^{V\cup\{\,Z\,\}}\) by Definition~\ref{tfte-def}.
Hence, by induction hypothesis,
\(C \geqtyp \O^{m_0}(B_1 \Impl \O^{m_1}(B_2 \Impl \O^{m_2}(B_3 \Impl
\ldots \Impl \O^{m_{n{-}1}}(B_n \Impl \O^{m_n} X)\ldots)))\)
for some \(n\), \(m_0\), \(m_1\), \(m_2\), \(\ldots\), \(m_n\),
\(B_1\), \(B_2\), \(\ldots\), \(B_n\) and \(X \not\in V \cup \{\,Z\,\}\).
Therefore, \(A = \fix{Z}{C} \geqtyp C[A/Z] \geqtyp
(\O^{m_0}(B_1 \Impl \O^{m_1}(B_2 \Impl \O^{m_2}(B_3 \Impl
\ldots \Impl \O^{m_{n{-}1}}(B_n \Impl \O^{m_n} X)\ldots))))[A/Z]
= \O^{m_0}(B_1[A/Z] \Impl \O^{m_1}(B_2[A/Z] \Impl \O^{m_2}(B_3[A/Z] \Impl
\ldots \Impl \O^{m_{n{-}1}}(B_n[A/Z] \Impl \O^{m_n} X)\ldots)))\)
by Proposition~\ref{geqtyp-subst}.
\qed\CHECKED{2014/05/14, 07/19}
\end{proof}

It is also not difficult to see that \(\TFTE\) enjoys the following properties.

\begin{proposition}\pushlabel
\begin{Enumerate}
\item \label{tfte-var-subseteq}
If\/ \(V \subseteq V'\), then \(\TFTE^{V'} \subseteq \TFTE^V\!\).
\item \label{tfte-var-petv}
If\/ \(X \not\in \PETV{A}\) and \(A \in \TFTE^V\!\), then
	    \(A \in \TFTE^{V\cup\{\,X\,\}}\!\).
\end{Enumerate}
\end{proposition}
\begin{proof}
By induction on \(h(A)\), and by cases on the form of \(A\).
For Item~2, note that
\(A \in \TFTE^V \) implies \(\ETV{A} \not= \{\}\)
by Propositions~\ref{tfte-tf} and \ref{geqtyp-etv}; and hence,
\(A\) is not a {\tvariant},
since \(\ETV{B} = \{\}\) for every {\tvariant} \(B\)
by Proposition~\ref{tvariant-etv}.
\qed\CHECKED{2014/05/14, 07/19}
\end{proof}

\begin{proposition}\pushlabel
Let \(A\) be a type expression, and \(V\) a set of type variables.
\begin{Enumerate}
\item \label{tfte-subst1}
    If\/ \(A \in \TFTE^{V\cup\{\,X\,\}}\!\), then \(A[B/X] \in \TFTE^V\!\).
\item \label{tfte-subst2}
    If\/ \(A \in \TFTE^V\!\) and \(B \in \TFTE^V\!\), then
	    \(A[B/X] \in \TFTE^V\!\).
\item \label{tfte-subst3}
    If\/ \(A[B/X] \in \TFTE^V\!\), then
	    \(A \in \TFTE^{V-\{\,X\,\}}\).
\item \label{tfte-subst4}
    If\/ \(A[B/X] \in \TFTE^V\!\), then
	    \(A \in \TFTE^{V\cup\{\,X\,\}}\) or \(B \in \TFTE^V\!\).
\end{Enumerate}
\end{proposition}
\begin{proof}
By induction on \(h(A)\), and by cases on the form of \(A\).
\ifdetail

\Case{\(A = Y\) for some \(Y\).}
\else
The only non-trivial case is when \(A = Y\) for some \(Y\).
\fi 
For Item~1, we get \(X \not= Y\) from \(A \in \TFTE^{V\cup\{\,X\,\}}\!\).
Therefore, \(A[B/X] = A \in \TFTE^{V\cup\{\,X\,\}} \subseteq \TFTE^V\) by
Proposition~\ref{tfte-var-subseteq}.
For Item~2, suppose that \(A \in \TFTE^V\) and \(B \in \TFTE^V\!\).
We similarly get \(A[B/X] = A \in \TFTE^V\) if \(X \not= Y\).
Otherwise, i.e., if \(X = Y\), then \(A[B/X] = B \in \TFTE^V\!\).
For Items~3 and 4, suppose that \(A[B/X] \in \TFTE^V\!\).
If \(X \not= Y\), then \(A = A[B/X] \in \TFTE^V\); therefore,
\(A \in \TFTE^{V-\{\,X\,\}}\) by Proposition~\ref{tfte-var-subseteq}, and
\(A \in \TFTE^{V\cup\{\,X\,\}}\) by Proposition~\ref{tfte-var-petv},
since \(X \not \in \FTV{A}\).
Otherwise, i.e., if \(X = Y\), then \(A = X \in \TFTE^{V-\{\,X\,\}}\)
by Definition~\ref{tfte-def}, and \(B = A[B/X] \in \TFTE^V\!\).
\ifdetail

\Case{\(A = \O C\) for some \(D\).}
For Item~1, suppose that \(A \in \TFTE^{V\cup\{\,X\,\}}\!\).
Then, \(C \in \TFTE^{V\cup\{\,X\,\}}\) by Definition~\ref{tfte-def}.
Therefore, \(C[B/X] \in \TFTE^V\) by induction hypothesis;
and hence, \(A[B/X] = \O C[B/X] \in \TFTE^V\) by Definition~\ref{tfte-def}
again.
Similar for Items~2, 3 and 4.

\Case{\(A = C \Impl D\) for some \(C\) and \(D\).}
For Item~1, suppose that \(A \in \TFTE^{V\cup\{\,X\,\}}\!\).
Then \(D \in \TFTE^{V\cup\{\,X\,\}}\) by Definition~\ref{tfte-def}.
Therefore, we get \(D[B/X] \in \TFTE^V\) by induction hypothesis;
and hence, \(A[B/X] = C[B/X] \Impl D[B/X] \in \TFTE^V\) by
Definition~\ref{tfte-def}.
Similar for Items~2, 3 and 4.

\Case{\(A = \fix{Y}C\) for some \(Y\) and \(C\).}
We can assume that \(Y \not\in \FTV{B} \cup \{\,X\,\}\).
For Item~1, we get \(C \in \TFTE^{V\cup\{\,X,\,Y\,\}}\)
from \(A \in \TFTE^{V\cup\{\,X\,\}}\) by Definition~\ref{tfte-def}.
Therefore, \(C[B/X] \in \TFTE^{V\cup\{\,Y\,\}}\) by induction hypothesis.
Hence, \(A[B/X] = (\fix{Y}C)[B/X] = \fix{Y}C[B/X] \in \TFTE^V\)
by Definition~\ref{tfte-def}.
Similar for Items~2, 3 and 4.
Use Proposition~\ref{tfte-var-petv} to get \(B \in \TFTE^{V\cup\{\,Y\,\}}\)
from \(B \in \TFTE^{V}\) in the proof of Item~2.
Use Proposition~\ref{tfte-var-subseteq} to get \(B \in \TFTE^{V}\)
from \(B \in \TFTE^{V\cup\{\,Y\,\}}\) in the proof of Item~4.
\fi 
\qed\CHECKED{2014/05/14, 07/20}
\end{proof}

Now we can establish the equivalence of the two definitions.
A type expression \(A\) is tail finite if and only if
\(A \in \TFTE\).
This is shown by proving the following proposition.

\begin{proposition}\label{tf-tfte-not-top-equiv}
Let \(V\) be a set of type variables.
The following three conditions are equivalent.
\begin{Enumerate}
\item[(a)] \(A \in \TFTE^V\)
\item[(b)] \(A \geqtyp \O^{m_0}(B_1 \Impl \O^{m_1}(B_2 \Impl \O^{m_2}(B_3 \Impl
\ldots \Impl \O^{m_{n{-}1}}(B_n \Impl \O^{m_n} X)\ldots)))\)
for some \(n\), \(m_0\), \(m_1\), \(m_2\), \(\ldots\), \(m_n\),
\(B_1\), \(B_2\), \(B_3\), \(\ldots\), \(B_n\) and \(X \not\in V\).
\item[(c)] \(A[\t/Y_1,\,\ldots,\,\t/Y_k] \ngeqtyp \t\), where
\(Y_1,\,\ldots,\,Y_k\) are distinct type variables such that
\(V = \{\,Y_1,\,\ldots,\,Y_k\,\}\).
\end{Enumerate}
Note that in case of\/ \(\peqtyp\), Condition~(b) is equivalent to
\begin{Enumerate}
\item[(b\({}'\))] \(A \peqtyp B_1 \Impl B_2 \Impl B_3 \Impl
\ldots \Impl B_n \Impl \O^{m} X\) for some \(n\), \(m\),
\(B_1\), \(B_2\), \(B_3\), \(\ldots\), \(B_n\) and \(X \not\in V\).
\end{Enumerate}
\end{proposition}
\begin{proof}
It suffices to show the following three propositions:
(a) \(\Rightarrow\) (b),
(b) \(\Rightarrow\) (c), and
(c) \(\Rightarrow\) (a).
The first one has been already shown as Proposition~\ref{tfte-tf}.
The rule \r{{\eqtyp}\mbox{-}{\Impl}{\t}} of Definition~\ref{eqtyp-def}
is crucial to prove the last part.
However, we can even show (b) \(\Rightarrow\) (a) independently of the rule.

\paragraph{Proof of ``(b) \(\Rightarrow\) (c)''}
Suppose that (b) holds, and let \(B_i' = B_i[\t/Y_1,\,\ldots,\,\t/Y_k]\)
(\(i = 1,\,2,\,3,\,\ldots,\,n\)).
Since \(X \not\in \{\,Y_1,\,\ldots,\,Y_k\,\}\),
we get \(A[\t/Y_1,\,\ldots,\,\t/Y_k] \geqtyp
\O^{m_0}(B_1' \Impl \O^{m_1}(B_2' \Impl \O^{m_2}(B_3' \Impl
\ldots \Impl \O^{m_{n{-}1}}(B_n' \Impl \O^{m_n} X)\ldots)))\)
by Proposition~\ref{geqtyp-subst}, which is not a {\tvariant}.
Therefore,
\(A[\t/Y_1,\,\ldots,\,\t/Y_k] \ngeqtyp \t\)
by Theorem~\ref{geqtyp-t-tvariant}.

\paragraph{Proof of ``(c) \(\Rightarrow\) (a)''}
We prove the contrapositive by induction on \(h(A)\),
and by cases on the form of \(A\).
Suppose that (a) is not the case, i.e., \(A \not\in \TFTE^V\).
Let \(Y_1,\,\ldots,\,Y_k\) be distinct type variables such that
\(V = \{\,Y_1,\,\ldots,\,Y_k\,\}\).
It suffices to show that \(A[\t/Y_1,\,\ldots,\,\t/Y_k] \geqtyp \t\).

\Case{\(A = Y\) for some \(Y\).}
We get \(Y \in V\) from \(A \not\in \TFTE^V\)
by Definition~\ref{tfte-def}.
Therefore, \(A[\t/Y_1,\,\ldots,\,\t/Y_k] = \t\).

\Case{\(A = \O C\) for some \(C\).}
By Definition~\ref{tfte-def}, \(C \not\in \TFTE^V\).
By induction hypothesis, \(C[\t/Y_1,\,\ldots,\,\t/Y_k] \geqtyp \t\).
Hence, \(A[\t/Y_1,\,\ldots,\,\t/Y_k] = \O C[\t/Y_1,\,\ldots,\,\t/Y_k]
\geqtyp \O \t \geqtyp \t\) by Proposition~\ref{geqtyp-t1}.

\Case{\(A = C \Impl D\) for some \(C\) and \(D\).}
Similarly, \(D \not\in \TFTE^V\) by Definition~\ref{tfte-def}.
Hence, by induction hypothesis, \(D[\t/Y_1,\,\ldots,\,\t/Y_k] \geqtyp \t\).
Therefore, \(A[\t/Y_1,\,\ldots,\,\t/Y_k]
= C[\t/Y_1,\,\ldots,\,\t/Y_k] \Impl D[\t/Y_1,\,\ldots,\,\t/Y_k]
\geqtyp C[\t/Y_1,\,\ldots,\,\t/Y_k] \Impl \t \geqtyp \t\)
by \r{{\eqtyp}\mbox{-}{\Impl}{\t}} of Definition~\ref{eqtyp-def}.

\Case{\(A = \fix{Z}{C}\) for some \(Z\) and \(C\).}
By Definition~\ref{tfte-def}, \(C \not\in \TFTE^{V\cup\{\,Z\,\}}\).
By induction hypothesis, \(C[\t/Z,\,\t/Y_1,\,\ldots,\,\t/Y_k] \geqtyp \t\).
Therefore, \(A[\t/Y_1,\,\ldots,\,\t/Y_k]
    = \fix{Z}{C[\t/Y_1,\,\ldots,\,\t/Y_k]} \geqtyp \t\)
by \r{{\eqtyp}\mbox{-uniq}} of Definition~\ref{eqtyp-def}.
\qed\CHECKED{2014/05/14, 07/20}
\end{proof}

\begin{theorem}\label{tf-tfte-not-top-theorem}
The following four conditions are equivalent.
\begin{Enumerate}
\item[(a)] \(A \in \TFTE\)
\item[(b)] \(A\) is tail finite.
\item[(c)] \(A \ngeqtyp \t\).
\item[(d)] \(A\) is not a {\tvariant}.
\end{Enumerate}
\end{theorem}
\begin{proof}
Obvious from Theorem~\ref{geqtyp-t-tvariant} and
Proposition~\ref{tf-tfte-not-top-equiv} considering the case that \(V = \{\}\).
This theorem fails if we adopt F-semantics of types
by dropping \r{{\eqtyp}\mbox{-}{\Impl}{\t}} from Definition~\ref{peqtyp-def}.
However, (a) \(\Leftrightarrow\) (b) remains to hold even in such a case.
\qed\CHECKED{2014/05/14, 07/20}
\end{proof}

Finally, we show that every \(\lambda\)-term of tail finite types
is head normalizable, by applying the soundness of {\lA}
with respect to the term model.
The proof is quite parallel to the standard method
using convertibility predicates due to Tait\cite{tait}.

\begin{definition}
    \ilabel{term-model-kernel-def}{Kh@$\protect\Kh$}
We define a subset \(\Kh\) of\/ \(\IE/{\beq}\,\) as follows.
\[
\Kh = \Zfset{[\,\app{x}{N_1 N_2 \ldots N_n}]}
		{\tabcolsep=0pt
		    \ifnarrow
			\begin{tabular}{p{180pt}}%
		    \else
			\begin{tabular}{l}%
		    \fi
		    \(x \in \IV\), \(n \ge 0\) and
		    \(N_i \in \IE\) for every \(i\)
		    (\(i = 1,\,2,\,\ldots,\,n\))\end{tabular}}.
\]
\end{definition}

\begin{lemma}\label{term-model-inc-kernel-lemma}
Consider the term model of\/ \(\IE\),
a well-founded frame \(\pair{\W}{\acc}\) and
a type environment \(\tenv\) such that
\(\Kh \subseteq \tenv(X)_p\) for every type variable \(X\) and \(p \in \W\).
Then,
\(\Kh \subseteq \I{A}^\tenv_p\,\) for every \(A\) and \(p \in \W\).
\end{lemma}
\begin{proof}
By induction on the lexicographic ordering of \(\pair{p}{r(A)}\),
and by cases on the form of \(A\).
We assume that \(A\) is not a {\tvariant},
since \(\Kh \subseteq \I{A}^\tenv_p = \V\) if \(A\) is.
For this proof, \(\tenv\) need not to be hereditary.

\Case{\(A = Y\) for some \(Y\).}
By Definition~\ref{rlz-def}, \(\I{A}^\tenv_p = \tenv(Y)_p \supseteq \Kh\).

\Case{\(A = \O B\) for some \(B\).}
In this case,
\(\I{A}^\tenv_p =
    \zfset{u}{u \in \I{B}^\tenv_q~~\mbox{for every}~q \opacc p}
    \supseteq \Kh\)
since \(\Kh \subseteq \I{B}^\tenv_q~~\mbox{for every}~q \opacc p\)
by induction hypothesis.

\Case{\(A = B \Impl C\) for some \(B\) and \(C\).}
By Definition~\ref{rlz-def},
\(\I{A}^\tenv_p =
        \zfset{u}{u \cdot v \in \I{C}^\tenv_q~~\mbox{for every}~
	    v \in \I{B}^\tenv_q~~\mbox{whenever}~p \tacc q}\).
Suppose that \([\,\app{x}{N_1 N_2 \ldots N_n}] \in \Kh\).
Then,
for every \(L \in \IE\),
\([\,\app{x}{N_1 N_2 \ldots N_n}] \cdot [L]
= [\,\app{x}{N_1 N_2 \ldots N_n L}] \in \Kh\).
Hence, for every \(q\) such that \(p \tacc q\),
we get \([\,\app{x}{N_1 N_2 \ldots N_n L}] \in \I{C}^\tenv_q\)
since \(\Kh \subseteq \I{C}^\tenv_q\strut\) by induction hypothesis.
Therefore,
\([\,\app{x}{N_1 N_2 \ldots N_n}] \in \I{B \Impl C}^\tenv_p\).

\Case{\(A = \fix{Y}{B}\) for some \(Y\) and \(B\).}
By Definition~\ref{rlz-def},
\(\I{A}^\tenv_p = \I{B[A/Y]}^\tenv_p\).
On the other hand,
since \(r(B[A/Y]) < r(A)\) by Proposition~\ref{rank-fix},
we get
\(\Kh \subseteq \I{B[A/Y]}^\tenv_p\) by induction hypothesis.
\qed\CHECKED{2014/05/14, 07/20}
\end{proof}

Note that since every {\lA}-frame is also a well-founded frame,
Lemma~\ref{term-model-inc-kernel-lemma} also holds in the case that
\(\pair{\W}{\acc}\) is a {\lA}-frame.
Now we show the main results of this subsection, which is stated
as the following theorem.

\begin{theorem}\label{typ-hnf-theorem}
Let \(V\) be a set of type variables, and
\(A\) a type expression such that \(A \in \TFTE^V\!\).
If\/ \(\typ{\G}{M:A}\) is derivable in {\lA},
then \(M\) has a head normal form.
\end{theorem}
\begin{proof}
Consider the term model and a {\lA}-frame \(\pair{\N}{>}\)
where \(\N\) is the set of natural numbers and
\({>}\) is the greater-than relation between natural numbers.
Note that \(\pair{\N}{>}\) is also a well-founded frame.
Let \(\ienv\) be an individual environment such that
\(\ienv(x) = [x]\) for every \(x \in \IV\).
Then, \(\ienv(x) \in \Kh \subseteq \I{\G(x)}^\tenv_p\)
for every \(x \in \Dom(\G)\), \(\tenv\) and \(p \in \N\)
by Lemma~\ref{term-model-inc-kernel-lemma}.
Hence, \([M] \in \I{A}^\tenv_p\) for every \(\tenv\) and \(p\strut\)
by Theorem~\ref{soundness-theorem} because \(\VI{M}{\rho} = [M]\).
Since \(\tenv\) can be any hereditary type environment,
it suffices to show that \(M\) has a head normal form whenever
\begin{Enumerate}
\item[(a)] \(A \in \TFTE^V\!\),
\item[(b)] \(\Kh \subseteq \tenv(X)_p\) for every \(p \in \N\) and \(X\),
\item[(c)] \(\tenv(X)_p = \Kh\) for every \(p \in \N\) and \(X \not\in V\), and
\item[(d)] \([M] \in \I{A}^\tenv_p\) for every \(p \in \N\).
\end{Enumerate}
The proof proceeds by induction on \(h(A)\),
and by cases on the form of \(A\).
Suppose (a) through (d).

\Case{\(A = X\) for some \(X\).}
In this case, \(X \not\in V\) from (a); and hence,
\([M] \in \I{A}^\tenv_p = \tenv(X)_p = \Kh\) by (c) and (d).
Therefore, \(M\) has a head normal form.

\Case{\(A = \O B\) for some \(B\).}
In this case, \(h(B) < h(A)\) and \(B \in \TFTE^V\) by (a).
Therefore, \(M\) has a head normal form by induction hypothesis,
since \([M] \in \I{\O B}^\tenv_{p{+}1} = \I{B}^\tenv_p\) for every \(p\).

\Case{\(A = B \Impl C\) for some \(B\) and \(C\).}
In this case, \(h(C) < h(A)\) and \(C \in \TFTE^V\) by (a).
Let \(y\) be a fresh individual variable.
Since \([M] \in \I{B \Impl C}^\tenv_p\) and
\([y] \in \Kh \subseteq \I{B}^\tenv_p\) for every \(p\)
by (b) and Lemma~\ref{term-model-inc-kernel-lemma},
we get
\([\app{M}{y}] \in \I{C}^\tenv_p\) for every \(p\).
Therefore, \(\app{M}{y}\) has a head normal form, say \(L\),
by induction hypothesis.
There are two possible cases: for some \(K\),
({\romannumeral 1}) \(M \ctc K\) and \(L = \app{K}{y}\), or
({\romannumeral 2}) \(M \ctc \lam{y}{K}\) and \(K \ctc L\).
In either case, \(M\) obviously has a head normal form.

\Case{\(A = \fix{Y}B\) for some \(Y\) and \(B\).}
In this case, \(h(B) < h(A)\) and \(B \in \TFTE^{V\cup\{\,Y\,\}}\) by (a).
By Definition~\ref{rlz-def} and Proposition~\ref{rlz-subst-env},
we get
\(\I{\fix{Y}B}^\tenv_p = \I{B[\fix{Y}B/Y]}^\tenv_p
= \I{B}^{\tenv'}_p\), where \(\tenv' = \tenv[\I{\fix{Y}B}^\tenv/Y]\).
Note that (\(\mbox{a}'\)) \(B \in \TFTE^{V\cup\{\,Y\,\}}\!\),
(\(\mbox{b}'\)) \(\Kh \subseteq \tenv'(X)_p\) for every \(p\) and \(X\)
since \(\Kh \subseteq \I{\fix{Y}B}^\tenv = \tenv'(Y)\)
by (b) and Lemma~\ref{term-model-inc-kernel-lemma},
(\(\mbox{c}'\)) \(\tenv'(X)_p = \Kh\) for every \(p\)
    and \(X \not\in V\cup\{\,Y\,\}\), and
(\(\mbox{d}'\)) \([M] \in \I{B}^{\tenv'}_p\) for every \(p\).
Therefore, \(M\) has a head normal form by induction hypothesis.
\qed\CHECKED{2014/05/14, 07/20, 07/22}
\end{proof}

\begin{corollary}
Every \(\lambda\)-term of tail finite types is head normalizable.
\end{corollary}
\begin{proof}
By Theorems~\ref{tf-tfte-not-top-theorem}
and \ref{typ-hnf-theorem} taking \(V\) as \(V =\{\}\)
in \ref{typ-hnf-theorem}.
\qed\CHECKED{2014/05/14, 07/20}
\end{proof}

\Subsection{Positively and negatively finite types}\label{typ-maximal-sec}

In the previous subsection, it has been shown that
every tail finite type, i.e., any type other than {\tvariant}s, is inhabited
only by head normalizable \(\lambda\)-terms.
It can be also shown that,
if the type does not involve any positive (effective) occurrence of
{\tvariant}s, then
the B\"{o}hm tree of any \(\lambda\)-term that belongs to the type
does not include any occurrence of \(\bot\).
For example, we can see that Curry's fixed-point combinator \(\Y\)
has such a B\"{o}hm tree since it has a type \((\O X \Impl X) \Impl X\).

\begin{definition}[Maximal \(\lambda\)-terms]
    \ilabel{max-term-def}{maximal $\lambda$-terms}
A \(\lambda\)-term \(M\) is {\em maximal} if and only if
the B\"{o}hm tree of \(M\) has no occurrence of \(\bot\), i.e.,
\(\lambda\)-terms that are not head normalizable.
\end{definition}
Note that the maximality of \(\lambda\)-terms is closed under \(\beq\).

\begin{definition}[Positively/negatively finite types]
    \ilabel{pnf-type-def}{positively finite}
    \ilabel*{negatively finite}
A type expression \(A\) is
{\em positively }(respectively, {\em negatively}) {\em finite}
if and only if\/ \(C\) is tail finite whenever
\(A \geqtyp B[C/X]\) for some \(B\) and \(X\) such that
\(X \in \PETV{B}\) and \(X \not\in \NETV{B}\)
(respectively, \(X \in \NETV{B}\) and \(X \not\in \PETV{B}\)).
\end{definition}

Although this definition might seem to be dependent on
the choice of \(\geqtyp\),
either of the two ways of definition gives the same notion,
and the following discussion in this subsection does not depend on
which it is.
That will be verified by showing that the equivalent notion can be defined
by the same grammar of type expressions
(cf. Definition~\ref{pfte-nfte-def} and Proposition~\ref{pfte-pf-equiv}).

Obviously, every positively finite type expression is tail finite, and
positively (negatively) finiteness is closed under \(\geqtyp\).
In this subsection,
we show that every \(\lambda\)-term of positively finite types is maximal.
To this end, we again employ alternative definitions of
positively and negatively finiteness as follows.

\begin{definition}
    \ilabel{pfte-nfte-def}{syntax!PF@$\protect\PFTE,\protect\NFTE$}
    \ilabel{pfte-def}{PF@$\protect\PFTE$}
    \ilabel{nfte-def}{NF@$\protect\NFTE$}
We define subsets \(\PFTE\) and \(\NFTE\) of\/ \(\TE\) as follows.
\begin{Eqnarray*}
\PFTE &\bnfdef& \TV \mskip-100mu \\
    &\bnfor& \O \PFTE \\
    &\bnfor& \NFTE \Impl \PFTE \\
    &\bnfor& \fix{X}{A} &
\ifnarrow
	\(\bigg(\begin{tabular}{p{160pt}}
	\(\fix{X}A \in \TFTE\), \(A \in \PFTE\), and
	(a) \(X \not\in \NETV{A}\) or (b) \(A \in \NFTE\).
	\end{tabular}\bigg)\)
\else
	(\(\fix{X}A \in \TFTE\), \(A \in \PFTE\), and
	(a) \(X \not\in \NETV{A}\) or (b) \(A \in \NFTE\))
\fi
	\\[4pt]
\NFTE &\bnfdef& \TV \mskip-100mu \\
    &\bnfor& \O \NFTE \\
    &\bnfor& \PFTE \Impl \NFTE \\
    &\bnfor& \fix{X}{A} &
\ifnarrow
	\(\bigg(\begin{tabular}{p{160pt}}
	\(\fix{X}{A} \in \TE\), \(A \in \NFTE\), and
	(a) \(X \not\in \NETV{A}\) or (b) \(\fix{X}A \in \PFTE\).
	\end{tabular}\bigg)\) \\
\else
	(\(\fix{X}{A} \in \TE\), \(A \in \NFTE\), and
	(a) \(X \not\in \NETV{A}\) or (b) \(\fix{X}A \in \PFTE\)) \\
\fi
    &\bnfor& A & \rm\,(\(A\) is a {\tvariant})
\end{Eqnarray*}
\end{definition}
We can easily check that \(\PFTE\) and \(\NFTE\) are closed under
\(\alpha\)-conversion of type expressions.
Obviously, every finite type expression, i.e.,
one that involves no recursive type, belongs to both \(\PFTE\) and \(\NFTE\).
The grammar of \(\NFTE\) is ambiguous since it says that
any {\tvariant} belongs to \(\NFTE\).
Hence, we should be careful with the fact that
\(A \Impl B \in \NFTE\) does not always imply \(A \in \PFTE\).

Our first task in this subsection is to verify that
\(A\) is positively (respectively, negatively) finite if and only if
\(A \in \PFTE\) (respectively, \(\NFTE\))
(cf. Proposition~\ref{pfte-pf-equiv}), which
also implies that it is decidable
whether a given type expression is positively (negatively) finite or not.

\begin{proposition}\label{pfte-tfte}
\(\PFTE \subseteq \TFTE\).
\end{proposition}
\begin{proof}
Prove that \(A \in \PFTE\) implies \(A \in \TFTE\),
by straightforward induction on \(h(A)\), and by cases on the form of \(A\).
\ifdetail
Suppose that \(A \in \PFTE\).

\Case{\(A = X\) for some \(X\).}
Trivial since \(X \in \TFTE\) by Definition~\ref{tfte-def}.

\Case{\(A = \O B\) for some \(B\).}
We get \(B \in \PFTE\) from \(A \in \PFTE\) by Definition~\ref{pfte-def}.
Therefore, \(B \in \TFTE\) by induction hypothesis; and hence,
\(\O B \in \TFTE\) by Definition~\ref{tfte-def}.

\Case{\(A = B \Impl C\) for some \(B\) and \(C\).}
Similarly, \(C \in \PFTE\) from \(A \in \PFTE\)
by Definition~\ref{pfte-def}.
Therefore, \(C \in \TFTE\) by induction hypothesis; and hence,
\(B \Impl C \in \TFTE\) by Definition~\ref{tfte-def}.

\Case{\(A = \fix{X}B\) for some \(X\) and \(B\).}
Trivial since \(\fix{X}B \in \PFTE\) implies \(\fix{X}B \in \TFTE\)
by Definition~\ref{pfte-def}.
\fi 
\qed\CHECKED{2014/05/15, 07/08, 07/17}
\end{proof}

The following proposition says that we can separate an effective
occurrence of an individual variable in a type expression from others so that
it becomes either positive or negative and not both.
For example, consider a type expression \(A = \fix{Y}\O Y \Impl X\),
in which the occurrence of \(X\) is both positive and negative.
We can obtain a positive and not negative occurrence of \(X\)
by unfolding \(A\) to \(\O\,(\fix{Y}\O Y \Impl X)\Impl X\), in which
so is the rightmost occurrence of \(X\).
Similarly, we can also obtain a negative and not positive occurrence of \(X\)
by unfolding \(A\) to \(\O\,(\O\,(\fix{Y}\O Y \Impl X) \Impl X) \Impl X\).

\begin{proposition}\label{etv-separate}
If\/ \(X \in \PNETV{A}\), then
there exist some \(A'\) and \(X'\) such that
\(A \geqtyp A'[X/X']\), \(X' \in \PNETV{A'}\) and \(X' \not\in \NPETV{A'}\).
\end{proposition}
\begin{proof}
By Propositions~\ref{geqtyp-canon} and \ref{geqtyp-etv},
we can assume that \(A\) is canonical without loss of generality.
\ifdetail
Suppose that \(X \in \PNETV{A}\), which
\else
Note that \(X \in \PNETV{A}\) 
\fi 
implies that \(A \not= \t\) and \(\PNIdp(A,\,X) < \infty\)
by Propositions~\ref{tvariant-etv} and \ref{depth-finite-etv}, respectively.
By induction on \(\PNIdp(A,\,X)\), and by cases on the form of \(A\).
\ifdetail

\Case{\(A = \O^n Y\) for some \(n\) and \(Y\).}
In this case, \(\NETV{A} = \{\}\).
On the other hand,
if \(X \in \PETV{A}\), then \(Y = X\);
and hence, it suffices to take \(A'\) and \(X'\) as \(A' = A\) and \(X' = X\).

\Case{\(A = \O^n (B \Impl C)\) for some \(n\), \(B\) and \(C\)
    such that \(C \ngeqtyp \t\), where \(n = 0\) in case of \(\peqtyp\).}
By Definition~\ref{depth-def} and Proposition~\ref{depth-finite-etv},
either
(a) \(\PNIdp(A) = \NPIdp(B,\,X)+1\) and \(X \in \NPETV{B}\), or
(b) \(\PNIdp(A) = \PNIdp(C,\,X)+1\) and \(X \in \PNETV{C}\).
In case (a), by induction hypothesis,
there exist \(B'\) and \(X'\) such that
\(B \geqtyp B'[X/X']\), \(X' \in \NPETV{B'}\)
    and \(X' \not\in \PNETV{B'}\), where
we can assume that \(X' \not\in \FTV{C}\) without loss of generality.
Therefore, taking \(A'\) as \(A' = \O^n(B' \Impl C)\),
we get
\(A = \O^n(B \Impl C) \geqtyp \O^n(B'[X/X'] \Impl C) = A'[X/X']\),
\(X' \in \NPETV{B'} \subseteq \NPETV{B'} \cup \PNETV{C} = \PNETV{A'}\)
and \(X' \not\in \PNETV{A'}\), since \(X' \not\in \NPETV{B'}\)
and \(X' \not\in \FTV{C}\).
The case (b) is similar.
\fi 
\qed\CHECKED{2014/07/15, 07/17}
\end{proof}

\begin{proposition}\label{pfte-nfte-subst-orig}
Suppose that \(B \ngeqtyp \t\).
If\/ \(A[B/X] \in \PFTE\) (respectively, \(\NFTE\)),
    then \(A \in \PFTE\) (respectively, \(\NFTE\)).
\end{proposition}
\begin{proof}
If \(A[B/X]\) is a {\tvariant}, then \(A[B/X] \not\in \PFTE\)
by Proposition~\ref{pfte-tfte} and Theorem~\ref{tf-tfte-not-top-theorem},
and \(A \in \NFTE\) since \(A\) is also a {\tvariant}
by Proposition~\ref{tvariant-subst2}.
Hence, it suffices to only consider the case that
\(A[B/X]\) is not a {\tvariant}.
By induction on \(h(A)\), and by cases on the form of \(A\).
\ifdetail

\Case{\(A = Y\) for some \(Y\).}
Trivial from Definition~\ref{pfte-nfte-def}.

\Case{\(A = \O C\) for some \(C\).}
In this case, \(A[B/X] = \O C[B/X]\).
Hence, \(C[B/X] \in \PFTE\) (respectively, \(\NFTE\)).
Therefore, \(C \in \PFTE\) (respectively, \(\NFTE\)) by induction hypothesis,
from which \(A \in \PFTE\) (respectively, \(\NFTE\)) follows.

\Case{\(A = C \Impl D\) for some \(C\) and \(D\).}
Since
\(A[B/X] = C[B/X] \Impl D[B/X] \in \PFTE\) (respectively, \(\NFTE\)), we get
\(C[B/X] \in \NFTE\) and \(D[B/X] \in \PFTE\)
(respectively, \(C[B/X] \in \PFTE\) and \(D[B/X] \in \NFTE\)).
Therefore, by induction hypothesis,
\(C \in \NFTE\) and \(D \in \PFTE\)
(respectively, \(C \in \PFTE\) and \(D \in \NFTE\)),
from which we get \(A \in \PFTE\) (respectively, \(\NFTE\)).

\Case{\(A = \fix{Y}C\) for some \(Y\) and \(C\).}
We 
\else
The only interesting case is when \(A = \fix{Y}C\) for some \(Y\) and \(C\).
In this case, we
\fi 
can assume that \(Y \not\in \FTV{B} \cup \{\,X\,\}\)
without loss of generality.
Since \(A[B/X] = \fix{Y}C[B/X] \in \PFTE\) (respectively, \(\NFTE\)),
\begin{eqnarray}
\label{pfte-nfte-subst-orig-01}
    \fix{Y}C[B/X] &\in& \TFTE~(\mbox{respectively,}~\TE) \\
\label{pfte-nfte-subst-orig-02}
    C[B/X] &\in& \PFTE~(\mbox{respectively,}~\NFTE),~\mbox{and} \\
\label{pfte-nfte-subst-orig-03}
    C[B/X] &\in& \NFTE~(\mbox{respectively,}~\fix{Y}C[B/X] \in \PFTE),
	~\mbox{if}~Y \in \NETV{C[B/X]}.
\end{eqnarray}
First, in case of \(A[B/X] \in \PFTE\),
we get \(\fix{Y}C \in \TFTE\) from (\ref{pfte-nfte-subst-orig-01})
by Theorem~\ref{tf-tfte-not-top-theorem} and Proposition~\ref{tvariant-subst1}.
Second, \(C \in \PFTE\) (respectively, \(\NFTE\))
from (\ref{pfte-nfte-subst-orig-02}) by induction hypothesis.
Finally, if \(Y \in \NETV{C}\), then
\(Y \in \NETV{C[B/X]}\)
by Proposition~\ref{petv-netv-subst}; and therefore,
\(C[B/X] \in \NFTE\) (respectively, \(\PFTE\))
from (\ref{pfte-nfte-subst-orig-03}) and Definition~\ref{pfte-nfte-def},
which implies \(C \in \NFTE\) (respectively, \(\PFTE\))
by induction hypothesis.
Thus, \(A \in \PFTE\) (respectively, \(\NFTE\)).
\qed\CHECKED{2014/07/12, 07/17}
\end{proof}

\begin{lemma}\label{pfte-nfte-subst-lemma}
Suppose that
\begin{Enumerate}
\item[(a)] \(A \in \TFTE^V\),
\item[(b)] \(V \cap \PETV{B} = \{\}\), and
\item[(c)] \(X \not\in \PETV{A}\) or \(B \in \TFTE\).
\end{Enumerate}
Then, \(A[B/X] \in \TFTE^V\).
\end{lemma}
\begin{proof}
If \(X \not\in \PETV{A}\), then \(A \in \TFTE^{V\cup\{\,X\,\}}\) by
(a) and Proposition~\ref{tfte-var-petv}; and therefore,
\(A[B/X] \in \TFTE^V\) by Proposition~\ref{tfte-subst1}.
On the other hand, if \(X \in \PETV{A}\), then \(B \in \TFTE\) by (c);
that is, \(B \in \TFTE^V\) by (b) and Proposition~\ref{tfte-var-petv}.
Hence, \(A[B/X] \in \TFTE^V\) by (a) and Proposition~\ref{tfte-subst2}.
\qed\CHECKED{2014/05/15, 07/17}
\end{proof}

\begin{proposition}\label{pfte-nfte-subst}
Suppose that
\begin{Enumerate}
\item[\it(a)] \(A \in \PFTE\) (respectively, \(\NFTE\)),
\item[\it(b)] if \(X \in \PETV{A}\)
    then \(B \in \PFTE\) (respectively, \(\NFTE\)), and
\item[\it(c)] if \(X \in \NETV{A}\)
    then \(B \in \NFTE\) (respectively, \(\PFTE\)).
\end{Enumerate}
Then \(A[B/X] \in \PFTE\) (respectively, \(\NFTE\)).
\end{proposition}
\begin{proof}
By induction on \(h(A)\), and by cases on the form of \(A\).
Use Proposition~\ref{tvariant-subst1} if \(A = C \Impl D\)
for some \(C\) and \(D\).
The only interesting case is
when \(A = \fix{Y}{C}\) for some \(Y\) and \(C\).
In this case, suppose that (a) through (c) hold.
We can assume that \(Y \not\in \FTV{B} \cup \{\,X\,\}\) without loss of
generality; that is, \(A[B/X] = \fix{Y}C[B/X]\) and
\(B\) is proper in \(Y\) by Proposition~\ref{etv-proper}.
By (a) and Definition~\ref{pfte-nfte-def}, we have
\begin{eqnarray}
\label{pfte-nfte-subst-01}
    && \fix{Y}C \in \TFTE~(\mbox{respectively},~\TE) \\
\label{pfte-nfte-subst-02}
    && C \in \PFTE~(\mbox{respectively},~\NFTE),~\mbox{and} \\
\label{pfte-nfte-subst-03}
    && Y \not\in \NETV{C} ~\mbox{or}~ C \in \NFTE~
	(\mbox{respectively},~\fix{Y}C \in \PFTE).
\end{eqnarray}
Moreover, since \(\PNETV{C} - \{\,Y\,\} \subseteq \PNETV{A}\), from (b) and (c),
\begin{Enumerate}
\item[\it(b\({}'\))]
    if \(X \in \PETV{C}\) then \(B \in \PFTE\) (respectively, \(\NFTE\)), and
\item[\it(c\({}'\))] if \(X \in \NETV{C}\) then \(B \in \NFTE\)
    (respectively, \(\PFTE\)).
\end{Enumerate}
By Definition~\ref{pfte-nfte-def}, it suffices to show that
\begin{eqnarray}
\label{pfte-nfte-subst-05}
    && \fix{Y}C[B/X] \in \TFTE~(\mbox{respectively},~\TE) \\
\label{pfte-nfte-subst-06}
    && C[B/X] \in \PFTE~(\mbox{respectively},~\NFTE),~\mbox{and} \\
\label{pfte-nfte-subst-07}
    && Y \not\in \NETV{C[B/X]} ~\mbox{or}~
	    C[B/X] \in \NFTE~(\mbox{respectively},~\fix{Y}C[B/X] \in \PFTE).
\end{eqnarray}
First, \(C[B/X]\) is proper in \(Y\)
by Proposition~\ref{proper-subst1},
since so is \(C\) by (\ref{pfte-nfte-subst-01}).
Furthermore,
in case of \(A \in \PFTE\), we also get
\(C \in \TFTE^{\{\,Y\,\}}\) from (\ref{pfte-nfte-subst-01}); and hence,
\(C[B/X] \in \TFTE^{\{\,Y\,\}}\)
by Lemma~\ref{pfte-nfte-subst-lemma}, Proposition~\ref{pfte-tfte}
and (\(\mbox{b}'\)).
Thus, we get (\ref{pfte-nfte-subst-05}).
Second, we get (\ref{pfte-nfte-subst-06})
from (\ref{pfte-nfte-subst-02}), (\(\mbox{b}'\)) and (\(\mbox{c}'\))
by induction hypothesis.
Finally, to show (\ref{pfte-nfte-subst-07}),
suppose that \(Y \in \NETV{C[B/X]}\).
Then, \(Y \in \NETV{C}\) by Proposition~\ref{petv-netv-subst2}
since \(Y \not\in \FTV{B}\); and therefore,
\(\PNETV{A} = \ETV{C} -\{\,Y\,\}\).
Hence, from (b) and (c), we get
\begin{eqnarray}
\label{pfte-nfte-subst-08}
    X \in \ETV{C} &~\mbox{implies}~& B \in \PFTE \cap \NFTE.
\end{eqnarray}
On the other hand, from \(Y \in \NETV{C}\) and (\ref{pfte-nfte-subst-03}),
we get \(C \in \NFTE\) (respectively, \(\fix{Y}C \in \PFTE\); and therefore,
\(C \in \PFTE\) by Definition~\ref{pfte-nfte-def}).
Hence, \(C[B/X] \in \NFTE\) (respectively, \(\PFTE\))
from (\ref{pfte-nfte-subst-08}) by induction hypothesis, where
in case of \(A \in \NFTE\),
\(C[B/X] \in \PFTE\) also implies \(\fix{Y}C[B/X] \in \PFTE\) because
\(C[B/X] \in \NFTE\) is already established as (\ref{pfte-nfte-subst-06}),
and because we can get \(\fix{Y}C[B/X] \in \TFTE\)
from \(\fix{Y}C \in \PFTE \subseteq \TFTE\)
by Lemma~\ref{pfte-nfte-subst-lemma},
Proposition~\ref{pfte-tfte} and
(\ref{pfte-nfte-subst-08}).
We thus establish (\ref{pfte-nfte-subst-07}).
\qed\CHECKED{2014/05/20, 07/17}
\end{proof}

\begin{proposition}\label{pfte-occur}
Suppose that \(B \ngeqtyp \t\).
\begin{Enumerate}
\item If\/ \(X \in \PETV{A}\)
	    and \(A[B/X] \in \PFTE\), then \(B \in \PFTE\).
\item If\/ \(X \in \NETV{A}\)
	    and \(A[B/X] \in \NFTE\), then \(B \in \PFTE\).
\end{Enumerate}
\end{proposition}
\begin{proof}
By simultaneous induction on \(h(A)\), and by cases on the form of \(A\).
Suppose that \(X \in \PETV{A}\) and \(A[B/X] \in \PFTE\)
    (respectively, \(X \in \NETV{A}\) and \(A[B/X] \in \NFTE\)).
Note that \(A \ngeqtyp \t\) since \(X \in \ETV{A}\); and hence,
\(A[B/X] \ngeqtyp \t\) from \(B \ngeqtyp \t\)
by Proposition~\ref{tvariant-subst2} and Theorem~\ref{geqtyp-t-tvariant}.

\Case{\(A = Y\) for some \(Y\).}
In this case, the assumption \(X \in \NETV{A}\) of Item~2 does not hold.
As for Item~1, \(Y = X\) since \(X \in \PETV{A}\).
Hence, \(B = A[B/X] \in \PFTE\).

\Case{\(A = \O C\) for some \(C\).}
In this case, \(A[B/X] = \O C[B/X]\).
Hence, \(X \in \PETV{C}\) and \(C[B/X] \in \PFTE\)
(respectively, \(X \in \NETV{C}\) and \(C[B/X] \in \NFTE\)).
Therefore, \(B \in \PFTE\) by induction hypothesis.

\Case{\(A = C \Impl D\) for some \(C\) and \(D\).}
Since
\(A[B/X] = C[B/X] \Impl D[B/X] \in \PFTE\) (respectively, \(\NFTE\)), we get
\(X \in \NETV{C} \cup \PETV{D}\), \(C[B/X] \in \NFTE\) and \(D[B/X] \in \PFTE\)
(respectively, \(X \in \PETV{C} \cup \NETV{D}\), \(C[B/X] \in \PFTE\)
    and \(D[B/X] \in \NFTE\)).
Therefore, \(B \in \PFTE\) by induction hypothesis.

\Case{\(A = \fix{Y}C\) for some \(Y\) and \(C\).}
We can assume that \(Y \not\in \FTV{B} \cup \{\,X\,\}\)
without loss of generality.
Since \(A[B/X] = \fix{Y}C[B/X] \in \PFTE\) (respectively, \(\NFTE\)),
we get
\begin{eqnarray}
\label{pfte-occur-01}
    C[B/X] &\in& \PFTE~(\mbox{respectively,}~\NFTE),~\mbox{and} \\
\label{pfte-occur-02}
    C[B/X] &\in& \NFTE~(\mbox{respectively,}~\fix{Y}C[B/X] \in \PFTE),
	~\mbox{if}~Y \in \NETV{C[B/X]}.
\end{eqnarray}
On the other hand, since \(X \in \PETV{A}\) (respectively, \(\NETV{A}\)),
either
\begin{eqnarray}
\label{pfte-occur-03}
    X &\in& \PETV{C}~(\mbox{respectively,}~\NETV{C}),~\mbox{or} \\
\label{pfte-occur-04}
    Y &\in& \NETV{C}~\mbox{and}~ X \in \NETV{C}
	~(\mbox{respectively,}~\PETV{C}).
\end{eqnarray}
Note that \(Y \in \NETV{C}\) implies \(Y \in\NETV{C[B/X]}\)
by Proposition~\ref{petv-netv-subst}, and that
\(\fix{Y}C[B/X] \in \PFTE\) implies \(C[B/X] \in \PFTE\).
Therefore, from
(\ref{pfte-occur-01}), (\ref{pfte-occur-02}),
(\ref{pfte-occur-03}) and (\ref{pfte-occur-04}),
we get either
\begin{eqnarray*}
    X &\in& \PETV{C}~(\mbox{respectively,}~\NETV{C})~\mbox{and}~
	C[B/X] \in \PFTE~(\mbox{respectively,}~\NFTE),~\mbox{or} \\
    X &\in& \NETV{C}~(\mbox{respectively,}~\PETV{C})~\mbox{and}~
	C[B/X] \in \NFTE~(\mbox{respectively,}~\PFTE).
\end{eqnarray*}
Hence, \(B \in \PFTE\) by induction hypothesis.
\qed\CHECKED{2014/07/12, 07/17}
\end{proof}

\begin{proposition}\label{nfte-occur}
\begin{Enumerate}
\item If\/ \(X \in \PETV{A}\)
	    and \(A[B/X] \in \NFTE\), then \(B \in \NFTE\).
\item If\/ \(X \in \NETV{A}\)
	    and \(A[B/X] \in \PFTE\), then \(B \in \NFTE\).
\end{Enumerate}
\end{proposition}
\begin{proof}
The proof is quite parallel to the one for Proposition~\ref{pfte-occur}.
Note that \(B \in \NFTE\) if \(B \geqtyp \t\)
by Definition~\ref{pfte-nfte-def}.
\ifdetail
Hence, we assume that \(B \ngeqtyp \t\) below.
The proof proceeds
by simultaneous induction on \(h(A)\), and by cases on the form of \(A\).
Suppose that \(X \in \PETV{A}\) and \(A[B/X] \in \NFTE\)
    (respectively, \(X \in \NETV{A}\) and \(A[B/X] \in \PFTE\)).
Note that \(A \ngeqtyp \t\) since \(X \in \ETV{A}\); and hence,
\(A[B/X] \ngeqtyp \t\) from \(B \ngeqtyp \t\)
by Proposition~\ref{tvariant-subst2} and Theorem~\ref{geqtyp-t-tvariant}.

\Case{\(A = Y\) for some \(Y\).}
In this case, the assumption \(X \in \NETV{A}\) of Item~2 does not hold.
As for Item~1, \(Y = X\) since \(X \in \PETV{A}\).
Hence, \(B = A[B/X] \in \NFTE\).

\Case{\(A = \O C\) for some \(C\).}
In this case, \(A[B/X] = \O C[B/X]\).
Hence, \(X \in \PETV{C}\) and \(C[B/X] \in \NFTE\)
(respectively, \(X \in \NETV{C}\) and \(C[B/X] \in \PFTE\)).
Therefore, \(B \in \NFTE\) by induction hypothesis.

\Case{\(A = C \Impl D\) for some \(C\) and \(D\).}
Since
\(A[B/X] = C[B/X] \Impl D[B/X] \in \NFTE\) (respectively, \(\PFTE\)), we get
\(X \in \NETV{C} \cup \PETV{D}\), \(C[B/X] \in \PFTE\) and \(D[B/X] \in \NFTE\)
(respectively, \(X \in \PETV{C} \cup \NETV{D}\), \(C[B/X] \in \NFTE\)
    and \(D[B/X] \in \PFTE\)).
Therefore, \(B \in \NFTE\) by induction hypothesis.

\Case{\(A = \fix{Y}C\) for some \(Y\) and \(C\).}
We can assume that \(Y \not\in \FTV{B} \cup \{\,X\,\}\)
without loss of generality.
Since \(A[B/X] = \fix{Y}C[B/X] \in \NFTE\) (respectively, \(\PFTE\)),
we get
\begin{eqnarray}
\label{nfte-occur-01}
    C[B/X] &\in& \NFTE~(\mbox{respectively,}~\PFTE),~\mbox{and} \\
\label{nfte-occur-02}
    \fix{Y}C[B/X] &\in& \PFTE~(\mbox{respectively,}~C[B/X] \in \NFTE),
	~\mbox{if}~Y \in \NETV{C[B/X]}.
\end{eqnarray}
On the other hand, since \(X \in \PETV{A}\) (respectively, \(\NETV{A}\)),
either
\begin{eqnarray}
\label{nfte-occur-03}
    X &\in& \PETV{C}~(\mbox{respectively,}~\NETV{C}),~\mbox{or} \\
\label{nfte-occur-04}
    Y &\in& \NETV{C}~\mbox{and}~ X \in \NETV{C}
	~(\mbox{respectively,}~\PETV{C}).
\end{eqnarray}
Note that \(Y \in \NETV{C}\) implies \(Y \in\NETV{C[B/X]}\)
by Proposition~\ref{petv-netv-subst}, and that
\(\fix{Y}C[B/X] \in \PFTE\) implies \(C[B/X] \in \PFTE\).
Therefore, from
(\ref{nfte-occur-01}), (\ref{nfte-occur-02}),
(\ref{nfte-occur-03}) and (\ref{nfte-occur-04}),
we get either
\begin{eqnarray*}
    X &\in& \PETV{C}~(\mbox{respectively,}~\NETV{C})~\mbox{and}~
	C[B/X] \in \NFTE~(\mbox{respectively,}~\PFTE),~\mbox{or} \\
    X &\in& \NETV{C}~(\mbox{respectively,}~\PETV{C})~\mbox{and}~
	C[B/X] \in \PFTE~(\mbox{respectively,}~\NFTE).
\end{eqnarray*}
Hence, \(B \in \NFTE\) by induction hypothesis.
\fi
\qed\CHECKED{2014/07/12, 07/17}
\end{proof}

\begin{proposition}\label{pfte-nfte-canon}
\(A \in \PFTE\) (respectively, \(\NFTE\)) if and only if
\(\Canong{A} \in \PFTE\) (respectively, \(\NFTE\)).
\end{proposition}
\begin{proof}
By induction on \(h(A)\), and by cases on the form of \(A\).
Note that the claim obviously holds when \(A\) is a {\tvariant}.
Hence, we only consider the case that \(A\) is not.

\Case{\(A = X\) for some \(X\).}
Trivial since \(\Canong{A} = A\) in this case.

\Case{\(A = \O B\) for some \(B\).}
In this case,
\(A \in \PFTE\) (respectively, \(\NFTE\)) iff
\(B \in \PFTE\) (respectively, \(\NFTE\)).
On the other hand,
\(B \in \PFTE\) (respectively, \(\NFTE\)) iff
\(\Canong{B} \in \PFTE\) (respectively, \(\NFTE\))
by induction hypothesis.
Therefore, it suffices to show that
\(\Canong{A} \in \PFTE\) (respectively, \(\NFTE\)) iff
\(\Canong{B} \in \PFTE\) (respectively, \(\NFTE\)),
which can be established by considering the form of \(\Canong{B}\) as follows.
Note that \(\Canong{B} \not= \t\) since \(A\) is not a {\tvariant}.
If \(\Canong{B} = \O^n X\) for some \(n\) and \(X\), then
\(\Canong{A} = \O \Canong{B}\), and
if \(\Canong{B} = \O^n(C \Impl D)\) for some \(n\), \(C\) and \(D\),
where \(n = 0\) is case of \(\peqtyp\), then
\(\Canon{A} = \O \Canon{B}\) and \(\Canonp{A} = \O C \Impl \O D\).
In either case, we can get
\(\Canong{A} \in \PFTE\) (respectively, \(\NFTE\)) iff
\(\Canong{B} \in \PFTE\) (respectively, \(\NFTE\)),
by Definition~\ref{pfte-nfte-def}.

\Case{\(A = B \Impl C\) for some \(B\) and \(C\).}
Trivial since \(\Canong{A} = A\) in this case.

\Case{\(A = \fix{X}B\) for some \(X\) and \(B\).}
Note that \(\Canong{A} = \Canong{B}[A/X]\) in this case.
Since \(A\) is not a {\tvariant}, \(B\) is not either; and hence,
\(B \ngeqtyp \t\) by Theorem~\ref{geqtyp-t-tvariant}.
First, we show the ``if'' part.
Suppose that \(\Canong{A} = \Canong{B}[A/X] \in \PFTE\)
(respectively, \(\NFTE\)), which also implies
\(\Canong{B} \in \PFTE\) (respectively, \(\NFTE\))
by Proposition~\ref{pfte-nfte-subst-orig}.
Hence, by induction hypothesis,
\begin{eqnarray}
    \label{pfte-nfte-canon-01}
    B &\in& \PFTE~(\mbox{respectively,}~\NFTE).
\end{eqnarray}
Furthermore, if \(X \in \NETV{B}\), then
\(X \in \NETV{\Canong{B}}\) by Propositions~\ref{geqtyp-canon}
and \ref{geqtyp-etv}; and therefore,
\(A \in \NFTE\) (respectively, \(\PFTE\))
by Proposition~\ref{nfte-occur} (respectively, \ref{pfte-occur}),
which also implies
\(B \in \NFTE\) (respectively, \(A \in \PFTE\))
by Definition~\ref{pfte-nfte-def}.
That is,
\begin{eqnarray}
    \label{pfte-nfte-canon-02}
    B &\in& \NFTE~(\mbox{respectively,}~A \in \PFTE),
	~\mbox{if}~X \in \NETV{B}.
\end{eqnarray}
On the other hand, \(A \in \TFTE\) by Theorem~\ref{tf-tfte-not-top-theorem}
since \(A\) is not a {\tvariant}.
Therefore, \(A \in \PFTE\) (respectively, \(\NFTE\))
from (\ref{pfte-nfte-canon-01}) and (\ref{pfte-nfte-canon-02})
by Definition~\ref{pfte-nfte-def}.

As for the ``only if'' part,
suppose that \(A \in \PFTE\) (respectively, \(\NFTE\)), which implies
\(B \in \PFTE\) (respectively, \(\NFTE\)).
Hence, by induction hypothesis,
\begin{eqnarray}
    \label{pfte-nfte-canon-03}
    \Canong{B} &\in& \PFTE~(\mbox{respectively,}~\NFTE).
\end{eqnarray}
On the other hand, if \(X \in \NETV{\Canong{B}}\), then
\(X \in \NETV{B}\) by Propositions~\ref{geqtyp-canon}
and \ref{geqtyp-etv}; and therefore,
\(B \in \NFTE\) (respectively, \(A \in \PFTE\))
by Definition~\ref{pfte-nfte-def}
since \(A \in \PFTE\) (respectively, \(\NFTE\)),
which also implies
\(A \in \NFTE\) (respectively, \(\PFTE\)).
That is,
\begin{eqnarray}
    \label{pfte-nfte-canon-04}
    A &\in& \NFTE~(\mbox{respectively,}~\PFTE),
	~\mbox{if}~X \in \NETV{\Canong{B}}.
\end{eqnarray}
Therefore, since \(A \in \PFTE\) (respectively, \(\NFTE\)),
we get \(\Canong{A} = \Canong{B}[A/X] \in \PFTE\) (respectively, \(\NFTE\))
from (\ref{pfte-nfte-canon-03}) and (\ref{pfte-nfte-canon-04})
by Proposition~\ref{pfte-nfte-subst}.
\qed\CHECKED{2014/07/12, 07/17}
\end{proof}

\begin{lemma}\label{pfte-pf-lemma}
If\/ \(A \in \PFTE\) (respectively, \(\NFTE\)), then
\(A\) is positively (respectively, negatively) finite.
\end{lemma}
\begin{proof}
By Theorem~\ref{tf-tfte-not-top-theorem},
it suffices to derive \(C \in \TFTE\) from the following assumptions:
\begin{Enumerate}
\item[\it(a)] \(A \in \PFTE\) (respectively, \(\NFTE\)),
\item[\it(b)] \(A \geqtyp B[C/X]\),
\item[\it(c)] \(X \in \PETV{B}\) (respectively, \(\NETV{B}\)), and
\item[\it(d)] \(X \not\in \NETV{B}\) (respectively, \(\PETV{B}\)).
\end{Enumerate}
We can assume that \(A\) and \(B\) are canonical
without loss of generality, because
\(\Canong{A} \in \PFTE\) (respectively, \(\NFTE\))
iff \(A \in \PFTE\) (respectively, \(\NFTE\))
by Proposition~\ref{pfte-nfte-canon}, and because
\(\Canong{A} \geqtyp A\),
\(\Canong{B}[C/X] \geqtyp B[C/X]\) and \(\PNETV{\Canong{B}} = \PNETV{B}\)
by Propositions~\ref{geqtyp-canon}, \ref{geqtyp-subst} and
\ref{geqtyp-etv}.
Note that \(B \not= \t\) from (c)
by Proposition~\ref{tvariant-etv},
and that Condition~(c) can be rewritten as
\begin{Enumerate}
\item[\it(c${}'$)]
    \(\PIdp(B,\,X) < \infty\) (respectively, \(\NIdp(B,\,X) < \infty\))
\end{Enumerate}
by Proposition~\ref{depth-finite-etv}.
Furthermore, if \(A = \t\), then \(A \not\in \PFTE\)
by Proposition~\ref{pfte-tfte} and Theorem~\ref{tf-tfte-not-top-theorem}
(respectively, \(X \in \PETV{B}\) from (b)
by Propositions~\ref{tvariant-subst2}, \ref{tail-etv}
and Theorem~\ref{geqtyp-t-tvariant}
since \(B \ngeqtyp \t\)).
That is, \(A = \t\) contradicts (a) (respectively, (d)).
Hence, we can also assume that \(A \not= \t\); and hence,
\(A \in \TFTE\) by Theorem~\ref{tf-tfte-not-top-theorem}.
The proof proceeds by induction on \(\PNIdp(B,\,X)\),
and by cases on the form of \(B\).

\Case{\(B = \O^n Y\) for some \(n\) and \(Y\).}
In this case, \(X \not\in \NETV{B}\) and \(Y = X\) from (c);
and hence, \(A \geqtyp B[C/X] = \O^n C\).
Since \(A \in \TFTE\), \(A\) is tail finite
by Theorem~\ref{tf-tfte-not-top-theorem}; and hence,
so is \(C\) by Definition~\ref{tf-def} and Proposition~\ref{geqtyp-congr}.
Therefore, \(C \in \TFTE\)
by Theorem~\ref{tf-tfte-not-top-theorem} again.

\Case{\(B = \O^n (D \Impl E)\) for some \(n\),
    \(D\) and \(E\) such that \(E \ngeqtyp \t\),
    where \(n = 0\) in case of \(\peqtyp\).}
In this case, \(A \geqtyp B[C/X] = \O^n (D[C/X] \Impl E[C/X])\)
and \(\Canong{B[C/X]} = \O^n (D[C/X] \Impl E[C/X])\).
Hence, \(A = \O^n (F \Impl G)\) for some \(F\) and \(G\)
such that \(F \geqtyp D[C/X]\) and \(G \geqtyp E[C/X] \ngeqtyp \t\)
by Proposition~\ref{canon-congr}.
Therefore,
since \(\PNETV{B} = \NPETV{D} \cup \PNETV{E}\) and
\(\PNIdp(B,\,X) = \min(\NPIdp(D,\,X),\,\PNIdp(E,\,X)) + 1\),
the assumptions (a), (c${}'$) and (d) imply the following.
\begin{eqnarray*}
    && \NIdp(D,\,X) < \PIdp(B,\,X)~\mbox{or}~\PIdp(E,\,X) < \PIdp(B,\,X) \\
    && \mskip130mu (\mbox{respectively,}~
	\PIdp(D,\,X) < \NIdp(B,\,X)~\mbox{or}~\NIdp(E,\,X) < \NIdp(B,\,X)), \\
    && F \in \NFTE~(\mbox{respectively,}~\PFTE),~F \geqtyp D[C/X]
	~\mbox{and}~X \not\in \PETV{D}~(\mbox{respectively,}~\NETV{D}),
    ~\mbox{and} \\
    && G \in \PFTE~(\mbox{respectively,}~\NFTE),~G \geqtyp E[C/X]
	~\mbox{and}~X \not\in \NETV{E}~(\mbox{respectively,}~\PETV{E}).
\end{eqnarray*}
Therefore, we get \(C \in \TFTE\) by induction hypothesis.
\qed\CHECKED{2014/07/15, 07/17}
\end{proof}

\begin{lemma}\label{pf-pfte-lemma}
If\/ \(A\) is positively (respectively, negatively) finite, then
\(A \in \PFTE\) (respectively, \(\NFTE\)).
\end{lemma}
\begin{proof}
We use Theorems~\ref{geqtyp-t-tvariant} and \ref{tf-tfte-not-top-theorem}
without specific mention, below.
Suppose that
\(A\) is positively (negatively) finite, i.e., \(C \in \TFTE\)
if there exists some \(B\) such that
\begin{eqnarray}
\label{pf-pfte-01}
    && A \geqtyp B[C/X], \\
\label{pf-pfte-02} 
    && X \in \PETV{B}~(\mbox{respectively,}~\NETV{B}),~\mbox{and} \\
\label{pf-pfte-03}
    && X \not\in \NETV{B}~(\mbox{respectively,}~\PETV{B}).
\end{eqnarray}
The proof proceeds by induction on \(h(A)\), and by cases on the form of \(A\).

\Case{\(A = Y\) for some \(Y\).}
\(A \in \PFTE \cap \NFTE\) by Definition~\ref{pfte-nfte-def}.

\Case{\(A = \O D\) for some \(D\).}
By induction hypothesis and Definition~\ref{pfte-nfte-def},
it suffices to show that \(D\) is positively (respectively, negatively) finite.
Suppose that
\(D \geqtyp B'[C/X]\), \(X \in \PETV{B'}\) and \(X \not\in \NETV{B'}\)
(respectively, \(X \in \NETV{B'}\) and \(X \not\in \PETV{B'}\))
for some \(X\) and \(B'\).
Let \(B = \O B'\).
Then we get (\ref{pf-pfte-01}) through (\ref{pf-pfte-03});
and hence, \(C \in \TFTE\).

\Case{\(A = D \Impl E\) for some \(D\) and \(E\).}
By induction hypothesis and Definition~\ref{pfte-nfte-def},
it suffices to show that (a) \(D\) is negatively (respectively, positively)
finite, and (b) \(E\) is positively (respectively, negatively) finite.
For (a), suppose that
\(D \geqtyp B'[C/X]\), \(X \in \NETV{B'}\) and \(X \not\in \PETV{B'}\)
(respectively, \(X \in \PETV{B'}\) and \(X \not\in \NETV{B'}\))
for some \(X\) and \(B'\).
We can assume that \(X \not\in \FTV{E}\)
without loss of generality.
Let \(B = B' \Impl E\).
Then we get (\ref{pf-pfte-01}) through (\ref{pf-pfte-03});
and hence, \(C \in \TFTE\).
We can similarly establish (b).

\Case{\(A = \fix{Y}D\) for some \(Y\) and \(D\).}
If \(A \geqtyp \t\), then obviously \(A \in \NFTE\); and
\(A \geqtyp X[\t/X]\), which means that
\(A\) is not positively finite.
Therefore, we now assume that \(A \ngeqtyp \t\), that is,
\(A\) is not a {\tvariant} and \(A \in \TFTE\). 
It suffices to show the following.
\begin{eqnarray}
\label{pf-pfte-04}
    && \fix{Y}{D} \in \TFTE~(\mbox{respectively,}~\TE), \\
\label{pf-pfte-05}
    && D \in \PFTE~(\mbox{respectively,}~\NFTE), ~\mbox{and} \\
\label{pf-pfte-06}
    && Y \not\in \NETV{D} ~\mbox{or}~ D \in \NFTE~
    (\mbox{respectively,}~\fix{Y}D \in \PFTE).
\end{eqnarray}
First, if \(A\) is positively finite, then
\(\fix{Y}D \in \TFTE\) since \(A \geqtyp X[\fix{Y}D/X]\).
Thus, we get (\ref{pf-pfte-04}).
Second, for (\ref{pf-pfte-05}),
by induction hypothesis, it is sufficient to show that
\(D\) is positively (respectively, negatively) finite.
To this end, suppose that
the following hold for some \(X\) and \(B'\).
\begin{eqnarray}
\label{pf-pfte-07}
    && D \geqtyp B'[C/X] \\
\label{pf-pfte-08}
    && X \in \PETV{B'}~(\mbox{respectively,}~\NETV{B'}),~\mbox{and} \\
\label{pf-pfte-09}
    && X \not\in \NETV{B'}~(\mbox{respectively,}~\PETV{B'}).
\end{eqnarray}
We show that (\ref{pf-pfte-07}) through (\ref{pf-pfte-09})
imply \(C \in \TFTE\).
We can assume that \(X \not\in \FTV{A} \cup \{\,Y\,\}\)
without loss of generality.
Let \(B_1 = B'[A/Y]\).
Then we get
\begin{Eqnarray*}
    A &\geqtyp& \fix{Y}B'[C/X]
	    & (by (\ref{pf-pfte-07}) and Proposition~\ref{geqtyp-fix-congr}) \\
	&\geqtyp& B'[C/X][\fix{Y}B'[C/X]/Y] & (by \r{{\eqtyp}\mbox{-fix}}) \\
	&\geqtyp& B'[C/X][A/Y]
	    & (by Proposition~\ref{geqtyp-subst}
		 since \(A \geqtyp \fix{Y}B'[C/X]\)) \\
	&=& B'[C[A/Y]/X,A/Y] & (since \(X \not= Y\)) \\
	&=& B'[A/Y][C[A/Y]/X] & (since \(X \not\in \FTV{A}\)) \\
	&=& B_1[C[A/Y]/X].
\end{Eqnarray*}
Furthermore, since \(A\) is not a {\tvariant}, we get
\[
    X \in \PETV{B_1}~\;(\mbox{respectively,}~\NETV{B_1})
\]
from (\ref{pf-pfte-08}) by Proposition~\ref{petv-netv-subst1}, and
\[
    X \not\in \NETV{B_1}~\;(\mbox{respectively,}~\PETV{B_1})
\]
from (\ref{pf-pfte-09}) by Proposition~\ref{petv-netv-subst2}
since \(X \not\in \FTV{A}\).
Therefore, \(C[A/Y] \in \TFTE\)
since \(A\) is positively (respectively, negatively) finite; and hence,
\(C \in \TFTE\) by Proposition~\ref{tfte-subst3}.
Thus, we establish that \(D\) is positively (respectively, negatively)
finite; and hence, we get (\ref{pf-pfte-05}) by induction hypothesis.

What is left is to establish (\ref{pf-pfte-06}).
Suppose that \(Y \in \NETV{D}\).
Then, by Proposition~\ref{etv-separate},
there exist some \(D'\) and \(Y'\) such that
\begin{eqnarray}
\label{pf-pfte-10}
    && D \geqtyp  D'[Y/Y'], \\
\label{pf-pfte-11}
    && Y' \in \NETV{D'},~\mbox{and} \\
\label{pf-pfte-12}
    && Y' \not\in \PETV{D'}.
\end{eqnarray}
We can assume that \(Y' \not\in \FTV{D} \cup \{Y\}\)
without loss of generality.
Since \(D\) is proper in \(Y\),
so is \(D'[Y/Y']\)
from (\ref{pf-pfte-10}) by Proposition~\ref{geqtyp-proper}.
As an important step for establishing (\ref{pf-pfte-06}),
we show that \(D\) is negatively (respectively, positively) finite.
To this end, suppose also that there exist some \(B''\) and \(X\) such that
\begin{eqnarray}
\label{pf-pfte-13}
    && D \geqtyp B''[C/X], \\
\label{pf-pfte-14}
    && X \in \NETV{B''}~(\mbox{respectively},~\PETV{B''}),~\mbox{and} \\
\label{pf-pfte-15}
    && X \not\in \PETV{B''}~(\mbox{respectively},~\NETV{B''}).
\end{eqnarray}
It suffices to show that \(C \in \TFTE\) follows
from (\ref{pf-pfte-13}), (\ref{pf-pfte-14}) and (\ref{pf-pfte-15}).
We can assume that \(X \not\in \FTV{A} \cup \FTV{D'} \cup \{\,Y,\,Y'\,\}\).
Therefore,
\begin{Eqnarray*}
&& \mskip-8mu A \geqtyp \fix{Y}D'[Y/Y']
	& (by (\ref{pf-pfte-10}) and Proposition~\ref{geqtyp-fix-congr}) \\
    &&\geqtyp \fix{Y}D'[D'[Y/Y']/Y']
	    & (by Proposition~\ref{geqtyp-fix-nest}) \\
    &&\geqtyp \fix{Y}D'[D/Y']
	    & (by (\ref{pf-pfte-10}) and Proposition~\ref{geqtyp-subst}) \\
    &&\geqtyp \fix{Y}D'[B''[C/X]/Y']
	    & (by (\ref{pf-pfte-13}) and Proposition~\ref{geqtyp-subst}) \\
    &&\geqtyp D'[B''[C/X]/Y'][\fix{Y}D'[B''[C/X]/Y']/Y]
	    & (by \r{{\eqtyp}\mbox{-fix}}) \\
    &&\geqtyp D'[B''[C/X]/Y'][A/Y] & \hskip-70pt
	    (by Proposition~\ref{geqtyp-subst}
		since \(A \geqtyp \fix{Y}D'[B''[C/X]/Y']\)) \\
    &&= D'[B''[C/X][A/Y]/Y',\,A/Y]
	    & (since \(Y' \not= Y\)) \\
    &&= D'[B''[A/Y][C[A/Y]/X]/Y',\,A/Y]
	    & (since \(X \not\in \FTV{A}\cup\{Y\}\)) \\
    &&= D'[B''[A/Y]/Y',\,A/Y][C[A/Y]/X]
	    & (since \(X \not\in \FTV{A}\cup\FTV{D'}\)) \\
    &&= D'[B''/Y'][A/Y][C[A/Y]/X]
	    & (since \(Y' \not= Y\))
\end{Eqnarray*}
We get \(X \in \PETV{D'[B''/Y']}\) (respectively, \(\NETV{D'[B''/Y']}\))
from (\ref{pf-pfte-11}) and (\ref{pf-pfte-14})
by Proposition~\ref{petv-netv-nest2}; and hence,
\[
    X \in \PETV{D'[B''/Y'][A/Y]}~\;
	(\mbox{respectively,}~ \NETV{D'[B''/Y'][A/Y]})
\]
by Proposition~\ref{petv-netv-subst1}
since \(Y \not= X\) and \(A\) is not a {\tvariant}.
On the other hand, since \(X \not\in \FTV{D'}\),
\(X \not\in \NETV{D'[B''/Y']}\) (respectively, \(\PETV{D'[B''/Y']}\))
from (\ref{pf-pfte-12}) and (\ref{pf-pfte-15})
by Proposition~\ref{petv-netv-nest3}.
Hence,
\[
    X \not\in \NETV{D'[B''/Y'][A/Y]}~\;
	(\mbox{respectively,}~\PETV{D'[B''/Y'][A/Y]})
\]
by Proposition~\ref{petv-netv-subst2}
because \(X \not\in \FTV{A}\).
Therefore, \(C[A/Y] \in \TFTE\)
since \(A\) is positively (respectively, negatively) finite;
and hence, \(C \in \TFTE\) by Proposition~\ref{tfte-subst3}.
Thus, \(D\) is negatively (respectively, positively) finite.

Now we can finally establish (\ref{pf-pfte-06}).
In fact, we get
\(D \in \NFTE\) (respectively, \(\PFTE\)) by induction hypothesis,
since \(D\) is negatively (positively) finite.
Furthermore, in the case that \(A\) is negatively finite,
\(D \in \PFTE\) also implies \(\fix{Y}D \in \PFTE\), because
\(\fix{Y}D \in \TFTE\) by assumption,
and because \(D \in \NFTE\) is already established as (\ref{pf-pfte-05}).
\qed\CHECKED{2014/07/14, 07/17}
\end{proof}

\begin{proposition}\label{pfte-pf-equiv}
\(A\) is positively (negatively) finite if and only if
\(A \in \PFTE\) (\(\NFTE\)).
\end{proposition}
\begin{proof}
Straightforward from Lemmas~\ref{pfte-pf-lemma} and \ref{pf-pfte-lemma}.
\qed\CHECKED{2014/05/23, 07/17}
\end{proof}

Thus, we verified that the notions of
positively finiteness and
negatively finiteness do not depend on the choice of \(\geqtyp\).
Now we can proceed to the main result of this subsection, namely,
the fact that every \(\lambda\)-term of positively finite types is maximal.
Before proving that, we have to prepare the following three lemmas.

\begin{lemma}\label{pfte-subtyp-lemma}
Let \(A\) and \(B\) be type expressions such that \(A \psubtyp B\).
\begin{Enumerate}
\item If \(B\) is positively finite, then so is \(A\).
\item If \(A\) is negatively finite, then so is \(B\).
\end{Enumerate}
\end{lemma}
\begin{proof}
By Propositions~\ref{depth-finite-etv}, \ref{pfte-pf-equiv} and
Theorem~\ref{tf-tfte-not-top-theorem},
it suffices to show that \(D \in \TFTE\)
if there exist some \(X\), \(A\), \(B\) and \(C\) such that either
\begin{Enumerate}
\item[(a)] \(\PIdp(C,\,X) < \infty\), \(\NIdp(C,\,X) = \infty\),
	    \(C[D/X] \psubtyp B\) and \(B \in \PFTE\), or
\item[(b)] \(\NIdp(C,\,X) < \infty\), \(\PIdp(C,\,X) = \infty\),
	    \(A \psubtyp C[D/X]\) and \(A \in \NFTE\).
\end{Enumerate}
Suppose that (a) or (b) holds.
We can assume that \(C\) is canonical with respect to \(\peqtyp\),
since \(\PNETV{C} = \PNETV{\Canonp{C}}\),
\(C[D/X] \peqtyp \Canonp{C}[D/X]\) and
\(\PNIdp(C,\,X) = \PNIdp(\Canonp{C},\,X)\)
by Propositions~\ref{geqtyp-canon}, \ref{geqtyp-etv},
\ref{geqtyp-subst} and \ref{geqtyp-depth}.
The proof proceeds by induction on \(\Idp(C,\,X)\),
and by cases on the form of \(C\).
We use Theorem~\ref{geqtyp-t-tvariant},
Propositions~\ref{tvariant-basic}, \ref{geqtyp-canon}
and \ref{subtyp-gt-t} without mention, below.

\Case{\(C = \t\).}
This case is impossible, since
\(C = \t\) implies \(\PNETV{C} = \{\}\) by Definition~\ref{etv-def},
which then implies \(\PNIdp(C,\,X) = \infty\)
by Proposition~\ref{depth-finite-etv}.

\Case{\(C = \O^n Y\) for some \(n\) and \(Y\).}
Note that (b) does not hold since \(\NIdp(C,\,X) = \infty\) in this case.
If (a) holds, then \(Y = X\) since \(\PIdp(C,\,X) < \infty\), and
\(B \npeqtyp \t\) by Proposition~\ref{pfte-tfte} and
Theorem~\ref{tf-tfte-not-top-theorem}.
Hence, \(C[D/X] = \O^n D \npeqtyp \t\)
since \(C[D/X] \psubtyp B\) in this case, which implies
\(D \npeqtyp \t\) by Proposition~\ref{geqtyp-t1}.
Therefore, \(D \in \TFTE\) by Theorem~\ref{tf-tfte-not-top-theorem}.

\Case{\(C = E \Impl F\) for some \(E\) and \(F\) such that
\(F \npeqtyp \t\).}
In this case,
\(C[D/X] = (E \Impl F)[D/X] = E[D/X] \Impl F[D/X]\).
Hence, if (a) holds, we get the following:
\begin{eqnarray}
\label{pfte-subtyp-01}
    && \NIdp(E,\,X) < \PIdp(C,\,X) ~\mbox{or}~ \PIdp(F,\,X) < \PIdp(C,\,X), \\
\label{pfte-subtyp-02}
    && \PIdp(E,\,X) = \NIdp(F,\,X) = \infty,~\mbox{and} \\
\label{pfte-subtyp-03}
    && E[D/X] \Impl F[D/X] \psubtyp B,
\end{eqnarray}
since \(\PNIdp(C,\,X) = \PNIdp(E \Impl F)
	= \min(\NPIdp(E,\,X),\,\PNIdp(F,\,X)) + 1\)
by Definition~\ref{depth-def}.
On the other hand, we get \(B \npeqtyp \t\)
from \(B \in \PFTE\) similarly to the previous case.
Therefore,
by (\ref{pfte-subtyp-03}) and Proposition~\ref{subtyp-O-Impl},
there exist some \(k\), \(G\) and \(H\) such that
\begin{enumerate}[{\kern8pt}({a}1)]
    \item \(\Canonp{B} = G \Impl H\),
    \item \(G \psubtyp \O^k E[D/X]\) and \(\O^k F[D/X] \psubtyp H\).
\end{enumerate}
Since \(B \in \PFTE\), we get
\(\Canonp{B} \in \PFTE\) by Proposition~\ref{pfte-pf-equiv}
and Definition~\ref{pnf-type-def}; and hence,
\(G \in \NFTE\) and \(H \in \PFTE\) by Definition~\ref{pfte-nfte-def}.
Therefore, \(D \in \TFTE\) from (a2),
(\ref{pfte-subtyp-01}) and (\ref{pfte-subtyp-02}) by induction hypothesis,
since \(\PNIdp(\O^k E,\,X) = \PNIdp(E,\,X)\) and
\(\PNIdp(\O^k F,\,X) = \PNIdp(F,\,X)\).

On the other hand, if (b) holds, then we similarly get
\begin{eqnarray}
\label{pfte-subtyp-04}
    && \PIdp(E,\,X) < \NIdp(C,\,X) ~\mbox{or}~\NIdp(F,\,X) < \NIdp(C,\,X), \\
\label{pfte-subtyp-05}
    && \NIdp(E,\,X) = \PIdp(F,\,X) = \infty ~\mbox{and} \\
\label{pfte-subtyp-06}
    && A \psubtyp E[D/X] \Impl F[D/X].
\end{eqnarray}
We get \(X \not\in \PETV{F}\) from \(\PIdp(F,\,X) = \infty\)
by Proposition~\ref{depth-finite-etv}; and hence, we get
\(F[D/X] \npeqtyp \t\) and \(E[D/X] \Impl F[D/X] \npeqtyp \t\)
from \(F \npeqtyp \t\)
by Propositions~\ref{tvariant-subst2} and \ref{tail-etv}.
Therefore,
by (\ref{pfte-subtyp-06}) and
Proposition~\ref{subtyp-O-Impl},
there exist some \(k\), \(G\) and \(H\) such that
\begin{enumerate}[{\kern8pt}({b}1)]
    \item \(\Canonp{A} = G \Impl H\),
    \item \(E[D/X] \psubtyp \O^k G\) and \(\O^k H \psubtyp F[D/X]\).
\end{enumerate}
We get \(\Canonp{A} \in \NFTE\) from \(A \in \NFTE\)
by Proposition~\ref{pfte-pf-equiv} and Definition~\ref{pnf-type-def}.
Furthermore, \(H \npeqtyp \t\) from (b2) since \(F[D/X] \npeqtyp \t\);
and hence, \(A \npeqtyp \t\) by (b1).
Therefore, \(G \in \PFTE\) and \(H \in \NFTE\)
by Definition~\ref{pfte-nfte-def};
and hence, \(\O^k G \in \PFTE\) and \(\O^k H \in \NFTE\).
Therefore, \(D \in \TFTE\) from (b2),
(\ref{pfte-subtyp-04}) and (\ref{pfte-subtyp-05}) by induction hypothesis.
\qed\CHECKED{2014/06/24, 07/16}
\end{proof}

\begin{lemma}\label{hnf-abst-lemma}
Let \(n\) be a non-negative integer, and
\(x_1,\,x_2,\,\ldots,\,x_n\) distinct individual variables
such that \(\{\,x_1,\,x_2,\,\ldots,\,x_n\,\} \cap \Dom(\G) = \{\}\).
If\/ \(\typ{\G}{\lam{x_1}\lam{x_2}\ldots\penalty0 \lam{x_n}M:A}\)
is derivable in {\lA},
then there exist some \(m\), \(B_1\), \(B_2\), \(\ldots\), \(B_n\) and \(C\)
such that
\begin{eqnarray*}
    && \typ{\O^m \G\cup\{\,x_1:B_1,\,x_2:B_2,\,\ldots,\,x_n:B_n\,\}}{M:C}
	~\;\mbox{is also derivable in {\lA}, and} \\
    && B_1 \Impl B_2 \Impl \ldots \Impl B_n \Impl C \psubtyp \O^m A.
\end{eqnarray*}
\end{lemma}
\begin{proof}
Proof proceeds by induction on \(n\).
It is trivial if \(n = 0\).
Hence, suppose that
\(\typ{\G}{\lam{x_1}\lam{x_2}\ldots\lam{x_n}M:A}\)
is derivable for some \(n > 0\).
By Lemma~\ref{lA-abst-lemma}, for some \(k\), \(B_1'\) and \(D\),
\begin{eqnarray}
\label{hnf-abst-01}
    && \typ{\O^k \G\cup\{\,x_1:B_1'\,\}}{\lam{x_2}\ldots\lam{x_n}M:D}
    ~\;\mbox{is derivable, and} \\
\label{hnf-abst-02}
    && B_1' \Impl D \psubtyp \O^k A.
\end{eqnarray}
Therefore, by applying the induction hypothesis to (\ref{hnf-abst-01}),
for some \(l\), \(B_2\), \(B_3\), \(\ldots\), \(B_n\) and \(C\),
\begin{eqnarray}
\label{hnf-abst-03}
    && \typ{\O^{k+l}\G\cup\{\,x_1:\O^l B_1',\,x_2:B_2,\,\ldots,\,x_n:B_n\,\}}%
	{M:C}
    ~\;\mbox{is derivable, and} \\
\label{hnf-abst-04}
    && B_2 \Impl \ldots \Impl B_n \Impl C \psubtyp \O^l D.
\end{eqnarray}
Let \(m = k+l\) and \(B_1 = \O^l B_1'\).
Then, (\ref{hnf-abst-03}) means that
\(\typ{\O^m \G\cup\{\,x_1:B_1,\,x_2:B_2,\,\ldots,\,x_n:B_n\,\}}{M:C}\)
    is derivable, and we get
\begin{Eqnarray*}
B_1 \Impl B_2 \Impl \ldots \Impl B_n \Impl C &=&
	\O^l B_1' \Impl B_2 \Impl \ldots \Impl B_n \Impl C \\
    &\psubtyp& \O^l B_1' \Impl \O^l D & (by  (\ref{hnf-abst-04})) \\
    &\peqtyp& \O^l(B_1' \Impl D) & (by \r{{\peqtyp}\mbox{-{\bf K}/{\bf L}}}) \\
    &\psubtyp& \O^{k+l} A & (by (\ref{hnf-abst-02})) \\
    &=& \O^m A.
\end{Eqnarray*}
\Qed\CHECKED{2014/06/24, 07/16}
\end{proof}

Note that the \r{{\peqtyp}\mbox{-{\bf K}/{\bf L}}}-rule
is crucial to the proof of this lemma.
However, if we remove the \r{\mbox{shift}}-rule from {\lA},
\(n\) in the statement of Lemma~\ref{lA-abst-lemma}, and hence,
\(m\) of this lemma, can be \(0\), and by which
the \r{{\peqtyp}\mbox{-{\bf K}/{\bf L}}}-rule becomes optional.

\begin{lemma}\label{peqtyp-O-nested-Impl-lemma}
If\/ \(\O A \peqtyp B_1 \Impl B_2 \Impl B_3 \Impl \ldots \Impl B_n \Impl C
\npeqtyp \t\), then
\(A \peqtyp B_1' \Impl B_2' \Impl B_3' \Impl \ldots \Impl B_n' \Impl C'\)
for some \(B_1'\), \(B_2'\), \(B_3'\), \(\ldots\), \(B_n'\) and \(C'\)
such that
\(\O B_i' \peqtyp B_i\) for every \(i \in \{\,1,\,2,\,3,\,\ldots,\,n\,\}\)
and \(\O C' \peqtyp C\).
\end{lemma}
\begin{proof}
By induction on \(n\).
It is trivial in case of \(n = 0\).
Suppose that
\(\O A \peqtyp B_1 \Impl B_2 \Impl B_3 \Impl \ldots \Impl B_n \Impl C\)
for some \(n > 0\).
By Proposition~\ref{peqtyp-O-Impl}, we get
\(A \peqtyp B_1' \Impl A'\) for some \(B_1'\) and \(A'\)
such that \(\O B_1' \peqtyp B_1\) and
\(\O A' \peqtyp B_2 \Impl B_3 \Impl \ldots \Impl B_n \Impl C\).
Note that \(A' \npeqtyp \t\) since \(A \npeqtyp \t\).
Therefore, by induction hypothesis,
\(A' \peqtyp B_2' \Impl B_3' \Impl \ldots \Impl B_n' \Impl C'\)
for some \(B_2'\), \(B_3'\), \(\ldots\), \(B_n'\) and \(C'\)
such that
\(\O B_i' \peqtyp B_i\) for every \(i \in \{\,2,\,3,\,\ldots,\,n\,\}\)
and \(\O C' \peqtyp C\).
Thus, we get
\(A \peqtyp B_1' \Impl B_2' \Impl B_3' \Impl \ldots \Impl B_n' \Impl C'\).
\qed\CHECKED{2014/07/16}
\end{proof}

\begin{lemma}\label{hnf-app-lemma}
Suppose that \(A \npeqtyp \t\).
If\/ \(\typ{\G}{\app{x}{N_1 N_2 \ldots N_n}:A}\) is derivable in {\lA},
then for some \(k\), \(l\), \(B_1\), \(B_2\), \(B_3\),
\(\ldots\), \(B_n\) and \(C\),
\begin{Enumerate}
\item[(a)] \(\G(x) \peqtyp B_1 \Impl B_2 \Impl
		    B_3 \Impl \ldots \Impl B_n \Impl C\),
\item[(b)] \(\O^k C \psubtyp \O^l A\),~\mbox{and}
\item[(c)] \(\typ{\O^l \G}{N_i:\O^k B_i}\) is derivable in {\lA}
    for every \(i \in \{\,1,\,2,\,\ldots,\,n\,\}\).
\end{Enumerate}
\end{lemma}
\begin{proof}
By induction on the derivation of
\(\typ{\G}{\app{x}{N_1 N_2 \ldots N_n}:A}\), and
by cases on the last rule used in the derivation.
Suppose that
\(\typ{\G}{\app{x}{N_1 N_2 \ldots N_n}:A}\) is derivable in {\lA}.

\Case{\r{\mbox{var}}.}
In this case, \(n = 0\), and the derivation has the form of
\[
    \ifr{\mbox{var}}
	{}
	{\typ{\G' \cup \{\,x:A\,\}}{x:A}}
\]
for some \(\G'\) such that \(\G = \G' \cup \{\,x:A\,\}\).
Therefore, we get (a) through (c) by taking \(C\), \(k\), \(l\) as \(C = A\)
and \(k = l = 0\).


\Case{\r{\mbox{shift}}.}
The derivation ends with
\[
    \Ifr{\r{\mbox{shift}}.}
	{\typ{\O \G}{\app{x}{N_1 N_2 \ldots N_n} : \O A}}
	{\typ{\G}{\app{x}{N_1 N_2 \ldots N_n} : A}}
\]
Note that \(\O A \npeqtyp \t\) since \(A \npeqtyp \t\).
By induction hypothesis,
for some \(k'\), \(l'\), \(B_1'\), \(B_2'\), \(\ldots\), \(B_n'\) and \(C'\),
\begin{eqnarray}
&& \O\G(x) \peqtyp B_1' \Impl B_2' \Impl
    B_3' \Impl \ldots \Impl B_n' \Impl C', \nonumber \\
\label{hnf-app-lemma-11}
&& \O^{k'}C' \psubtyp \O^{l'}\O A,~\mbox{and} \\
\label{hnf-app-lemma-12}
&& \typ{\O^{l'}\O\G}{N_i:\O^{k'}B_i'}\,
    ~\mbox{is derivable for every \(i\) (\(1 \le i \le n\)).}
\end{eqnarray}
Therefore, by Lemma~\ref{peqtyp-O-nested-Impl-lemma}, (a) holds
for some \(B_1\), \(B_2\), \(\ldots\), \(B_n\) and \(C\) such that
\begin{eqnarray}
\label{hnf-app-lemma-13}
    && \O B_i \peqtyp B_i'\;
	~\mbox{for every \(i\) (\(1 \le i \le n\)), and} \\
\label{hnf-app-lemma-14}
    && \O C \peqtyp C'.
\end{eqnarray}
Taking \(k\) and \(l\) as \(k = k'{+}1\) and \(l = l'{+}1\),
respectively, we get (b) from
(\ref{hnf-app-lemma-11}) and (\ref{hnf-app-lemma-14}), and get
(c) from (\ref{hnf-app-lemma-12}) and (\ref{hnf-app-lemma-13}).

\Case{\r{\t}.}
This case is impossible since \(A \npeqtyp \t\).

\Case{\r{\rsubtyp}.}
In this case, the derivation ends with
\[
    \ifr{\rsubtyp}
	{\typ{\G}{\app{x}{N_1 N_2 \ldots N_n} : A'} & A' \psubtyp A}
	{\typ{\G}{\app{x}{N_1 N_2 \ldots N_n} : A}}
\]
for some \(A'\).
Note that \(A' \npeqtyp \t\) by Proposition~\ref{subtyp-gt-t}.
By induction hypothesis, there exist some \(k\), \(l\),
\(B_1\), \(B_2\), \(B_3\), \(\ldots\), \(B_n\) and \(C\) such that
(a), \(\O^{k}C \psubtyp \O^{l}A'\) and (c) hold, from the second of which
we can also get (b) since \(A' \psubtyp A\).

\Case{\r{{\Impl}\mbox{I}}.}
This case is impossible by the form of \(\app{x}{N_1 N_2 \ldots N_n}\).

\Case{\r{{\Impl}\mbox{E}}.}
In this case, \(n > 0\), and the derivation ends with
\[
    \ifr{\Impl\mbox{E}}
	{\typ{\G_1}{\app{x}{N_1 N_2 \ldots N_{n{-}1}} : D\Impl A}
	 & \typ{\G_2}{N_n : D}}
	{\typ{\G_1 \cup \G_2}{\app{x}{N_1 N_2 \ldots N_n} : A}}
\]
for some \(\G_1\), \(\G_2\) and \(D\) such that \(\G = \G_1 \cup \G_2\).
Note that \(D \Impl A \npeqtyp \t\) by Proposition~\ref{geqtyp-Impl-t}
since \(A \npeqtyp \t\).
Hence, by induction hypothesis,
for some \(k'\), \(l\), \(B_1\), \(B_2\), \(B_3\),
\(\ldots\), \(B_{n{-}1}\) and \(E\),
\begin{eqnarray}
\label{hnf-app-lemma-21}
    && \G_1(x) \peqtyp B_1 \Impl B_2 \Impl
	B_3 \Impl \ldots \Impl B_{n{-}1} \Impl E, \\
\label{hnf-app-lemma-22}
    && \O^{k'}E \psubtyp \O^l (D \Impl A),~\mbox{and} \\
\label{hnf-app-lemma-23}
    && \typ{\O^{l} \G_1}{N_i:\O^{k'} B_i}
	~\mbox{is derivable for every \(i\) (\(1 \le i < n\)).}
\end{eqnarray}
Since \(\Canonp{(\O^l (D \Impl A))} = \O^l \Canonp{D} \Impl \O^l \Canonp{A}\),
by Propositions~\ref{subtyp-O-Impl} and \ref{geqtyp-canon},
there exist some \(B_n'\), \(C'\) and \(l'\) such that
\begin{eqnarray}
\label{hnf-app-lemma-24}
&& \O^{k'} E \peqtyp B_n' \Impl C', \\
\label{hnf-app-lemma-25}
&& \O^l D \psubtyp \O^{l'} B_n'\;~\mbox{and}~
	\O^{l'} C' \psubtyp \O^l A
\end{eqnarray}
from (\ref{hnf-app-lemma-22}).
Note that \(E \npeqtyp \t\) by (\ref{hnf-app-lemma-22}) since
\(D \Impl A \npeqtyp \t\).
Hence, by applying Proposition~\ref{peqtyp-O-Impl}
to (\ref{hnf-app-lemma-24}) repeatedly,
there also exist some \(B_n\) and \(C\) such that
\begin{eqnarray}
\label{hnf-app-lemma-26}
&& E \peqtyp B_n \Impl C, \\
\label{hnf-app-lemma-27}
&& \O^{k'} B_n \peqtyp B_n'\;~\mbox{and}~\O^{k'} C \peqtyp C'.
\end{eqnarray}
We now get (a) from (\ref{hnf-app-lemma-21}) and (\ref{hnf-app-lemma-26})
since \(\G(x) = \G_1(x)\).
Let \(k = k'{+}l'\).
Then we get (b) from (\ref{hnf-app-lemma-25}) and (\ref{hnf-app-lemma-27}).
On the other hand,
\(\typ{\O^l \G_2}{N_n : \O^l D}\) is derivable
by \r{\mbox{nec}} from \(\typ{\G_2}{N_n : D}\); and therefore, we get (c)
from (\ref{hnf-app-lemma-23})
by Proposition~\ref{typ-weakening} and \r{\rsubtyp},
since \(\G = \G_1 \cup \G_2\),
\(\O^{k'} B_i \psubtyp \O^k B_i\)
for every \(i \in \{\,1,\,2,\,\ldots,\,n{-}1\,\}\), and since
\(\O^l D \psubtyp \O^{l'} B_n' \peqtyp \O^{k'{+}l'} B_n = \O^k B_n\)
from (\ref{hnf-app-lemma-25}) and (\ref{hnf-app-lemma-27}).
\qed\CHECKED{2014/06/24, 07/16}
\end{proof}

\begin{theorem}\label{typ-maximal}
Let\/ \(\typ{\G}{M:A}\) be a derivable judgment of {\lA}.
If \(A\) is positively finite and \(\G(x)\) is
negatively finite for every \(x \in \Dom(\G)\), then \(M\) is maximal.
\end{theorem}
\begin{proof}
We show that for every \(n\), every node of the B\"{o}hm-tree of \(M\)
at the level \(n\) is head normalizable, by induction on \(n\).
Since \(A\) is positively finite, \(M\) is head normalizable
by Theorem~\ref{typ-hnf-theorem}, that is,
\[
    M \ctc \lam{x_1}\lam{x_2}\ldots\lam{x_m}\app{y}{N_1\,N_2\,\ldots\,N_l}
\]
for some \(x_1\), \(x_2\), \(\ldots\), \(x_m\), \(y\), \(N_1\), \(N_2\),
\(\ldots\), \(N_l\).
We can assume, without loss of generality,
that \(x_1\), \(x_2\), \(\ldots\), and \(x_m\) are
distinct individual variables such that
\(\{\,x_1,\,x_2,\,\ldots,\,x_m\,\} \cap \Dom(\G) = \{\}\).
It suffices to show that every node of the B\"{o}hm-tree of \(N_i\)
at a level less than \(n\) is head normalizable
for every \(i\) (\(1 \le i \le l\)).
By Theorem~\ref{subj-red-theorem},
\(\typ{\G}{\lam{x_1}\lam{x_2}\ldots\lam{x_m}\app{y}{N_1\,N_2\,\ldots\,N_l}:A}\)
is also derivable in {\lA}; and therefore,
by Lemma~\ref{hnf-abst-lemma},
for some \(k\), \(B_1\), \(B_2\), \(\ldots\), \(B_m\) and \(C\),
\begin{eqnarray}
\label{typ-maximal-01}
    && \typ{\O^k \G\cup\{\,x_1:B_1,\,x_2:B_2,\,\ldots,\,x_m:B_m\,\}}%
	{\app{y}{N_1\,N_2\,\ldots\,N_l}:C}~\;\mbox{is derivable, and} \\
\label{typ-maximal-02}
    && B_1 \Impl B_2 \Impl \ldots \Impl
		B_m \Impl C \psubtyp \O^k A.
\end{eqnarray}
Note that since \(A\) is positively finite,
so is \(\O^k A\)
by Proposition~\ref{pfte-pf-equiv} and Definition~\ref{pfte-nfte-def}; and
hence, \(B_1 \Impl B_2 \Impl \ldots \Impl B_m \Impl C\)
is also positively finite
from (\ref{typ-maximal-02}) by Lemma~\ref{pfte-subtyp-lemma}.
Therefore, \(C\) is positively finite, and
\(B_1\), \(B_2\), \(\ldots\), \(B_m\) are negatively finite
by Proposition~\ref{pfte-pf-equiv} and Definition~\ref{pfte-nfte-def}.
Note also that \(C \not\peqtyp \t\) since \(C\) is positively finite.
Let \(\G' = \O^k\G\cup\{\,x_1:B_1,\,x_2:B_2,\,\ldots,\,x_m:B_m\,\}\).
Note that \(\G'(x)\) is negatively finite
for every \(x \in \Dom(\G')\).
By applying Lemma~\ref{hnf-app-lemma} to (\ref{typ-maximal-01}),
for some \(p\), \(q\), \(D_1\), \(D_2\), \(\ldots\), \(D_l\) and \(E\),
\begin{eqnarray}
\label{typ-maximal-03}
    && \G'(y) \peqtyp D_1 \Impl D_2 \Impl \ldots \Impl D_l \Impl E, \\
\label{typ-maximal-04}
    && \O^p E \psubtyp \O^q C,~\mbox{and} \\
\label{typ-maximal-05}
    && \mbox{\(\typ{\O^q \G'}{N_i:\O^p D_i}\) is derivable
	for every \(i\) (\(1 \le i \le l\))}.
\end{eqnarray}
Since \(\O^q C\) is also positively finite,
so is \(E\) from (\ref{typ-maximal-04})
by Lemma~\ref{pfte-subtyp-lemma}, Proposition~\ref{pfte-pf-equiv}
and Definition~\ref{pfte-nfte-def}, i.e., \(E \npeqtyp \t\); and therefore,
\(D_i\) is positively finite, and so is \(\O^p D_i\),
for every \(i\) (\(1 \le i \le l\))
from (\ref{typ-maximal-03}) because \(\G'(y)\) is negatively finite.
Since \(\O^q \G'(x)\) is also negatively finite for every \(x \in \Dom(\G')\),
by applying induction hypothesis to (\ref{typ-maximal-05}),
every node of the B\"{o}hm-tree of \(N_i\)
at a level less than \(n\) is head normalizable
for every \(i\) (\(1 \le i \le l\)); that is,
so is one of \(M\) at a level less than or equal to \(n\).
\qed\CHECKED{2014/06/24, 07/16}
\end{proof}

\Subsection{Simple types}\label{typ-finite-sec}

By the standard method due to Tait\cite{tait},
it can be also shown that \(\lambda\)-terms of
simple types, namely, those having no occurrence of the modal operator,
are normalizable.

\begin{definition}
    \ilabel{NF-def}{NF@$\protect\NF$}
    \ilabel{term-model-normalizable-kernel-def}{Kn@$\protect\Kn$}
We denote by \(\NF\) the subset of normalizable \(\lambda\)-terms of \(\IE\),
and define a subset \(\Kn\) of\/ \(\IE/{\beq}\,\) as follows.
\[
\Kn = \Zfset{[\,\app{x}{N_1 N_2 \ldots N_n}]}
		{\tabcolsep=0pt
		    \ifnarrow
			\begin{tabular}{p{180pt}}%
		    \else
			\begin{tabular}{l}%
		    \fi
		    \(x \in \IV\), \(n \ge 0\) and
		    \(N_i \in \NF\) for every \(i\)
		    (\(i = 1,\,2,\,\ldots,\,n\))\end{tabular}}.
\]
\end{definition}

\begin{lemma}\label{normalizable-lemma}
Consider the term model of\/ \(\IE\),
an arbitrary  {\lA}-frame \(\pair{\W}{\acc}\) and
a hereditary type environment \(\tenv\) such that
\(\Kn \subseteq \tenv(X)_p \subseteq \NF\)
for every type variable \(X\) and \(p \in \W\).
Let \(A\) be a type expression with no occurrence of
the modal operator~\(\O\).
Then,
\(\Kn \subseteq \I{A}^\tenv_p \subseteq \NF\,\) for every \(p \in \W\).
\end{lemma}
\begin{proof}
By induction on \(h(A)\),
and by cases on the form of \(A\).

\Case{\(A = Y\) for some \(Y\).}
Obvious since \(\I{A}^\tenv_p = \tenv(Y)_p\)
by Proposition~\ref{rlz-another-def}.

\Case{\(A = B \Impl C\) for some \(B\) and \(C\).}
First, we show \(\Kn \subseteq \I{B \Impl C}^\tenv_p\).
Let \(u = [\,\app{x}{N_1 N_2 \ldots N_n}] \in \Kn\),
and suppose that \(p \tacc q\) and \([L] \in \I{B}^\tenv_q\).
By induction hypothesis, \([L] \in \NF\).
Then, \(u \cdot [L] = [\,\app{x}{N_1 N_2 \ldots N_n L}] \in \Kn
\subseteq \I{C}^\tenv_q\strut\)
by induction hypothesis again; and hence,
\(u \in \I{B \Impl C}^\tenv_p\) by Proposition~\ref{rlz-another-def}.
Second, to show \(\I{B \Impl C}^\tenv_p \subseteq \NF\strut\),
suppose that \([M] \in \I{B \Impl C}^\tenv_p\).
Let \(y\) be a fresh individual variable.
Since \([y] \in \Kn \subseteq \I{B}^\tenv_p\)
by induction hypothesis,
we get \([M]\cdot[y] = [\app{M}{y}] \in \I{C}^\tenv_p\)
by Proposition~\ref{rlz-another-def};
and hence, \([\app{M}{y}] \in \NF\) by induction hypothesis.
Therefore, \(\app{M}{y}\) has a normal form, say \(L\).
There are two possible cases: for some \(K\),
({\romannumeral 1}) \(M \ctc K\) and \(L = \app{K}{y}\), or
({\romannumeral 2}) \(M \ctc \lam{y}{K}\) and \(K \ctc L\).
In either case, \(M\) obviously has a normal form.

\Case{\(A = \fix{X}B\) for some \(X\) and \(B\).}
In this case, \(X \not\in \FTV{B}\)
since \(B\) is proper in \(X\), and since \(A\) has no occurrence of \(\O\).
Therefore, \(\I{\fix{X}B}^\tenv_p = \I{B}^\tenv_p\)
by Proposition~\ref{rlz-another-def}; and hence,
obvious from induction hypothesis.
\qed\CHECKED{2014/07/16}
\end{proof}

\begin{theorem}\label{normalizable-theorem}
Suppose that the modal operator does not occur in \(\G\) or \(A\).
If\/ \(\typ{\G}{M:A}\) is derivable in {\lA}, then
\(M\) is normalizable.
\end{theorem}
\begin{proof}
Let \(\G = \{\,x_1:B_1,\,x_2:B_2,\,\ldots,\,x_n:B_n\,\}\) and
let \(\pair{\W}{\acc}\) be an arbitrary {\lA}-frame.
Consider the term model
\(\tuple{\V,\,\cdot,\,\VI{\;}{}}\) of \(\IE\) and
a type environment \(\tenv\) such that
\(\tenv(X)_p = \Kn\) for every \(X \in \TV\) and \(p \in \W\).
Furthermore, let \(\ienv\) be an individual environment such that
\(\ienv(x) = [x]\) for every \(x \in \IV\).
For every \(i \in \{\,1,\,2,\,\ldots,\,n\,\}\),
since \(\ienv(x_i) = [x_i] \in \Kn\),
we get \(\ienv(x_i) \in \I{B_i}^\tenv_p\,\) for every \(p \in \W\)
by Lemma~\ref{normalizable-lemma}.
Hence, \(\VI{M}{\rho} = [M] \in \I{A}^\tenv_p \subseteq \NF\)
by Theorem~\ref{soundness-theorem} and
Lemma~\ref{normalizable-lemma} again.
\qed\CHECKED{2014/07/16}
\end{proof}

Finiteness of type expressions, namely,
not including recursive types, is not sufficient
for assuring the normalizability of inhabitants.
In the proof of Lemma~\ref{normalizable-lemma}, we cannot show that
\(\I{\O B}^\tenv_p \subseteq \NF\) from
the induction hypothesis \(\I{B}^\tenv_p \subseteq \NF\),
since it might be the case that \(\I{\O B}^\tenv_p \supsetneq \I{B}^\tenv_p\).
In fact, we have already observed that Curry's \(\Y\) has such a type
\((\O X \Impl X) \Impl X\) in Example~\ref{Y-derivable}.

\Section{{\lA} as a basis for logic of programming}\label{program-sec}

The typing system {\lA} can be easily
extended to cover full propositional and second-order types.
For example, we can add the following equality
and typing rules for product types.
\begin{eqnarray*}
    \O(A_1 \Conj A_2) &\eqtyp& \O A_1 \Conj \O A_2
\end{eqnarray*}
\[
\ifr{{\Conj}\mbox{I}}
    {\typ{\G_1}{M : A_1}
     & \typ{\G_2}{N : A_2}}
    {\typ{\G_1 \cup \G_2}{\pair{M}{N} : A_1 \Conj A_2}}
\ifnarrow\mskip20mu\else\mskip80mu\fi
\Ifr{\r{{\Conj}\mbox{E}}\quad(i = 1,\,2)}
    {\typ{\G}{M : A_1 \Conj A_2}}
    {\typ{\G}{\proj{\it i}{M} : A_i}}
\]
However, even with such extensions,
type expressions discussed so far can only refer to
approximative worlds at a fixed distance from the current world,
and do not have enough power to
express the manner in which approximation proceeds in
general recursive programs.
In order to handle such programs,
we have to extend our logic to a predicate logic,
and allow type expressions
such as \(\O^{t} A\) as well-formed type expressions,
with \(t\) being a numeric expression.
We also need enough arithmetic endowing the typing system
to discuss such \(t\).
With the help of such extensions, {\lA}
can be a basis for logic of a wide range of programs.
In this section we give some examples.

\Subsection{Recursive programs over non-negative integers}
\begingroup
\def\preprime{\mathit{preprime}}
\def\primep{\mathit{prime}}
\def\npp{\mathit{npp}}
\def\np{\mathit{np}}
\def\PPS{\mathtt{PPS}}
\def\PS{\mathtt{PS}}
\def\Sieve{\mathtt{Sieve}}
\def\primes{\mathtt{primes}}
\def\sieve{\mathtt{sieve}}
\def\enum{\mathtt{enum}}
\def\zero{\mathtt{0}}
\def\one{\mathtt{1}}
\def\two{\mathtt{2}}
\def\pl{\mathbin{\mathtt{+}}}
\def\mn{\mathbin{\mathtt{-}}}
\def\eq{\mathbin{\mathtt{=}}}

Provided such an extension to predicate logic endowed with an arithmetic,
we can construct a wide range
of recursive programs with fixed-point combinators
ensuring their termination.
For example,
let \(\nat(n)\) represent the type of (implementations of)
the non-negative integer \(n\), and
suppose that \(f\) is a primitive recursive functional defined as
\[
    f\:x \equiv \:\mathbf{if}~(x \eq \zero)~\mathbf{then}~c~\mathbf{else}~
	g\:(f\:(x \mn \one))\:x,
\]
where
\(c\) is a program of a type \(A(0)\),
\(g\) a type \(A(n \dminus 1) \Impl \nat(n) \Impl A(n)\)
for every positive integer \(n\)
with a primitive recursive function \(\dminus\) defined in the arithmetic as
\begin{eqnarray*}
    x \dminus y &\equiv& \Choice{%
	x - y & \mbox{(if \(x \ge y\))} \\
	\hbox to 60pt {\(0\)\hfil} & \mbox{(otherwise)}
    }
\end{eqnarray*}
We also assume that \(c\) or \(g\) includes no free occurrence of \(x\),
and that the infix operator \(-\) is a program that satisfies
\(\forall m.\,\forall n.\,\nat(m) \Impl \nat(n) \Impl \nat(n \dminus m)\).
We will show that \(f\) has the type
\(\forall n.\,\nat(n) \Impl \O^n A(n)\), where, in general,
\(\forall n.\, B(n)\) is interpreted as
\[
    \I{\forall n.\, B(n)}^\tenv_p =
    \zfset{u}{u \in \I{B(n)}^\tenv_p\,~\mbox{for every non-negative integer}~n}.
\]
By virtue of the fixed-point combinator \(\Y\), we can reformulate
the definition of \(f\) to
\[
    f \equiv \Y\:(\lam{f}\lam{x}
	\mathbf{if}~(x \eq \zero)~\mathbf{then}~c~\mathbf{else}~
	    g\:(f\:(x \mn \one))\:x).
\]
Let \(C = \forall n.\,\nat(n) \Impl \O^n A(n)\).
It suffices to show that
    \(\mathbf{if}\;(x \eq \zero)
	\;\mathbf{else}\;g\;(f\;(x \mn \one))\;x : \O^{n}A(n)\)
assuming \(f : \O\,C\) and \(x : \nat(n)\),
since \(\Y: (\O C \Impl C) \Impl C\).
If \(n = 0\), then this is straightforward from \(c : A(0)\).
The case that \(n > 0\) can be accomplished as follows.
\begingroup
\ifnarrow\mtight\fi
\begin{Eqnarray*}
&& f : \forall n.\,\O (\nat(n) \Impl \O^n A(n)) &
	(since \(\O\) is interchangeable with \(\forall n\)) \\
&& f : \O (\nat(n \dminus 1) \Impl \O^{n \dminus 1} A(n \dminus 1)) &
	(by instantiation, i.e., \r{{\forall}\mbox{E}}) \\
&& f : \O\nat(n \dminus 1) \Impl \O^{n\dminus1+1}A(n \dminus 1) &
	(by \r{{\peqtyp}\mbox{-{\bf K}/{\bf L}}} or
	    \r{{\subtyp}\mbox{-\bf K}}) \\
&& f : \O\nat(n \dminus 1) \Impl \O^n A(n \dminus 1) &
	(since \(n \dminus 1+1 = n\)) \\
&& f\:(x \mn \one) : \O^n A(n \dminus 1) &
	(since \(x-1 : \O\nat(n \dminus 1)\)) \\
&& g\:(f\:(x \mn \one)) : \O^{n}\nat(n) \Impl \O^{n}A(n)
    \ifnarrow\\[-3pt]\fi
    & \ifnarrow&&\hskip-100pt\hfill\fi
    (since \(g:\O^{n}A(n \dminus 1) \Impl \O^{n}\nat(n) \Impl \O^{n}A(n)\)) \\
&& g\:(f\:(x \mn \one))\:x : \O^{n}A(n) &
	(since \(x:\O^{n}\nat(n)\))
\end{Eqnarray*}
\endgroup
Thus, we get \(\mathbf{if}~(x \eq \zero)~\mathbf{then}~c~\mathbf{else}~
	g\;(f\;(x \mn \one))\:x : \O^{n}A(n)\).

In general, if a function \(f\) of a type \(A(\vec{n}) \Impl B\)
can be recursively defined based on a well-founded measure
\(\pi(\vec{n})\) on the input,
then we can give \(f\) a type \(A(\vec{n}) \Impl \O^{\pi(\vec{n})} B\).
Consider the following recursive program,
which represents McCarthy's 91-function.
\[
    f\:x \equiv \mathbf{if}~(x \mathbin{\mathtt{>}} \mathtt{100})~
	    \mathbf{then}~x \mn \mathtt{10}~
	    \mathbf{else}~f\:(f\:(x \pl \mathtt{11}))
\]
We can show that \(f\) has a type, or satisfies a specification,
\(\forall n.\,\nat(n) \Impl \O^{101 \dminus n}\nat(g(n))\),
where \(n\) ranges over non-negative integers, and
\(g\) is a primitive recursive function defined in the arithmetic as
follows.
\begin{eqnarray*}
    g(x) &\equiv& \Choice{%
	x - 10 & \mbox{(if \(x > 100\))} \\[-2pt]
	\hbox to 60pt {\(91\)\hfil} & \mbox{(otherwise)}
    }
\end{eqnarray*}
Suppose that
\(f:\O\,\forall n.\,\nat(n) \Impl \O^{101 \dminus n}\penalty0\nat(g(n))\) and
\(x:\nat(n)\).
The type of \(\Y\) assures that it suffices to show
\begin{eqnarray}
    \mathbf{if}~(x \mathbin{\mathtt{>}} \mathtt{100})~
	    \mathbf{then}~x \mn \mathtt{10}~
	    \mathbf{else}~f\:(f\:(x \pl \mathtt{11}))
	\ifnarrow\mskip130mu\nonumber \\[-2pt]\fi
    \label{f91-type}
	    {}: \forall n.\,\nat(n) \Impl \O^{101 \dminus n}\nat(g(n)).
\end{eqnarray}
We assume that the infix operator \(\pl\) is a program of the type
\(\forall m.\,\forall n.\,\nat(m) \Impl \nat(n) \Impl \nat(n+m)\).
First, we get
\[
    f\;(x\pl\mathtt{11}) : \O\O^{101 \dminus (n+11)}\nat(g(n+11)).
\]
If \(n \le 90\), then
we get \(f\,(x\pl\mathtt{11}) : \O^{91 \dminus n}\,\nat(91)\)
by the definitions of \(\dminus\) and \(g\); and therefore,
\[
    f\;(f\;(x\pl\mathtt{11})) :
     \O^{91 \dminus n}\penalty0\O^{101 \dminus 91}\penalty0\nat(g(91)),
\]
which type is equivalent to \(\O^{101 \dminus n}\penalty0\nat(91)\).
On the other hand, if \(90 < n \le 100\),
then we similarly get \(f\;(x\pl\mathtt{11}) : \O \nat(n+1)\); and therefore,
\[
    f\;(f\;(x\pl\mathtt{11})) :
	\O\O^{101 \dminus (n+1)}\penalty0\nat(g(n+1)),
\]
which type is also equivalent to
\(\O^{101 \dminus n}\penalty0\nat(91)\).
Otherwise, i.e., if \(100 < n\), obviously \(x\mn\mathtt{10} : \nat(g(n))\).
Thus, we establish (\ref{f91-type}).

In this derivation, the fixed-point combinator worked
as the induction scheme discussed in Section~\ref{intro-sec} with
a sequence \(S_0\), \(S_1\), \(S_2\), \(\ldots\), \(S_n\) as follows.
\begin{eqnarray*}
    S_0 &=& \V \\
    S_{n+1} &=& \zfset{f}{\forall x \ge 101 \dminus n.~f(x) = g(x)}
    \qquad(n = 0,\,1,\,2,\,\ldots)
\end{eqnarray*}
The typing
\(\typ{}{f:\forall n.\,\nat(n) \Impl \O^{101 \dminus n}\nat(g(n))}\)
does not mean that the function \(f\) computes the output in
\(101 \dminus n\) steps.
Actually, \(f(n)\) involves \(2(101 \dminus n)\) recursive calls.
Recall that our interpretations of types are closed under \(\ct\).
It only means that the output satisfies a specification that
is at most \(101 \dminus n\) steps {\em weaker}, in a sense which
only the programmer knows, than
the specification of the input.
However, on the contrary, if a recursively defined program \(f\)
over non-negative integers computes its output \(g(n)\)
within \(\pi(n)\) levels of recursion for every input \(n\),
the typing judgment
\(\typ{}{f:\forall n.\,\nat(n) \Impl \O^{\pi(n)}\nat(g(n))}\)
should be derivable.

Note also that if we only consider {\lA}-frames such that
\begin{eqnarray}
   \label{successive-frame}
    \mbox{for every \(p \in \W\),
    there exists some \(q \in \W\) such that \(q \acc p\),\xfootnote{%
	For example, the set of non-negative integers, or ordinals, and
	the ``greater than'' relation \(>\) constitutes such a frame.}}
\end{eqnarray}
\emitxfootnote
and if \(\,\I{\nat(n)}^\tenv_p\) does not depend on \(p\),
then the interpretation of
\(\typ{}{f:\forall n.\,\nat(n) \Impl \O^{101 \dminus n}\nat(g(n))}\)
implies \(f \in \I{\forall n.\,\nat(n) \Impl \nat(g(n))}^\tenv_p\,\)
for every \(p \in \W\).
For, in such a frame, for each \(n\),
every possible world \(p\) has another world \(q\) from which
\(p\) is accessible in more than \(101 \dminus n\) steps.
Because the typing judgment is valid in every possible world,
we get \(f \in \I{\nat(n) \Impl \O^{101 \dminus n}\nat(g(n))}^\tenv_q\)
for such \(q\), and which implies that
\(f \in \I{\nat(n) \Impl \nat(g(n))}^\tenv_p\,\) provided that
\(\I{\nat(n)}^\tenv_p\) does not depend on \(p\).
Thus, we get
\(f \in \I{\nat(n) \Impl \nat(g(n))}^\tenv_p\,\) for every \(p \in \W\)
and \(n \in \N\).
It can be also observed that \(\typ{}{f:\forall n.\,\nat(n) \Impl \nat(g(n))}\)
becomes formally derivable from
\(\typ{}{f:\forall n.\,\nat(n) \Impl \O^{101 \dminus n}\penalty0\nat(g(n))}\)
by introducing another modality, say \(\B\), which is interpreted as
\[
   \I{\B A}^\tenv_p = \zfset{u}{u \in \I{A}^\tenv_q~\mbox{for every}~q \in \W},
\]
and accordingly enjoys the following subtyping relations
and typing rules:
\begin{Enumerate}
\item \(A \psubtyp B\)\, implies \,\(\B A \psubtyp \B B\)
\item \(\B (A \Impl B) \psubtyp \B A \Impl \B B\)
\item \(\B A \psubtyp A \)
\item \(\B\B A \eqtyp \B A\)
\item \(\B \O^t A \eqtyp \B A\)
\item \(\B \nat(n) \eqtyp \nat(n)\)
\end{Enumerate}
\[
    \ifr{\B}
	{\typ{\G_1}{M : A}}
	{\typ{\B \G_1 \cup \G_2}{M : \B A}}
    \mskip80mu
    \ifr{\mbox{shift}}
	{\typ{\B \G_1 \cup\O \G_2}{M : \O A}}
	{\typ{\B \G_1 \cup \G_2}{M : A}}
\]
where the new \r{\mbox{shift}}-rule supersedes
the original one in Definition~\ref{typ-rules}.
    While the equality \(\B \O^t A \eqtyp \B A\) and
    the new \r{\mbox{shift}}-rule
    depend on the restriction (\ref{successive-frame}),
    the other rules are valid for any frame.
    We might also need the \r{\mbox{subst}}-rule
    since Proposition~\ref{lA-subst-redundant} would not hold
    for the extended system.
Recursive type variables are not allowed to occur in
scopes of the \(\B\)-operator, and
\(\I{A}^{\tenv}_p\) is now defined by induction
on the lexicographic ordering of \(\tuple{b(A),\,p,\,r(A)}\),
where \(b(A)\) is the depth of nesting occurrences of \(\B\) in \(A\).

\Subsection{Infinite data structures}
Streams, or infinite sequences, of data of a type \(X\)
are representable by the type \(\fix{Y}X \Conj \O Y\).
We can construct recursive programs for streams
with fixed-point combinators such as Curry's \(\Y\).
For example,
a program that generates a stream
of a given constant of type \(X\) as follows, where
\(A = \fix{Y}X \Conj \O Y\).
\[
    \mskip-10mu
    \ifr{\Impl\mbox{I}}
	{\ifr{\Impl\mbox{E}}
	    {\stack{\Derive{30}{}}
		   {\typ{}{\Y : (\O A\Impl A)\Impl A}}
	     & \mskip-10mu
	     \ifr{\Impl\mbox{I}}
		    {\ifr{\eqtyp}
			{\ifr{\Conj\mbox{I}}
			    {\ifr{\mbox{var}}{}{\typ{x:X}{x:X}}
			     & \ifr{\mbox{var}}{}{\typ{y:\O A}{y : \O A}}}
			    {\typ{x:X,\,y:\O A}{\pair{x}{y} : X \Conj \O A}}}
			{\typ{x:X,\,y:\O A}{\pair{x}{y} : A}}}
		    {\typ{x : X}{\lam{y}{\pair{x}{y}} : \O A\Impl A}}}
	    {\typ{x : X}{\app{\Y}{(\lam{y}{\pair{x}{y}})}:A}}}
	{\typ{}{\lam{x}{\app{\Y}{(\lam{y}{\pair{x}{y}})}} : X\Impl A}}
\]
The following shows the derivation of a program
that merges two streams, where
\(B = A\Impl A \Impl A\) and \(\G = \{\,f:\O B,\, x:A,\,y:A\,\}\).
\[
\ifnarrow\Mtight\fi
\def\MM{\pair{\pleft{x}}{\app{\app{f}{y}}{(\pright{x})}}}
\ifr{\Impl\mbox{E}}
    {\stack{\mskip0mu\Derive{30}{}}{\typ{}{\Y : (\O B \Impl B)\Impl B}}
     &\mskip-10mu
      \ifr{\Impl\mbox{I}}
	    {\mskip-20mu\ifr{\Impl\mbox{I}}
		{\ifr{\Impl\mbox{I}}
		    {\mskip-105mu
		     \ifr{\eqtyp}
			{\ifr{\Conj\mbox{I}}
			    {\ifr{\Conj\mbox{E}}
				{\ifr{\eqtyp}
				     {\ifr{\mbox{var}}{}{\typ{\G}{x:A}}}
				     {\typ{\G}{x:X \Conj \O A}}}
				{\typ{\G}{\pleft{x}:X}}
			    & \ifnarrow\mskip-40mu\fi
			      \ifr{\Impl\mbox{E}}
				{\ifnarrow
				    \mskip-60mu
				\else
				     \def\stack#1#2{#1}%
				\fi
				\stack{\ifr{\Impl\mbox{E}}
				    {\ifr{\rsubtyp}
					 {\ifr{\mbox{var}}{}
					      {\typ{\G}{f:\O B}}}
					 {\typ{\G}{f:\O A\Impl\O A\Impl\O A}}
				     & \ifr{\rsubtyp}
					   {\ifr{\mbox{var}}{}{\typ{\G}{y:A}}}
					   {\typ{\G}{y:\O A}}}
				    {\typ{\G}{\app{f}{y}:\O A \Impl \O A}}}%
				    {\Derive{50}}
				 & \ifnarrow\mskip-90mu\fi
				 \ifr{\Conj\mbox{E}}
				    {\ifr{\eqtyp}
					 {\ifr{\mbox{var}}{}{\typ{\G}{x:A}}}
					 {\typ{\G}{x:X \Conj \O A}}}
				    {\typ{\G}{\pright{x}:\O A}}}
				{\typ{\G}{\app{\app{f}{y}}
						{(\pright{x})}:\O A}}}
			    {\typ{\G}{\pair{\pleft{x}}
					    {\app{\app{f}{y}}{(\pright{x})}}
					:X \Conj \O A}}}
			{\typ{\G}{\pair{\pleft{x}}
			    {\app{\app{f}{y}}{(\pright{x})}}:A}}}
		{\typ{f:\O B,\,x:A}{\lam{y}{\MM}:A \Impl A}}}
	    {\typ{f:\O B}
		 {\lam{x}{\lam{y}{\MM}}:B}}}
      {\typ{}{\lam{f}{\lam{x}\lam{y}{\MM}}
		:\O B \Impl B}}}
    {\typ{}{\app{\Y\,}{(\lam{f}\lam{x}\lam{y}
			\pair{\pleft{x}}{\app{\app{f}{y}}{(\pright{x})}})}
	    : B}}
    \mskip-20mu
\]

To capture more complicated recursion over streams,
we naturally need an extension to predicate logic such as used in the case of
the 91-function.
For example, the prime number generator based on the sieve of Eratosthenes
is derivable in our framework with such an extension.
First, by using \(\Y\),
define three recursive programs \(\sieve\), \(\primes\) and \(\enum\)
as follows.
\begin{eqnarray*}
    \enum~x &\equiv& \pair{x}{\enum~(x\pl\one)} \\
    \sieve~x~w &\equiv&
	\mathbf{if}~((\pleft{w}) \mathbin{\mathbf{mod}} x \eq \zero)~
	\mathbf{then}~\sieve~x~(\pright{w})~\NL{\mskip80mu}
	\mathbf{else}~\pair{\pleft{w}}{\sieve~x~(\pright{w})} \\
\primes~w
    &\equiv& \pair{\pleft{w}}{\primes~(\sieve~(\pleft w)\,(\pright{w}))}
\end{eqnarray*}
Then, \(\primes~(\enum~\two)\) generates the stream of
all the prime numbers in ascending order, which can be formally verified
by the extended typing system.
Let \(m\), \(n\), \(k\) and \(l\) range over non-negative integers, and
let \(\preprime\), \(\primep\), \(\npp\), and \(\np\)
be defined in the arithmetic as follows.
\begin{eqnarray*}
    \preprime(m,\:n) &\equiv& 1 < m \conj
	\forall{k}.\;\forall{l}.\;
	    1 < k < m \conj k < n \impl m \not= k \cdot l \\
    \primep(n) &\equiv& \preprime(n,\:n) \\
    \npp(m,\:n) &\equiv& \min\zfset{k}{m \le k \conj \preprime(k,\:n)} \\
    \np(n) &\equiv& \min\zfset{k}{n \le k \conj \primep(k)}
\end{eqnarray*}
Note that \(\primep(n)\) means that \(n\) is a prime number, and
\(\np(n)\) represents the least prime number that is greater than
or equal to \(n\).
We use these two predicates and two functions for annotating type expression,
and define three types \(\PPS\), \(\PS\) and \(\Sieve\)
as follows.
\begin{eqnarray*}
\PPS(m,\:n) &\equiv& \nat(m)\conj \preprime(m,\:n)
	\times \NL{\mskip80mu}
	    \O\O^{\npp(m+1,\:n) \dminus m \dminus 1}\,\PPS(\npp(m+1,\,n),\,n) \\
\PS(n) &\equiv& \nat(n) \conj{\primep(n)}
      \times \O\O^{\np(n+1) \dminus n \dminus 1}\,\PS(\np(n+1)) \\
\Sieve(m,\:n) &\equiv&
	\nat(n) \conj \primep(n) \Impl \PPS(m,\,n) \Impl \\
	&& \mskip80mu \O^{\npp(m,\:np(n+1)) \dminus m}\,
	      \PPS(\npp(m,\:\np(n+1)),\,\np(n+1))
\end{eqnarray*}
The two recursive types, namely \(\PPS\) and \(\PS\),
are defined without explicit occurrence of \(\fx\) for readability.
The streams of all the prime numbers in ascending order
should have the type \(\PS(2)\).
Actually, we can derive the following typing.
\begin{eqnarray*}
    \enum &:& \forall{n}.\;\nat(n) \conj 1 < n \Impl \PPS(n,~2) \\
    \sieve &:& \forall m.\;\forall n.\;\Sieve(m,\:n) \\
    \primes &:& \forall n.\,\PPS(n,\,n) \Impl \PS(n) \\
    \primes~(\enum~\two) &:& \PS(2)
\end{eqnarray*}

\ifdetail
\begin{proof}\leavevmode
\paragraph{Proof of\/ $\protect\enum : \protect\forall{n}.\;\protect\nat(n)
    \protect\conj 1 < n \protect\Impl \protect\PPS(n,~2)$}
Suppose that
\begin{eqnarray}
\label{enum-01}
    \enum &:& \O\,\forall{n}.\;\nat(n) \conj 1 < n \protect\Impl \PPS(n,~2),
	~\mbox{and} \\
\label{enum-02}
    x &:& \nat(n) \conj 1 < n. \\
\na{Then, since \(1 < n\) implies \(\preprime(n,\:2)\),}
\label{enum-03}
    x &:& \nat(n) \conj \preprime(n,\:2).
\end{eqnarray}
The typing of \(\enum\) can be shown as follows.
\begin{Eqnarray*}
\enum &:& \forall{n}.\; \O (\nat(n) \conj 1 < n \Impl \PPS(n,~2))
	& (from (\ref{enum-01})
	    since \(\forall{n}.\,\O A \eqtyp \O\,\forall{n}.\,A\)) \nonumber \\
\enum &:& \O (\nat(n+1) \conj 1 < n+1 \Impl \PPS(n+1,~2))
	& (by \r{{\forall}\mbox{E}}) \nonumber \\
    &\subtyp& \O (\nat(n+1) \conj 1 < n+1) \Impl \O \PPS(n+1,~2)
	& (by \r{{\peqtyp}\mbox{-{\bf K}/{\bf L}}} or
	    \r{{\subtyp}\mbox{-\bf K}}). \nonumber \\
\na{Since \(x + 1 : \O(\nat(n+1) \conj 1 < n+1)\) from (\ref{enum-02}),}
\enum~(x \pl \one) &:& \O \PPS(n+1,~2) & (by \r{{\Impl}\mbox{E}}) \\
\pair{x}{\enum~(x \pl \one)} &:&
	\nat(n) \conj \preprime(n,\:2) \times\O \PPS(n+1,~2)
	    & (by \r{{\Conj}\mbox{I}} with (\ref{enum-03})) \nonumber \\
    &\eqtyp& \nat(n) \conj \preprime(n,\:2) \Conj
    \O\O^{\npp(n+1,\:2) \dminus n \dminus 1}\,\PPS(\npp(n+1,\,2),\,2)
	\mskip-500mu \nonumber \\
    &&& (since \(\npp(n+1,\:2) = n+1\)) \nonumber \\
    &\eqtyp& \PPS(n,\;2) & (by the definition of \(\PPS\)) \nonumber \\
\enum &:& \nat(n) \conj 1 < n \Impl \PPS(n,~2)
	& (by \r{{\Impl}\mbox{I}} with (\ref{enum-02})) \nonumber
\end{Eqnarray*}
Hence, by \r{{\forall}\mbox{I}}, \r{{\Impl}\mbox{I}} with (\ref{enum-01}),
	and by \(\Y\), we get
\begin{eqnarray}
\label{enum-04}
    \enum &:& \forall{n}.\; \nat(n) \conj 1 < n \Impl \PPS(n,~2).
\end{eqnarray}

\paragraph{Proof of\/ $\protect\sieve : \protect\forall m.\;
    \protect\forall n.\;\protect\Sieve(m,\:n)$}
In general, the following propositions hold.
\begin{eqnarray}
\label{sieve-01}
    && m+1 \le \npp(m+1,\,n) \\
\label{sieve-02}
    && k \le l ~\mbox{implies}~\npp(\npp(m,\,k),\:l) = \npp(m,\:l) \\
\label{sieve-03}
    && \preprime(m,\,n+1) \impl \preprime(m,\,\np(n+1))
\end{eqnarray}
Assume the following.
\begin{eqnarray}
\label{sieve-04}
    \sieve &:& \O\,\forall m.\;\forall n.\;\Sieve(m,\:n) \\
\label{sieve-05}
    x &:& \nat(n) \conj prime(n) \\
\label{sieve-06}
    w &:& \PPS(m,\:n)
\end{eqnarray}
Then, by the definitions of \(\Sieve(m,\:n)\) and \(\PPS(m,\:n)\),
\begin{eqnarray}
    \sieve &:&\O\,\forall m.\;\forall n.\;
	    \nat(n) \conj \primep(n) \Impl \PPS(m,\,n) \nonumber \\
\label{sieve-07}
	&& \mskip100mu {}\Impl \O^{\npp(m,\:np(n+1)) \dminus m}\,
	      \PPS(\npp(m,\:\np(n+1)),\,\np(n+1)) \\
\label{sieve-08}
    \pleft{w} &:& \nat(m) \conj \preprime(m,\:n),~\mbox{and} \\
\label{sieve-09}
    \pright w &:&
	\O\O^{\npp(m+1,\:n) \dminus m \dminus 1}\,\PPS(\npp(m+1,\,n),\,n).
\end{eqnarray}
Note also that since \(n < \np(n+1)\), by (\ref{sieve-02}),
\begin{eqnarray}
\label{sieve-10}
	\npp(\npp(m+1,\,n),\:\np(n+1)) &=& \npp(m+1,\:\np(n+1)).
\end{eqnarray}
Hence, from (\ref{sieve-05}), (\ref{sieve-07}) and (\ref{sieve-09}),
\begin{Eqnarray}
    \sieve~x~(\pright w) &:&
	\O\O^{\npp(\npp(m+1,\,n),\:np(n+1)) \dminus \npp(m+1,\,n)}
			    \nonumber \\[-6pt]
	      &&\mskip20mu \PPS(\npp(\npp(m+1,\,n),\:\np(n+1)),\,\np(n+1))
	& (by \r{{\forall}\mbox{E}}), \r{\subtyp}
		and \r{{\Impl}\mbox{E}})\hskip-20pt \nonumber \\[3pt]
    &\subtyp& \O\O^{\npp(\npp(m+1,\,n),\:np(n+1)) \dminus m \dminus 1}
			    \nonumber \\[-6pt]
	      &&\mskip20mu \PPS(\npp(\npp(m+1,\,n),\:\np(n+1)),\,\np(n+1))
		    & (by (\ref{sieve-01})) \nonumber \\
    &\eqtyp& \O\O^{\npp(m+1,\,np(n+1)) \dminus m \dminus 1}
			    \nonumber \\[-6pt]
\label{sieve-11}
	      &&\mskip20mu \PPS(\npp(m+1,\,\np(n+1)),\,\np(n+1))
		      & (by (\ref{sieve-10})).
\end{Eqnarray}
If \(n\) is a factor of \(m\), then \((\pleft{w}) \mathbin{\mathbf{mod}} x=0\),
and \(\npp(m+1,\:\np(n+1)) = \npp(m,\:np(n+1))\)
	since \(n < \np(n+1)\).
Hence, in this case,
\begin{Eqnarray*}
    \sieve~x~(\pright w) &:&
    \O\O^{\npp(m,\:np(n+1)) \dminus m \dminus 1}
		  \PPS(\npp(m,\:\np(n+1)),\,\np(n+1)) \\
     &\eqtyp& \O^{\npp(m,\:np(n+1)) \dminus m}
		  \PPS(\npp(m,\:\np(n+1)),\,\np(n+1)) \\
	    &&& \hskip-100pt (since \(m + 1 \le \npp(m,\:np(n+1))\)).
\end{Eqnarray*}
On the other hand,
if \(n\) is not a factor of \(m\),
then \((\pleft{w}) \mathbin{\mathbf{mod}} x=0\) does not hold.
In this case, we get
\(\preprime(m,\,n+1)\) from (\ref{sieve-08}); and hence, by (\ref{sieve-03}), 
\begin{Eqnarray}
    \label{sieve-12}
    && \preprime(m,\,\np(n+1)), \\
\na{which also implies}
    \label{sieve-13}
    && \npp(m,\:\np(n+1)) = m.
\end{Eqnarray}
Therefore,
\begin{Eqnarray*}
    \pleft w &:& \nat(m) & (from (\ref{sieve-08})) \\
    \pleft w &:& \nat(m) \conj \preprime(m,\,\np(n+1))
	    & (by (\ref{sieve-12})). \\
\na{Hence,}
    \pair{(\pleft w)}{\sieve~x~(\pright w)} \mskip-160mu \\
    &:& \nat(m) \conj \preprime(m,\:\np(n+1))
	\Conj \O\O^{\npp(m+1,\,np(n+1)) \dminus m \dminus 1} \mskip-300mu \\
    && \mskip40mu \PPS(\npp(m+1,\,\np(n+1)),\,\np(n+1))
	& (by \r{{\Conj}\mbox{I}} with (\ref{sieve-11})) \\
    &\eqtyp& \PPS(m,\:\np(n+1))
	& (by the definition of \(\PPS\)) \\
    &\eqtyp& \PPS(\npp(m,\:\np(n+1)), \np(n+1))
	& (by (\ref{sieve-13})) \\
    &\subtyp& \O^{\npp(m,\:np(n+1)) \dminus m}
		  \PPS(\npp(m,\:\np(n+1)),\,\np(n+1)).
\end{Eqnarray*}
Thus, we get the same type for both the \(\mathbf{then}\) and
the \(\mathbf{else}\) clauses.
Therefore, by the definition of \(\sieve\),
\begin{Eqnarray*}
    \sieve~x~w &:& \O^{\npp(m,\:np(n+1)) \dminus m}
		  \PPS(\npp(m,\:\np(n+1)),\,\np(n+1)) \\
    \sieve &:& \nat(n) \conj prime(n) \Impl \PPS(m,\:n) \Impl
			\O^{\npp(m,\:np(n+1)) \dminus m} \\
    && \mskip40mu \PPS(\npp(m,\:\np(n+1)),\,\np(n+1))
		& (by \r{{\Impl}\mbox{I}} with
		    (\ref{sieve-05}) and (\ref{sieve-06})) \\
	&\eqtyp& \Sieve(m,\:n)
		& (by the definition of \(\Sieve\))
\end{Eqnarray*}
Hence, by \r{{\forall}\mbox{I}}, \r{{\Impl}\mbox{I}} with (\ref{sieve-04}),
	and by \(\Y\), we get
\begin{eqnarray}
\label{sieve-14}
    \sieve &:& \forall{m}.\;\forall{n}.\;\Sieve(m,\:n).
\end{eqnarray}

\paragraph{Proof of\/ $\protect\primes : \protect\forall n.\;
    \protect\PPS(n,\:n)\protect\Impl \protect\PS(n)$}
In general, the following two propositions hold.
\begin{eqnarray}
\label{primes-01}
    && \preprime(n,\:n) \impl \primep(n) \\
\label{primes-02}
    && \npp(\npp(n+1,\:n),\:\np(n+1)) = \np(n+1)
\end{eqnarray}
Assume the following.
\begin{eqnarray}
\label{primes-03}
    \primes &:& \O\,\forall n.\;\PPS(n,\:n) \Impl \PS(n) \\
\label{primes-04}
    w &:& \PPS(n,\:n)
\end{eqnarray}
Hence, from (\ref{primes-04}) and the definition of \(\PPS\),
we get the following.
\begin{Eqnarray}
    \pleft w &:& \nat(n) \conj \preprime(n,\:n) \nonumber \\
\label{primes-05}
    \pleft w &:& \nat(n) \conj \primep(n) & (by (\ref{primes-01})) \\
\label{primes-06}
    \pright w &:& \O\O^{\npp(n+1,\:n)\dminus n \dminus 1}
	\PPS(\npp(n+1,\:n),\:n)
\end{Eqnarray}
On the other hand,
\begin{Eqnarray*}
\sieve &:& \Sieve(\npp(n+1,\:n),\:n)
	& (by \r{{\forall}\mbox{E}} from (\ref{sieve-14})) \\
    &\eqtyp& \nat(n) \conj prime(n) \Impl \PPS(\npp(n+1,\:n),\:n) \Impl
	    \O^{\npp(\npp(n+1,\:n),\:np(n+1)) \dminus \npp(n+1,\:n)}
		\mskip-200mu \\
	&& \mskip20mu
		\PPS(\npp(\npp(n+1,\:n),\:\np(n+1)),\,\np(n+1))
	& (by the definition of \(\Sieve\)) \\
    &\eqtyp& \nat(n) \conj prime(n) \Impl \PPS(\npp(n+1,\:n),\:n) \Impl \\
	&& \mskip20mu
	    \O^{np(n+1) \dminus \npp(n+1,\:n)}\PPS(\np(n+1),\,\np(n+1))
	& (by (\ref{primes-02})).
\end{Eqnarray*}
Hence, the typing of \(\primes\) is established as follows.
\begin{Eqnarray}
\sieve~(\pleft w)~(\pright w) &:&
	    \O\O^{\npp(n+1,\:n)\dminus n \dminus 1}
	    \O^{np(n+1) \dminus \npp(n+1,\:n)}\PPS(\np(n+1),\,\np(n+1))
		\mskip-500mu\nonumber  \\
	&& \span\hfill (by \r{\subtyp}, \r{{\Impl}\mbox{E}} with
		(\ref{primes-05}) and (\ref{primes-06})) \nonumber \\
     &\eqtyp& \O\O^{np(n+1)\dminus n \dminus 1}\PPS(\np(n+1),\,\np(n+1))
		\mskip-30mu\nonumber \\
	&&\span\hfill (since \(\npp(n+1,\:n) \ge n + 1\)
		    and \(np(n+1) \ge \npp(n+1,\:n)\)) \nonumber \\[5pt]
\primes~(\sieve~(\pleft w)~(\pright w)) &:&
	 \O\O^{np(n+1)\dminus n \dminus 1}\PS(\np(n+1))
	& (by \r{\subtyp}, \r{{\Impl}\mbox{E}}
	    with (\ref{primes-03}) \nonumber \\[5pt]
\primes~w &:& \nat(n) \conj \primep(n) \Conj
	 \O\O^{np(n+1)\dminus n \dminus 1}\PS(\np(n+1))
	     \mskip-300mu\nonumber \\
	&&\span\hfill (by \r{{\Conj}\mbox{I}} with (\ref{primes-05})
			    and the definition of \(\primes\)) \nonumber \\
    &\eqtyp& \PS(n) & (by definition of \(\PS\)) \nonumber \\
\primes &:& \PPS(n,\:n) \Impl \PS(n) & (by \r{{\Impl}\mbox{I}}) \nonumber
\end{Eqnarray}
Hence, by \r{{\forall}\mbox{I}},
    \r{{\Impl}\mbox{I}} with (\ref{primes-03}), and by \(\Y\), we get
\begin{eqnarray}
\label{primes-07}
    \primes &:& \forall{n}.\;\PPS(n,\:n) \Impl \PS(n).
\end{eqnarray}

\paragraph{Proof of\/ $\protect\primes~(\protect\enum~\protect\two)
    : \protect\PS(2)$}
\begin{Eqnarray*}
    \enum &:& \nat(2) \conj 1 < 2 \Impl \PPS(2,\:2)
	& (by \r{{\forall}\mbox{E}} from (\ref{enum-04})) \\
    \enum~\two &:& \PPS(2,\:2)
	& (by \r{{\Impl}\mbox{E}}, since \(\two : \nat(2) \conj 1<2\)) \\
    \primes~(\enum~\two) &:& \PS(2,\:2)
	& (by \r{{\forall}\mbox{E}},
	    \r{{\Impl}\mbox{E}} with (\ref{primes-07}))
\end{Eqnarray*}
\Qed
\end{proof}
\fi 

As can be seen in this example,
the intrinsic complexity of ensuring the convergence, or the productivity,
of programs is not mitigated by use of the approximation modality.
Its importance consists in the facts that (a)
it enables us to discuss the convergence without going into the computational
behavior of programs, and by which (b) it is possible to keep the modularity
of programs.
The derivations above can be done even if
the left hand sides of ``\(:\)'' in the typing judgments are hidden, and
through the proofs-as-programs notion, the derivation itself can be
regarded as an executable program.

\Subsection{The $\Nat(n)$-example}

We now reconsider the example of object-oriented natural numbers
with an addition method.
We revise the definition of \(\Nat(n)\) as follows.
\begin{eqnarray*}
\Nat(n) &\equiv& ((n = 0) + (n > 0 \conj \O\Nat(n-1)) \times \NL{\mskip80mu}
    (\forall m.\: \O\Nat(m) \Impl \O\Nat(n+m)))
\end{eqnarray*}
Then, the specifications of \(\mathbf{add}\) and \(\mathbf{add}'\) are
now different as follows.
\begin{eqnarray}
\label{Nat-01}
    \mathbf{add} &:&
    \forall n.\: \forall m.\: \Nat(n) \Impl \O\Nat(m) \Impl \O\Nat(n+m) \\
\label{Nat-02}
    \mathbf{add}' &:&
    \forall n.\: \forall m.\: \O\Nat(n) \Impl \Nat(m) \Impl \O\Nat(n+m)
\end{eqnarray}
We can show \(\suc : \forall n.\, \Nat(n) \Impl \Nat(n+1)\)
    by deriving \(\suc~x: \Nat(n+1)\)
    from the following assumptions.
\begin{Eqnarray}
\label{Nat-03}
    \suc &:& \O\,\forall n.\, \Nat(n) \Impl \Nat(n+1) \\
\label{Nat-04}
    x &:& \Nat(n)
\end{Eqnarray}
In fact, from \(x: \Nat(n)\), we get
\begin{Eqnarray}
\label{Nat-05}
    \iright{x} &:& (n+1 = 0) + (n+1 > 0 \conj  \O\Nat(n+1-1)).
\end{Eqnarray}
Furthermore, if \(y : \O\Nat(m)\), then
\begin{Eqnarray*}
    \suc &:& \O\Nat(m) \Impl \O\Nat(m+1)
	    & (by \r{{\forall}\mbox{E}} from (\ref{Nat-03}),
		    and \r{\rsubtyp}) \\
    \suc\;y &:& \O\Nat(m+1)
	    & (by \r{{\Impl}\mbox{E}}) \\
    \mathbf{add}\;x\;(\suc\;y) &:& \O \Nat(n+m+1)
	& (by (\ref{Nat-01}) and (\ref{Nat-04})) \\
    &\eqtyp& \O \Nat(n+1+m) & (since \(\Mtight n+m+1 = n+1+m\)). \\
\na{Hence,}
    \lam{y}\mathbf{add}\;x\;(\suc\;y) &:& \O \Nat{m} \Impl \O \Nat(n+1+m)
	    & (by \r{{\Impl}\mbox{I}}) \\
    \lam{y}\mathbf{add}\;x\;(\suc\;y)
	&:& \forall{m}.\:\O \Nat(m) \Impl \O \Nat(n+1+m)
	    & (by \r{{\forall}\mbox{I}}) \\
    \pair{\iright{x}}{\lam{y}{\mathbf{add}\;x\;(\suc\;y)}}
	&:& \Nat(n+1)
	& (by \r{{\Conj}\mbox{I}} with (\ref{Nat-05})).
\end{Eqnarray*}
Thus, we get \(\suc : \forall{n}.\:\Nat(n) : \Nat(n+1)\).
However, on the other hand, under similar assumptions, we can only get
\[
    \mathbf{add}'\:x\;(\suc'\:y) : \O\O\Nat(n+m+1),
\]
from (\ref{Nat-02}),
and fail to derive \(\suc' : \forall n.\, \Nat(n) \Impl \Nat(n+1)\).

\endgroup

\Section{Modal logics behind {\lA}}\label{logic-sec}

In this section, we consider {\lA} as a modal logic
by ignoring left hand sides of ``\(:\)'' from typing judgments, and
show that it corresponds to an intuitionistic version of
the logic of provability {\GL} (cf. \cite{boolos}).
The modal logic {\GL} is
also denoted by {\bf G} (for G\"{o}del), {\bf L} (for L\"{o}b),
{\bf PrL}, {\bf KW}, or {\bf K4W}, in the literature.

\begin{definition}[Formal systems {\miK4} and {\LAm}]
    \ilabel{miK4-def}{miK4@\protect\miK4}
    \ilabel{LAm-def}{LA-mu@\protect\LAm}
    \ilabel*{0 yields-LA@$\protect\typ{\G}{A}$}
Regarding type expressions as logical formulae,
we define a formal system {\miK4},
the minimal and implication fragment of the modal logic {\bf K4},
by the following inference rules,
where \(\G\) denotes a finite set of formulae.
\[
\ifr{\mbox{assump}}
    {}
    {\typ{\G \cup \{\,A\,\}}{A}}
\ifnarrow\mskip80mu\else\mskip120mu\fi
\ifr{\mbox{nec}}
    {\typ{\G_1}{A}}
    {\typ{\O \G_1 \cup \G_2}{\O A}}
\]
\[
\ifr{\mbox{\bf 4}}
    {\typ{\G}{\O A}}
    {\typ{\G}{\O\O A }}
\ifnarrow\mskip30mu\else\mskip80mu\fi
\ifr{\Impl\,\mbox{I}}
    {\typ{\G \cup \{\,A\,\}}{B}}
    {\typ{\G}{A \Impl B}}
\ifnarrow\mskip30mu\else\mskip80mu\fi
\ifr{\Impl\mbox{E}}
    {\typ{\G_1}{A \Impl B}
     & \typ{\G_2}{A}}
    {\typ{\G_1 \cup \G_2}{B}}
\]
Note that the following rule is derivable in {\miK4}.
\[
\ifr{\mbox{\bf K}}
    {\typ{\G}{\O(A \Impl B)}}
    {\typ{\G}{\O A \Impl \O B}}
\]
Similarly, we define {\LAm} as the formal system
obtained from \miK4 by adding the following four additional rules.
\[
\ifr{\mbox{fold}}
    {\typ{\G}{A[\fix{X}A/X]}}
    {\typ{\G}{\fix{X}A}}
\ifnarrow\mskip60mu\else\mskip30mu\fi
\ifr{\mbox{unfold}}
    {\typ{\G}{\fix{X}A}}
    {\typ{\G}{A[\fix{X}A/X]}}
\ifnarrow\]\[\else\mskip30mu\fi
\ifr{\mbox{\bf L}}
    {\typ{\G}{\O A \Impl \O B}}
    {\typ{\G}{\O(A \Impl B)}}
\ifnarrow\mskip60mu\else\mskip30mu\fi
\ifr{\mbox{approx}}
    {\typ{\G}{A}}
    {\typ{\G}{\O A}}
\]
\end{definition}
We occasionally use \(\typ{}{A}\) as a synonym to \(\typ{\{\}\!}{A}\).
Note that the \r{\mbox{\bf 4}}-rule is redundant to {\LAm}
since it has \r{\mbox{approx}}.
When {\miK4} is considered, we usually suppose that formulae, i.e.,
type expressions, are finite, that is, do not include any occurrence of $\fx$.
It will be shown, later in Theorem~\ref{logical-equiv-theorem},
that the modal logic {\LAm} corresponds to the typing system {\lA}.
In {\LAm}, \r{\mbox{fold}}, \r{\mbox{unfold}},
\r{\mbox{\bf K}}, \r{\mbox{\bf L}} and \r{\mbox{approx}} substitute for
the \r{{\rsubtyp}}-rule of {\lA}.
The formal system {\LAm} does not have
a rule directly corresponding to \r{{\eqtyp}\mbox{-uniq}} of {\lA}.
This is because the logical equivalence between formulae is a weaker notion
than the one as types, i.e., sets of realizers,
and is derivable from other inference rules.
The first thing we should confirm is the following proposition.

\begin{proposition}\label{LAm-lA}
If\/ \(\typ{\{\,A_1,\,\ldots,\,A_n\,\}}{B}\) is derivable in\/ {\LAm}, then
\(\typ{\{\,x_1 : A_1,\,\ldots,\,x_n : A_n\,\}}{M:B}\) is derivable
in {\lA} for some \(\lambda\)-term \(M\) and
distinct individual variables \(x_1\), \(\ldots\), \(x_n\)
such that \(\FV{M} \subseteq \{\,x_1,\,\ldots,\,x_n\,\}\).
\end{proposition}
\begin{proof}
Straightforward induction on the derivation.
Use Proposition~\ref{lA-nec-redundant} to handle \r{\mbox{nec}}.
The rules \r{\mbox{assump}}, \r{{\Impl}\mbox{I}} and \r{{\Impl}\mbox{E}}
have direct correspondents in {\lA}, which involve term construction.
The other rules, namely \r{\mbox{fold}},
\r{\mbox{unfold}},
\r{\mbox{\bf L}} and
\r{\mbox{approx}}, are imitable by the \r{{\rsubtyp}}-rule of {\lA}
since it has
\r{{\eqtyp}\mbox{-fix}},
\r{{\peqtyp}\mbox{-{\bf K}/{\bf L}}} and
\r{{\rsubtyp}\mbox{-approx}}.
\qed\CHECKED{2014/07/10, 07/17}
\end{proof}

To show the opposite, we will employ the completeness of {\LAm} with respect
to a certain interpretation of formulae.
As a preparation for that, we first show that in {\LAm},
\begin{Enumerate}
\item \((\O A \Impl A) \Impl A\),
\item \(A\), for every {\tvariant} \(A\), and
\item \(A \eqtyp B\) implies \(A \Impl B\) and \(B \Impl A\).
\end{Enumerate}

\begin{proposition}\label{LAm-Y}
The following inference rule is derivable in {\LAm}.
\[
\ifr{\mbox{\bf Y}}
    {\typ{\G}{\O A\Impl A}}
    {\typ{\G}{A}}
\]
\end{proposition}
\begin{proof}
It suffices to show that
\(\typ{}{(\O A \Impl A) \Impl A}\) is derivable in {\LAm} for every \(A\).
In fact,
we have a derivation corresponding to the one of Example~\ref{Y-derivable}.
Let a formula \(B = \fix{X}\O X \Impl A\) and
a derivation \(\P\) as follows.
\[
\ifnarrow\lower 22pt\else\lower10pt\fi
\hbox{\(\Pi = {}\)}\;
\begin{array}[c]{c}
	\ifr{\Impl\mbox{I}}
	    {\ifr{\Impl\mbox{E}}
		{\ifr{\mbox{assump}}
			{}
			{\typ{\O A \Impl A}{\O A \Impl A}}
		 & \ifnarrow
		     \mskip-220mu
		     \def\Dr#1{\stack{\mskip70mu#1}{\Derive{30}}}%
		 \else
		     \def\Dr#1{#1}%
		 \fi
		 \Dr{\ifr{\Impl\mbox{E}}
			{\ifr{\mbox{\bf K}}
			    {\ifr{\mbox{nec}}
				{\ifr{\mbox{unfold}}
				    {\ifr{\mbox{assump}}{}{\typ{B}{B}}}
				    {\typ{B}{\O B\Impl A}}}
				{\typ{\O B}{\O(\O B\Impl A)}}}
			    {\typ{\O B}{\O\O B\Impl \O A}}
			& \ifr{\mbox{approx}}
				{\ifr{\mbox{assump}}{}{\typ{\O B}{\O B}}}
				{\typ{\O B}{\O\O B}}}
			{\typ{\O B}{\O A}}}}
		{\typ{\O A \Impl A,\,\O B}{A}}}
	    {\typ{\O A \Impl A}{\O B \Impl A}}
\end{array}
\]
Then we can derive \(\typ{}{(\O A \Impl A) \Impl A}\) as follows.
\vskip-15pt
\[
\ifr{\Impl\mbox{I}}
    {\ifr{\Impl\mbox{E}}
	{\stack{\derive{\Pi}}
		{\typ{\O A \Impl A}{\O B \Impl A}}
	& \ifr{\mbox{approx}}
	      {\ifr{\mbox{fold}}
		   {\stack{\derive{\Pi}}
			  {\typ{\O A \Impl A}{\O B \Impl A}}}
		   {\typ{\O A \Impl A}{B}}}
	       {\typ{\O A \Impl A}{\O B}}}
	{\typ{\O A \Impl A}{A}}}
    {\typ{}{(\O A \Impl A) \Impl A}}
\]
\Qed\CHECKED{2014/07/17}
\end{proof}

\begin{definition}
    \ilabel{LAm-typ-def}{0 yields-LA-mu@$\protect\LAmtyp{\G}{A}$}
We write \(\LAmtyp{\G}{\G'}\) if and only if
\(\typ{\G}{A}\) is derivable in\/ {\LAm} for every \(A \in \G'\), and
write \(\LAmtyp{\G}{A \Lequiv B}\)
if and only if \(\LAmtyp{\G}{\!\{A \Impl B,\,B \Impl A\}}\).
We also define a binary relation \(\LAmequiv\) between formulae,
i.e., type expressions,
as \(A \LAmequiv B\) if and only if \(\LAmtyp{}{A \Lequiv B}\).
\end{definition}

\begin{proposition}\label{Lequiv-Y}
If\/ \(\LAmtyp{\G\cup\{\,\O (A \Impl B),\,\O (B \Impl A)\,\}}{A \Lequiv B}\),
then \(\LAmtyp{\G}{A \Lequiv B}\).
\end{proposition}
\begin{proof}
Suppose that we have
\(\LAmtyp{\G\cup\{\,\O (A \Impl B),\,\O (B \Impl A)\,\}}{A \Lequiv B}\).
Since \(\typ{\G\cup\{\,\O (A \Impl B),\,\O (B \Impl A)\,\}}{A \Impl B}\) is
derivable, so is
\begin{eqnarray}
    \label{Lequiv-Y-01}
    \typ{\G\cup\{\,\O (B \Impl A)\,\}}{A \Impl B}
\end{eqnarray}
by Proposition~\ref{LAm-Y}, namely \r{\mbox{\bf Y}}.
Similarly,
\(\typ{\G\cup\{\,\O (A \Impl B)\,\}}{B \Impl A}\) is also derivable;
and hence, so is
\begin{eqnarray}
    \label{Lequiv-Y-02}
    \typ{\G\cup\{\,\O (A \Impl B)\,\}}{\O (B \Impl A)}
\end{eqnarray}
by \r{\mbox{approx}}.
Therefore, we can derive
\(\typ{\G\cup\{\,\O (A \Impl B)\,\}}{A \Impl B}\) from
(\ref{Lequiv-Y-01}) and (\ref{Lequiv-Y-02}); and hence,
so is \(\typ{\G}{A \Impl B}\) by \r{\mbox{\bf Y}}.
Symmetrically, \(\typ{\G}{B \Impl A}\) is also derivable in {\LAm}.
\qed\CHECKED{2014/06/25, 07/17}
\end{proof}

\begin{proposition}\label{LAm-basic}

\begin{Enumerate}\pushlabel
\item \label{LAm-weakening}
    Suppose that \(\LAmtyp{\G'}{\G}\).
    If\/ \(\typ{\G}{A}\) is derivable in {\LAm},  then
    so is \(\typ{\G'}{A}\).
\item \label{LAm-subst}
    If\/ \(\typ{\{\,A_1,\,A_2,\,\ldots,\,A_n\,\}}{B}\) is derivable
    in\/ {\LAm},
    then so is\/
    \(\typ{\{\,A_1[C/X],\,A_2[C/X],\,\ldots,\,A_n[C/X]\,\}}{B[C/X]}\)
    for every \(C\) and \(X\).
\end{Enumerate}
\end{proposition}
\begin{proof}
\ifdetail
For Item~1, let \(\G = \{\,B_1,\,B_2,\,\ldots,\,B_n\,\}\).
Since \(\typ{\G}{A}\) is derivable,
so is \(\typ{}{B_1\Impl B_2 \Impl \ldots B_n \Impl A}\)
by applying \r{{\Impl}\mbox{I}} \(n\) times.
On the other hand, since \(\typ{\G'}{B_i}\) is derivable for every \(i\),
\(\typ{\G'}{A}\) is derivable by \r{{\Impl}\mbox{E}}.
The second item can be shown by straightforward induction on the derivation.
\else
Straightforward.
\fi
\qed\CHECKED{2014/07/10, 07/17}
\end{proof}

\begin{proposition}\label{LAm-tail}
\(\typ{\{\,\tail{A}\,\}}{A}\) is derivable in\/ {\LAm} for every \(A\).
\end{proposition}
\begin{proof}
By induction on \(h(A)\), and by cases on the form of \(A\).

\Case{\(A = X\) for some \(X\).}
Trivial since \(\tail{A} = A\) by Definition~\ref{tail-def}.

\Case{\(A = \O B\) for some \(B\).}
In this case, \(\tail{A} = \O \tail{B}\).
By induction hypothesis,
\(\typ{\{\,\tail{B}\,\}}{B}\) is derivable; hence,
so is \(\typ{\{\,\O \tail{B}\,\}}{\O B}\) by \r{\mbox{nec}}.

\Case{\(A = B \Impl C\) for some \(B\) and \(C\).}
In this case, \(\tail{A} = \tail{C}\).
By induction hypothesis and Proposition~\ref{LAm-weakening},
\(\typ{\{\,B,\,\tail{C}\,\}}{C}\) is derivable, from which we get
\(\typ{\{\,\tail{C}\,\}}{B \Impl C}\) by \r{{\Impl\,}\mbox{I}}.

\Case{\(A = \fix{X}B\) for some \(X\) and \(B\).}
In this case, \(\tail{A} = \fix{X}\tail{B}\), and
by induction hypothesis,
\begin{eqnarray}
    \label{LAm-tail-01}
    \typ{\{\,\tail{B}\,\}}{B}~\mbox{is derivable.}
\end{eqnarray}
If \(X \not\in \FTV{\tail{B}}\), then
\(\typ{\{\,\tail{A}\,\}}{\tail{B}}\) is derivable from
\r{\mbox{assump}} by \r{\mbox{unfold}}
since \(\tail{B}[\tail{A}/X] = \tail{B}\); and hence,
so is \(\typ{\{\,\tail{A}\,\}}{B}\) by (\ref{LAm-tail-01})
and Proposition~\ref{LAm-weakening}.
In this case, we can derive
\(\typ{\{\,\tail{A}\,\}}{B[A/X]}\) by Proposition~\ref{LAm-subst},
from which we get
\(\typ{\{\,\tail{A}\,\}}{A}\) by \r{\mbox{fold}}.
On the other hand, in case of \(X \in \FTV{\tail{B}}\),
let \(\tail{B} =
    \O^{m_0}\fix{Y_1}\O^{m_1}\fix{Y_2}\O^{m_2}\ldots\fix{Y_n}\O^{m_n}X\),
where \(X \not\in \{\,Y_1,\,Y_2,\,\ldots,\,Y_n\,\}\),
and let \(k = m_0+m_1+m_2+\ldots+m_n\).
Note that \(A\) is a {\tvariant} and \(k > 0\) in this case,
since \(\tail{B}\) is also proper in \(X\) by Proposition~\ref{tail-proper}.
Hence, \(\typ{\{\,\O^k X\}}{\tail{B}}\) is derivable from \(\typ{\{\,X\,\}}{X}\)
by repeatedly applying \r{\mbox{nec}} or \r{\mbox{fold}}; and therefore, so is
\(\typ{\{\,\O X\}}{\tail{B}}\)
by Proposition~\ref{LAm-weakening} since
\(\typ{\{\,\O X\,\}}{\O^k X}\) is derivable from \(\typ{\{\,\O X\,\}}{\O X}\)
by applying \r{\mbox{approx}}
\(k-1\) times.
Hence, by (\ref{LAm-tail-01}),
\begin{eqnarray*}
    \typ{\{\,\O X\,\}}{B}~\mbox{is also derivable.}
\end{eqnarray*}
Therefore, by Proposition~\ref{LAm-subst}, so is
\(\typ{\{\,\O A\,\}}{B[A/X]}\), from which we get
\(\typ{\{\,\O A\,\}}{A}\) by \r{{\mbox{fold}}}.
Then, \(\typ{}{\O A \Impl A}\) is derivable
by \r{{\Impl}\mbox{I}}; and hence,
so is \(\typ{}{A}\) by Proposition~\ref{LAm-Y}.
Finally, we get
\(\typ{\{\,\tail{A}\,\}}{A}\) by Proposition~\ref{LAm-weakening}.
\qed\CHECKED{2014/06/25, 07/17}
\end{proof}

\begin{proposition}\label{LAm-tvariant}
If\/ \(A\) is a {\tvariant}, then
\(\typ{}{A}\) is derivable in {\LAm}.
\end{proposition}
\begin{proof}
Suppose that \(A\) is a {\tvariant}, i.e.,
\(\tail{A}
    = \O^{m_0}\fix{X_1}\O^{m_1}\fix{X_2}\O^{m_2}\ldots\fix{X_n}\O^{m_n}X_i\)
for some \(n\), \(m_0\), \(m_1\), \(m_2\), \(\ldots\), \(m_n\),
\(X_1\), \(X_2\), \(\ldots\), \(X_n\) and \(i\) such that
\(1 \le i \le n\),
\(X_i \notin \{\,\,X_{i{+}1},\,X_{i{+}2},\,\ldots,X_n\,\}\)
and \(m_i+m_{i{+}1}+m_{i{+}2}+\ldots+m_n \ge 1\).
By Proposition~\ref{LAm-tail},
it suffices to show that \(\typ{}{\tail{A}}\) is derivable.
Let \(k\) be the largest integer such that \(m_k > 0\), and let \(C\) and \(D\)
as \(C = \fix{X_{k+1}}\fix{X_{k+2}}\ldots\fix{X_n}X_i\)
and \(D = \O^{m_i}\fix{X_{i+1}}
    \O^{m_{i+1}}\fix{X_{i+2}}\ldots\O^{m_{k-1}}\fix{X_k}\O^{m_k}C\).
Note that
\(\tail{A} = \O^{m_0}\fix{X_1}\O^{m_1}\fix{X_2}\ldots\O^{m_{i-1}}\fix{X_i}D\).
We can derive \(\typ{\{\,X_i\,\}}{C}\)
from \(\typ{\{\,X_i\,\}}{X_i}\) by \r{\mbox{fold}}, from which
\(\typ{\{\,\O X_i\,\}}{\O C}\) follows by \r{\mbox{nec}}.
Hence,
continuing the derivation by repeatedly applying
\r{\mbox{approx}}
or \r{\mbox{fold}}, we get
\(\typ{\{\,\O X_i\}}{D}\).
Therefore, by Proposition~\ref{LAm-subst},
\(\typ{\{\,\O \fix{X_i}D\}}{D[\fix{X_i}D/X_i]}\) is derivable, and from which
\(\typ{\{\,\O \fix{X_i}D\}}{\fix{X_i}D}\) follows by \r{\mbox{fold}}.
Hence,
\(\typ{}{\fix{X_i}D}\) is derivable by Proposition~\ref{LAm-Y},
and therefore, so is \(\typ{}{\tail{A}}\) by applying
\r{\mbox{nec}} and \r{\mbox{fold}}.
\qed\CHECKED{2014/06/25, 07/17}
\end{proof}

\begin{proposition}\label{Lequiv-subst}
Let \([\vec{B}/\vec{X}]\) and \([\vec{C}/\vec{X}]\) be abbreviations
for \([B_1/X_1,\,B_2/X_2,\,\ldots,\,B_n/X_n]\)
and \([C_1/X_1,\,C_2/X_2,\,\ldots,\,C_n/X_n]\), respectively.
Suppose that for any \(i \in \{\,1,\,2,\,\ldots,\,n\,\}\), either
\begin{enumerate}[{\kern8pt}(a)]
\item
    \(\{\,B_i \Impl C_i,\,C_i \Impl B_i\,\} \subseteq \G\), or
\item
    \(\{\,\O(B_i \Impl C_i),\,\O(C_i \Impl B_i)\,\} \subseteq \G\) and
    \(A\) is proper in \(X_i\).
\end{enumerate}
Then, \(\LAmtyp{\G}{A[\vec{B}/\vec{X}] \Lequiv A[\vec{C}/\vec{X}]}\).
\end{proposition}
\begin{proof}
By induction on \(h(A)\), and by cases on the form of \(A\).
If \(A\) is a {\tvariant}, then so are both
\(A[\vec{B}/X]\) and \(A[\vec{C}/X]\) by Proposition~\ref{tvariant-subst1};
and hence, it is trivial by Proposition~\ref{LAm-tvariant}.
Therefore, we only consider the case that \(A\) is not.

\Case{\(A = Y\) for some \(Y\).}
Trivial if \(Y \not\in \{\,X_1,\,X_2,\,\ldots,\,X_n\,\}\).
If \(Y = X_i\) for some \(i\), then
\(A[\vec{B}/\vec{X}] = B_i\) and \(A[\vec{C}/\vec{X}] = C_i\); and hence,
\(\LAmtyp{\G}{A[\vec{B}/\vec{X}] \Lequiv A[\vec{C}/\vec{X}]}\)
from (a) since \(A\) is not proper in \(X_i\).

\Case{\(A = \O D\) for some \(D\).}
Let \(\G'= \{\,B_1 \Impl C_1,\,C_1 \Impl B_1,\,
    B_2 \Impl C_2,\,C_2 \Impl B_2,\,\ldots,\,
    B_n \Impl C_n,\,C_n \Impl B_n\,\}\).
By induction hypothesis,
\(\LAmtyp{\G\cup\G'}{D[\vec{B}/\vec{X}] \Lequiv D[\vec{C}/\vec{X}]}\);
and hence, both
\(\typ{\O \G\cup\O\G'}{\O (D[\vec{B}/\vec{X}] \Impl D[\vec{C}/\vec{X}])}\) and
\(\typ{\O \G\cup\O\G'}{\O (D[\vec{C}/\vec{X}] \Impl D[\vec{B}/\vec{X}])}\)
are derivable by \r{\mbox{nec}}; and so are
\(\typ{\O \G\cup\O\G'}{A[\vec{B}/\vec{X}] \Impl A[\vec{C}/\vec{X}]}\) and
\(\typ{\O \G\cup\O\G'}{A[\vec{C}/\vec{X}] \Impl A[\vec{B}/\vec{X}]}\)
by \r{\mbox{\bf K}}.
Therefore, since \(\LAmtyp{\G}{\O \G\cup\O \G'}\) by \r{\mbox{approx}},
\(\LAmtyp{\G}{A[\vec{B}/\vec{X}] \Lequiv A[\vec{C}/\vec{X}]}\)
by Proposition~\ref{LAm-weakening}.

\Case{\(A = D \Impl E\) for some \(D\) and \(E\).}
Since \(A\) is not a {\tvariant},
\(A\) is proper in \(X_i\) if and only if so are \(D\) and \(E\).
Therefore,
\(\LAmtyp{\G}{D[\vec{B}/\vec{X}] \Lequiv D[\vec{C}/\vec{X}]}\) and
\(\LAmtyp{\G}{E[\vec{B}/\vec{X}] \Lequiv E[\vec{C}/\vec{X}]}\)
by induction hypothesis; and hence,
\(\LAmtyp{\G}{A[\vec{B}/\vec{X}] \Lequiv A[\vec{C}/\vec{X}]}\).

\Case{\(A = \fix{Y}D\) for some \(Y\) and \(D\).}
We can assume that \(Y \not\in \FTV{\G} \cup \{\,\vec{X}\,\}\)
without loss of generality.
\begin{Eqnarray*}
    A[\vec{B}/\vec{X}] &=& \fix{Y}D[\vec{B}/\vec{X}] \\
	&\mskip-20mu\LAmequiv& D[\vec{B}/\vec{X}][A[\vec{B}/\vec{X}]/Y]
	    & (by \r{\mbox{fold}} and \r{\mbox{unfold}}) \\
	&=& D[\vec{B}/\vec{X},A[\vec{B}/\vec{X}]/Y]
	    & (since \(Y \not\in \FTV{\G} \cup \{\,\vec{X}\,\}\))
\end{Eqnarray*}
Similarly, we get
\(A[\vec{C}/\vec{X}] \LAmequiv D[\vec{C}/\vec{X},A[\vec{C}/\vec{X}]/Y]\).
On the other hand, by Definition~\ref{proper-def},
\(A\) is proper in \(X_i\) if and only if
so is \(D\), since \(A\) is not a {\tvariant}.
Hence, by induction hypothesis,
\[
    \LAmtyp{\G \cup \{\,\O(A[\vec{B}/\vec{X}] \Impl A[\vec{C}/\vec{X}]),\,
	\O(A[\vec{C}/\vec{X}] \Impl A[\vec{B}/\vec{X}])\,\}}%
    {D[\vec{B}/\vec{X},A[\vec{B}/\vec{X}]/Y]
	\Lequiv D[\vec{C}/\vec{X},A[\vec{C}/\vec{X}]/Y]}.
\]
Therefore,
\(\LAmtyp{\G \cup \{\,\O(A[\vec{B}/\vec{X}] \Impl A[\vec{C}/\vec{X}]),\,
	\O(A[\vec{C}/\vec{X}] \Impl A[\vec{B}/\vec{X}])\,\}}%
    {A[\vec{B}/\vec{X}] \Lequiv A[\vec{C}/\vec{X}]}\); and hence,
\(\LAmtyp{\G}{A[\vec{B}/\vec{X}] \Lequiv A[\vec{C}/\vec{X}]}\)
by Proposition~\ref{Lequiv-Y}.
\qed\CHECKED{2014/06/25, 07/17}
\end{proof}

\begin{lemma}\label{eqtyp-Lequiv}
If\/ \(A \eqtyp B\), then \(A \LAmequiv B\).
\end{lemma}
\begin{proof}
By induction on the derivation of \(A \eqtyp B\),
and by cases on the rule applied last.
Most cases are straightforward.
Use Proposition~\ref{LAm-tvariant}
for the case of \r{{\eqtyp}\mbox{-}{\Impl}{\t}}.
The only interesting case is \r{{\eqtyp}\mbox{-uniq}}.
In this case, \(B = \fix{X}C\) for some \(X\) and \(C\) such that
\(A \eqtyp C[A/X]\) and \(C\) is proper in \(X\).
By Proposition~\ref{Lequiv-subst},
\[
\LAmtyp{\{\,\O (A \Impl B),\,\O(B \Impl A)\,\}}{C[A/X] \Lequiv C[B/X]}.
\]
Since \(A \LAmequiv C[A/X]\) by induction hypothesis,
and since \(B \LAmequiv C[B/X]\) by \r{\mbox{fold}} and \r{\mbox{unfold}},
we get
\[
\LAmtyp{\{\,\O (A \Impl B),\,\O(B \Impl A)\,\}}{A \Lequiv B}.
\]
Therefore, \(A \LAmequiv B\) by Proposition~\ref{Lequiv-Y}.
\qed\CHECKED{2014/06/25, 07/17}
\end{proof}

\Subsection{Kripke semantics of {\LAm}}

Now we turn to the semantics of {\LAm}, and proceed to show
the completeness.

\begin{definition}[{\iGL}-frames and {\LA}-frames]
    \ilabel{iGL-frame-def}{iGL-frame@\protect\iGL-frames}
    \ilabel*{frames!iGL-frame@\protect\iGL}
    \ilabel{LA-frame-def}{LA-frame@\protect\LA-frames}
    \ilabel*{frames!LA-frame@\protect\LA}
    \ilabel{iwf-frame-def}{intuitionistic well-founded frames}
    \ilabel*{frames!LA-frame@intuitionistic well-founded}
An\/ {\em intuitionistic well-founded frame}\/ is
a triple \(\tuple{\W,\,{\acc},\,{\mathrel{R}}}\),
which consists of a non-empty set \(\W\) of possible worlds and
two accessibility relations \({\acc}\) and \(\mathrel{R}\)
on \(\W\) such that
\begin{Enumerate}
    \item \(\acc\) is a (conversely) well-founded
	binary relation on \(\W\),
    \item \(\mathrel{R\,}\) is a transitive and reflexive
	binary relation on \(\W\), and
    \item \(p \mathrel{R} q \acc r\) implies \(p \acc r\).
\end{Enumerate}
An\/ {\em {\iGL}-frame}\/ is
an\/ intuitionistic well-founded frame that satisfies
the following condition.
\begin{Enumerate}
    \item[4.] \(\acc\) is transitive.
\end{Enumerate}
An\/ {\em {\LA}-frame}\/ is
an\/ intuitionistic well-founded frame that satisfies
the following two conditions.
\begin{Enumerate}
    \item[5.] \(p \acc q\) implies \(p \mathrel{R} q\).
    \item[6.] if \(p \acc q \mathrel{R} q'\), then there exists
	some \(r \in \W\) such that
	\begin{Enumerate}
	\item[(a)] \(p \mathrel{R} r \acc q'\), and
	\item[(b)] \(r \acc s\) implies \(q' \mathrel{R} s\)
	    for every \(s \in \W\).
	\end{Enumerate}
\end{Enumerate}
\end{definition}
Conditions~1 through 4 constitute the class of frames, to which
the intuitionistic variant of the logic of provability
is sound and complete (cf. \cite{ursini-au79,ursini-sl79}).
Condition~5 means that the interpretation is hereditary
with respect to the well-founded relation \(\acc\) as well as \(\mathrel{R}\).
Condition~6 corresponds to
Condition~2 of Definition~\ref{wf-frame-def}, which indicates that
\(\acc\) is {\em locally} linear.

\begin{definition}[Kripke semantics of the logics]
    \ilabel{LAm-semantics}{semantics!LA-mu@\protect\LAm}
Let \(\tuple{\W,\,{\acc},\,\mathrel{R}}\) be
an intuitionistic well-founded frame.
A mapping \(f\)
from \(\W\) to \(\{\,\mathbf{t},\,\mathbf{f}\,\}\) is
{\em hereditary} if and only if
\[
    \mbox{if}~ p \mathrel{R} q,~\mbox{then}~
        f(p) = \mathbf{t} ~\,\mbox{implies}\,~ f(q) = \mathbf{t}.
\]
A mapping \(\ttenv\) that assigns a mapping
(from \(\W\) to \(\{\,\mathbf{t},\,\mathbf{f}\,\}\)) to each
propositional variable, i.e., type variable, is called {\em a valuation}.
We say a valuation \(\ttenv\) is {\em hereditary} if and only if
\(\ttenv(X)\) is hereditary for every \(X\).
In this paper, only hereditary valuations will be considered.
We define a hereditary mapping \(\Il{A}^\ttenv\)
from \(\W\) to \(\{\,\mathbf{t},\,\mathbf{f}\,\}\)
for each formula \(A\) by extending \(\ttenv\) as follows,
    \ilabel{logic-models-def}%
	{0 models-logic@$\protect\models^{\protect\ttenv}_p\,A$}%
where we use \({\models}^{\ttenv}_p\,A\) to denote that
\(\Il{A}^{\ttenv}(p) = \mathbf{t}\).
\[
\renewcommand\arraystretch{1.5}
\begin{array}{l@{\kern10pt}l}
{\models}^\ttenv_p\,{A} & \mbox{(\(A\) is a {\tvariant})} \\
{\models}^\ttenv_p\,{X} \mathrel{~\mbox{iff}~~} \ttenv(X)(p) = \mathbf{t} \\
{\models}^\ttenv_p\,{\O A} \mathrel{~\mbox{iff}~~}
	{\models}^\ttenv_q\,{A}~~\mbox{for every}~q~\mbox{such that}~p \acc q
    & \mbox{(\(\,\O A\) is not a {\tvariant})} \\
{\models}^\ttenv_p\,{A \Impl B} \mathrel{~\mbox{iff}~~}
    {\models}^\ttenv_q\,{A}~~\mbox{implies}~~ {\models}^\ttenv_q\,{B}~
	\mbox{for every}~q~\mbox{such that}~p \mathrel{R} q
	    \ifnarrow\mskip-220mu\else\mskip-120mu\fi \\[-4pt]
    & \mbox{(\(A \Impl B\) is not a {\tvariant})} \\
{\models}^\ttenv_p\,{\fix{X}A} \mathrel{~\mbox{iff}~~}
	{\models}^\ttenv_p\,{A[\fix{X}A/X]}
    & \mbox{(\(\fix{X}A\) is not a {\tvariant})} \\
\end{array}
\]
\end{definition}
The interpretation \({\models}^{\ttenv}_p\,{A}\) is defined by induction
on \(\pair{p}{r(A)}\), where we use the ordering \(\sqsupset\)
defined as follows.
\[
    \pair{p}{A} \sqsupset \pair{q}{B}
	~~\mbox{iff}~~
	    p \mathrel{R} q ~\mbox{and}~r(A) > r(B),~\mbox{or}\,~
	    p \acc q.
\]
Note that \(\sqsupset\) is also well-founded, since
so is \(\acc\), and since Condition~3 of Definition~\ref{iGL-frame-def}
is satisfied.
    \ilabel{logic-relative-models-def}{0 models-logic relative@%
	$\protect\G\,\protect\models^{\protect\ttenv}_p\,A$}%
We write \(\G \models^\ttenv_p\, A\) if and only if
\({\models}^\ttenv_p\, A\) whenever
\({\models}^\ttenv_p\, B\) for every \(B \in \G\strut\).
It is easy to verify that for every \(p,\, q \in \W\), \(\ttenv\) and \(A\),
\begin{quote}
    if \(p \mathrel{R} q\), then
    \({\models}^\ttenv_p\,{A}\) implies \({\models}^\ttenv_q\,{A}\).
\end{quote}
By a discussion similar to Theorem~\ref{soundness-theorem},
we can also observe the soundness of {\lA} as a modal logic
with respect to this semantics of formulae.
The proof proceeds as follows.
In the sequel, we write \({\models}^\ttenv\,{A}\)
if and only if \({\models}^\ttenv_p\,{A}\) for every \(p \in \W\).

\begin{proposition}\pushlabel
Let \(\tuple{\W,\,{\acc},\,\mathrel{R}}\) be an\/ {\LA}-frame, and
\(\ttenv\) a hereditary valuation.
\begin{Enumerate}
\item \itemlabel{logic-subst-env}
    \({\models}^\ttenv_p\,{A[B/X]}\)
    iff\/ \({\models}^{\ttenv[\Il{B}^\ttenv/X]}_p\,{A}\).
\item \itemlabel{logic-proper-subst-lemma}
    Let\/ \(\strut\ttenv'\) be a hereditary valuation such that\/
    \(\ttenv(X)(q) = \ttenv'(X)(q)\)\/ for every \(X\) and \(q \opacc p\).
    If\/ for every \(X\), either (a) \(A\) is proper in \(X\),
    or (b)\/ \(\ttenv(X)(p)  = \ttenv'(X)(p)\),
    then \({\models}^\ttenv_p\,{A}\) iff\/ \({\models}^{\ttenv'}_p\,{A}\).
\item \itemlabel{logic-peqtyp-soundness}
    If\/ \(A \peqtyp B\), then
    \({\models}^\ttenv\,{A}\) iff\/ \({\models}^\ttenv\,{B}\).
\item \itemlabel{logic-psubtyp-soundness}
    Let \(\g =
	\{\,X_1\psubtyp Y_1,\,X_2\psubtyp Y_2,\,\ldots,\,X_n\psubtyp Y_n\,\}\).
    If\/ \(\subt{\g}{A \psubtyp B}\) is derivable, and
    \(\{\,X_i\,\}\mathrel{\models^\ttenv}{Y_i}\)
    for every \(i \in \{\,1,\,2,\,\ldots\,n\,\}\),
    then \(\{\,A\,\}\mathrel{\models^\ttenv}B\).
\item \itemlabel{lA-logic-soundness}
    If\/ \(\typ{\{\,x_1 : A_1,\,\ldots,\,x_n : A_n\,\}}{M:B}\) is derivable
    in {\lA}, then
    \(\{\,A_1,\,\ldots,\,A_n\,\}\mathrel{\models^\ttenv}B\).
\end{Enumerate}
\end{proposition}
\begin{proof}
The proofs are quite parallel to those
of Proposition~\ref{rlz-subst-env}, Lemma~\ref{rlz-proper-subst-lemma},
Theorems~\ref{peqtyp-soundness},
\ref{psubtyp-soundness} and \ref{soundness-theorem}, respectively.
The proofs of Items~\itemref{logic-subst-env}
and \itemref{logic-proper-subst-lemma} proceed
by induction on the ordering \(\sqsupset\) of \(\pair{p}{r(A)}\).
Those of Items~\itemref{logic-peqtyp-soundness},
\itemref{logic-psubtyp-soundness} and \itemref{lA-logic-soundness} are
by induction on the derivations.
\ifdetail

\paragraph{Proof of \protect\itemref{logic-subst-env}}
Induction step proceeds by cases on the form of \(A\), where
the proof does not depend on the hereditarity of \(\ttenv\).
Let \(\ttenv' = \ttenv[\Il{B}^\ttenv/X]\).
If \(A\) is a {\tvariant}, then so is \(A[B/X]\)
by Proposition~\ref{tvariant-subst1}; and hence,
both \({\models}^\ttenv\,{A[B/X]}\) and
    \({\models}^{\ttenv[\Il{B}^\ttenv/X]}\,{A}\)
by Definition~\ref{LAm-semantics}.
Therefore, we only consider the case when \(A\) is not.

\Case{\(A = Y\) for some \(Y\).}
If \(Y = X\), then \(A[B/X] = B\); and hence,
    \(\Il{A[B/X]}^\ttenv = \Il{B}^\ttenv = \Il{X}^{\ttenv'}\).
Otherwise, \(Y \not= X\) and \(A[B/X] = Y\).
Hence,
    \(\Il{A[B/X]}^\ttenv = \Il{Y}^\ttenv = \Il{Y}^{\ttenv'}\).

\Case{\(A = \O C\) for some \(C\).}
In this case,
\begin{Eqnarray*}
    {\models}^\ttenv_p\,{A[B/X]}
    &~\mbox{iff}~& {\models}^\ttenv_p\,{\O(C[B/X])} \\
    &~\mbox{iff}~& {\models}^\ttenv_q\,{C[B/X]}~
		    \mbox{for every}~q~\mbox{such that}~p \acc q
	& (by Definition~\ref{LAm-semantics}) \\
    &~\mbox{iff}~& {\models}^{\ttenv'}_q\,{C}~
		    \mbox{for every}~q~\mbox{such that}~p \acc q
	& (by induction hypothesis) \\
    &~\mbox{iff}~& {\models}^{\ttenv'}_p\,{\O C}
	& (by Definition~\ref{LAm-semantics}).
\end{Eqnarray*}

\Case{\(A = C \Impl D\) for some \(C\) and \(D\).}
In this case,
\begin{Eqnarray*}
    {\models}^\ttenv_p\,{A[B/X]} \mathrel{~\mbox{iff}~}
	{\models}^\ttenv_p\,{C[B/X] \Impl D[B/X]} \mskip-280mu \\
    &~\mbox{iff}~& {\models}^\ttenv_q\,{C[B/X]}~\mbox{implies}~
		    {\models}^\ttenv_q\,{D[B/X]}~\mbox{for every}~q~
			\mbox{such that}~p \mathrel{R} q
	& (by Definition~\ref{LAm-semantics}) \\
    &~\mbox{iff}~& {\models}^{\ttenv'}_q\,{C}~\mbox{implies}~
		    {\models}^{\ttenv'}_q\,{D}~\mbox{for every}~q~
			\mbox{such that}~p \mathrel{R} q
	& (by induction hypothesis) \\
    &~\mbox{iff}~& {\models}^{\ttenv'}_p\,{C \Impl D}
	& (by Definition~\ref{LAm-semantics}).
\end{Eqnarray*}

\Case{\(A = \fix{Y}C\) for some \(Y\) and \(C\).}
We can assume that \(Y \not\in \FTV{B} \cup \{\,X\,\}\)
without loss of generality.
\begin{Eqnarray*}
    {\models}^\ttenv_p\,{A[B/X]}
    &~\mbox{iff}~&{\models}^\ttenv_p\,{\fix{Y}C[B/X]}
	& (since \(Y \not\in \FTV{B} \cup \{\,X\,\}\)) \\
    &~\mbox{iff}~& {\models}^\ttenv_p\,{C[B/X][\fix{Y}C[B/X]/Y]}
	& (by Definition~\ref{LAm-semantics}) \\
    &~\mbox{iff}~& {\models}^\ttenv_p\,{C[\fix{Y}C/Y][B/X]}
	& (since \(Y \not\in \FTV{B} \cup \{\,X\,\}\)) \\
    &~\mbox{iff}~& {\models}^{\ttenv'}_p\,{C[\fix{Y}C/Y]}
	& (by induction hypothesis) \\
    &~\mbox{iff}~& {\models}^{\ttenv'}_p\,{\fix{Y}C}
	& (by Definition~\ref{LAm-semantics}).
\end{Eqnarray*}
Note that \(r(C[\fix{Y}C]) < r(\fix{Y}C)\)
by Proposition~\ref{rank-fix},
since \(\fix{Y}C\) is not a {\tvariant}.

\paragraph{Proof of \protect\itemref{logic-proper-subst-lemma}}
Induction step proceeds by cases on the form of \(A\).
Suppose that (a) or (b) holds for every \(X \in \TV\).
If \(A\) is a {\tvariant}, then
\({\models}^\ttenv_p\,{A}\) and \({\models}^{\ttenv'}_p\,{A}\)
by Definition~\ref{LAm-semantics}.
Hence, we assume that \(A\) is not.

\Case{\(A = Y\) for some \(Y\).}
In this case, \(A\) is not proper in \(Y\).
Hence,
\(\Il{A}^\ttenv(p) = \ttenv(Y)(p) = \ttenv'(Y)(p) = \Il{A}^{\ttenv'}(p)\)
from (b).

\Case{\(A = \O A'\) for some \(A'\).}
Since \(\ttenv(X)(q) = \ttenv'(X)(q)\) for every \(X\)
and \(q \opacc p\),
we get that \({\models}^\ttenv_q\,{A'}\) iff \({\models}^{\ttenv'}_q\,{A'}\)
for every \(q \opacc p\) by induction hypothesis.
Therefore,
\begin{Eqnarray*}
    {\models}^\ttenv_p\,{\O A'}
    &~\mbox{iff}~& {\models}^\ttenv_q\,{A'}~
		    \mbox{for every}~q~\mbox{such that}~p \acc q
	& (by Definition~\ref{LAm-semantics}) \\
    &~\mbox{iff}~& {\models}^{\ttenv'}_q\,{A'}~
		    \mbox{for every}~q~\mbox{such that}~p \acc q \\
    &~\mbox{iff}~& {\models}^{\ttenv'}_p\,{\O A'}
	& (by Definition~\ref{LAm-semantics})
\end{Eqnarray*}

\Case{\(A = B \Impl C\) for some \(B\) and \(C\).}
Note that \(C\) is not a {\tvariant} since \(A\) is not.
Therefore, \(r(B)\), \(r(C) < r(A)\) and (a) implies both \(B\) and \(C\) are
also proper in \(X\).
Hence, by induction hypothesis, we get that
\({\models}^\ttenv_q\,{B}\) iff \({\models}^{\ttenv'}_q\,{B}\), and that
\({\models}^\ttenv_q\,{C}\) iff \({\models}^{\ttenv'}_q\,{C}\), for every \(q\)
such that \(p \mathrel{R} q\).
Therefore, \({\models}^\ttenv_p\,{B \Impl C}\) iff
\({\models}^{\ttenv'}_p\,{B \Impl C}\)
by Definition~\ref{LAm-semantics}.

\Case{\(A = \fix{Y}C\) for some \(Y\) and \(C\).}
We can assume that \(X \not= Y\) without loss of generality.
Note that \(r(C[\fix{Y}C/Y]) < r(\fix{Y}C)\)
by Proposition~\ref{rank-fix}, and that
\(C[\fix{Y}C/Y]\) is proper in \(X\) if so is \(A\)
by Definition~\ref{proper-def} and Proposition~\ref{proper-subst1}.
Therefore,
\begin{Eqnarray*}
{\models}^\ttenv_p\,{\fix{Y}C}
    &~\mbox{iff}~& {\models}^\ttenv_p\,{C[\fix{Y}C/Y]}
	& (by Definition~\ref{LAm-semantics}) \\
    &~\mbox{iff}~& {\models}^{\ttenv'}_p\,{C[\fix{Y}C/Y]}
	& (by induction hypothesis) \\
    &~\mbox{iff}~& {\models}^{\ttenv'}_p\,{\fix{Y}C}
	& (by Definition~\ref{LAm-semantics})
\end{Eqnarray*}

\paragraph{Proof of \protect\itemref{logic-peqtyp-soundness}}
By induction on the derivation of \(A \peqtyp B\),
and by cases on the last rule applied in the derivation.
Suppose that \(A \peqtyp B \).
If either \(A\) or \(B\) is a {\tvariant}, then so are both by
Proposition~\ref{geqtyp-tvariant}; and hence,
\({\models}^\ttenv\,{A}\) and \({\models}^\ttenv\,{B}\)
by Definition~\ref{LAm-semantics}.
Therefore, we assume that neither is a {\tvariant}.

\Cases{\r{{\eqtyp}\mbox{-reflex}},
    \r{{\eqtyp}\mbox{-symm}} and \r{{\eqtyp}\mbox{-trans}}.}
Trivial.

\Case{\r{{\eqtyp}\mbox{-}{\O}}.}
In this case, there exist some \(A'\) and \(B'\) such that
\(A = \O A'\), \(B = \O B'\) and \(A' \peqtyp B'\).
By induction hypothesis,
\({\models}^\ttenv\,{A'}\) iff \({\models}^\ttenv\,{B'}\).
Therefore,
\({\models}^\ttenv\,{\O A'}\) iff \({\models}^\ttenv\,{\O B'}\)
by Definition~\ref{LAm-semantics}.

\Case{\r{{\eqtyp}\mbox{-}{\Impl}}.}
Similar to the previous case.

\Case{\r{{\eqtyp}\mbox{-}{\Impl}{\t}}.}
Impossible because we assumed that neither \(A\) nor \(B\) is a {\tvariant}.

\Case{\r{{\eqtyp}\mbox{-fix}}.}
Obvious from Definition~\ref{LAm-semantics}.

\Case{\r{{\eqtyp}\mbox{-uniq}}.}
There exist some \(X\) and \(C\) such that
\(B = \fix{X}C\), \(A \peqtyp C[A/X]\) and \(C\) is proper in \(X\).
By induction hypothesis,
\begin{eqnarray}\label{logic-eqtyp-soundness-01}
    {\models}^{\ttenv'}\,{A} &~\mbox{iff}~& {\models}^{\ttenv'}\,{C[A/X]}~~
    \mbox{for every \(\ttenv'\)}.
\end{eqnarray}
We show that \({\models}^\ttenv_p\,{A}\) iff \({\models}^\ttenv_p\,{\fix{X}C}\)
for every \(p \in \W\) by induction on \(p\).
By induction hypothesis,
\begin{eqnarray}
\label{logic-eqtyp-soundness-02}
    {\models}^\ttenv_q\,{A} &~\mbox{iff}~& {\models}^\ttenv_q\,{\fix{X}C}
    ~~\mbox{for every \(q \opacc p\)}.
\end{eqnarray}
Therefore,
\begin{Eqnarray*}
{\models}^\ttenv_p\,{A}
    &~\mbox{iff}~& {\models}^\ttenv_p\,{C[A/X]}
	& (by (\ref{logic-eqtyp-soundness-01})) \\
    &~\mbox{iff}~& {\models}^{\ttenv[\Il{A}^\ttenv/X]}_p\,{C}
	& (by Item~\itemref{logic-subst-env} of this proposition) \\
    &~\mbox{iff}~& {\models}^{\ttenv[\Il{\fix{X}C}^\ttenv/X]}_p\,{C}
	& (by (\ref{logic-eqtyp-soundness-02})
	    and Item~\itemref{logic-proper-subst-lemma}
		of this proposition) \\
    &~\mbox{iff}~& {\models}^\ttenv_p\,{C[\fix{X}C/X]}
	& (by Item~\itemref{logic-subst-env} of this proposition) \\
    &~\mbox{iff}~& {\models}^\ttenv_p\,{\fix{X}C}
	& (by Definition~\ref{LAm-semantics}).
\end{Eqnarray*}

\Case{\r{{\peqtyp}\mbox{-{\bf K}/{\bf L}}}.}
In this case, \(A = \O(C \Impl D)\) and \(B = \O C \Impl \O D\) for some
\(C\) and \(D\).
First, suppose that
\({\models}^\ttenv_p\,{\O(C \Impl D)}\), i.e.,
\begin{eqnarray}
    \label{logic-K-soundness-01}
    {\models}^\ttenv_r\,{C}~~\mbox{implies}~~
    {\models}^\ttenv_r\,{D}~~\mbox{for every \(r\), whenever
    \(p \acc q \mathrel{R} r\) for some \(q\)}.
\end{eqnarray}
To show \({\models}^\ttenv_p\,{\O C \Impl \O D}\),
suppose also that \(p \mathrel{R} s\) and \({\models}^\ttenv_s\,{\O C}\).
It suffices to show
\({\models}^\ttenv_r\,{D}\) for every \(r \opacc s\),
which can be established by (\ref{logic-K-soundness-01}),
since \(p \mathrel{R} s \acc r\) implies \(p \acc r\)
by Condition~3 of Definition~\ref{LA-frame-def}, and hence,
\(p \acc r \mathrel{R} r\) from the reflexivity of \(\mathrel{R}\).

On the other hand, to show the converse,
suppose that \({\models}^\ttenv_p\,{\O C \Impl \O D}\), i.e.,
\begin{eqnarray}
    \label{logic-L-soundness-01}
    {\models}^\ttenv_r\,{\O C}~~\mbox{implies}~~
    {\models}^\ttenv_r\,{\O D}~~\mbox{for every \(r\) such that
    \(p \mathrel{R} r\)}.
\end{eqnarray}
To show \({\models}^\ttenv_p\,{\O (C \Impl D)}\),
suppose also that \(p \acc q \mathrel{R} q'\) and
\({\models}^\ttenv_{q'}\,{C}\).
It suffices to show \({\models}^\ttenv_{q'}\,{D}\).
By Condition~6 of Definition~\ref{LA-frame-def},
there exists some \(r \in \W\) such that
\begin{Eqnarray}
    && p \mathrel{R} r \acc q',~\mbox{and} \label{logic-L-soundness-02}\\
    && r \acc s~\mbox{implies}~q' \mathrel{R} s~\mbox{for every}~s.
	\label{logic-L-soundness-03}
\end{Eqnarray}
Since \(\Il{C}^\ttenv\) is hereditary,
and since
\(r \acc s\) implies \(q' \mathrel{R} s\) by (\ref{logic-L-soundness-03}),
we get \({\models}^\ttenv_r\,{\O C}\) from \({\models}^\ttenv_{q'}\,{C}\).
Therefore,
\({\models}^\ttenv_r\,{\O D}\) by (\ref{logic-L-soundness-01}); and
hence,
\({\models}^\ttenv_{q'}\,{D}\) by (\ref{logic-L-soundness-02}).

\paragraph{Proof of \protect\itemref{logic-psubtyp-soundness}}
Suppose that \(\subt{\g}{A \psubtyp B}\) is derivable, and
    \(\{\,X_i\,\}\mathrel{\models^\ttenv}{Y_i}\)
    for every \(i \in \{\,1,\,2,\,\ldots\,n\,\}\).
By induction on the derivation of \(\subt{\g}{A \psubtyp B}\),
and by cases on the last subtyping rule applied in the derivation.

\Case{\r{{\rsubtyp}\mbox{-assump}}.}
Obvious since
\(A = X\) and \(B = Y\) for some \(X\) and \(Y\) such that
\(\{\,X \psubtyp Y\,\} \subseteq \g\)
in this case.

\Case{\r{{\rsubtyp}\mbox{-}\t}.}
Obvious from Definition~\ref{LAm-semantics}.

\Case{\r{{\rsubtyp}\mbox{-reflex}}.}
Obvious from Item~\itemref{logic-peqtyp-soundness} of this proposition.

\Case{\r{{\rsubtyp}\mbox{-trans}}.}
The derivation ends with
\[
\ifr{{\rsubtyp}\mbox{-trans}}
    {\subt{\g_1}{A \psubtyp C}
     & \subt{\g_2}{C \psubtyp B}}
    {\subt{\g_1 \cup \g_2}{A \psubtyp B}}
\]
for some \(C\), \(\g_1\) and \(\g_2\) such that \(\g = \g_1 \cup \g_2\).
By induction hypothesis,
\(\{\,A\,\}\,{\models}^\ttenv\,{C}\) and
\(\{\,C\,\}\,{\models}^\ttenv\,{B}\).
Therefore, \(\{\,A\,\}\,{\models}^\ttenv\,{B}\).

\Case{\r{{\rsubtyp}\mbox{-}{\O}}.}
The derivation ends with
\[
\ifr{{\rsubtyp}\mbox{-}{\O}}
    {\subt{\g}{A' \psubtyp B'}}
    {\subt{\g}{\O A' \psubtyp \O B'}}
\]
for some \(A'\) and \(B'\) such that \(A = \O A'\) and \(B = \O B'\).
Hence, \(\{\,A\,\}\,{\models}^\ttenv\,{B}\)
by Definition~\ref{LAm-semantics},
since \(\{\,A'\,\}\,{\models}^\ttenv\,{B'}\) by induction hypothesis.

\Case{\r{{\rsubtyp}\mbox{-}{\Impl}}.}
The derivation ends with
\[
\ifr{{\rsubtyp}\mbox{-}{\Impl}}
    {\subt{\g_1}{B_1 \psubtyp A_1}
     & \subt{\g_2}{A_2 \psubtyp B_2}}
    {\subt{\g_1 \cup \g_2}{A_1 \Impl A_2 \psubtyp B_1 \Impl B_2}}
\]
for some \(A_1\), \(A_2\), \(B_1\), \(B_2\), \(\g_1\) and \(\g_2\)
such that \(A = A_1 \Impl A_2\), \(B = B_1 \Impl B_2\) and
\(\g = \g_1 \cup \g_2\).
By induction hypothesis,
\(\{\,B_1\,\}\,{\models}^\ttenv\,{A_1}\) and
\(\{\,A_2\,\}\,{\models}^\ttenv\,{B_2}\).
Therefore, \(\{\,A_1 \Impl A_2\,\}\,{\models}^\ttenv\,{B_1 \Impl B_2}\)
by Definition~\ref{LAm-semantics}.

\Case{\r{{\rsubtyp}\mbox{-}{\fx}}.}
The derivation ends with
\[
\ifr{{\rsubtyp}\mbox{-}{\fx}}
    {\subt{\g \cup \{\,X' \psubtyp Y'\,\}}{A' \psubtyp B'}}
    {\subt{\g}{\fix{X'} A' \psubtyp \fix{Y'} B'}}
\]
for some \(X'\), \(Y'\), \(A'\) and \(B'\) such that
\(A = \fix{X'}A'\), \(B = \fix{Y'}B'\),
\(X' \not\in \FTV{\g} \cup \FTV{B'}\) and
\(Y' \not\in \FTV{\g} \cup \FTV{A'}\), where
\(A'\) and \(B'\) are proper in \(X'\) and \(Y'\), respectively.
Note that
\(A'\) and \(B'\) are also proper in \(Y'\) and \(X'\), respectively,
by Proposition~\ref{etv-proper},
since \(Y' \not\in \FTV{A'}\) and \(X' \not\in \FTV{B'}\).
We show that \(\{\,A\,\}\,{\models}^\ttenv_p\,{B}\)
for every \(p \in \W\) by induction on \(p\).
By the induction hypothesis on \(p\),
\begin{eqnarray}
\label{logic-psubtyp-soundness-01}
    \{\,A\,\}\,{\models}^\ttenv_q\,{B}~~\mbox{for every}~ q \opacc p.
\end{eqnarray}
Let \(\ttenv'\) be the (hereditary) valuation defined as follows.
\begin{eqnarray*}
    \ttenv'(X')(q) &=& \Choice{%
		\Il{A}^\ttenv_q \quad& (p \acc q) \\
		\mathbf{f} & (p \not\acc q)
		} \\[4pt]
    \ttenv'(Y')(q) &=& \Choice{%
		\Il{B}^\ttenv_q \quad& (p \acc q) \\
		\mathbf{f} & (p \not\acc q)
		} \\[4pt]
    \ttenv'(Z)(q) &=& \ttenv(Z)(q)\mskip65mu(Z \not\in \{\,X',\,Y'\,\})
\end{eqnarray*}
Note that
\begin{eqnarray}\label{logic-psubtyp-soundness-02}
    \ttenv'(Z)(q) &=& \ttenv[\Il{A}^\ttenv/X',\Il{B}^\ttenv/Y'](Z)(q)
    ~~\mbox{for every \(Z\) and \(q \opacc p\)}.
\end{eqnarray}
Since \(\{\,X'\,\}\,{\models}^{\ttenv'}\,{Y'}\)
by (\ref{logic-psubtyp-soundness-01}), and since
\(\{\,X_i\,\}\,{\models}^{\ttenv'}\,{Y_i}\)
    for every \(i \in \{\,1,\,2,\,\ldots\,n\,\}\),
by the induction hypothesis on the derivation,
\begin{eqnarray}
\label{logic-psubtyp-soundness-03}
    \{\,A'\,\}\,{\models}^{\ttenv'}_p\,{B'}.
\end{eqnarray}
On the other hand,
\begin{Eqnarray*}
{\models}^\ttenv_p\,{A}
    &~\mbox{iff}~& {\models}^\ttenv_p\,{A'[A/X']}
	    & (by Definition~\ref{LAm-semantics}) \\
    &~\mbox{iff}~& {\models}^{\ttenv[\Il{A}^\ttenv/X']}_p\,{A'}
	    & (by Item~\itemref{logic-subst-env} of this proposition) \\
    &~\mbox{iff}~& {\models}^{\ttenv[\Il{A}^\ttenv/X']}_p\,{A'[B/Y']}
	    & (by Proposition~\ref{geqtyp-no-etv-subst}
		and Item~\itemref{logic-peqtyp-soundness}
		of this proposition) \\
    &~\mbox{iff}~& {\models}^{\ttenv[\Il{A}^\ttenv/X']%
		[\Il{B}^{\ttenv[\Il{A}^\ttenv/X']}/Y']}_p\,{A'}
	    & (by Item~\itemref{logic-subst-env} of this proposition) \\
    &~\mbox{iff}~& {\models}^{\ttenv[\Il{A}^\ttenv/X']%
		[\Il{B[A/X']}^\ttenv/Y']}_p\,{A'}
	    & (by Item~\itemref{logic-subst-env} of this proposition) \\
    &~\mbox{iff}~& {\models}^{\ttenv[\Il{A}^\ttenv/X']%
		[\Il{B}^\ttenv/Y']}_p\,{A'}
	    & (by Proposition~\ref{geqtyp-no-etv-subst}
		and Item~\itemref{logic-peqtyp-soundness}
		of this proposition) \\
    &~\mbox{iff}~& {\models}^{\ttenv[\Il{A}^\ttenv/X',\,
		\Il{B}^\ttenv/Y']}_p\,{A'}
	    & (since \(X' \not= Y'\)) \\
    &~\mbox{iff}~& {\models}^{\ttenv'}_p\,{A'}
	    & (by (\ref{logic-psubtyp-soundness-02})
		Item~\itemref{logic-proper-subst-lemma}
		of this proposition).
\end{Eqnarray*}
Similarly, \({\models}^\ttenv_p\,{B}\) iff \({\models}^{\ttenv'}_p\,{B'}\).
Therefore, \(\{\,A\,\}\,{\models}^\ttenv_p\,{B}\)
by (\ref{logic-psubtyp-soundness-03}).

\Case{\r{{\rsubtyp}\mbox{-approx}}.}
\(B = \O A\) in this case.
By Condition~5 of Definition~\ref{LA-frame-def},
\(\Il{A}^\ttenv\) is also hereditary with respect to \(\acc\).
Hence, \(\{\,A\,\}\,{\models}^\ttenv\,{\O A}\)
by Definition~\ref{LAm-semantics}.

\paragraph{Proof of \protect\itemref{lA-logic-soundness}}
Let \(\G = \{\,x_1 : A_1,\,\ldots,\,x_n : A_n\,\}\), and
\(\typ{\G}{M:B}\) be a derivable judgment of {\lA}.
By induction on the derivation, and by cases on the last rule applied in it.
Suppose that \({\models}^\ttenv_p\,{A_i}\) for every
\(i \in \{\,1,\,2,\,\ldots,n\,\}\).
It suffices to show that \({\models}^\ttenv_p\,{B}\).

\Case{\r{\mbox{var}}.}
In this case,
\(M = x_i\) and \(B = A_i\) for some \(i \in \{\,1,\,2,\,\ldots,n\,\}\).
Therefore, \({\models}^\ttenv_p\,{A_i}\) by assumption.

\Case{\r{\mbox{shift}}.} 
\else 
The only non-trivial difference is the case for the \r{\mbox{shift}}-rules
in the proof of Item~\itemref{lA-logic-soundness}.
\fi 
\def\Ip#1{\II'(#1)}%
In this case, the derivation ends with
\[
    \Ifr{\r{\mbox{shift}}.}
        {\typ{\{\,x_1:\O A_1,\,x_2:\O A_2,\,\ldots,\,x_n:\O A_n,\,\}}
	     {M : \O B}}
        {\typ{\{\,x_1:A_1,\,x_2:A_2,\,\ldots,\,x_n:A_n\,\}}
	     {M : B}}
\]
For the given
\(\tuple{\W,\,{\acc},\,\mathrel{R}}\), \(p\) and \(\ttenv\),
we construct another {\LA}-frame \(\tuple{\W',\,{\acc'},\,{\mathrel{R'}}}\)
and a hereditary valuation \(\ttenv'\) for it by extending
\(\tuple{\W,\,{\acc},\,{\mathrel{R}}}\) and \(\ttenv\) as follows.
\begin{eqnarray*}
    \W' &=& \Zfset{\,\pair{a}{q}}{\strut
	    a \in \{\,0,\,1\,\}~\:\mbox{and}\:~q \in \W} \\[4pt]
    \pair{a}{q} \acc' \pair{b}{r}
	&~\mbox{iff}~& \Choice{%
		a = b = 0~\:\mbox{and}\:~q \acc r,~\mbox{or} \\[-2pt]
		a = 1,~b = 0~\:\mbox{and}\:~q \mathrel{R} r
	    } \\[3pt]
    \pair{a}{q} \mathrel{R'} \pair{b}{r}
	&~\mbox{iff}~&  a \ge b~~\mbox{and}~~q \mathrel{R} r \\[3pt]
    \ttenv'(X)(\pair{a}{q}) &=& \ttenv(X)(q)
\end{eqnarray*}
We can easily verify that
\(\tuple{\W',\,{\acc'},\,{\mathrel{R'}}}\) satisfies
Conditions~1, 2, 3 and 5 of Definition~\ref{LA-frame-def},
and that \(\ttenv'\) is hereditary.
Furthermore, Condition~6 is also satisfied.
In fact, suppose that
\(\pair{a}{p'} \acc' \pair{b}{q} \mathrel{R'} \pair{c}{q'}\).
If \(a = 0\), then it is obvious, since \(a = 0\) implies \(b = c = 0\), and
since \(\tuple{\W,\,{\acc},\,{\mathrel{R}}}\) satisfies Condition~6.
On the other hand, if \(a = 1\), then
we get \(b = c = 0\) by the definition of \(\acc'\);
and hence,
\[
    p' \mathrel{R} q \mathrel{R} q'.
\]
Then, it suffices to let \(r\) be \(\pair{1}{q'}\), because
\[
    \pair{a}{p'} \mathrel{R'} \pair{1}{q'} \acc' \pair{c}{q'},
\]
and because \(\pair{1}{q'} \acc' \pair{d}{s}\) implies
\(d = 0\) and \(q' \mathrel{R} s\) by the definition of \(\acc'\);
and hence, \(\pair{c}{q'} \mathrel{R'} \pair{d}{s}\).

Now, let \(\Ilp{A}^{\ttenv'}\) be the interpretation of \(A\)
in \(\tuple{\W',\,{\acc'},\,{\mathrel{R'}}}\)
under the valuation \(\ttenv'\).
Then, \(\Ilp{A}^{\ttenv'}_{\pair{0}{q}}\! = \Il{A}^\ttenv_q\,\)
for every \(q \in \W\) and \(A\).
Since \(\Il{A_i}^\ttenv_p = \mathbf{t}\,\) implies
\(\Ilp{\O A_i}^{\ttenv'}_{\pair{1}{p}}\! = \mathbf{t}\,\)
for every \(i \in \{\,1,\,2,\,\ldots,n\,\}\),
we get \(\Ilp{\O B}^{\ttenv'}_{\pair{1}{p}}\! = \mathbf{t}\,\)
by the induction hypothesis on the derivation.
Therefore,
\(\Ilp{B}^{\ttenv'}_{\pair{0}{p}} = \mathbf{t}\), that is,
\(\Il{B}^\ttenv_p = \mathbf{t}\).
\ifdetail

\Case{\r{\t}.}
Obviously \({\models}^\ttenv_p\,{B}\) by Definition~\ref{LAm-semantics}
since \(B = \t\) in this case.

\Case{\r{\rsubtyp}.}
In this case, \(\typ{\G}{M : B'}\) is derivable for some \(B'\)
such that \(B' \psubtyp B\).
Hence, \({\models}^\ttenv_p\,{B}\), since
\({\models}^\ttenv_p\,{B'}\) by induction hypothesis,
and since
\(\{\,B'\,\,\}\,{\models}^\ttenv\,{B}\)
by Item~\itemref{logic-psubtyp-soundness} of this proposition.

\Case{\r{\Impl\mbox{I}}.}
The derivation ends with
\[
\ifr{\Impl\mbox{I}}
    {\typ{\G \cup \{\,y:B_1\,\}}{L:B_2}}
    {\typ{\G}{\lam{y}{L} : B_1 \Impl B_2}}
\]
for some \(y\), \(L\), \(B_1\) and \(B_2\) such that
\(M = \lam{y}{L}\) and \(B = B_1 \Impl B_2\).
Suppose that \(p \mathrel{R} r\) and \({\models}^\ttenv_r\,{B_1}\).
Since \(\Il{A}\) is hereditary for any \(A\),
\({\models}^\ttenv_r\,{A_i}\) for every \(i \in \{\,1,\,2,\,\ldots,n\,\}\).
Hence, by induction hypothesis, \({\models}^\ttenv_r\,{B_2}\).
Thus, we get \({\models}^\ttenv_p\,{B_1 \Impl B_2}\).

\Case{\r{\Impl\mbox{E}}.}
The derivation ends with
\[
    \ifr{\Impl\mbox{E}}
	{\typ{\G_1}{M_1:C \Impl B}
	 & \typ{\G_2}{M_2:C}}
        {\typ{\G_1\cup\G_2}{\app{M_1}{M_2} : B}}
\]
for some \(\G_1\), \(\G_2\), \(M_1\), \(M_2\) and \(C\)
such that \(\G = \G_1 \cup \G_2\) and \(M = \app{M_1}{M_2}\).
By induction hypothesis,
\({\models}^\ttenv_p\,{C \Impl B}\) and \({\models}^\ttenv_p\,{C}\).
Hence, \({\models}^\ttenv_p\,{B}\) by Definition~\ref{LAm-semantics}.
\fi 
\qed\CHECKED{2014/07/10, 07/17}
\end{proof}

\begin{theorem}\label{logical-equiv-theorem}
The following three conditions are equivalent.
\begin{enumerate}[{\kern 8pt}(a)]\pushlabel
\item \(\typ{\{\,A_1,\,A_2,\,\ldots,\,A_n\,\}}{B}\)
    is derivable in {\LAm}.
\item \(\typ{\{\,x_1:A_1,\,x_2:A_2,\,\ldots,\,x_n:A_n\,\}}{M:B}\)
    is derivable in {\lA} for some \(x_1\), \(x_2\),
    \(\ldots\), \(x_n\) and \(M\).
\item \(\{\,A_1,\,A_2,\,\ldots,\,A_n\,\}\models^\ttenv{B}\,\)
for every\/ {\LA}-frame \(\tuple{\W,\,{\acc},\,\mathrel{R}}\) and
hereditary valuation \(\ttenv\).
\end{enumerate}
\end{theorem}
\begin{proof}
We get (a) \(\Rightarrow\) (b) and
(b) \(\Rightarrow\) (c) by
Propositions~\ref{LAm-lA} and \ref{lA-logic-soundness}, respectively.
Hence, it suffices to show that (c) \(\Rightarrow\) (a),
which is established by Theorem~\ref{LAm-completeness-theorem} below.
\qed\CHECKED{2014/07/10, 07/17}
\end{proof}
Note that this theorem implies that the \r{\mbox{shift}}-rule of {\lA}
is derivable in {\LAm}, that is,
if \(\typ{\O \G}{\O A}\) is derivable in {\LAm}, then so is
\(\typ{\G}{A}\)\footnote{%
This is also the case for the formal systems {\miGLC} and {\LA},
which will be defined later.}.

\begin{definition}
    \ilabel{comp-def}{Comp@$\protect\Comp(A)$}
    \ilabel*{0 subty@$\protect\component$}
Let \(A\) and \(B\) type expressions, i.e., formulae of {\LAm}.
We call \(A\) a {\em component} of \(B\), and write
\(A \component B\), if and only if
\(C[A/X] \eqtyp B\) for some \(X\) and \(C\) such that \(X \in \ETV{C}\).
We also define \(\Comp(B)\) as
\[
    \Comp(B) = \zfset{A}{A \component B}.
\]
\end{definition}
That is, \(\Comp(B)\) is the set of all type expressions
that is equivalent, in the sense of \(\eqtyp\),
to some subexpression occurring
freely and {\em effectively} in \(B\).
For example, let \(C = \fix{X}Y \Impl \O (X \Impl Z)\).
Then, \(A \in \Comp(Y \Impl \O C)\) if and only if
\[
    A \eqtyp A'
    ~\mbox{for some}~A'\in \{\,Y,\,Z,\,C,\,\O C,\,
	    C \Impl Z,\,\O(C \Impl Z),\,Y \Impl \O C\,\}.
\]
It can be shown that for every \(B\),
\(\Comp(B)/{\eqtyp}\) is a finite set.
We refer the reader for the proof to Appendix~\ref{comp-finite-sec}.
We do not define the corresponding notions using \(\peqtyp\)
instead of \(\eqtyp\).
If we define such a set, say \(Comp^{\peqtyp}(B)\), by \(\peqtyp\),
then \(Comp^{\peqtyp}(B)/{\peqtyp}\) is not always finite\footnote{%
    However, \(Comp(B)/{\peqtyp}\) is still finite, since
    \(A_1 \eqtyp A_2\) implies \(A_1 \peqtyp A_2\).}
because of the \r{{\peqtyp}\mbox{-{\bf K}/{\bf L}}}-rule.
For example, consider \(B = \fix{X}\O(X \Impl Y)\).
Then, \(A \in Comp^{\peqtyp}(B)/{\peqtyp}\) iff \(A \peqtyp A'\) for some
\(A' \in \{\,
    Y,\,B,\,B \Impl Y,\,
    \O Y,\,\O B,\,\O B \Impl \O Y,\,
    \O\O Y,\,\O\O B,\,\O\O B \Impl \O\O Y,\,
    \O\O\O Y,\,\O\O\O B,\,\O\O\O B \Impl \O\O\O Y,\,\ldots\,\}\), since
\begin{eqnarray*}
B &\peqtyp& \O(B \Impl Y) \\
    &\peqtyp& \O B \Impl \O Y \\
    &\peqtyp& \O\O (B \Impl Y) \Impl \O Y \\
    &\peqtyp& (\O\O B \Impl \O \O Y) \Impl \O Y \\
    &\peqtyp& (\O\O\O (B \Impl Y) \Impl \O \O Y) \Impl \O Y \\
    &\peqtyp& ((\O\O\O B \Impl \O\O\O Y) \Impl \O \O Y) \Impl \O Y \\
    &\peqtyp& \ldots.
\end{eqnarray*}

\begin{theorem}[Kripke completeness of {\LAm}]
    \ilabel{LAm-completeness-theorem}{completeness!LA-mu@\protect\LAm}
If\/ \(\,\typ{\{\,A_1,\,A_2,\,\ldots,\,A_n\,\}}{B}\) is not derivable
in\/ {\LAm}, then there exist some {\LA}-frame
\(\tuple{\W_0,\,{\acc_0},\,R_0}\),
hereditary valuation \(\ttenv_0\), and \(p_0 \in \W_0\)
such that \({\not\models}^{\ttenv_0}_{p_0} B\) while
\({\models}^{\ttenv_0}_{p_0}A_i\)
for every \(i \in \{\,1,\,2,\,\ldots\,n\,\}\).
\end{theorem}
\begin{proof}
Suppose that \(\typ{\{\,A_1,\,\ldots,\,A_n\,\}}{B}\)
is not derivable in {\LAm}.
Let
\[
    \F = \zfset{C}{C \in \Comp(B)~\mbox{or}~C \in \Comp(A_i)~\mbox{for some}~i},
\]
and define \(\W_0\) and \(p_0\) as follows.
\begin{eqnarray*}
    \W_0 &=& \zfset{p \subseteq \F}{C \in p ~\mbox{whenever}~
        C \in \F ~\mbox{and}~\typ{p}{C}~
        \mbox{is derivable\xfootnote{More precisely,
	    \(\typ{\G }{C}\) derivable
	    for some finite subset \(\G\) of \(p\).}}} \\
    p_0 &=& \zfset{C \in \F}{\typ{\{\,A_1,\,\ldots,\,A_n\,\}}{C}
        ~\mbox{is derivable}}
\end{eqnarray*}
\emitxfootnote
Note that \(\F/{\eqtyp}\) is a finite set, since so are \(\Comp(B)/{\eqtyp}\)
and \(\Comp(A_i)/{\eqtyp}\) for every \(i\);
and hence, \(\W_0\) is also a finite set
because
\(C \eqtyp D\) implies \(C \LAmequiv D\) by Lemma~\ref{eqtyp-Lequiv}.
Then, for each \(p \in \W_0\), define \(\tilde{p}\) as
\[
    \tilde{p} = \zfset{C \in \F}{\typ{p}{\O C}~\mbox{is derivable}}.
\]
Observe that
\begin{eqnarray}
    \label{LAm-completeness-theorem-01}
	&& \LAmtyp{p}{\O\tilde{p}}
\end{eqnarray}
by the definition, and that \(\tilde{p} \in \W_0\).
For, if \(\typ{\tilde{p}}{C}\) is derivable for some \(C \in \F\), then
so is \(\typ{\O \tilde{p}}{\O C}\) by \r{\mbox{nec}}; and therefore,
\(\typ{p}{\O C}\) is also derivable, i.e., \(C \in \tilde{p}\),
by Proposition~\ref{LAm-weakening} and
(\ref{LAm-completeness-theorem-01}).
Note also that
\begin{eqnarray*}
    p &\subseteq& \tilde{p}
\end{eqnarray*}
since {\LAm} has the \r{\mbox{approx}}-rule.
We then define accessibility relations \(\acc_0\) and \(R_0\) as follows.
\begin{eqnarray*}
    p \acc_0 q &~~\mbox{iff}~~&
	\tilde{p} \subseteq q ~\mbox{and}~ \tilde{q} \not\subseteq q \\
    p \mathrel{R_0} q &~~\mbox{iff}~~& p = q~\mbox{or,}~
	    p \subsetneq q ~\mbox{and}~\tilde{q} \not\subseteq q
\end{eqnarray*}
Since \(\W_0\) is finite, and since \(\acc_0\) is transitive and irreflexive,
\(\acc_0\) is (conversely) well-founded.
We can easily verify that \(\acc_0\) and \(R_0\) satisfy
Conditions~1, 2, 3 and 5 of Definition~\ref{LA-frame-def}.
Furthermore, Condition~6 is also satisfied.
For, suppose that \(p \acc_0 q \mathrel{R_0} q'\).
Then, let \(r = \zfset{C\in \F}{\typ{p\cup \O q'}{C}~\mbox{is derivable}}\).
We show that (a) and (b) of Condition~6
of Definition~\ref{LA-frame-def} hold.
First, obviously \(p \subseteq r\) by the definition.
Furthermore, it can be shown that
\(\tilde{r} \not\subseteq r\) by contradiction.
Suppose that \(\tilde{r} \subseteq r\).
Since \(\ptilde{q} \not\subseteq q'\) from \(p \acc_0 q \mathrel{R_0} q'\),
there exists some \(D \in \ptilde{q}\) such that \(D \not\in q'\).
Hence, we can derive the following.
\begin{Eqnarray*}
    && \typ{q'}{\O D} & (since \(D \in \ptilde{q}\)) \\
    && \typ{\O q'}{\O\O D} & (by \r{\mbox{nec}}) \\
    && \typ{r}{\O\O D} & (by the definition of \(r\)) \\
    && \typ{\tilde{r}}{\O D} & (by the definition of \(\tilde{r}\)) \\
    && \typ{r}{\O D} & (since \(\tilde{r} \subseteq r\) by assumption) \\
    && \typ{\tilde{r}}{D} & (by the definition of \(\tilde{r}\)) \\
    && \typ{r}{D} & (since \(\tilde{r} \subseteq r\) by assumption) \\
    && \typ{p\cup\O q'}{D} & (by the definition of \(r\)) \\
    && \typ{q'\cup\O q'}{D} & (since \(p \subseteq q'\) from
		    \(p \acc_0 q \mathrel{R_0} q'\)) \\
    && \typ{q'}{D} & (since \(\typ{q'}{\O q'}\) by \r{\mbox{approx}})
\end{Eqnarray*}
This contradicts \(D \not\in q'\); and hence,
\(\tilde{r} \not\subseteq r\), and \(p \mathrel{R_0} r\).
Second, we show that \(r \acc_0 q'\).
To this end, suppose that \(C \in \tilde{r}\).
Then, the following judgments are derivable.
\begin{Eqnarray*}
    && \typ{r}{\O C} & (since \(C \in \tilde{r}\)) \\
    && \typ{p\cup\O q'}{\O C} & (by the definition of \(r\)) \\
    && \typ{p}{\O q' \Impl \O C} & (by \r{{\Impl}\mbox{I}}\footnote{%
	More precisely,
	\(\typ{\G}{\O A_1 \Impl \O A_2 \Impl \ldots \Impl \O A_n \Impl \O C}\)
	is derivable for some finite subset \(\G\) of \(p\) and some
	finite subset \(\{\,A_1,\,A_2,\,\ldots,\,A_n\,\}\) of \(q\).}) \\
    && \typ{p}{\O (q' \Impl C)} & (by \r{\mbox{\bf L}}) \\
    && \typ{\tilde{p}}{q' \Impl C} & (by the definition of \(\tilde{p}\)) \\
    && \typ{\tilde{p}\cup q'}{C} & (by \r{{\Impl}\mbox{E}}) \\
    && \typ{q'}{C} & (since \(\tilde{p} \subseteq q'\)
	    from \(p \acc_0 q \mathrel{R_0} q'\))
\end{Eqnarray*}
Hence, \(C \in q'\); that is, \(\tilde{r} \subseteq q'\); and
therefore, \(r \acc_0 q'\),
since \(\ptilde{q} \not\subseteq q'\)
	    from \(p \acc_0 q \mathrel{R_0} q'\),
Condition~(a) is thus established.
Next, we show that (b) also holds.
Suppose that \(r \acc_0 s\) and \(C \in q'\).
It suffices to show that \(C \in s\).
We can derive
\begin{Eqnarray*}
    && \typ{q'}{C} & (since \(C \in q'\)) \\
    && \typ{\O q'}{\O C} & (by \r{\mbox{nec}}) \\
    && \typ{r}{\O C} & (by the definition of \(r\)) \\
    && \typ{\tilde{r}}{C} & (by the definition of \(\tilde{r}\)) \\
    && \typ{s}{C} &
	(since \(\tilde{r} \subseteq s\) from \(r \acc_0 s\)).
\end{Eqnarray*}
That is, \(C \in s\).
Thus, Condition~(b) is also established.
We finally define the valuation \(\ttenv_0\) as follows.
\[
  \ttenv_0(X)_p = \Choice{%
      \mathbf{t} & (X \in p) \\
      \mathbf{f} & (X \not\in p)
      }
\]
Obviously, \(\ttenv_0\) is hereditary by the definition of \(R_0\).
To finish the proof of the completeness theorem, it suffices to
prove the following lemma,
because \(B \not\in p_0\) while \(A_i \in p_0\) for every \(i\).
\qed\CHECKED{2014/07/10, 07/17}
\end{proof}

\begin{lemma}\label{LAm-crucial-lemma}
Consider the frame \(\tuple{\W_0,\,{\acc_0},\,R_0}\) and \(\F\)
defined in the proof of Theorem~\ref{LAm-completeness-theorem}.
Let\/ \(C \in \F\) and \(p \in \W_0\).
Then, \(C \in p\) if and only if\/ \({\models}^{\ttenv_0}_p{C}\).
\end{lemma}
\begin{proof}
The proof proceeds by induction on
the ordering \(\sqsupset\) of \(\pair{p}{r(C)}\),
and by cases on the form of \(C\).
If \(C\) is a {\tvariant}, then \(C \in p\) by the definition of \(\W_0\)
and Proposition~\ref{LAm-tvariant}, and
\({\models}^{\ttenv_0}_p{C}\) for every \(p\)
by Definition~\ref{LAm-semantics}.
Hence, we only consider the case that \(C\) is not.

\Case{\(C = X\) for some \(X\).}
Trivial from the definition of \(\ttenv_0(X)\).

\Case{\(C = \O D\) for some \(D\).}
For the ``only if'' part, suppose that
\(\O D \in p\).
If \(p \acc_0 q\), then since \(\tilde{p} \subseteq q\) and
since \(\typ{p}{\O D}\) is derivable, we get
\(D \in q\).
Hence, \({\models}^{\ttenv_0}_q{D}\) by induction hypothesis.
Thus, we get \({\models}^{\ttenv_0}_p{\O D}\) from \(\O D \in p\).
For the ``if'' part,
suppose that \({\models}^{\ttenv_0}_p{\O D}\), i.e.,
\begin{eqnarray}\label{LAm-crucial-lemma-01}
    &&\mbox{\({\models}^{\ttenv_0}_q{D}\;\) for any \(q\;\) such that
    \(p \acc_0 q\)}.
\end{eqnarray}
Let
\[
    q = \zfset{E \in \F}{\typ{\tilde{p}\cup\{\O D\}}{E}~\mbox{is derivable}}.
\]
Note that \(q \in \W_0\) and \(\tilde{p} \subseteq q\).
If \(p \acc_0 q\), then
\({\models}^{\ttenv_0}_q{D}\) by (\ref{LAm-crucial-lemma-01}); therefore,
\(D \in q\) by induction hypothesis.
Otherwise, \(\tilde{q} \subseteq q\).
In this case, again \(D \in q\) since \(D \in \tilde{q}\).
Therefore, \(\typ{\tilde{p}\cup\{\O D\}}{D}\) is derivable in either case;
and hence, we can derive the following judgments.
\begin{Eqnarray*}
    && \typ{\tilde{p}}{\O D \Impl D} & (by \r{{\Impl}\mbox{I}}) \\
    && \typ{\tilde{p}}{D}
	& (by Proposition~\ref{LAm-Y}, namely \r{\mbox{\bf Y}}) \\
    && \typ{\O \tilde{p}}{\O D} & (by \r{\mbox{nec}}) \\
    && \typ{p}{\O D} & (by (\ref{LAm-completeness-theorem-01}) and
	    Proposition~\ref{LAm-weakening})
\end{Eqnarray*}
Thus, we get \(\O D \in p\) from \({\models}^{\ttenv_0}_p{\O D}\).

\Case{\(C = D \Impl E\) for some \(D\) and \(E\).}
For the ``only if'' part, suppose that \(D \Impl E \in p\),
\(p \mathrel{R_0} q\) and \({\models}^{\ttenv_0}_q{D}\).
We get \(D \in q\) from \({\models}^{\ttenv_0}_q{D}\)
by induction hypothesis.
On the other hand,
\(D \Impl E \in q\), since \(p \subseteq q\) from \(p \mathrel{R_0} q\).
Therefore, \(E \in q\); and by induction hypothesis again,
\({\models}^{\ttenv_0}_q{E}\).
Thus, we get \({\models}^{\ttenv_0}_p{D \Impl E}\).
As for the ``if'' part,
suppose that \({\models}^{\ttenv_0}_p{D \Impl E}\), i.e.,
\begin{eqnarray}\label{iLA-completeness-theorem-01}
    &&\mbox{\({\models}^{\ttenv_0}_q{D}\;\) implies
    \(\;{\models}^{\ttenv_0}_q{E}\;\) whenever \(\;p \mathrel{R_0} q\)}.
\end{eqnarray}
Let
\[
    q = \zfset{F \in \F}{\typ{p \cup \{\,D,\,\O E\,\}}{F} ~\mbox{is derivable}}.
\]
Note that \(q \in \W_0\) and \(p \subseteq q\).
If \(\tilde{q} \subseteq q\), then we can derive
\begin{Eqnarray*}
    && \typ{q}{\O E} & (by the definition of \(q\)) \\
    && \typ{\tilde{q}}{E} & (by the definition of \(\tilde{q}\)) \\
    && \typ{q}{E} & (since \(\tilde{q} \subseteq q\) by assumption) \\
    && \typ{p \cup \{\,D,\,\O E\,\}}{E} & (by the definition of \(q\)) \\
    && \typ{p \cup \{\,D\,\}}{\O E \Impl E} & (by \r{{\Impl}\mbox{I}}) \\
    && \typ{p \cup \{\,D\,\}}{E}
	& (by Proposition~\ref{LAm-Y}, namely \r{\mbox{\bf Y}}) \\
    && \typ{p}{D \Impl E} & (by \r{{\Impl}\mbox{I}}).
\end{Eqnarray*}
That is, \(D \Impl E \in p\) in this case.
On the other hand, if \(\tilde{q} \not\subseteq q\),
then \(p \mathrel{R_0} q\) holds.
Hence, \({\models}^{\ttenv_0}_q{D}\) from \(D \in q\)
by induction hypothesis.
Then,
\({\models}^{\ttenv_0}_q{E}\) from (\ref{iLA-completeness-theorem-01});
and therefore, by induction hypothesis again, \(E \in q\), i.e.,
\(\strut\typ{p \cup \{\,D,\,\O E\,\}}{E}\) is derivable, which implies
that so is
\(\typ{p}{D \Impl E}\) by \r{{\Impl}\mbox{I}} and \r{\mbox{\bf Y}}.
Thus, \(D \Impl E \in p\) also in this case.

\Case{\(C = \fix{X}{D}\) for some \(X\) and \(D\).}
For the ``if'' part,
suppose that \({\models}^{\ttenv_0}_p{\fix{X}{D}}\),
i.e., \({\models}^{\ttenv_0}_p{D[\fix{X}{D}/X]}\).
Since \(r(D[\fix{X}{D}/X]) < r(\fix{X}{D})\),
we get \(D[\fix{X}{D}/X] \in p\) by induction hypothesis; and therefore,
\(\fix{X}{D} \in p\strut\) since {\LAm} has \r{\mbox{fold}}.
For the ``only if'' part, suppose that
\(\fix{X}{D} \in p\), i.e., also \(D[\fix{X}{D}/X] \in p\)
by \r{\mbox{unfold}}.
Hence,
\({\models}^{\ttenv_0}_p{D[\fix{X}{D}/X]}\) by induction hypothesis;
and therefore,
\({\models}^{\ttenv_0}_p{\fix{X}{D}}\) by Definition~\ref{LAm-semantics}.
\qed\CHECKED{2014/05/30, 07/10, 07/17}
\end{proof}
This completes the proofs of Theorems~\ref{LAm-completeness-theorem}
and \ref{logical-equiv-theorem}.
Since the counter model constructed in the proof of
Lemma~\ref{LAm-crucial-lemma} is based on a finite frame,
the logic {\LAm} has the finite model property,
and we therefore get the following corollary.
\begin{corollary}
The following problems are decidable.
\begin{Enumerate}
\item Provability in {\LAm}.
\item Type inhabitance in {\lA}.
\end{Enumerate}
\end{corollary}

\Subsection{Relationship to intuitionistic modal logic of provability}

The logic {\LAm} allows self-referential formulae.
In this subsection, we show that
if {\LAm} is restricted to finite formulae, i.e.,
those without any occurrence of \(\fx\), then
we get an intuitionistic
version of the logic of provability {\GL}.

\begin{definition}[{\miGL}, {\miGLC} and {\LA}]
    \ilabel{miGL-def}{miGL@\protect\miGL, \protect\miGLC}
    \ilabel{miGLC-def}{LA@\protect\LA}
We define {\miGL}, which only allows
finite formulae, to be the formal system obtained from {\miK4}
(Definition~\ref{miK4-def}) by adding the following inference rule called
L\"{o}b's Principle.
\[
\ifr{\mbox{\bf W}}
    {\typ{\G}{\O(\O A\Impl A)}}
    {\typ{\G}{\O A}}
\]
We also define two more formal systems: {\miGLC} and {\LA},
which again only allow finite formulae,
as \(\mbox{\miGL}+\r{\mbox{approx}}\) and
\(\mbox{\miGLC}+\r{\mbox{\bf L}}\), respectively, where
\r{\mbox{approx}} and \r{\mbox{\bf L}} are those of
Definition~\ref{LAm-def}.
\end{definition}

The formal system {\miGL} corresponds to the minimal and implicational fragment
of {\GL}\footnote{%
The intuitionistic version of {\GL} was introduced
    by Ursini\cite{ursini-au79}, who also showed
    its completeness and finite model property\cite{ursini-sl79}.
    The \r{\mbox{\bf 4}}-rule is redundant if
    conjunctive formulae are available.}, and
it can be shown that {\miGL} is sound and complete
with respect to the Kripke semantics over {\iGL}-frames
(Definitions~\ref{iGL-frame-def} and \ref{LAm-semantics}).
The \r{\mbox{approx}}-rule of {\miGLC} and {\LA}
is called the ``Completeness Principle''
in the context of the logic of provability (cf. \cite{visser-aml82}).

Note that \r{\mbox{\bf W}}, i.e., L\"{o}b's Principle, is
a derivable rule of {\LAm}.
In fact, \(\typ{}{(\O A\Impl A) \Impl A}\)
is derivable in {\LAm} by Proposition~\ref{LAm-Y},
and from which we get \(\typ{}{\O (\O A\Impl A) \Impl \O A}\)
by applying \r{\mbox{nec}} and \r{\mbox{\bf K}}.
Conversely,
\(\typ{}{(\O A\Impl A) \Impl A}\) is derivable in {\miGLC} and {\LA}
by \r{\mbox{\bf W}} and \r{\mbox{approx}} as follows\footnote{%
    If the system allows conjunctive formulae,
    the pair of \r{\mbox{approx}} and \r{\mbox{\bf W}}
    is equivalent to the \r{\mbox{\bf Y}}-rule of Proposition~\ref{LAm-Y},
    which is called the ``Strong L\"{o}b's Principle''. }.
\[
\Mtight
\ifr{{\Impl}\mbox{I}}
    {\ifr{{\Impl}\mbox{E}}
       {\ifr{\mbox{assump}}
	   {}
	   {\typ{\{\,\O A\Impl A\,\}}{\O A \Impl A}}
       & \ifr{\mbox{\bf W}}
             {\Ifr{\mbox{\r{\mbox{approx}}}}
                  {\ifr{\mbox{assump}}
                       {}
                       {\typ{\{\,\O A\Impl A\,\}}{\O A \Impl A}}}
                  {\typ{\{\,\O A\Impl A\,\}}{\O(\O A \Impl A)}}}
             {\typ{\{\,\O A\Impl A\,\}}{\O A}}}
       {\typ{\{\,\O A\Impl A\,\}}{A}}}
    {\typ{}{(\O A \Impl A) \Impl A}}
\]
Therefore, since the only role of \r{\mbox{fold}} and \r{\mbox{unfold}}
for finite formulae in the proof of
Theorem~\ref{LAm-completeness-theorem} is the derivability of
\((\O A \Impl A) \Impl A\), i.e., Proposition~\ref{LAm-Y},
which is used in the ``if'' part of the proof
of Lemma~\ref{LAm-crucial-lemma},
we get the following.
\begin{theorem}[Kripke completeness of {\LA}]
    \ilabel{LA-completeness-theorem}{completeness!LA@\protect\LA}
The formal system {\LA} is also Kripke complete with respect to
{\LA}-frames.
\end{theorem}
\begin{table}[tb]
\caption{Summary of the systems}\label{comparison-table}
\begin{center}
\renewcommand\arraystretch{1.3}\small
\def\L#1#2{\setbox0=\hbox{\lower #1pt\hbox{#2}}\dp0=0pt\copy0}%
\def\H#1{\hbox to 45pt{\hfill#1\hfill}}
\begin{tabular}{|c|c|c|c|c|c|} \cline{2-6}
     \multicolumn{1}{c|}{}
     & \H{\miGL} & \H{\miGLC} & \H{\LA} & \H{\LAm} & \H{\lA}
     \\ \cline{2-6}\noalign{\vskip2pt}\hline
\multicolumn{1}{|c}{system}
 & \multicolumn{4}{|c|}{logic} & type \\ \hline\hline
\multicolumn{1}{|c}{recursive types}
 & \multicolumn{3}{|c|}{no} & \multicolumn{2}{|c|}{yes} \\ \hline\hline
\L{7}{frames}
 &\multicolumn{5}{|c|}{well-founded} \\ \cline{2-6}
 &\multicolumn{2}{|c|}{\em n/a}
     &\multicolumn{3}{|c|}{{\em locally} linear} \\ \hline\hline
interpretation
 & \em n/a & \multicolumn{4}{|c|}{hereditary w.r.t. $\acc$} \\ \hline\hline
 &\multicolumn{4}{|c|}{\r{\mbox{assump}}} & \r{\mbox{var}} \\ \cline{2-6}
 &\multicolumn{5}{|c|}{\r{{\Impl}\mbox{I}}} \\ \cline{2-6}
 &\multicolumn{5}{|c|}{\r{{\Impl}\mbox{E}}} \\ \cline{2-6}
 &\multicolumn{4}{|c|}{\r{\mbox{nec}}}
     & \em derivable \\ \cline{2-6}
 &\multicolumn{4}{|c|}{\r{\mbox{\bf K}} {\em drivable}}
     & \L{7}{\r{{\peqtyp}\mbox{-{\bf K}/{\bf L}}}} \\ \cline{2-5}
\L{7}{rules}
 &\multicolumn{2}{|c|}{\em n/a}
     & \multicolumn{2}{|c|}{\r{\mbox{\bf L}}} & \\ \cline{2-6}
 &\multicolumn{1}{|c|}{\r{\mbox{\bf 4}}}
     &\multicolumn{3}{|c|}{\L{7}{\r{\mbox{approx}}}}
	&\L{7}{\r{{\psubtyp}\mbox{-approx}}} \\ \cline{2-2}
 & \em n/a & \multicolumn{3}{|c|}{} & \\ \cline{2-6}
 & \em n/a & \multicolumn{3}{|c|}{\em derivable}
	 & \r{\mbox{shift}} \\ \cline{2-6}
 & \multicolumn{3}{|c|}{\r{\mbox{\bf W}}}
     & \multicolumn{2}{|c|}{\em derivable} \\ \cline{2-6}
 & \multicolumn{3}{|c|}{\em \L{4}{derivable}}
     & \r{\mbox{fold/unfold}}
	 & \r{{\eqtyp}\mbox{-fix}} \\ \cline{5-6}
 & \multicolumn{4}{|l|}{%
	     \L{-1}{\em (for logical fixed-points of finite formulae)}}
	 & \r{{\eqtyp}\mbox{-uniq}} \\ \hline\hline
 & \multicolumn{5}{|c|}{\(\O (A\Impl B) \Impl \O A \Impl \O B\)} \\ \cline{2-6}
 & \multicolumn{2}{|c|}{\em n/a}
     & \multicolumn{3}{|c|}{\((\O A \Impl \O B) \Impl \O(A \Impl B)\)}
     \\ \cline{2-6}
\L{7}{theorems}
 & \multicolumn{5}{|c|}{\(\O A\Impl \O\O A\)} \\ \cline{2-6}
 & \em n/a & \multicolumn{4}{|c|}{\(A \Impl \O A\)} \\ \cline{2-6}
 & \multicolumn{5}{|c|}{\(\O (\O A\Impl A) \Impl \O A\)} \\ \cline{2-6}
 & \em n/a & \multicolumn{4}{|c|}{\((\O A \Impl A) \Impl A\)} \\ \hline
\end{tabular}
\end{center}
\end{table}
This theorem also implies, with Theorem~\ref{LAm-completeness-theorem},
that {\LAm} is a conservative extension of {\LA}.
Sambin\cite{sambin-sl76} proved that
in the intuitionistic version of {\GL}, for any \(X\) and \(A\)
such that \(X\) only occurs in \(A\)
within the scope of an occurrence of the modal operator,
we can construct a formula that acts as the (logically) unique fixed-point
of the propositional function that maps \(X\) to \(A\).
In case of systems involving the \r{\mbox{approx}}-rule,
such as {\miGLC} and {\LA}, it is also known that
the fixed-point is quite simple, and can be \(A[\t/X]\)
(cf. \cite{dejongh-visser94}).
For example, in such a logic, the recursive type
\(\fix{X}A \Conj \O X\), which represents the type of infinite streams of \(A\),
is {\em logically} representable by \(A \Conj \O \t\)\footnote{%
    \(A\) can be also a fixed-point of \(X \mapsto A \Conj \O X\). },
since
\(\typ{}{(A \Conj \O \t) \Lequiv (A \Conj \O (A \Conj \O \t))}\).
That is, as a logic, any recursive type is replaceable by
a finite formula that is logically equivalent.
However, this is not the case from the type theoretical point of view,
since we need a fixed-point as a set of realizers.

The correspondence between {\lA} and the logical systems can be
summarized as Table~\ref{comparison-table}\footnote{%
In this table, ``{\em n/a}'' should be read as
``not always'' or ``not available''.}.
As seen in this table, the logic behind {\lA} can be considered
the minimal and implicational fragment of the logic of provability
with the ``Completeness Principle'' and the axiom schema
\((\O A \Impl \O B) \Impl \O (A \Impl B)\).

\Section{Concluding remarks}\label{conclusion-sec}

In this paper, a modal typed \(\lambda\)-calculus {\lA}
with recursive types has been presented,
and its soundness with respect to a realizability interpretation and
the convergence of well-typed \(\lambda\)-terms
according to their types have been shown.
We have also shown that the modal logic behind {\lA}
can be regarded as an intuitionistic fragment of the logic of provability,
and shown its Kripke completeness
with respect to intuitionistic, (conversely) well-founded and
{\rm locally} linear frames.
By the completeness theorem,
decidability of type inhabitance in {\lA} has been also shown.
However, the decidability questions for type checking and typability
of \(\lambda\)-terms in {\lA} are still open.

The connection to the logic of provability suggests
certain subsystems of {\lA} which
includes \(\O A \subtyp \O\O A\) instead of \(A \subtyp \O A\), or
does not have \(\O A \Impl \O B \subtyp \O (A \Impl B)\), namely
the converse of \r{\mbox{\bf K}}.
For example, it seems that
all the examples we have seen in Section~\ref{program-sec} can be also
captured in the subsystem without \(\O A \Impl \O B \subtyp \O (A \Impl B)\).
In this sense, the author does not see any apparent significance
of this rule in practice so far.
However, he also thinks {\lA}, with this rule, more preferable
for a basis for logic of programming, because it has much simpler formulation.
Note that we would need two additional rules, namely \r{\mbox{nec}} and
\r{\mbox{subst}} without the equality
\(\O A \Impl \O B \peqtyp \O (A \Impl B)\).

Although {\lA} was presented as a typed \(\lambda\)-calculus,
the author does not think that it is directly applicable to
type systems of programming languages.
Since our framework can assert the convergence of derived programs,
typing general recursive programs naturally requires
some (classical) arithmetic
as seen in the cases of the 91-function and the sieve of Eratosthenes,
which would make mechanical type checking impossible.
Our goal is rather to capture a wider range of programs
in the proofs-as-programs paradigm and give an axiomatic semantics
to them preserving the compositionality of programs.
We have seen that our approach is
applicable to some interesting programs such as fixed-point combinators
and objects with binary methods,
which have not been captured in the conventional frameworks.

\section*{References}
\addcontentsline{toc}{section}{References}
%
%
\bibliographystyle{elsart-alpha}
\bibliography{paper}

\begin{thebibliography}{10}
\expandafter\ifx\csname url\endcsname\relax
  \def\url#1{\texttt{#1}}\fi
\expandafter\ifx\csname urlprefix\endcsname\relax\def\urlprefix{URL }\fi

\bibitem{abadi:cardelli}
M.~Abadi, L.~Cardelli, A theory of objects, Springer-Verlag, 1996.

\bibitem{amadio:cardelli}
R.~M. Amadio, L.~Cardelli, Subtyping recursive types, ACM Transactions on
  Programming Languages and Systems 15~(4) (1993) 575--631.

\bibitem{barendregt}
H.~P. Barendregt, Lambda calculi with types, in: S.~Abramsky, D.~M. Gabbay,
  T.~S.~E. Maibaum (Eds.), Handbook of Logic in Computer Science, Vol.~2,
  Oxford University Press, 1992, pp. 118--309.

\bibitem{boolos}
G.~Boolos, The logic of provability, Cambridge University Press, 1993.

\bibitem{bruce:cardelli:pierce}
K.~B. Bruce, L.~Cardelli, B.~C. Pierce, Comparing object encodings, Information
  and Computation 155 (1999) 108--133.

\bibitem{cardelli}
L.~Cardelli, Amber, in: G.~Cousineau, P.-L. Curien, B.~Robinet (Eds.),
  Combinators and functional programming languages, Vol. 242 of Lecture Notes
  in Computer Science, Springer-Verlag, 1986, pp. 21--47.

\bibitem{cardone:coppo}
F.~Cardone, M.~Coppo, Type inference with recursive types: syntax and
  semantics, Information and Computation 92~(1) (1991) 48--80.

\bibitem{constable}
R.~L. Constable, S.~Allen, H.~Bromely, W.~Cleveland, et~al., Implementing
  Mathematics with the Nuprl Proof Development System, Prentice-Hall, 1986.

\bibitem{constable:mendler}
R.~L. Constable, N.~P. Mendler, Recursive definitions in type theory, in:
  Logics of Programs, Vol. 193 of Lecture Notes in Computer Science,
  Springer-Verlag, 1985, pp. 61--78.

\bibitem{constable:smith}
R.~L. Constable, S.~F. Smith, Partial objects in constructive type theory, in:
  Proceedings of the 2nd IEEE Symposium on Logic in Computer Science, IEEE
  Computer Society Press, 1987, pp. 183--193.

\bibitem{courcelle}
B.~Courcelle, Fundamental properties of infinite trees, Theoretical Computer
  Science 25 (1983) 95--169.

\bibitem{dejongh-visser94}
D.~de~Jongh, A.~Visser, Embeddings of Heyting algebras, revised version, Vol.
  115 of Logic Group Preprint Series, Utrecht University, 1994.

\bibitem{gurevich:shelah}
Y.~Gurevich, S.~Shelah, Fixed-point extensions of first-order logic, Annals of
  Pure and Applied Logic 32~(3) (1986) 265--280.

\bibitem{hayashi:nakano}
S.~Hayashi, H.~Nakano, {PX}: A Computational Logic, The {MIT} {P}ress, 1988.

\bibitem{hindley}
R.~Hindley, The completeness theorem for typing \(\lambda\)-terms, Theoretical
  Computer Science 22 (1983) 1--17.

\bibitem{hindley-F}
R.~Hindley, Curry's type-rules are complete with respect to \hbox{F}-sematics
  too, Theoretical Computer Science 22 (1983) 127--133.

\bibitem{howard}
W.~A. Howard, The formulae-as-types notion of construction, in: R.~J. Hindley,
  J.~P. Seldin (Eds.), To H.B.\,Curry: Essays on Combinatory Logic, Lambda
  Calculus and Formalism, Academic Press, 1980, pp. 480--490.

\bibitem{kobayashi:tatsuta}
S.~Kobayashi, M.~Tatsuta, Realizability interpretation of generalized inductive
  definitions, Theoretical Computer Science 131~(1) (1994) 121--138.

\bibitem{kozen}
D.~C. Kozen, Results on the propositional \(\mu\)-calculus, Theoretical
  Computer Science 27~(3) (1983) 333--354.

\bibitem{mps}
D.~B. MacQueen, G.~D. Plotkin, R.~Sethi, An ideal model for recursive
  polymorphic types, Information and Computation 71 (1986) 95--130.

\bibitem{nakano-lics00}
H.~Nakano, A modality for recursion, in: Proceedings of the 15th IEEE Symposium
  on Logic in Computer Science, IEEE Computer Society Press, 2000, pp.
  255--266.

\bibitem{nakano-tacs01}
H.~Nakano, Fixed-point logic with the approximation modality and its {K}ripke
  completeness, in: Proceedings of the 4th International Symposium on
  Theoretical Aspects of Computer Software, Springer-Verlag, 2001, pp.
  165--182.

\bibitem{paulin-mohring}
C.~Paulin-Mohring, Extracting \({F}_\omega\)'s programs from proofs in the
  calculus of constructions, in: Proceedings of the 16th ACM Symposium on
  Principles of Programming Languages, 1989, pp. 89--104.

\bibitem{pratt}
V.~R. Pratt, A decidable \(\mu\)-calculus (preliminary report), in: Proceedings
  of the 22nd IEEE Symposium on Foundation of Computer Science, 1981, pp.
  421--427.

\bibitem{sambin-sl76}
G.~Sambin, An effective fixed-point theorem in intutionistic diagonalizable
  algebras, Studia Logica 35 (1976) 345--361.

\bibitem{scott}
D.~S. Scott, A type-theoretical alternative to {ISWIM}, {CUCH}, {OWHY},
  Theoretical Computer Science 121 (1993) 411--440.

\bibitem{tait}
W.~W. Tait, Intensional interpretations of functionals of finite type {I},
  Journal of Symbolic Logic 32 (1967) 198--212.

\bibitem{tatsuta}
M.~Tatsuta, Realizability interpretation of coinductive definitions and program
  synthesis with streams, Theoretical Computer Science 122 (1994) 119--136.

\bibitem{ursini-au79}
A.~Ursini, Intutionistic diagonalizable algebras, Algebra Universalis 9 (1979)
  229--237.

\bibitem{ursini-sl79}
A.~Ursini, A modal calculus analogous to {K4W}, based on intutionistic
  propositional logic, {I}${}^\circ$, Studia Logica 38 (1979) 297--311.

\bibitem{visser-aml82}
A.~Visser, On the completeness principle: a study of provability in
  {Heytings}'s arithmetic and extensions, Annals of Mathematical Logic 22
  (1982) 263--295.

\end{thebibliography}

\begin{appendices}
\appendixeqno
\Section{Appendix}\label{appendix-sec}

\Subsection{Proof of Proposition~\ref{geqtyp-subst}}\label{geqtyp-subst-sec}

We prove a more general propositions as follows.
The \r{{\eqtyp}\mbox{-uniq}}-rule is crucial to
the proof.

\begin{proposition}\label{appendix-geqtyp-subst}
Let \(n\) be a non-negative integer,
\(X_1\), \(X_2\), \(\ldots\), \(X_n\) distinct type variables, and
\(A\), \(B\), \(C_1\), \(C_2\), \(\ldots\), \(C_n\),
\(D_1\), \(D_2\), \(\ldots\), \(D_n\) type expressions.
Let \([\vec{C}/\vec{X}]\) and \([\vec{D}/\vec{X}]\) be abbreviations
for \([C_1/X_1,\,C_2/X_2,\,\ldots,\,C_n/X_n]\)
and \([D_1/X_1,\,D_2/X_2,\,\ldots,\,D_n/X_n]\), respectively.
If\/ \(A \geqtyp B\) and
\(C_i \geqtyp D_i\) for every \(i\) (\(i = 1,\,2,\,\ldots,\,n\)),
then \(A[\vec{C}/\vec{X}] \geqtyp B[\vec{D}/\vec{X}]\).
\end{proposition}
\begin{proof}
By induction on the derivation of \(A \geqtyp B\),
and by cases on the last rule applied in the derivation.
Most cases are straightforward, where
in the cases of \r{{\eqtyp}\mbox{-reflex}} and \r{{\eqtyp}\mbox{-trans}},
use Lemma~\ref{appendix-eqtyp-subst-0-lemma}, which will be shown later.
If the last rule is \r{{\eqtyp}\mbox{-fix}}, then
\(A = \fix{Y}A'\) and \(B = A'[\fix{Y}A'/Y]\) for some \(Y\) and \(A'\).
We can assume that \(Y \not\in \{\,X_i\,\} \cup \FTV{C_i} \cup \FTV{D_i}\)
for every \(i\).
Then,
\begin{Eqnarray*}
\fix{Y}A'[\vec{C}/\vec{X}]
    &\geqtyp& A'[\vec{C}/\vec{X}][\fix{Y}A'[\vec{C}/\vec{X}]/Y]
			& (by \r{{\eqtyp}\mbox{-fix}}) \\
    &=& A'[\vec{C}/\vec{X},\,\fix{Y}A'[\vec{C}/\vec{X}]/Y]
			& (since \(Y \not\in \FTV{\vec{C}}\)) \\
    &\geqtyp& A'[\vec{D}/\vec{X},\,\fix{Y}A'[\vec{C}/\vec{X}]/Y]
		& (by Lemma~\ref{appendix-eqtyp-subst-0-lemma}) \\
    &=& A'[\vec{D}/\vec{X}][\fix{Y}A'[\vec{C}/\vec{X}]/Y]
			& (since \(Y \not\in \FTV{\vec{D}}\)).
\end{Eqnarray*}
Since \(A'[\vec{D}/\vec{X}]\) is also proper in \(Y\)
by Proposition~\ref{proper-subst1}, we now get
\(\fix{Y}A'[\vec{C}/\vec{X}] \geqtyp \fix{Y}A'[\vec{D}/\vec{X}]\)
by \r{{\eqtyp}\mbox{-uniq}}; and therefore,
\begin{Eqnarray*}
A[\vec{C}/\vec{X}] &=& \fix{Y}A'[\vec{C}/\vec{X}] \\
    &\geqtyp& \fix{Y}A'[\vec{D}/\vec{X}] &  \\
    &\geqtyp& A'[\vec{D}/\vec{X}][\fix{Y}A'[\vec{D}/\vec{X}]/Y]
			& (by \r{{\eqtyp}\mbox{-fix}}) \\
    &=& A'[\vec{D}/\vec{X},\,\fix{Y}A'[\vec{D}/\vec{X}]/Y]
			& (since \(Y \not\in \FTV{\vec{D}}\)) \\
    &=& A'[\fix{Y}A'/Y][\vec{D}/\vec{X}] \\
    &=& B[\vec{D}/\vec{X}].
\end{Eqnarray*}
If the last rule is \r{{\eqtyp}\mbox{-uniq}}, then
\(B = \fix{Y}B'\) for some \(Y\) and \(B'\) such that
\(A \geqtyp B'[A/Y]\) and \(B'\) is proper in \(Y\).
We can assume that \(Y \not\in \{\,X_i\,\} \cup \FTV{C_i} \cup \FTV{D_i}\)
for every \(i\).
Then,
\begin{Eqnarray*}
A[\vec{C}/\vec{X}]
    &\geqtyp& B'[A/Y][\vec{D}/\vec{X}]
			& (by ind. hyp.) \\
    &=& B'[\vec{D}/\vec{X},\,A[\vec{D}/\vec{X}]/Y] \\
    &=& B'[\vec{D}/\vec{X}][A[\vec{D}/\vec{X}]/Y]
			& (since \(Y \not\in \FTV{\vec{D}}\)) \\
    &\geqtyp& B'[\vec{D}/\vec{X}][A[\vec{C}/\vec{X}]/Y]
		& (by using Lemma~\ref{appendix-eqtyp-subst-0-lemma} twice).
\end{Eqnarray*}
Note that \(B'[\vec{D}/\vec{X}]\) is also proper in \(Y\)
by Proposition~\ref{proper-subst1}.
Hence, we now get
\(A[\vec{C}/\vec{X}] \geqtyp \fix{Y}B'[\vec{D}/\vec{X}]\)
by \r{{\eqtyp}\mbox{-uniq}}; and therefore,
\(A[\vec{C}/\vec{X}] \geqtyp B[\vec{D}/\vec{X}]\).
\qed\CHECKED{2014/04/23}
\end{proof}

\begin{lemma}\label{appendix-eqtyp-subst-0-lemma}
Let \(n\) be a non-negative integer,
\(X_1\), \(X_2\), \(\ldots\), \(X_n\) distinct type variables, and
\(A\), \(C_1\), \(C_2\), \(\ldots\), \(C_n\),
\(D_1\), \(D_2\), \(\ldots\), \(D_n\) type expressions.
If\/ \(C_i \geqtyp D_i\) for every \(i\) (\(i = 1,\,2,\,\ldots,\,n\)),
then \(A[C_1/X_1,\,C_2/X_2,\,\ldots,\,C_n/X_n] \geqtyp
	A[D_1/X_1,\,D_2/X_2,\,\ldots,\,D_n/X_n]\).
\end{lemma}
\begin{proof}
By induction on \(h(A)\), and by cases on the form of \(A\).
The only interesting case is when \(A = \fix{Y}A'\) for some \(Y\) and \(A'\).
We can assume that \(Y \not\in \{\,X_i\,\} \cup \FTV{C_i} \cup \FTV{D_i}\)
for every \(i\).
\begin{Eqnarray*}
A[\vec{C}/\vec{X}] &=& \fix{Y}A'[\vec{C}/\vec{X}]
		    & (since \(Y \not\in \{\,\vec{X}\,\}\cup\FTV{\vec{C}}\)) \\
    &\geqtyp& A'[\vec{C}/\vec{X}][\fix{Y}A'[\vec{C}/\vec{X}]/Y]
			& (by \r{{\eqtyp}\mbox{-fix}}) \\
    &=& A'[\vec{C}/\vec{X},\,\fix{Y}A'[\vec{C}/\vec{X}]/Y]
		    & (since \(Y \not\in \{\,\vec{X}\,\}\cup\FTV{\vec{C}}\)) \\
    &\geqtyp& A'[\vec{D}/\vec{X},\,\fix{Y}A'[\vec{C}/\vec{X}]/Y]
			& (by ind. hyp.) \\
    &=& A'[\vec{D}/\vec{X}][\fix{Y}A'[\vec{C}/\vec{X}]/Y]
			& (since \(Y \not\in \FTV{\vec{D}}\)) \\
    &=& A'[\vec{D}/\vec{X}][A[\vec{C}/\vec{X}]/Y]
		    & (since \(Y \not\in \{\,\vec{X}\,\}\cup\FTV{\vec{C}}\))
\end{Eqnarray*}
Since \(A'[\vec{D}/\vec{X}]\) is also proper in \(Y\)
by Proposition~\ref{proper-subst1}, we now get
\(A[\vec{C}/\vec{X}] \geqtyp \fix{Y}A'[\vec{D}/\vec{X}]\)
by \r{{\eqtyp}\mbox{-uniq}}; and therefore,
\(A[\vec{C}/\vec{X}] \geqtyp A[\vec{D}/\vec{X}]\).
\qed\CHECKED{2014/04/23}
\end{proof}

\Subsection{Finiteness of $\Comp(A)/{\eqtyp}$}\label{comp-finite-sec}

In this section, we show that
\(\Comp(A)/{\eqtyp}\) and \(\Comp(A)/{\peqtyp}\) are finite sets
for every type expression \(A\), which is parallel to the fact that
the set of all subtrees of a regular tree
is finite\cite{courcelle,amadio:cardelli}.

\begin{definition}\label{k-comp-def}
Let \(k\) be a non-negative integer.
We write \(A \component_k B\) if and only if
\(C[A/X] \eqtyp B\) for some \(C\) and \(X\) such that \(\Idp(C,\,X) \le k\).
\end{definition}
Note that \(A \component B\) if and only if
\(A \component_k B\) for some \(k\), by Proposition~\ref{depth-finite-etv}.

\begin{proposition}\label{k-comp-basic}\pushlabel
\begin{Enumerate}
\item \itemlabel{k-comp-reflex}
    \(A \component_k A\) for every \(k\).
\item \itemlabel{eqtyp-k-comp}
    If\/ \(A \eqtyp A' \component_k B' \eqtyp B\), then \(A \component_k B\).
\item \itemlabel{k-comp-monotonic}
    If\/ \(k \le l\), then \(A \component_k B\) implies \(A \component_l B\).
\item \itemlabel{k-comp-etv}
    If\/ \(A \component_k B\), then \(\ETV{A} \subseteq \ETV{B}\).
\item \itemlabel{k-comp-trans}
    If\/ \(A \component_k B\) and \(B \component_l C\),
    then \(A \component_{k+l} C\).
\end{Enumerate}
\end{proposition}
\begin{proof}
Straightforward from Definition~\ref{k-comp-def}.
Use Proposition~\ref{geqtyp-subst} for Item~\itemref{eqtyp-k-comp},
Propositions~\ref{depth-finite-etv}, \ref{petv-netv-nest} and
\ref{geqtyp-etv} for Item~\itemref{k-comp-etv},
and Propositions~\ref{depth-subst1} and \ref{depth-finite-etv}
for Item~\itemref{k-comp-trans}.
\qed\CHECKED{2014/07/13}
\end{proof}

\begin{proposition}\label{k-comp-congr}\pushlabel
Suppose that \(A \component_k B\).
\begin{Enumerate}
\item \itemlabel{k-comp-var-congr}
    If \(\Canon{B} = \t\), then \(\Canon{A} = \t\).
\item \itemlabel{k-comp-O-congr}
    If \(\Canon{B} = \O^n X\),
    then \(\Canon{A} = \O^m X\) for some \(m \le n\).
\item \itemlabel{k-comp-Impl-congr}
    If \(\Canon{B} = \O^n (C \Impl D)\), then
    \begin{enumerate}[{\kern8pt}(a)]
    \item \(\Canon{A} = \O^m (C' \Impl D')\) for some \(m\), \(C'\) and \(D'\)
	such that \(m \le n\), \(C' \eqtyp C\) and \(D' \eqtyp D\), or
    \item \(0 < k\), and
       either \(A \component_{k-1} C\) or \(A \component_{k-1} D\).
    \end{enumerate}
\end{Enumerate}
\end{proposition}
\begin{proof}
By Definition~\ref{k-comp-def},
there exist some \(E\) and \(Y\) such that
\begin{eqnarray}
\label{k-comp-congr-01}
    && B \eqtyp E[A/Y],~\mbox{and} \\
\label{k-comp-congr-02}
    && \Idp(E,\,Y) \le k.
\end{eqnarray}
Note that \(E\) is not a {\tvariant} from (\ref{k-comp-congr-02})
by Propositions~\ref{tvariant-etv} and \ref{depth-finite-etv}.
We can assume that \(E\) is canonical by Propositions~\ref{geqtyp-canon},
\ref{geqtyp-subst} and \ref{geqtyp-depth}.
For Item~\itemref{k-comp-var-congr}, if \(\Canon{B} = \t\), then
\(\t \eqtyp B \eqtyp E[A/X]\) by Proposition~\ref{geqtyp-canon}; and
hence, \(A \eqtyp \t\) by Proposition~\ref{tvariant-subst2}
and Theorem~\ref{geqtyp-t-tvariant}.
Thus, we get \(\Canon{A} = \t\).
For Item~\itemref{k-comp-O-congr}, suppose that \(\Canon{B} = \O^n X\).
By Propositions~\ref{not-geqtyp-var-Impl} and \ref{not-geqtyp-var-t},
we get \(E = \O^l Y\) for some \(l \le n\)
from (\ref{k-comp-congr-01}); and therefore,
\(\O^n X \eqtyp E[A/Y] = \O^l A\)
by Proposition~\ref{geqtyp-canon}, from which we get
\(A \eqtyp \O^{n-l} X\) by Proposition~\ref{geqtyp-O-congr}.
Thus, we get \(\Canon{A} = \O^m X\) taking \(m\) as \(m = n-l\)
by Proposition~\ref{canon-congr}.
As for Item~\itemref{k-comp-Impl-congr}, suppose similarly that
\(\Canon{B} = \O^n (C \Impl D)\), where \(D \not\eqtyp \t\).
By Propositions~\ref{canon-congr}, \ref{geqtyp-O-congr} and
\ref{eqtyp-O-Impl}, we get either
\begin{eqnarray}
\label{k-comp-congr-03}
    && E = \O^l Y~\mbox{for some}~l \le n,~\mbox{or} \\
\label{k-comp-congr-04}
    && E = \O^n (C' \Impl D')~\mbox{for some}~C'~\mbox{and}~D' \not\eqtyp \t
\end{eqnarray}
from (\ref{k-comp-congr-01}).
If (\ref{k-comp-congr-03}) is the case, then
\(\O^n (C \Impl D) \eqtyp B \eqtyp E[A/Y] = \O^l A\)
from (\ref{k-comp-congr-01}); and hence,
\(A \eqtyp \O^{n-l} (C \Impl D)\) by Proposition~\ref{geqtyp-O-congr}.
Therefore, (a) holds by Proposition~\ref{canon-congr}
considering \(n-l\) as \(m\).
On the other hand, in case of (\ref{k-comp-congr-04}),
we get
\(\O^n (C \Impl D) \eqtyp B \eqtyp E[A/Y] = \O^n (C'[A/Y] \Impl D'[A/Y])\)
from (\ref{k-comp-congr-01}); and hence,
\[
    C \eqtyp C'[A/Y]~\mbox{and}~D \eqtyp D'[A/Y]
\]
by Propositions~\ref{geqtyp-O-congr} and \ref{geqtyp-Impl-congr}.
Therefore, (b) holds,
since (\ref{k-comp-congr-02}) and (\ref{k-comp-congr-04}) imply
\(\Idp(C',\,Y) \le k-1\) or \(\Idp(D',\,Y) \le k-1\)
by Definition~\ref{depth-def}.
\qed\CHECKED{2014/07/15}
\end{proof}

\begin{proposition}\label{comp-congr}
Suppose that \(A \component B\).
\begin{Enumerate}
\item If \(\Canon{B} = \t\), then \(\Canon{A} = \t\).
\item If \(\Canon{B} = \O^n X\),
    then \(\Canon{A} = \O^m X\) for some \(m \le n\).
\item If \(\Canon{B} = \O^n (C \Impl D)\), then
    \begin{enumerate}[{\kern8pt}(a)]
    \item \(\Canon{A} = \O^m (C' \Impl D')\) for some \(m\), \(C'\) and \(D'\)
	such that \(m \le n\), \(C' \eqtyp C\) and \(D' \eqtyp D\),
    \item \(A \component C\), or
    \item \(A \component D\).
    \end{enumerate}
\end{Enumerate}
\end{proposition}
\begin{proof}
Obvious from Proposition~\ref{k-comp-congr}
\qed\CHECKED{2014/07/13}
\end{proof}

\begin{proposition}\label{k-comp-subst}
If\/ \(A \component_k B[C/X]\), then either
\begin{enumerate}[{\kern8pt}(a)]
\item there exists some \(A'\) such that \(A' \component_k B\) and
    \(A \eqtyp A'[C/X]\), or
\item \(\Idp(B,\,X) \le k\) and \(A \component_{k-\Idp(B,\,X)} C\).
\end{enumerate}
\end{proposition}
\begin{proof}
We can assume that \(A\), \(B\) and \(C\) are canonical
by Propositions~\ref{eqtyp-k-comp}, \ref{geqtyp-canon},
\ref{geqtyp-depth} and \ref{geqtyp-subst}.
Suppose that \(A \component_k B[C/X]\).
If \(X \not\in \ETV{B}\), then \(B[C/X] \eqtyp B\) and
\(A[C/X] \eqtyp A\)
by Propositions~\ref{geqtyp-no-etv-subst}, \ref{geqtyp-etv}
and \ref{k-comp-etv}.
Hence, (a) holds considering \(A\) as \(A'\) in such a case.
Therefore, we only consider the case that \(X \in \ETV{B}\).
The proof proceeds by induction on \(k\),
and by cases on the forms of \(B\) and \(C\).

\Case{\(B = \t\).}
In this case, \(A = \t\)
by Proposition~\ref{k-comp-congr} since \(B[C/X] = \t\);
and hence, (a) holds considering \(\t\) as \(A'\).

\Case{\(B = \O^n Y\) for some \(n\) and \(Y\).}
Note that \(X = Y\) by the assumption that \(X \in \ETV{B}\); and hence,
\(\Idp(B,\,X) = 0\) by Definition~\ref{depth-def}.
We now consider the form of \(C\).

\Subcase{\(C = \t\).}
In this subcase, \(B[C/X] = \O^n \t \eqtyp \t\).
Hence, (b) holds by Propositions~\ref{k-comp-reflex} and \ref{eqtyp-k-comp}
since \(A \eqtyp \t\) by Proposition~\ref{k-comp-congr}.

\Subcase{\(C = \O^m Z\) for some \(m\) and \(Z\).}
In this subcase, \(B[C/X] = \O^{n+m} Z\).
Hence, \(A = \O^l Z\) for some \(l\) such that \(l \le n+m\)
by Proposition~\ref{k-comp-congr}.
Therefore, if \(l \le m\), then
\(A \component_0 C\); and hence,
(b) holds, because
\(A \component_0 C\) implies
\(A \component_{k-0} C\) by Proposition~\ref{k-comp-monotonic}.
Otherwise, i.e., if \(l > m\), then let \(A' = \O^{l-m} X\).
Then, \(A' \component_0 B\) and
\(A'[C/X] = (\O^{l-m} X)[\O^m Z/X] = \O^l Z = A\).
Thus, we get (a), since
\(A' \component_0 B\) implies
\(A' \component_k B\) by Proposition~\ref{k-comp-monotonic}.

\Subcase{\(C = \O^m (D \Impl E)\) for some \(m\), \(D\) and \(E\) such that
    \(E \not\eqtyp \t\).}
Since \(B[C/X] = \O^{n+m} (D \Impl E)\),
by Proposition~\ref{k-comp-congr}, either
\begin{eqnarray}
    \label{comp-subst-01}
    && A = \O^l (D' \Impl E'),~
	    l \le n+m,~D' \eqtyp D~\,\mbox{and}~E' \eqtyp E
	    ~\;\mbox{for some}~l,~D'~\mbox{and}~E',~\mbox{or} \\
    \label{comp-subst-02}
    && 0 < k,~\mbox{and either}~
	A \component_{k-1} D~\mbox{or}~ A \component_{k-1} E.
\end{eqnarray}
If (\ref{comp-subst-01}) is the case, we get (a) or (b) similarly to
the previous subcase.
On the other hand,
in case of (\ref{comp-subst-02}),
either \(A \component_{k-1} D\) or \(A \component_{k-1} E\) implies
\(A \component_{k-0} C\) by Proposition~\ref{k-comp-trans}, since
\(D \component_1 C\) and \(E \component_1 C\).
Therefore, (b) holds in this case.

\Case{\(B = \O^n (D \Impl E)\) for some \(n\), \(D\) and \(E\) such that
    \(E \not\eqtyp \t\).}
Since \(B[C/X] = \O^n (D[C/X] \Impl E[C/X])\),
by Proposition~\ref{k-comp-congr},
\begin{eqnarray}
\label{comp-subst-03}
    && A = \O^m (D' \Impl E'),~
	    m \le n,~D' \eqtyp D[C/X]~\mbox{and}~E' \eqtyp E[C/X]
	    ~\;\mbox{for some}~m,~D' ~\mbox{and}~E', \\
\label{comp-subst-04}
    && 0 < k~\mbox{and}~
	A \component_{k-1} D[C/X],~\mbox{or} \\
\label{comp-subst-05}
    && 0 < k~\mbox{and}~A \component_{k-1} E[C/X].
\end{eqnarray}
In case of (\ref{comp-subst-03}),
we get (a) by Proposition~\ref{k-comp-monotonic}
taking \(A'\) as \(A' = \O^m(D \Impl E)\).
If (\ref{comp-subst-04}) is the case,
by induction hypothesis, we have either
\begin{enumerate}[{\kern8pt}(a{${}'$)}]
    \item there exists some \(A'\) such that
	\(A' \component_{k-1} D\) and \(A \eqtyp A'[C/X]\), or
    \item \(\Idp(D,\,X) \le k-1\) and \(A \component_{k-1-\Idp(D,\,X)} C\).
\end{enumerate}
If (a${}'$) holds,
we get \(A' \component_k B\) from \(A' \component_{k-1} D\)
by Proposition~\ref{k-comp-trans}, since \(D \component_1 B\).
Thus, (a${}'$) yields (a).
On the other hand, note that
\(\Idp(B,\,X) = \min(\Idp(D,\,X),\,\Idp(E,\,X)) + 1 \le \Idp(D,\,X) + 1\)
by Definition~\ref{depth-def}.
Hence, if (b${}'$) holds, we get
\(\Idp(B,\,X) \le k\) from \(\Idp(D,\,X) \le k-1\),
and also get \(A \component_{k-\Idp(B,\,X)} C\)
from \(A \component_{k-1-\Idp(D,\,X)} C\)
by Proposition~\ref{k-comp-monotonic}.
Thus, (b${}'$) yields (b).
The proof for the case (\ref{comp-subst-05}) is similar.
\qed\CHECKED{2014/07/15}
\end{proof}

\begin{proposition}\label{comp-subst}
If \(A \component B[C/X]\), then either
\begin{enumerate}[{\kern8pt}(a)]
\item there exists some \(A'\) such that \(A' \component B\) and
    \(A \eqtyp A'[C/X]\), or
\item \(A \component C\).
\end{enumerate}
\end{proposition}
\begin{proof}
Obvious from Proposition~\ref{k-comp-subst}
\qed\CHECKED{2014/07/13}
\end{proof}

\begin{lemma}\label{k-comp-fix}
If \(A \component_k \fix{X}B\), then
there exists some \(A'\) such that \(A' \component_k B\) and
    \(A \eqtyp A'[\fix{X}B/X]\).
\end{lemma}
\begin{proof}
Suppose that \(A \component_k \fix{X}B\).
The proof proceeds by induction on \(k\).
If \(\fix{X}B \eqtyp \t\), then it suffices to let \(A' = B\).
In fact, \(A' \component_k B\) by Proposition~\ref{k-comp-reflex}.
Furthermore,
\(A \component_k \fix{X}B\) implies \(A \eqtyp \t\)
by Propositions~\ref{eqtyp-k-comp}, \ref{k-comp-congr} and
\ref{geqtyp-canon}; and hence, we get
\(A \eqtyp \fix{X}B \eqtyp A'[\fix{X}B/X]\)
by \r{{\eqtyp}\mbox{-fix}}.
Therefore, we only consider the case that \(\fix{X}B \not\eqtyp \t\),
that is, \(\fix{X}B\) is not a {\tvariant}, below.
Since \(\fix{X}B \eqtyp B[\fix{X}B/X]\),
by Propositions~\ref{eqtyp-k-comp} and \ref{k-comp-subst},
\begin{enumerate}[{\kern8pt}(a)]
\item there exists some \(A'\) such that \(A' \component_k B\) and
    \(A \eqtyp A'[\fix{X}B/X]\), or
\item \(\Idp(B,\,X) \le k\) and \(A \component_{k-\Idp(B,\,X)} \fix{X}B\).
\end{enumerate}
The proof is completed by showing that (b) also implies (a).
Suppose that (b) is the case.
Since \(\fix{X}B\) is not a {\tvariant},
we get \(\Idp(B,\,X) > 0\)
by Propositions~\ref{NIdp-positive} and \ref{PIdp-positive}.
Hence, by induction hypothesis,
\(A' \component_{k-\Idp(B,\,X)} B\), which implies
\(A' \component_k B\) by Proposition~\ref{k-comp-monotonic}, and
\(A \eqtyp A'[\fix{X}B/X]\) for some \(A'\).
\qed\CHECKED{2014/07/15}
\end{proof}

\begin{proposition}\label{comp-finite}
\(\Comp(A)/{\eqtyp}\) and \(\Comp(A)/{\peqtyp}\) are finite sets.
\end{proposition}
\begin{proof}
It suffices to show that \(\Comp(A)/{\eqtyp}\) is finite,
since \(\eqtyp\) implies \(\peqtyp\).
The proof proceeds by induction on \(h(A)\), and by cases on the form of \(A\).

\Case{\(A = X\) for some \(X\).}
In this case, by Proposition~\ref{comp-congr},
\(B \component A\) if and only if \(B \eqtyp X\).
Therefore, \(\Comp(A)/{\eqtyp}\) is a singleton.

\Case{\(A = \O C\) for some \(C\).}
In this case, since \(\O C = (\O X)[C/X]\),
by Proposition~\ref{comp-subst},
\(B \component A\) implies
\begin{enumerate}[{\kern8pt}(a)]
\item there exists some \(B'\) such that \(B' \component \O X\) and
    \(B \eqtyp B'[C/X]\), or
\item \(B \component C\).
\end{enumerate}
Note that (a) implies \(B' \eqtyp \O X\) or \(B' \eqtyp X\)
by Proposition~\ref{comp-congr}.
Hence,
\(\Comp(A)/{\eqtyp}\) is finite
by Proposition~\ref{geqtyp-subst},
since so is \(\Comp(C)/{\eqtyp}\)
by induction hypothesis.

\Case{\(A = C \Impl D\) for some \(C\) and \(D\).}
Similarly, by Proposition~\ref{comp-congr},
\(B \component A\) implies either
(a) \(B \eqtyp A\), (b) \(B \component C\), or (c) \(B \component D\).
Hence, \(\Comp(A)/{\eqtyp}\) is finite since so are
\(\Comp(C)/{\eqtyp}\) and \(\Comp(D)/{\eqtyp}\)
by induction hypothesis.

\Case{\(A = \fix{X}C\) for some \(X\) and \(C\).}
In this case, by Lemma~\ref{k-comp-fix},
\(B \component A\) implies there exists some \(B'\) such that
\(B' \component C\) and \(B \eqtyp B'[A/X]\).
Hence, \(\Comp(A)/{\eqtyp}\) is finite since so is
\(\Comp(C)/{\eqtyp}\) by induction hypothesis.
\qed\CHECKED{2014/07/13}
\end{proof}

\end{appendices}

\ifdetail
\clearpage
\phantomsection
\addcontentsline{toc}{section}{Index}
\printindex
\fi 

\end{document}